\documentclass[superscriptaddress,amsmath,amssymb,longbibliography,secnumarabic]{revtex4-2}
\usepackage{array}[=2016-10-06]

\usepackage[margin=25mm]{geometry}
\usepackage[T1]{fontenc}
\usepackage{lmodern}
\usepackage{amsmath,amssymb,amsthm,mathtools}
\usepackage{microtype,booktabs,longtable,tikz}
\usetikzlibrary{arrows.meta,positioning}
\usepackage[colorlinks=true,linkcolor=blue,citecolor=blue,urlcolor=blue]{hyperref}
\usepackage{bookmark}
\usepackage{placeins}
\hypersetup{pdftitle={Root-sparsity scaling for Hamiltonian simulation, differential equations and linear-system solvers},pdfauthor={Dominic W. Berry, Abhishek S. Rawat, Sophia Simon and Guang Hao Low}}
\usepackage{xcolor}

\newtheorem{theorem}{Theorem}[section]
\newtheorem{lemma}[theorem]{Lemma}
\newtheorem{corollary}[theorem]{Corollary}
\newtheorem{proposition}[theorem]{Proposition}
\theoremstyle{remark}
\newtheorem{remark}[theorem]{Remark}
\renewcommand{\thetheorem}{\arabic{section}.\arabic{theorem}}
\renewcommand{\thesection}{\arabic{section}}
\numberwithin{equation}{section}
\newcommand{\norm}[1]{\left\lVert#1\right\rVert}
\newcommand{\ket}[1]{\lvert#1\rangle}
\newcommand{\bra}[1]{\langle#1\rvert}
\newcommand{\eps}{\varepsilon}
\newcommand{\ii}{\mathrm i}
\newcommand{\ip}{\ii}
\newcommand{\Reop}{\operatorname{Re}}
\newcommand{\Lip}{\operatorname{Lip}}

\newcommand{\E}{\mathbb E}
\newcommand{\trdist}{D_{\rm tr}}

\makeatletter
\renewcommand{\l@section}{\@dottedtocline{1}{0em}{3em}}
\renewcommand{\l@subsection}{\@dottedtocline{2}{1.5em}{3em}}
\renewcommand{\l@subsubsection}{\@dottedtocline{3}{3em}{3.8em}}
\makeatother

\newcommand{\google}{\affiliation{Google Quantum AI, Venice, CA 90291, United States}
}
\newcommand{\MQ}{\affiliation{School of Mathematical and Physical Sciences,
Macquarie University, Sydney, NSW 2109, Australia}}
\newcommand{\UQ}{\affiliation{School of Electrical Engineering and Computer Science, University of Queensland, Brisbane, QLD 4072, Australia}}

\newcommand{\UofT}{\affiliation{
Department of Computer Science, University of Toronto, Toronto, Ontario, Canada}}

\newcommand{\abs}[1]{\left\lvert#1\right\rvert}
\newcommand{\cO}{\mathcal O}
\newcommand{\ceil}[1]{\left\lceil#1\right\rceil}
\newcommand{\Ces}[1]{\mathcal C_{#1}}
\DeclareMathOperator{\diag}{diag}
\DeclareMathOperator{\spec}{spec}
\DeclareMathOperator{\dist}{dist}
\DeclareMathOperator{\arsinh}{arsinh}

\begin{document}
\title{Root-sparsity scaling for Hamiltonian simulation, differential equations and linear-system solvers}
\author{Dominic W. Berry}\email{dominic.berry@uq.edu.au}
\MQ\UQ
\author{Abhishek S. Rawat}\MQ
\author{Sophia Simon}\UofT\UQ
\author{Guang Hao Low}\google

\begin{abstract}
We give quantum algorithms for time-independent and time-dependent Hamiltonian simulation, linear differential equations and linear systems, all with complexity that has only square-root dependence on matrix sparsity.
Our algorithms retain optimal scaling in other relevant parameters such as time, condition number and precision. For several cases we prove matching lower bounds on the joint parameters and show that the additional gate overhead is only polylogarithmic in the inverse precision.
The common primitive connecting our algorithms is the Cayley transform, which provides a cheap unitary embedding of the sparse matrices. For Hamiltonian simulation and differential equations, we then use transducers to construct the desired evolution operators. For linear systems, we instead encode the solution in the kernel of a matrix and use the Cayley transform to construct a projector onto that kernel.
\end{abstract}
\maketitle

{\small\tableofcontents}

\section{Introduction}\label{sec:introduction}

Hamiltonian simulation and the solution of linear equations are basic
components of quantum algorithms for physical simulation and scientific
computation. Hamiltonian simulation provides the time evolution used in
studying quantum dynamics and estimating energies, while quantum
linear-system algorithms prepare states whose amplitudes encode the
solution of a linear system~\cite{Lloyd1996,BACS,HHL}. Related methods can prepare
solutions of differential equations, including equations describing
non-unitary dynamics~\cite{BerryODE,BCOW}. The complexity of these
algorithms depends on how the matrix is provided, as well as its norm,
the evolution time or condition number, and the required accuracy.
Improving these dependences is necessary to determine which problems
can benefit from quantum computation.

An important setting is sparse access, where oracles return matrix
entries and the locations of potentially nonzero entries in each row
and column. We denote the maximum length of these lists by $s$, meaning that every row and column has at most $s$ nonzero entries.
Efficient simulation in this setting was established by Aharonov and
Ta-Shma~\cite{AharonovTaShma}, with subsequent improvements using product
formulae and quantum walks~\cite{BACS,BerryChilds}. Later algorithms
obtained logarithmic or nearly logarithmic dependence on the inverse
error~\cite{BCK,LowChuang}. Quantum signal processing and qubitisation
give optimal simulation complexity when the cost is expressed in terms
of an appropriate block encoding~\cite{LowChuang,Qubitization}, and
quantum singular value transformation provides a general framework for
implementing functions of block-encoded matrices~\cite{QSVT}.

For time-dependent Hamiltonians, Chen, Gao, Wang, and Zhou used
transducers to obtain optimal query complexity for time-dependent Hamiltonian simulation~\cite{CGWZ}. Their result includes
the additive bound
$\mathcal{O}(\alpha_{\rm BE}t+\log(1/\eps))$, where
$\alpha_{\rm BE}$ is the block-encoding normalisation, $t$ is the
evolution duration and $\eps$ is the simulation error, with a sharper
precision dependence given below. Our general simulation framework
recovers this additive-logarithmic bound as a special case and extends
the construction to exploit sparse access directly.

There are many ways of block encoding, and in the sparse-access model standard constructions lead to a normalisation and corresponding cost linear in the sparsity $s$~\cite{BCK,Qubitization}.
In 2012, Berry and Childs \cite{BerryChilds} proposed methods to reduce the complexity scaling to $\sqrt{s}$ for Hamiltonian simulation, but were only able to prove this scaling (up to logarithmic factors) provided the Hamiltonian satisfied a particular condition on the behaviour of the norm when breaking it into parts.
Realistic examples satisfy this condition, but it is possible to adversarially create counterexamples that violate the condition, for which that analysis does not give $\sqrt{s}$ scaling.
Since then it has been an open problem how to prove the $\sqrt{s}$ complexity scaling in general.
In 2019, Low established square-root dependence for general sparse Hamiltonians up to subpolynomial factors by using an iterated interaction picture approach~\cite{Low}.
However, these subpolynomial factors can be large at practical problem sizes.

Our first result removes those subpolynomial factors from the leading
sparsity and time dependence. If $\eta$ bounds the Euclidean norm of
every row and column of a Hermitian matrix, and
$\tau=\eta t$, then evolution for time $t \ge 0$ can be simulated to error
$\eps$ using
\begin{equation}
 \mathcal{O}\!\left(\sqrt{s}\,\tau+
   \sqrt{s\tau\log(1/\eps)}+\log(1/\eps)\right)
 \label{eq:intro-hs}
\end{equation}
sparse-oracle queries. At constant precision and $\tau\ge1$, this is
$\mathcal{O}(\sqrt{s}\,\tau)$. The same bound holds for time-dependent
Hamiltonians with coherent time-indexed sparse access and uniform sparsity
and row/column norm bounds. For time-independent Hamiltonians, Theorem~\ref{hs:thm:optimal} improves the precision dependence to the optimal joint bound in Table~\ref{tab:query-bounds}.
We obtain a further refinement in the complexity when provided with a bound $\mu$ on the maximum magnitude of any entry. The precise result is given in Theorem~\ref{hs:thm:two-norm}, with a matching lower bound shown in Proposition~\ref{hsl:prop:two-norm}.

At constant precision, the simulation theorem also implements any
$N\times N$ unitary using $\mathcal{O}(\sqrt N)$ queries to its matrix
entries. This matches the search lower bound and answers the question
posed by Berry and Childs~\cite{BerryChilds}. Bravo-Prieto, Harrow, and
Kothari also obtain this optimal bound in their independent
work~\cite[Corollary~3]{BPHK26}. Our result is stated as Corollary~\ref{hs:cor:unitary}. In Corollary~\ref{hs:cor:unitary-optimal}
we improve the precision-scaling of the upper bound.

We also obtain root-sparsity scaling for dissipative linear differential
equations by using our root-sparsity Hamiltonian simulation construction in the Linear Combination of Hamiltonian Simulation (LCHS) framework~\cite{ALL, ACL, LS}. We consider
\begin{equation}
 \frac{d x(t')}{dt'}=A(t')x(t')+b(t'),\qquad
 -A(t')=D(t')+\ii H(t'),\quad D(t')\succeq0,
\end{equation}
with $D(t')$ and $H(t')$ Hermitian. Let $\eta_D$ and $\eta_H$ bound the row and column Euclidean norms of $D(t')$ and $H(t')$, respectively, and let $\bar g$ bound the ratio between the combined input and source norm and the final-solution norm. With
$\ell_g=\log(C\bar g/\eps)$ for some constant $C > 1$ and
$\Theta=t\eta_H+t\eta_D\ell_g$, the matrix-query complexity for preparing $x(t)/\norm{x(t)}$ within error $\eps$ and with constant success probability is 
\begin{equation}
 \mathcal{O}\!\left(\bar g\left[
   \sqrt{s}\,\Theta+\sqrt{s\Theta\ell_g}+\ell_g\right]\right).
 \label{eq:intro-ode}
\end{equation}
The baseline Hamiltonian simulation result is recovered when $b=0$ and $D=0$. We give lower bounds that capture the joint time, sparsity, and normalisation dependence for the time-dependent problem, as well as precision-dependent hardness for a constant positive semidefinite dissipative generator. To obtain root-sparsity scaling for more general non-dissipative linear differential equations, we build on history-state based quantum differential equation solvers~\cite{BerryODE, Krovi, BC}, see Theorem~\ref{ode:thm:general}.

For linear systems of the form $Ax=b$, the algorithm of Harrow, Hassidim, and
Lloyd~\cite{HHL} initiated a sequence of improvements in the condition
number and precision dependence~\cite{AmbainisVT,CKS,Costa,Dalzell}.
The optimal block-encoding complexity is
$\mathcal{O}(\kappa\log(1/\eps))$, but a standard sparse block encoding introduces a factor $s$. Low and Su~\cite{LowSu} obtained optimal dependence on
the solution overlap in the number of right-hand-side preparation
queries, while retaining nearly optimal matrix-query complexity.
Recent work of Gao, Ji, and Liu, given an estimate of the solution norm, simultaneously achieves
the optimal matrix- and right-hand-side query counts in the block-encoding
model~\cite{GJL26}. We give a sparse-oracle algorithm with matrix-query complexity
\begin{equation}
 \mathcal{O}\!\left(\sqrt{s}\,\kappa\log(1/\eps)\right)
 \label{eq:intro-qls}
\end{equation}
for $\norm A\le1$ and $\norm{A^{-1}}\le\kappa$, with at most the same order of right-hand-side preparation calls.
If a constant-factor estimate of $\norm{A^{-1}b}$ is supplied,
Corollary~\ref{qls:cor:separate} combines our Cayley walk with the
algorithm of Gao, Ji and Liu to retain this matrix-query bound while
using $\mathcal O(\kappa/\norm{A^{-1}b})$ preparation calls.
The matrix-query bound matches the joint lower bound in its stated parameter regime. The
$\Omega(\kappa\log(1/\eps))$ lower bound had previously been communicated
by Harrow and Kothari, and Costa, An, Babbush, and Berry later supplied
an explicit proof~\cite{CABB25}. The constant-error sparsity bound is due
to Mori \emph{et al.}~\cite{Mori}, while a joint sparse-access
lower bound was also proved by Patel~\cite{Patel}. An alternative algorithm for linear systems with the same complexity, together with a matching lower bound, was obtained
independently and concurrently to the work presented here by Bravo-Prieto, Harrow, and
Kothari~\cite{BPHK26}.

\begin{table}[tb]
\caption{Summary of the main matrix-query bounds. Here $s$ is the sparsity, $t$ is the evolution time, $\eta$ is a bound on the Euclidean norm of the rows and columns, $\mu$ is an upper bound on the maximum value of any entry, and $\eps$ is the error tolerance. We then define $\tau=\eta t$, $L=\log(1/\eps)$, $r=L/\log(e+L/\tau)$, $r_s=L/\log(e+L/(s\tau))$, $r_\mu=L/\log(e+\mu L/(\eta^2 t))$.
For the ODE results, $\bar g$ denotes an upper bound on the ratio of the norm of the initial state plus the norm of the source term to the norm of the final state. We then define $\ell_g=\log(C\bar g/\eps)$ for some constant $C>1$.
For linear systems, $\kappa$ is an upper bound on the condition number.
The complete access and output guarantees are given in
the cited results.}\label{tab:query-bounds}
\centering\small
\begin{tabular}{@{}p{.36\linewidth}p{.29\linewidth}p{.29\linewidth}@{}}
\toprule
Problem & Upper bound & Lower bound\\\midrule
Time-independent Hamiltonian\newline (only $\eta$ known) &
$\mathcal O ( \sqrt{s}\tau+\sqrt{s\tau r}+r )$\newline Theorem~\ref{hs:thm:optimal} &
$ \Omega (\sqrt{s}\tau+\sqrt{s\tau r}+r)$\newline Theorem~\ref{hsl:thm:main}\\
\addlinespace[8pt]
Time-independent Hamiltonian\newline ($\eta$ and $\mu$ known) &
$\mathcal O (\sqrt{s}\tau+\sqrt{s\mu|t|r_\mu}+r_\mu)$\newline Theorem~\ref{hs:thm:two-norm} &
$\Omega (\sqrt{s}\tau+\sqrt{s\mu|t|r_\mu}+r_\mu)$\newline Proposition~\ref{hsl:prop:two-norm} \\
\addlinespace[8pt]
Time-dependent Hamiltonian &
$\mathcal O (\sqrt{s}\tau+\sqrt{s\tau L}+r_s)$\newline Corollary~\ref{hs:cor:hybrid} &
$\Omega (\sqrt{s}\tau+\sqrt{s\tau r}+r)$\newline Theorem~\ref{hsl:thm:main}\\
\addlinespace[8pt]
Dissipative ODE &
$\mathcal O (\bar g\sqrt{s}\tau\ell_g)$\newline Theorems~\ref{ode:thm:main}, \ref{ode:thm:time} &
$\Omega (\bar g\sqrt{s}\tau+\sqrt{s\tau r}+r)$\newline Theorem~\ref{odel:thm:combined}\\
\addlinespace[8pt]
Linear system &
$\mathcal O (\sqrt{s}\kappa L)$\newline Theorem~\ref{qls:thm:main} &
$\Omega (\sqrt{s}\kappa L)$\newline Theorem~\ref{qlsl:thm:main}\\\bottomrule
\end{tabular}
\end{table}

Table~\ref{tab:query-bounds} summarises the principal sparse-oracle
query bounds. The ODE row uses a common norm bound on the generator for simplicity, while Eq.~\eqref{eq:intro-ode} gives the sharper separation of its Hamiltonian and dissipative parts.
The principal results concern sparse-oracle query complexity, with input
and source preparation calls counted separately where specified. We also give
gate-efficient implementations of time-independent Hamiltonian
simulation, including both improved-precision norm refinements
(Theorem~\ref{gate:thm:main}), and dissipative differential-equation
solving. Using the ordered-rotation factorisation of Chen, Gao, Ji, Li, Wang,
and Zhou~\cite{CGJLWZ}, these implementations retain the query bounds
with only polylogarithmic overhead in inverse precision under common access assumptions. These additional assumptions are stated separately from those of the query theorems.

The remainder of the paper is organised as follows.
Section~\ref{sec:methods} introduces the Cayley representation and discusses the
baseline Hamiltonian simulation construction through to its query guarantee, then
explains its uses for differential equations and linear systems.
Section~\ref{sec:overview} specifies the access model and error
conventions. Sections~\ref{part:hs} and~\ref{part:hsl} give the
Hamiltonian simulation refinements and lower bounds. Sections~\ref{part:ode}
and~\ref{part:odel} treat differential equations. Sections~\ref{part:qls}
and~\ref{part:qlsl} establish the linear-system algorithm and its lower
bound. Appendices~\ref{opt:sec:proof}, \ref{twonorm:sec:proof}
and~\ref{sec:conversion-proof} prove the precision, two-norm and general
conversion estimates. Appendix~\ref{gate:sec:proof} supplies the
improved-precision gate implementation. Appendix~\ref{app:formalisation}
describes the scope of the supplied Lean formalisation.

\FloatBarrier

\section{Cayley representation and algorithmic overview}\label{sec:methods}

The main idea behind our improved sparsity scalings is to process the
entries of the input matrix separately, while retaining the information
in its row and column Euclidean norms. We place each entry in a
two-dimensional Hermitian matrix and use its Cayley transform to obtain
a unitary rotation that costs only a constant number of sparse-oracle
queries. The remaining task is to turn these inexpensive entry rotations
into the operation we want on the input state. For Hamiltonian simulation,
we use a transducer: a unitary interaction that produces the desired
evolution while returning an additional, input-dependent component
unchanged. We then explain how to obtain a circuit without actually having to prepare
that additional component. For differential equations, we use the Hamiltonian simulation construction as a subroutine. For linear systems, we use the Cayley entry rotations in a reflected walk whose fixed vectors encode the solution. We first develop the Hamiltonian simulation construction and its baseline bound~\eqref{eq:intro-hs}, and then explain the other two applications.

\subsection{Entrywise unitary representation and norm bounds}
\label{ov:sec:representation}

Let $H=H^\dagger\in\mathbb C^{N\times N}$ have at most $s$
nonzero entries per row and column, with every row and column Euclidean
norm bounded by a known $\eta>0$. Our simulation task is to apply
$e^{-\ii tH}$ to an arbitrary input in $\mathbb C^N$. We shall call
this $N$-dimensional input space the system space, to distinguish it
from the larger spaces in which we process the matrix entries.

We access $H$ through an entry-value oracle and reversible row and
column locators. A value query writes a finite bit string encoding
$H_{ij}$; throughout the construction, $H_{ij}$ is the exact complex
value assigned to that string by the specified decoder. A row locator
takes a row index and a position in its list and returns the corresponding
column index; a column locator performs the analogous operation with
rows and columns interchanged. Both locators act in place, replacing
the list-position register by the located index; their supplied
inverses recover the position from that index. Each list contains all nonzero entries
of its row or column and is padded to length $s$ with distinct indices.
The row and column padding need not agree. Section~\ref{sec:overview}
specifies these operations as reversible maps and accounts separately
for any discrepancy between the represented matrix and an intended
matrix whose entries have been approximated.

We would like the cost to depend on the norms of these rows and columns,
rather than on a bound obtained by adding the magnitudes of all $s$
entries. To keep the entries separate, let $\mathcal E\cong\mathbb C^{N^2}$
be the space with basis $\ket{i,j}$, where $i,j\in\{0,\ldots,N-1\}$.
We call $(i,j)$ an edge because it connects a row label to a column
label; no graph structure beyond these pairs is required. We can enter
this space either from a row or from a column. The following maps
$L,R:\mathbb C^N\to\mathcal E$ describe the two choices, while the
diagonal operator $Y$ attaches the corresponding matrix value:

\begin{equation}\label{ov:eq:frames}
 \begin{aligned}
 L\ket i&:=\sum_{j\text{ in row list }i}\ket{i,j},&
 R\ket j&:=\sum_{i\text{ in column list }j}\ket{i,j},\\
 Y&:=\sum_{i,j}\frac{H_{ij}}\eta\ket{i,j}\bra{i,j}.
 \end{aligned}
\end{equation}
The maps $L$ and $R$ are unnormalised uniform sums over the relevant
lists. For example, $L\ket i$ has coefficient one on each of the $s$
edges listed in row $i$, and zero on the other edges. Different row
labels have disjoint supports, as do different column labels, so
$L^\dagger L=R^\dagger R=sI$. More significantly, a nonzero edge
$(i,j)$ occurs in both the list for row $i$ and the list for column $j$.
Entering through the column list, multiplying by its value, and returning
through the row list therefore gives $L^\dagger YR=H/\eta$.
An edge used only for padding contributes zero to this identity.

We use an orientation register to distinguish these two ways of
entering an edge. In the ordered orientation basis
$\ket0$ for a row and $\ket1$ for a column, set
\begin{equation}\label{hs:eq:JK}
 J:=\frac1{\sqrt{2s}}\binom L R,
 \qquad K_0:=\sqrt{s}\begin{pmatrix}0&Y\\Y^\dagger&0\end{pmatrix}.
\end{equation}
Here $J:\mathbb C^N\to\mathcal E\oplus\mathcal E$ is a
$2N^2\times N$ isometry: the factor $1/\sqrt{2s}$ makes
$J^\dagger J=I$. The Hermitian operator $K_0$ acts on the
$2N^2$-dimensional doubled edge space. For register descriptions we
identify this space with $\mathcal E\otimes\mathbb C^2$; the two
summands correspond to the two values of the orientation register.
For $\psi=\sum_{i=0}^{N-1}\psi_i\ket i$, the row orientation of
$J\psi$ carries amplitude $\psi_i/\sqrt{2s}$ on every edge listed
in row $i$, while the column orientation carries
$\psi_j/\sqrt{2s}$ on every edge listed in column $j$.
The operator $K_0$ changes the orientation within a fixed edge and
weights that change by the entry associated with the edge. It never
changes the edge label.

Let us now verify how the representation retains the two pieces of
information we need. If we embed with $J$, apply an operator $X$ on
the doubled edge space, and return with $J^\dagger$, the resulting
system operator is $J^\dagger XJ$, called the compression of $X$.
We will use the first two compressed powers of $K_0$, also called its
first and second compressed moments. The first recovers $H$ itself:
indeed, the matrix element

\begin{equation}
 \bra iL^\dagger YR\ket j=H_{ij}/\eta
\end{equation}
receives a contribution only from the edge where the two lists meet.
Inserting the factors in~\eqref{hs:eq:JK} thus gives
\[
 J^\dagger K_0J
 =\frac{L^\dagger YR+R^\dagger Y^\dagger L}{2\sqrt{s}}
 =\frac{H+H^\dagger}{2\eta\sqrt{s}}.
\]
For the second compressed power, we instead use
$K_0^2=s\,\operatorname{diag}(YY^\dagger,Y^\dagger Y)$.
Two orientation changes return us to the original orientation and
replace the complex entry by its squared magnitude. Compression
therefore sums the squared entries within a row or a column, giving

\begin{equation}\label{hs:eq:moments}
 \begin{aligned}
 J^\dagger K_0J&=\frac{H}{\eta\sqrt{s}},\\
 J^\dagger K_0^2J
 &=\frac1{2\eta^2}\diag_i\!\left(
 \sum_j|H_{ij}|^2+\sum_j|H_{ji}|^2\right)\preceq I,
 &\norm{K_0}&\le\sqrt{s}.
 \end{aligned}
\end{equation}
The first identity identifies the scale $\eta\sqrt{s}$ at which
the represented operator recovers the desired Hamiltonian. The second
is the reason for retaining the entry space: for every unit input
$\psi$, it gives
$\norm{K_0J\psi}^2=\langle\psi,J^\dagger K_0^2J\psi\rangle\le1$.
Thus the embedded inputs see a bounded squared norm even though other
vectors in the larger space can see a norm as large as $\sqrt{s}$.
We will use these two identities for different purposes. The first
will determine the simulated system action; the second will bound the
extra amplitude needed to support that action and, ultimately, its
query cost.

We cannot obtain the simulation merely by evolving under $K_0$ and
compressing afterwards. The subspace containing embedded inputs is
$\operatorname{ran}J=\{J\psi:\psi\in\mathbb C^N\}$, with
orthogonal projector $JJ^\dagger$. This subspace need not be invariant
under $K_0$: one application can take amplitude out of it, and a later
application can bring that amplitude back. Inserting
$I=JJ^\dagger+(I-JJ^\dagger)$ between two copies of $K_0$ makes this
extra contribution explicit:

\begin{equation}\label{ov:eq:excursion}
 J^\dagger K_0^2J=(J^\dagger K_0J)^2+
 J^\dagger K_0(I-JJ^\dagger)K_0J.
\end{equation}
The second term describes precisely such an excursion out of the
embedded subspace and back. It need not vanish, so the compressed
second power need not equal the square of the compressed first power.
The same problem occurs in the power series for the exponential:
$J^\dagger e^{-\ii uK_0}J$ generally differs from
$e^{-\ii uJ^\dagger K_0J}$. On the other hand, using the full norm
bound $\sqrt{s}$ to normalise $K_0$ before compression would give
$J^\dagger(K_0/\sqrt{s})J=H/(s\eta)$, losing the desired scale
$\eta\sqrt{s}$. We therefore need a different way to connect the
inexpensive entry operation to a system evolution.

The first part of that connection is a unitary that preserves the
entrywise structure. For a known $0<\gamma\le1$, define
\begin{equation}\label{hs:eq:scaled}
 K_\gamma:=\gamma K_0,
 \qquad D_0:=(K_\gamma+\ii I)(K_\gamma-\ii I)^{-1}.
\end{equation}
The unitary $D_0$ is the Cayley transform of $K_\gamma$. It turns each real
eigenvalue $k$ into the complex number $(k+\ii)/(k-\ii)$ on the
unit circle, so $D_0$ is unitary. In particular, $k=0$ is sent to
$-1$, and no finite $k$ is sent to $1$. We leave the scale
$\gamma$ free because it controls a tradeoff: reducing $\gamma$
brings the phases closer to $-1$ and farther from $1$, which will
help suppress unwanted evolution later, but it also decreases
$J^\dagger K_\gamma J$ and increases the interaction needed for a
fixed physical time. We will choose $\gamma$ after deriving the
cost of the complete circuit.

Although the Cayley formula contains a matrix inverse, implementing
$D_0$ does not require a linear-system subroutine. Each fixed edge
has only two orientations, and we can compute its entire Cayley
unitary from the queried entry. More explicitly, for edge $(i,j)$
write $z:=\gamma\sqrt{s}H_{ij}/\eta$ and let
\begin{equation}
 K_{ij}:=\begin{pmatrix}0&z\\\overline z&0\end{pmatrix}.
\end{equation}
Since $K_{ij}^2=|z|^2I$, its inverse in the Cayley formula is
$(K_{ij}-\ii I)^{-1}=(K_{ij}+\ii I)/(1+|z|^2)$.
Multiplying by the other factor gives the unitary we actually apply:
\begin{equation}\label{hs:eq:edge}
 C_{ij}:=(K_{ij}+\ii I)(K_{ij}-\ii I)^{-1}
 =\frac{(|z|^2-1)I+2\ii K_{ij}}{1+|z|^2}.
\end{equation}
The displayed inverse explains the algebra; $C_{ij}$, rather than
that inverse alone, is the operation on the orientation register.

To see why the construction works coherently, let $w_{ij}$ be the
bit string encoding $H_{ij}$ and let $\ket\chi$ be any state of the
orientation register. With the edge label held fixed, the circuit acts as
\begin{equation}\label{ov:eq:entry-circuit}
 \begin{aligned}
 \ket{i,j}\ket\chi\ket0
 &\xrightarrow{\ O_{\rm val}\ }
       \ket{i,j}\ket\chi\ket{w_{ij}}\\
 &\xrightarrow{\ \text{encoding-controlled rotation}\ }
       \ket{i,j}(C_{ij}\ket\chi)\ket{w_{ij}}\\
 &\xrightarrow{\ O_{\rm val}^{\dagger}\ }
       \ket{i,j}(C_{ij}\ket\chi)\ket0.
 \end{aligned}
\end{equation}
The middle operation changes only the orientation. It changes neither
the indices used by the value oracle nor the stored encoding, so the
inverse query erases the work register on every edge. By linearity the
same sequence acts on superpositions of edges and implements
$D_0=\sum_{i,j}\ket{i,j}\bra{i,j}\otimes C_{ij}$.
At a zero entry $C_{ij}=-I$, including any unmatched padded edges.
This is a constant number of value queries; the remaining task is to
couple the edge unitary to the system so that its action uses the
compressed identities in Eq.~\eqref{hs:eq:moments}.

\subsection{Transducer coupling and chronological evolution}
\label{ov:sec:circuit}

The entry unitary is inexpensive, but it acts on the doubled edge
space rather than applying the desired evolution to a system input.
We now connect these spaces using the transducer framework of
Refs.~\cite{BJY,CGWZ}. We place the system in a public summand
$\mathcal H_{\rm pub}=\mathbb C^N$ and the doubled edge space in
a private summand $\mathcal H_{\rm priv}=\mathcal E\oplus\mathcal E$.
The words public and private describe the roles of these summands,
not restrictions on access. One can think of the private space as a
workbench: individual entries are easy to process there, and a suitable
coupling will allow that processing to produce the system operation.

Our design criterion is that the private component should return to
its initial value after the interaction. More precisely, we seek a
unitary $S$ on $\mathcal H_{\rm pub}\oplus\mathcal H_{\rm priv}$
and a linear map $\Gamma:\mathcal H_{\rm pub}\to\mathcal H_{\rm priv}$
such that

\begin{equation}\label{ov:eq:transducer}
 S(\psi\oplus\Gamma\psi)=U\psi\oplus\Gamma\psi
 \qquad\text{for every }\psi.
\end{equation}
The notation $\psi\oplus\Gamma\psi$ denotes a vector with public
component $\psi$ and private component $\Gamma\psi$; it is a direct
sum, not a tensor product. The vector need not be normalised.
Equation~\eqref{ov:eq:transducer} says that $S$ produces the desired
public action $U$ while returning the private component unchanged.
We call the latter component stationary for this reason. It is not
an independently prepared ancillary state: it depends on the unknown
input $\psi$, and at this stage we do not assume that we can prepare
it. We first use the stationary relation to identify an exact public
operation, and will then convert that relation into a circuit with
known initialisation.

To design the interaction, let us express $S$ in blocks relative to
the public and private summands. Substituting this block matrix into
\eqref{ov:eq:transducer} separates the stationary condition from the
public output:

\begin{equation}\label{ov:eq:balance-general}
 S=\begin{pmatrix}A&B\\C&D\end{pmatrix},
 \qquad C+D\Gamma=\Gamma,
 \qquad U=A+B\Gamma.
\end{equation}
Here $A$ maps the public space to itself, $C$ maps public amplitude
into the private space, $D$ acts within the private space, and $B$
returns private amplitude to the public space. The equation
$C+D\Gamma=\Gamma$ says that the injected amplitude exactly
balances the change in the private component. If $I-D$ is invertible,
we can solve this equation algebraically as
$\Gamma=(I-D)^{-1}C$ and obtain the public action
$U=A+B(I-D)^{-1}C$. This calculation tells us what an interaction
would do under its stationary condition; it does not require us to
implement an inverse as a quantum operation.

The Cayley representation is useful because the inverse needed for
this stationary calculation has a simple dependence on the represented
Hamiltonian. As no finite real eigenvalue $k$ is mapped to $1$,
$I-D_0$ is invertible. We introduce the resolvent
$R_0(z):=(zI-D_0)^{-1}$, defined for complex $z$ outside the spectrum
of $D_0$, to describe this inverse at $z=1$. Rearranging the Cayley
formula gives
\[
 I-D_0=-2\ii(K_\gamma-\ii I)^{-1},\qquad
 (I-D_0)^{-1}=\frac{I+\ii K_\gamma}{2}.
\]
Compressing this expression with $J$ uses only the first moment in
\eqref{hs:eq:moments}. In terms of the rescaled Hamiltonian $h$ and
its scale $\alpha$, we obtain

\begin{equation}\label{ov:eq:response}
 R_0(1)=\frac{I+\ii K_\gamma}{2},
 \qquad J^\dagger R_0(1)J=\frac{I+\ii h}{2},
 \qquad h:=H/\alpha,\quad \alpha:=\frac{\eta\sqrt{s}}\gamma.
\end{equation}
Thus a stationary excursion into the private space, followed by return
through $J^\dagger$, produces an expression linear in $H$. This is
the capability that direct evolution under $K_0$ did not provide:
the calculation uses $J^\dagger K_\gamma J$ without introducing
higher compressed powers and the unwanted excursions in
\eqref{ov:eq:excursion}. Our next task is to choose a unitary coupling
whose stationary equation uses this identity. Later, the norm of the
same resolvent will serve a different purpose, controlling how quickly
we can remove private amplitude from a circuit started at zero.

We cannot choose the four blocks in~\eqref{ov:eq:balance-general}
independently, because together they must form one unitary. Instead,
we build the coupling from an ordinary rotation between a public
vector and its embedded image under $J$. For real $c,\beta$ with
$c^2+\beta^2=1$ and $\beta>0$, define

\begin{equation}\label{ov:eq:coupling}
 G:=\begin{pmatrix}
 cI&-\beta J^\dagger\\
 \beta J&I-(1-c)JJ^\dagger
 \end{pmatrix},\qquad S:=G(I\oplus D_0).
\end{equation}
On the two-dimensional span generated by $\psi\oplus0$ and
$0\oplus J\psi$, the map $G$ has rotation matrix
$\bigl(\begin{smallmatrix}c&-\beta\\\beta&c\end{smallmatrix}\bigr)$.
On private directions orthogonal to $\operatorname{ran}J$, it is
the identity. These actions show directly that $G$ is unitary.
We apply the entry unitary first and the coupling second, as specified
by $S=G(I\oplus D_0)$. The free rotation parameter will determine
the strength of the resulting public step.

Let $y:=J^\dagger D_0\Gamma$ be the linear map giving the private
amplitude that returns in embedded-system coordinates. From the private
block row of~\eqref{ov:eq:coupling}, stationarity requires
$\Gamma=\beta J+[I-(1-c)JJ^\dagger]D_0\Gamma$.
Rearranging it and then compressing with the resolvent yields

\begin{equation}\label{ov:eq:balance}
 (I-D_0)\Gamma=J[\beta I-(1-c)y],
 \qquad
 y=\frac{-I+\ii h}{2}[\beta I-(1-c)y].
\end{equation}
For the second equality, we first solve the preceding equation for
$\Gamma$ using $R_0(1)$ and then apply $J^\dagger D_0$. The identity
$D_0R_0(1)=R_0(1)-I$ and~\eqref{ov:eq:response} give
$J^\dagger D_0R_0(1)J=(-I+\ii h)/2$.
We can now solve for the returning amplitude entirely on the system
space. Collecting the terms containing $y$ gives
\[
 [(1+c)I+\ii(1-c)h]y=\beta(-I+\ii h).
\]
The public output of the coupling is $cI-\beta y$. Substituting the
solution for $y$, using $\beta^2=1-c^2$, and collecting its numerator
gives

\begin{equation}\label{ov:eq:design-ratio}
 U=[(1+c)I-\ii(1-c)h]
     [(1+c)I+\ii(1-c)h]^{-1}.
\end{equation}
This rational expression lets us choose the coupling from the public
operation we want. For a physical time step $\Delta t>0$, put
$w:=\alpha\Delta t$ and require $w\le1/4$, as needed by the
later conversion estimate. Since $wh=\Delta tH$, a ratio
$(1-c)/(1+c)=w/2$ turns~\eqref{ov:eq:design-ratio} into a
Cayley approximation to $e^{-\ii\Delta tH}$. Solving this ratio
for $c$ and taking the positive $\beta=\sqrt{1-c^2}$ gives

\begin{equation}\label{ov:eq:rotation-parameters}
 c:=\frac{1-w/2}{1+w/2},\qquad
 \beta:=\frac{\sqrt{2w}}{1+w/2},
\end{equation}
which yields the exact relation~\eqref{ov:eq:transducer} with
\begin{equation}\label{ov:eq:step}
 U_C:=(I-\ii\Delta tH/2)(I+\ii\Delta tH/2)^{-1},
 \qquad
 \Gamma=\sqrt{2w}\,R_0(1)J(I+\ii wh/2)^{-1}.
\end{equation}
The first formula is the resulting public operation, and the second
is the stationary private map supporting it. If $H\psi=E\psi$,
then $U_C$ multiplies $\psi$ by
$(1-\ii\Delta tE/2)/(1+\ii\Delta tE/2)$, whose phase is
$-2\arctan(\Delta tE/2)$. Since
$\arctan x=x+\mathcal O(x^3)$ near zero, this agrees with the
physical phase $-\Delta tE$ up to cubic local error. We have therefore
constructed an exact transducer for a rational approximation to one
evolution step. We have not yet constructed the simulation circuit,
because its stationary private component remains input dependent.

For a positive total time $t$, we take $M$ steps of duration
$\Delta t=t/M$. Their public product $U_C^M$ approaches
$e^{-\ii tH}$ as the time mesh is refined; we will bound the accumulated
error below. We compose the stationary relations first and remove the
private components only after composition. Converting each fine step
separately would make us pay for the physical mesh, even though a
coherent implementation can process all time slots with the same
query sequence.

For two steps, assign a separate private summand to each step. The
first step receives public input $\psi$, whereas the second receives
$U_C\psi$, so their required stationary components are
$\Gamma\psi$ and $\Gamma U_C\psi$. The two interactions act as

\begin{equation}\label{ov:eq:two-time-slots}
 (\psi;\Gamma\psi,\Gamma U_C\psi)
 \longmapsto(U_C\psi;\Gamma\psi,\Gamma U_C\psi)
 \longmapsto(U_C^2\psi;\Gamma\psi,\Gamma U_C\psi).
\end{equation}
The semicolon separates the public component from the private
components of a direct-sum vector; it does not separate tensor factors.
The first interaction acts on the first private component and returns
it unchanged, while leaving the second private component untouched.
The second interaction acts on the second component and returns that
one unchanged, while leaving the first untouched. This distinction
explains why each step can use its own stationary relation without
interfering with the other. Continuing in chronological order gives
public action $U:=U_C^M$ and private map

\begin{equation}\label{ov:eq:total-catalyst}
 \Gamma_{\rm tot}\psi:=
 \bigoplus_{r=0}^{M-1}\Gamma U_C^r\psi.
\end{equation}

We next show why one call to this composed interaction need not make
$M$ separate sets of matrix queries. The key is to place every query
inside the private operation, leaving the public--private coupling
independent of the unknown entries. There are two coordinate choices
to reconcile: a uniform superposition over list positions is easy to
prepare, but the entry rotations use row and column indices.
Let $J_{\rm free}$ prepare the uniform list position and orientation,
and let $F$ apply the row or column locator according to that
orientation, with the corresponding routing of index registers.
Then

\begin{equation}
 J=FJ_{\rm free},\qquad
 J\longmapsto J_{\rm free},\qquad D_0\longmapsto F^\dagger D_0F
\end{equation}
describes a change of coordinates in the private space. In list
coordinates, the embedding is the known map $J_{\rm free}$.
To apply the private operation, we use $F$ to obtain edge coordinates,
apply the value-query, rotation, and inverse-query sequence, and then
apply $F^\dagger$ to return. The coupling now uses only a known
preparation and its inverse; the locator and value queries occur
inside $F^\dagger D_0F$. As this is a unitary change of private
coordinates, it preserves the public action and all the compressed
identities used to derive it.

This coordinate change uses two features of the access model. First, on every valid list input the locator
performs its specified basis-state permutation with coefficient $+1$.
This is the \emph{phase-one} convention: there is no
input-dependent phase that would change the amplitudes when a uniform
list superposition is converted into $J\psi$. Second, each locator is supplied as a unitary on its full register space, together with its inverse. This makes $F$ a unitary change of coordinates and $F^\dagger D_0F$ a unitary on the entire private space. The supplied locators also satisfy the full-unitary access assumption of Lemma~\ref{access:lem:control}. We work in list coordinates from now on, retaining $J,D_0,G$ for the transformed operators.

For each time-slot index $r\in\{0,\ldots,M-1\}$, let $G_r$
couple the public space to private slot $r$, and let $\widehat D_r$
apply $D_0$ to that slot while acting as the identity elsewhere.
If $j\ne r$, then $\widehat D_r$ commutes with $G_j$: the former
is the identity on both the public space and private slot $j$, which
are the only spaces on which $G_j$ acts. Moving these disjoint
operations past one another gives

\begin{equation}\label{ov:eq:chronological-circuit}
 S_{\rm tot}:=(G_{M-1}\widehat D_{M-1})\cdots
                 (G_0\widehat D_0)
 =G_{M-1}\cdots G_0
 \left(I_{\rm pub}\oplus\bigoplus_{r=0}^{M-1}D_0\right).
\end{equation}
The coupling rotations still appear in chronological order because
they share the public space and need not commute. The rightmost
operator, however, is a direct sum of identical private operations.
The slot label simply accompanies the same locator--value--rotation
sequence coherently; we do not run that sequence separately for each
label. Thus one call to $S_{\rm tot}$ costs a constant number of
sparse queries, independent of $M$.

We also need controlled calls and inverses for the later reuse circuit.
To control $S_{\rm tot}$, we control the known coupling rotations
and the private operation. Within the private sequence, only the middle
entry rotation needs the external enable control: when it is disabled,
the locator and value computations cancel against their inverses.
The known coupling rotations require no matrix queries, so controlling
the complete interaction preserves the constant sparse-query cost.
Reversing its entire sequence implements the adjoint at the same cost.
This statement concerns queries only. The product of known coupling
rotations can still contain many gates; the later gate-efficient
construction addresses that separate cost.

\subsection{Coherent reuse and circuit conversion}
\label{ov:sec:conversion}

We have now constructed a unitary interaction $S_{\rm tot}$ that
costs a constant number of sparse queries and has public action
$U:=U_C^M$, provided its input contains the stationary private
component $\Gamma_{\rm tot}\psi$. That condition is the remaining
obstacle: the component depends on the unknown input $\psi$, so we
cannot simply initialise it as an ordinary work register. We shall
instead start with zero private amplitude and allow amplitude injected
by one interaction to return during later interactions. Combining
these returns coherently will approximate the public action that the
stationary component would have produced.

For this subsection, $A,B,C,D$ denote the blocks of the complete
interaction $S_{\rm tot}$, rather than those of a single physical
time step. To see how the returns arise, let us introduce two labelled
copies of the public space sharing one private space. We prepare an
equal superposition of the public labels, with the same system input
carried in each component. The initial vector and the two successive
interactions are

\begin{equation}\label{ov:eq:two-reuses}
 \begin{aligned}
 (\psi/\sqrt2,\psi/\sqrt2;0)
 &\longmapsto(A\psi/\sqrt2,\psi/\sqrt2;C\psi/\sqrt2)\\
 &\longmapsto(A\psi/\sqrt2,(A+BC)\psi/\sqrt2;
                    (C+DC)\psi/\sqrt2).
 \end{aligned}
\end{equation}
The semicolon again separates components of a direct sum. The first
interaction acts on public label $0$ and the private space, injecting
$C\psi/\sqrt2$ while leaving public label $1$ unchanged. The second
acts on label $1$ and that same private space. It returns
$BC\psi/\sqrt2$ to label $1$, in addition to the direct contribution
$A\psi/\sqrt2$, and leaves label $0$ unchanged. Crucially, we never
reset the private component between these interactions.

We then undo the uniform-label preparation. In its initial label state,
the two public outputs add with another factor $1/\sqrt2$, giving
$P_2\psi=(A+BC/2)\psi$. The term $BC/2$ is the return missing from
a single interaction started with zero private amplitude. Other label
states and the remaining private amplitude are still present in the
unitary output; $P_2$ describes only the component selected by the
initial work-register state.

For $n$ labels, the same recurrence gives
\begin{equation}\label{ov:eq:reuse}
 P_n=U-Bg_n(D)\Gamma_{\rm tot},\qquad
 g_n(z):=\frac1n\sum_{j=0}^{n-1}z^j.
\end{equation}
Here $n$ is a positive integer, and the polynomial $g_n$ averages
the first $n$ powers of its argument. To derive the formula, let
$j\in\{0,\ldots,n-1\}$ index the interactions, and use the
stationary relation $C=(I-D)\Gamma_{\rm tot}$. Before interaction
$j$, the shared private component is the sum of all earlier injections:
\[
 \frac1{\sqrt n}\sum_{\ell=0}^{j-1}D^\ell C\psi
 =\frac1{\sqrt n}(I-D^j)\Gamma_{\rm tot}\psi.
\]
The sum is empty when $j=0$, as required by zero initialisation.
The selected public component therefore leaves that interaction as
\[
 \frac1{\sqrt n}\bigl[A+B(I-D^j)\Gamma_{\rm tot}\bigr]\psi
 =\frac1{\sqrt n}(U-BD^j\Gamma_{\rm tot})\psi,
\]
where we used $U=A+B\Gamma_{\rm tot}$. Undoing the uniform-label
preparation averages these $n$ outputs and gives~\eqref{ov:eq:reuse}.
These labels distribute the amplitudes of one unknown input; they do
not require copies of that input.

We call this complete unitary $\mathcal R_n$ the length-$n$ reuse
circuit. It includes the direct-sum sector encoding, uniform-label
preparation, the $n$ interactions and unpreparation. The sector encoding
uses orthogonal register labels to distinguish the public copies from
the shared private space, so their direct-sum components can be held
in an ordinary register circuit. In tensor-register
notation its selected block is
\begin{equation}
 P_n=(I\otimes\bra{0_{\rm work}})\mathcal R_n
                       (I\otimes\ket{0_{\rm work}}).
\end{equation}
The other work-register outputs are retained in the complete circuit;
selecting this block is a description of its action, not a measurement
performed by the algorithm.

The error in~\eqref{ov:eq:reuse} involves $B$, which returns private
amplitude to the public space, whereas the stationary equation directly
describes injection through $C$. Unitarity links these two directions.
For any public vector $\psi$ and private vector $v$, apply the same
unitary $S_{\rm tot}$ to both entries of the inner product:

\begin{equation}
 \langle\psi\oplus\Gamma_{\rm tot}\psi,0\oplus v\rangle
 =\langle U\psi\oplus\Gamma_{\rm tot}\psi,
                 Bv\oplus Dv\rangle.
\end{equation}
Expanding the left side gives $\langle\psi,\Gamma_{\rm tot}^\dagger
v\rangle$, while the right side gives
$\langle\psi,(U^\dagger B+\Gamma_{\rm tot}^\dagger D)v\rangle$.
Since the equality holds for all $\psi$ and $v$, we obtain
$\Gamma_{\rm tot}^\dagger=U^\dagger B+
\Gamma_{\rm tot}^\dagger D$, or
$B=U\Gamma_{\rm tot}^\dagger(I-D)$. Substituting this expression
into~\eqref{ov:eq:reuse} and using
$(I-D)g_n(D)=(I-D^n)/n$ yields

\begin{equation}\label{ov:eq:reuse-error}
 U-P_n=\frac Un\Gamma_{\rm tot}^\dagger(I-D^n)\Gamma_{\rm tot},
 \qquad \norm{U-P_n}\le\frac{2W}{n},
\end{equation}
where $W$ is a chosen upper bound on $\norm{\Gamma_{\rm tot}}^2$;
equivalently, $\norm{\Gamma_{\rm tot}\psi}^2\le W$ for every unit
public input. The norm estimate follows from $\norm D\le1$, since
$D$ is a block of a unitary, and hence $\norm{I-D^n}\le2$.
We have therefore related an algebraic property of the transducer to
a circuit cost: to make this selected-block error at most $\delta$,
it suffices to take a positive integer $n\ge2W/\delta$.
The quantity $W$ itself is a squared-amplitude bound, not a query
count. Each of the $n$ interactions makes one call to the complete
transducer, whose constant sparse-query cost was established above.

This averaging construction removes the initialisation obstacle, but
its inverse dependence on the error would be costly at high precision.
We can improve it by applying a polynomial filter to the private
evolution: we seek a polynomial whose value is small on the unwanted
private motion, so that its contribution to the public error is small.
Let us first see the mechanism for a scalar phase $z$ on the unit
circle. For a positive integer $b$ and $z\ne1$, the geometric-sum
identity gives

\begin{equation}\label{ov:eq:scalar-filter}
 |g_b(z)|=\left|\frac{1-z^b}{b(1-z)}\right|
 \le\frac2{b|1-z|}.
\end{equation}
If $|1-z|\ge1/Q$ for some $Q\ge1$, choosing an integer $b\ge4Q$
makes the average at most $1/2$ in magnitude. Taking a positive integer
power $k$ then reduces it to at most $2^{-k}$. This suggests the
polynomial $f(z):=g_b(z)^k$: averaging first gives a fixed reduction,
and repeated multiplication makes that reduction exponential in $k$.

Two further steps are needed before this scalar observation gives an
algorithm. First, we need a bound for the actual operator $f(D)$.
The private block $D$ need not be normal, meaning that it need not
commute with its adjoint, so its eigenvalues alone need not control
$\norm{f(D)}$. Second, the available circuits implement the reuse
blocks $P_n$, not arbitrary powers of $D$. We address the operator
bound first and then express the desired filter through those circuits.

For the operator bound we use the resolvent, the inverse
$(zI-D_0)^{-1}$ already introduced when solving the stationary
equation. The chronological resolvent argument in
Appendix~\ref{sec:conversion-proof} starts from an integer $Q\ge1$
with $\norm{(I-D_0)^{-1}}\le Q$ for the entry unitary before coupling.
It then controls the private block of the complete chronological
interaction, including the coupling and all time slots. For every
$0<\delta<1$, it shows that, with $b=32Q$, we can choose $k$ so
that $f=g_b^k$ has degree $d\ge2W$ and

\begin{equation}\label{ov:eq:filter-guarantee}
 \norm{f(D)}\le\delta,\qquad
 d=\mathcal O\!\left(W+Q\log(2/\delta)\right).
\end{equation}
The first term in the degree bound accounts for the size of the
stationary component, and the second gives logarithmic dependence on
the desired suppression. We will evaluate $W$ and $Q$ from the sparse
representation after completing the circuit. At this point the
appendix result supplies the operator guarantee that the scalar phase
argument alone could not establish.

We now turn the polynomial into a combination of the reuse circuits.
Write $f(z)=\sum_{j=0}^d a_jz^j$. Each factor $g_b$ has equal,
nonnegative coefficients summing to one. Multiplying $k$ such factors
therefore convolves $k$ uniform coefficient lists: the resulting
$a_j$ are nonnegative, sum to one, and satisfy $a_j=a_{d-j}$ because
reversing each list leaves it unchanged.

These coefficients specify a useful signed combination of the averages
$g_n$ appearing in~\eqref{ov:eq:reuse}. We include one further average
$g_d$ because the identity $(I-D)g_d(D)=(I-D^d)/d$ will cancel
the factor $I-D$ in $B$. The required combination is

\begin{equation}\label{ov:eq:smoothing}
 f(z)g_d(z)=\sum_{n=1}^{2d}\lambda_n g_n(z),\qquad
 \lambda_n:=\frac nd(a_{n-d}-a_n),
\end{equation}
where $a_j=0$ outside $0\le j\le d$. We can see the coefficient
difference in this formula by multiplying both sides by $1-z$:
each $(1-z)g_n(z)$ equals $(1-z^n)/n$, while the left side becomes
$f(z)(1-z^d)/d$. Matching the coefficient of $z^n$ gives precisely
$\lambda_n/n=(a_{n-d}-a_n)/d$. The constant coefficients agree
because $\sum_n\lambda_n/n=a_0/d$.
Evaluating the identity at $z=1$ shows that $\sum_n\lambda_n=1$.
The symmetry and nonnegativity of the $a_j$ also give
$\sum_n|\lambda_n|=2-2a_0<2$. This bounded absolute weight is what
will keep the linear-combination implementation inexpensive.

Let $\widetilde P:=\sum_n\lambda_nP_n$. Because the weights sum
to one, the ideal term $U$ is retained. Its error, by
\eqref{ov:eq:reuse} and~\eqref{ov:eq:smoothing}, is
$Bf(D)g_d(D)\Gamma_{\rm tot}$. Substituting
$B=U\Gamma_{\rm tot}^\dagger(I-D)$ and commuting the polynomials
in $D$ gives

\begin{equation}\label{ov:eq:filtered-block}
 U-\widetilde P=\frac Ud\Gamma_{\rm tot}^\dagger
 f(D)(I-D^d)\Gamma_{\rm tot},\qquad
 \norm{U-\widetilde P}\le\frac{2W}{d}\norm{f(D)}\le\delta.
\end{equation}
Here the two-sided bound on $\Gamma_{\rm tot}$ contributes $W$,
contractivity gives $\norm{I-D^d}\le2$, and the choice $d\ge2W$
makes the remaining prefactor at most one. Thus the suppression
$\norm{f(D)}\le\delta$ has become a selected-block approximation
to the desired public unitary. The full coefficient argument and the
operator estimate are given in Appendix~\ref{sec:conversion-proof}.

To implement this signed sum, we use a linear combination of unitaries
(LCU). We prepare a register labelling the reuse length $n$, with
amplitudes proportional to $\sqrt{|\lambda_n|}$, apply the selected
unitary $\mathcal R_n$ together with the known sign of $\lambda_n$,
and undo the label preparation. Selecting both this label register
and the reuse work registers in their initial states produces the
weighted sum of the blocks $P_n$. We pad the absolute coefficient
weight to exactly two by adding cancelling contributions, so this
selected block is $\widetilde P/2$.

All lengths can be selected within one circuit of $2d$ controlled
transducer calls: at call $j$, a branch is active only if its selected
length has not yet been reached. Shorter branches remain idle for the
remaining calls. Thus the number of different lengths does not multiply
the maximum reuse length. One round of oblivious amplitude amplification,
using the reflections about the initial work-register state, would
raise a block exactly equal to $U/2$ to amplitude one. Applied here,
it uses three calls to the LCU circuit or its inverse and produces
the desired output with error at most $2\delta$ on the complete
state, including the residual work-register amplitudes. The calculation
in Appendix~\ref{sec:conversion-proof} establishes this stronger
statement; a selected-block estimate alone would not establish it.

We have now removed the unknown stationary component from the
initialisation: the circuit starts with known zero work registers and
uses at most $6d$ calls to the transducer. Its output approximates
$U_C^M\psi$ with the work registers returned close to zero. It remains
to choose the physical time mesh, evaluate $W$ and $Q$, and balance
the scale $\gamma$ to obtain the claimed sparse-query cost.

\subsection{Baseline simulation error and query complexity}
\label{ov:sec:cost}

We have reduced simulation to three quantitative choices: the time
mesh must make $U_C^M$ close to the physical evolution, the reuse
filter must make the realised circuit close to $U_C^M$, and the scale
$\gamma$ must balance their query cost. We first evaluate the two
bounds $W$ and $Q$ required by the filter. They measure different
properties: $W$ bounds the stationary private amplitude generated
from a system input, while $Q$ bounds the full resolvent of the
entry unitary.

For the amplitude bound, recall from~\eqref{ov:eq:step} that
$\Gamma=\sqrt{2w}\,R_0(1)J(I+\ii wh/2)^{-1}$, with
$w=\alpha t/M$ and $R_0(1)=(I+\ii K_\gamma)/2$.
Hermiticity of $K_\gamma$ cancels the cross terms in
$R_0(1)^\dagger R_0(1)$, and the second moment
in~\eqref{hs:eq:moments} then gives

\begin{equation}\label{ov:eq:compressed-response}
 R_0(1)^\dagger R_0(1)=\frac{I+K_\gamma^2}{4},
 \qquad
 J^\dagger R_0(1)^\dagger R_0(1)J
 \preceq\frac{1+\gamma^2}{4}I.
\end{equation}
Because $h$ is Hermitian, every singular value of
$(I+\ii wh/2)^{-1}$ is at most one. Inserting the displayed
compressed bound on either side of this inverse therefore yields
$\Gamma^\dagger\Gamma\preceq w(1+\gamma^2)I/2$.
To pass from one step to the whole evolution, use the orthogonality
of the private time slots in~\eqref{ov:eq:total-catalyst}:
\[
 \norm{\Gamma_{\rm tot}\psi}^2
 =\sum_{r=0}^{M-1}\norm{\Gamma U_C^r\psi}^2
 \le \frac{Mw(1+\gamma^2)}2\norm\psi^2.
\]
Each $U_C^r$ is unitary, so the same local bound applies to every
term. Since $Mw=\alpha t$, we may choose

\begin{equation}\label{ov:eq:W}
 W:=\frac{\alpha t(1+\gamma^2)}2
   =\frac{\eta t\sqrt{s}}2(\gamma+\gamma^{-1}).
\end{equation}
The cancellation of $M$ is the important point: refining the mesh
reduces the squared stationary amplitude per step in exactly the
proportion by which it increases the number of private slots.
The total bound consequently remains fixed. This is where the row
and column Euclidean-norm promise enters the cost. It bounds the
compressed squared response of embedded inputs, even when the full
edge operator has much larger norm.

The time mesh has a different job: it controls the error in replacing
physical evolution by rational steps. For an eigenvalue $E$ of $H$,
apply $|y-\arctan y|\le|y|^3/3$ with $y=tE/(2M)$.
The phase difference for one step is at most
$t^3|E|^3/(12M^3)$. Summing over $M$ steps and using
$|e^{\ii a}-e^{\ii b}|\le|a-b|$ gives

\begin{equation}\label{ov:eq:discretisation}
 \norm{U_C^M-e^{-\ii tH}}
 \le\frac{t^3\norm H^3}{12M^2}.
\end{equation}
We choose a finite positive integer $M$ large enough that this error
is at most $\eps/4$ and $\alpha t/M\le1/4$.
For example, the sufficient requirements are
$M^2\ge t^3\norm H^3/(3\eps)$ and $M\ge4\alpha t$;
a supplied upper bound on $\norm H$ can be used in this choice.
The direct sum in~\eqref{ov:eq:chronological-circuit} ensures that
refining $M$ does not increase the query cost of one transducer call,
although it can increase the number of known gates.

We next bound the full resolvent, which controls the high-precision
filter rather than the stationary amplitude. Since $K_\gamma$ is
Hermitian, the eigenvalues of $R_0(1)^\dagger R_0(1)$ are
$(1+k^2)/4$, where $k$ runs over the eigenvalues of $K_\gamma$.
Using $\norm{K_\gamma}\le\gamma\sqrt{s}$, we can take the
integer bound

\begin{equation}\label{ov:eq:Q}
 \norm{R_0(1)}=\frac{\sqrt{1+\norm{K_\gamma}^2}}2\le Q,
 \qquad Q:=\left\lceil\max\{1,\gamma\sqrt{s}\}\right\rceil.
\end{equation}
For a unitary $D_0$, the norm of $(I-D_0)^{-1}$ is the reciprocal
of the smallest distance between $1$ and an eigenvalue of $D_0$.
Thus $Q$ controls how close the entry-unitary phases can approach
$1$. The chronological resolvent proof then carries this information
through the coupling and all time slots to establish
\eqref{ov:eq:filter-guarantee}. A bound on $W$ alone would establish
the averaging estimate~\eqref{ov:eq:reuse-error}, but would not give
this logarithmic precision dependence.

We allocate the errors before optimising the cost. Set
$\delta=\eps/8$ in the filter, so the amplified circuit has
complete-state error at most $2\delta=\eps/4$ relative to $U_C^M$.
Adding the mesh error gives at most $\eps/2$, leaving $\eps/2$
for approximating the known rotations and preparations by finite
gate sequences when necessary.

We can now choose $\gamma$. Substituting $W$ and $Q$ into the
filter degree and using the constant query cost of every transducer
call gives the following number of primitive sparse queries:

\begin{equation}\label{ov:eq:tradeoff}
 \mathcal O\!\left(
 \frac{\eta|t|\sqrt{s}}2(\gamma+\gamma^{-1})
 +\left\lceil\max\{1,\gamma\sqrt{s}\}\right\rceil
     \log(1/\eps)\right).
\end{equation}
For $0<\gamma\le1$, decreasing $\gamma$ increases the amplitude
term through $\gamma^{-1}$, but reduces the precision term until
$\gamma\sqrt{s}$ reaches its constant floor. In the range
$s^{-1/2}\le\gamma\le1$, the competing contributions are
$\eta|t|\sqrt{s}/\gamma$ and
$\gamma\sqrt{s}\log(1/\eps)$. Equating them suggests
$\gamma^2=\eta|t|/\log(1/\eps)$. Clipping this value to that
range gives, for $t\ne0$,

\begin{equation}\label{ov:eq:gamma}
 \gamma:=\min\!\left\{1,
 \max\!\left\{s^{-1/2},
 \sqrt{\frac{\eta|t|}{\log(1/\eps)}}\right\}\right\}
\end{equation}
There are three regimes. If $\eta|t|\ge\log(1/\eps)$, we take
$\gamma=1$; then the precision contribution is dominated by
$\mathcal O(\sqrt{s}\eta|t|)$. If
$\log(1/\eps)/s\le\eta|t|\le\log(1/\eps)$, the balancing
value lies in the permitted interval and gives
$\mathcal O(\sqrt{s\eta|t|\log(1/\eps)})$. Finally, below that
interval we take $\gamma=s^{-1/2}$. The resolvent bound is then
constant and the amplitude contribution is at most of order
$\log(1/\eps)$, giving that same order for the total cost.
Combining these regimes yields the baseline theorem. Its complete-state
guarantee includes the work registers, so the simulation can be used
inside a larger coherent algorithm.

\begin{theorem}[Baseline sparse Hamiltonian simulation]\label{hs:thm:main}
Let $H\in\mathbb C^{N\times N}$ be Hermitian, with at most $s\ge1$
nonzero entries per row and column and their Euclidean norms at most
a supplied $\eta>0$. Assume the reversible entry-value oracle and
phase-one in-place row and column locators, with injective padded
lists, full unitary extensions and supplied inverses, described in
Sections~\ref{ov:sec:representation}, \ref{ov:sec:circuit}
and~\ref{hs:sec:control}.
For $t\in\mathbb R$ and
$0<\eps\le1/2$, a unitary circuit $\mathcal A$ satisfies
\begin{equation}\label{hs:eq:fullerror}
 \sup_{\norm\psi=1}
 \norm{\mathcal A(\psi\otimes\ket0)
       -(e^{-\ii tH}\psi)\otimes\ket0}\le\eps
\end{equation}
using
\begin{equation}\label{hs:eq:mainbound}
 \mathcal O\!\left(\sqrt{s}\,\eta|t|
 +\sqrt{s\eta|t|\log(1/\eps)}+\log(1/\eps)\right)
\end{equation}
primitive sparse-oracle queries. The error includes all work registers
and is unchanged by an untouched reference system. A supplied
spectral-norm bound $\norm H\le\Lambda$ permits $\eta=\Lambda$.
For $t=0$ or a known zero Hamiltonian, the identity circuit suffices.
\end{theorem}

\begin{proof}
The construction above supplies the ingredients of
Lemma~\ref{hs:lem:abstract}, whose proof in
Appendix~\ref{sec:conversion-proof} establishes the operator-filter
and complete-state conversion bounds; we now verify its assumptions. For positive time, use
$\alpha=\eta\sqrt{s}/\gamma$ and $h=H/\alpha$.
Equation~\eqref{ov:eq:response} gives the required first compressed
resolvent, and~\eqref{ov:eq:compressed-response} gives its squared
norm bound with the lemma's parameter $\nu=(1+\gamma^2)/2$.
The full resolvent bound is~\eqref{ov:eq:Q}.
Moreover,
\[
 \norm h=\norm{J^\dagger K_\gamma J}
 \le\norm{K_\gamma J}\le\gamma\le1,
\]
where the second moment proves the penultimate inequality. Since
$0<\gamma\le1$ and $Q\ge1$, we also have $0<\nu\le1\le Q$,
as required.

The locator change of coordinates makes $J$ a known preparation.
The controlled query--rotation--unquery construction provides the
private unitary, its adjoint, and their coherent controls with constant
sparse-query cost. The lemma therefore applies and gives
\eqref{ov:eq:tradeoff}, including the complete-state error bound.
Choosing $\gamma$ as in~\eqref{ov:eq:gamma} and applying the three
regimes just analysed gives~\eqref{hs:eq:mainbound}.
All estimates are operator-norm bounds, so tensoring the construction
with the identity on an untouched reference system preserves them.
A supplied $\norm H\le\Lambda$ bounds every row and column
Euclidean norm and hence permits $\eta=\Lambda$.
The identity circuit handles $t=0$ or a known zero matrix. For
negative time, take the adjoint of the positive-time circuit, which
preserves the error and query count and reverses the simulated
evolution.
\end{proof}

The theorem applies to the matrix represented exactly by the oracle
decoder. If the encoded values approximate an intended matrix, the
resulting change in evolution requires a separate perturbation budget.
Also, the query-free operations used above are independent of the
unknown entries; they need not have negligible gate cost. These
distinctions become important when we turn the query construction
into a gate-efficient implementation.

Before turning to the other applications, let us identify what this
construction makes reusable. We need an inexpensive unitary on a
larger space, a coupling that we can implement, and a stationary
equation whose elimination produces the desired public action.
We must then bound both the size of the stationary private component
and the private dynamics that allow a short filter from zero
initialisation. The first bound controls how much amplitude must be
recovered; the second controls how rapidly the residual error can be
suppressed.

The stationary relation also imposes a useful limitation. Unitarity
of $S$ in~\eqref{ov:eq:transducer} gives
$\norm\psi^2+\norm{\Gamma\psi}^2
=\norm{U\psi}^2+\norm{\Gamma\psi}^2$, so $U$ must preserve norms.
For a square public action, this means $U^\dagger U=I$.
An arbitrary nonunitary target therefore cannot be obtained merely
by substituting it for $U$ in that relation. We must represent it as
a selected block of a larger unitary, or use a success flag, and
account for the associated normalisation and success cost. That is
exactly what the differential-equation construction does next.

\subsection{Dissipative propagation and source assembly}
\label{ov:sec:applications}

We next use the simulation algorithm to solve differential
equations. The new difficulty is that their evolution need not be
unitary: the solution can lose norm, and an inhomogeneous term can
add a source at intermediate times. We account for both effects
by forming a linear combination of the unitary simulations already
constructed, then preparing the state represented by the resulting
selected output.
Consider the differential equation
\begin{equation}
    \dot x(t)=A(t)x(t)+b(t)
\end{equation}
on $\mathbb C^N$, with $x(0)=x_0$,
and suppose $-A(t)=D(t)+\ii H(t)$, where $D(t)$ and $H(t)$ are
Hermitian and $D(t)\succeq0$. Here $D(t)$ denotes the dissipative
part of the generator, rather than the private block of a transducer.
Its positivity means that the homogeneous evolution cannot increase
the squared norm, since
$d\norm{x(t)}^2/dt=-2\langle x(t),D(t)x(t)\rangle\le0$ when
$b(t)=0$. Our target is the normalised quantum state proportional
to the assumed nonzero vector $x(T)$.

Let $V(t,u)$ be the propagator taking an initial vector at time $u$
to its homogeneous solution at time $t$. Variation of constants then
separates the propagated initial condition from the source contributions:

\begin{equation}\label{ov:eq:source}
 x(T)=V(T,0)x_0+\int_0^T V(T,u)b(u)\,du.
\end{equation}
Thus we must prepare both an evolution starting at time $0$ and
evolutions of source vectors inserted at different times $u$.
For each such interval, we use the linear combination of Hamiltonian
simulations (LCHS) representation~\cite{ALL,ACL,LS}. It expresses
$V(T,u)$ as a weighted integral of unitary propagators generated by
the Hermitian matrices $H(t)+kD(t)$, with a real integration
parameter $k$. The role of the representation is to produce decay
through interference between unitary branches, without requiring a
unitary to implement a norm-decreasing map by itself.

For constant $H$ and $D$, each branch is a simulation of a fixed
Hamiltonian. No locator for $H+kD$ is supplied, so we build the branch
from the tagged entry representation of Lemma~\ref{ode:lem:tags} and
apply Lemma~\ref{hs:lem:abstract}; this gives the scaling of
Theorem~\ref{hs:thm:main} up to constant factors. When they vary with
time, the branch instead requires a time-ordered propagator, and we use
the chronological construction of Section~\ref{hs:sec:time} with the
same tagged representation. The oracles for $A(t)$ retain a time
address, so they can select the matrix for a sampled interval
coherently even when its support changes. Theorem~\ref{ode:thm:time}
assumes known bounds that control variation between samples, for
example an operator-norm Lipschitz bound, together with a known error bound for the represented
samples; those samples must also retain the required sparsity and
Euclidean-norm promises. These bounds determine a mesh and an error
allocation that approximate the continuous branch evolution.

The construction follows the private-slot mechanism already described:
we assign a sampled matrix to each time interval, construct its local
Cayley step, and compose the steps in chronological order before
removing the private components. The time address permits one coherent
query sequence to apply the appropriate private operation in every
slot. The chronological conversion theorem then gives the same baseline
query scaling without a multiplicative mesh-size factor; the sampled
matrices need not commute. Section~\ref{hs:sec:time} proves this
extension, and Theorem~\ref{ode:thm:time} specifies the corresponding
time-indexed access, regularity, and sampling guarantees for the
differential-equation solver.

After discretising the integral, we prepare a coefficient register,
select the corresponding simulation coherently, and undo the
coefficient preparation. The block selected by returning that register
to its initial state approximates $V(T,u)$ after the prescribed
coefficient normalisation. We assemble the source integral by a
second coherent selection over insertion times, preparing each
$b(u)$ with a weight proportional to its norm and the quadrature
weight. The quadrature comes with a bound on its discretisation
error. Combining the initial-condition branch and these source
branches gives the vector in~\eqref{ov:eq:source} in a selected
output component. The remaining components are part of the full
unitary circuit, rather than an assumed deterministic nonunitary gate.

The solution may be much smaller than the sum of the norms of these
contributions, because evolution decreases norms or because source
contributions cancel. This determines both the required branch
accuracy and the cost of selecting the desired output. Define

\begin{equation}
 B:=\norm{x_0}+\int_0^T\norm{b(u)}\,du,\qquad
 g:=\frac B{\norm{x(T)}},\qquad \bar g\ge g,
\end{equation}
where $\bar g$ is a supplied upper bound on $g$.
If every propagator has operator error at most $\xi$, the error
in the unnormalised assembled vector is at most $\xi B$: the initial
branch contributes at most $\xi\norm{x_0}$ and the source branches
at most $\xi\int_0^T\norm{b(u)}\,du$. We therefore choose
$\xi$ of order $\eps/\bar g$, giving vector error of order
$\eps B/\bar g\le\eps\norm{x(T)}$ before normalisation.
Quadrature and preparation errors are allocated within the same
relative error budget.

The LCHS construction attains this accuracy with frequencies
$|k|=\mathcal O(\log(\bar g/\eps))$ and with a constant bound
on the total absolute coefficient weight. If $\eta_H$ and $\eta_D$
bound the row and column Euclidean norms of $H(t)$ and $D(t)$
uniformly in time, the triangle inequality gives the branch bound
$\eta_H+|k|\eta_D$, and hence the common bound
$\eta_H+\mathcal O(\log(\bar g/\eps))\eta_D$.
Only the dissipative contribution acquires this frequency factor.

We choose common simulation parameters for all frequencies and source
intervals. A coherently selected direct sum has operator error equal
to the maximum branch error, so the number of quadrature points does
not multiply the matrix-query count. It can affect the cost of
preparing the coefficient registers, which is distinct from that
query count. After assembling the sums, the successful amplitude is
of order $\norm{x(T)}/B=1/g$. Amplitude amplification therefore
costs a factor of order $\bar g$.

To display the resulting cost, write
$\Theta:=T\eta_H+T\eta_D\log(10^4\bar g/\eps)$ for the
common time-norm scale of the selected simulations. Applying the
baseline simulation bound, with the time-dependent extension just
described when needed, at their required accuracy and including the
amplification factor gives

\begin{equation}\label{ov:eq:ode-bound}
 \begin{aligned}
 &\mathcal O\!\left(\bar g\left[
 \sqrt{s}\,\Theta+\sqrt{s\Theta\log(\bar g/\eps)}
 +\log(\bar g/\eps)\right]\right)
 \quad\text{matrix queries},\\
 &\mathcal O(\bar g)\quad\text{input/source calls}.
 \end{aligned}
\end{equation}
Conditional on heralded success with constant probability, the
prepared state has Euclidean error at most $\eps$, up to a global
phase. Theorems~\ref{ode:thm:main} and~\ref{ode:thm:time} specify
the input and source preparations, quadrature assumptions, and
complete guarantees.

There is also an access-model issue when we form the Hermitian
parts of $A(t)$. A locator for $A(t)$ does not automatically give
a locator for the union of the supports of $A(t)$ and
$A(t)^\dagger$. We avoid requiring such an oracle by retaining
two tagged copies of the original lists, so the tag records which
contribution is being used. Lemma~\ref{ode:lem:tags} shows how to
construct the Hermitian branches from that supplied access.

If $\eta_H,\eta_D\le\eta$, the bound becomes
$\mathcal O(\bar g[\sqrt{s}(\eta T+\sqrt{\eta T})+1]
\log(\bar g/\eps))$. If $D=0$, the underlying propagation
instead recovers the baseline Hamiltonian precision dependence.
These observations distinguish the two main effects in the
construction: coherent selection avoids paying separately for each
unitary branch, whereas amplification pays for the small norm of
their sum.

For a general constant generator, the Hermitian dissipative part
need not be positive, so this dissipative argument no longer applies
in the same form. We use a separate linear system that encodes a
Taylor approximation to the solution at successive times. Its
conditioning depends on the growth of the propagator, and extracting
the final-time state depends on the norm of the encoded trajectory.
Section~\ref{ode:sec:general} develops this construction and its bounds.

\subsection{A reflected Cayley walk for linear systems}

For linear systems, our task is to prepare
$A^{-1}b/\norm{A^{-1}b}$ for an invertible
$A\in\mathbb C^{N\times N}$ and a supplied preparation of the
normalised vector $b\in\mathbb C^N$. We again use the entry Cayley
unitary, but combine it with a reflection to recognise a subspace
containing the solution. The starting point is to append a column
to $A$ so that the desired solution becomes part of a vector in the
kernel:

\begin{equation}\label{ov:eq:kernel-augmentation}
 C_{\rm ker}:=[A,-b/\varrho]\in\mathbb C^{N\times(N+1)},
 \qquad
 C_{\rm ker}\binom{A^{-1}b/\varrho}{1}=0,
\end{equation}
Here $\varrho>0$ is a guess for $\norm{A^{-1}b}$, and
$\ker C_{\rm ker}=\{v\in\mathbb C^{N+1}:C_{\rm ker}v=0\}$
is one dimensional. To see the dimension claim, a vector
$\binom{x}{a}$ is in the kernel precisely when
$Ax=ab/\varrho$; invertibility then forces
$x=aA^{-1}b/\varrho$. The appended coordinate gives us an initial
state that is easy to prepare, while the other $N$ coordinates
contain the solution.

The scale $\varrho$ balances these two parts. If
$\varrho$ is within a constant factor of $\norm{A^{-1}b}$,
the vector $A^{-1}b/\varrho$ has norm of order one. Consequently,
the normalised kernel vector has constant overlap with the appended
coordinate and constant weight on the original coordinates. Projecting
the appended-coordinate state onto the kernel, and then selecting
the original coordinates, would therefore prepare the solution with
constant probability. We must still implement that projection and
find a suitable guess without already knowing the solution norm.
The actual algorithm first regularises a Hermitian dilation, an
enlarged Hermitian matrix encoding the action of $A$ and $A^\dagger$;
Section~\ref{qls:sec:augmentation} gives the matrices and dimensions.
This regularisation varies a known term on a further two-level register
to connect an easily inverted matrix to one encoding the original
inverse, while maintaining a controlled lower bound on the smallest
singular value along the path.
The possibly dense column containing $b$ is implemented using its
supplied preparation circuit, not a sparse locator for its entries.

To implement the projection, write $H_{\rm ker}\in\mathbb C^{m\times m}$ for the
Hermitian matrix whose kernel we want to select, where $m$ is its
dimension after augmentation and any register padding. When padding
introduces unused basis states, we include a known identity on those
states so that they do not create unwanted kernel vectors.
The augmented matrix includes the possibly dense source column,
so we do not apply the sparse-entry construction to the whole matrix.
Instead, we separate its sparse contribution, its known regularising
contribution, and its source contribution, together with the padding
identity when needed.

We assign these contributions orthogonal representation branches:
a branch label records which contribution is to be processed, and
its basis states distinguish the corresponding subspaces of the
larger representation. The sparse branch uses the entry construction
already developed. The regularising and padding branches have explicitly known
Cayley operations. On the source branch, the Hermitian operator
exchanges the appended coordinate with the prepared source state.
To implement its Cayley transform, we use the supplied preparation
circuit to identify these active vectors with known basis states,
perform the corresponding two-dimensional rotation, and undo the
change of basis; the Cayley transform acts as $-I$ on the complement.
Thus processing this branch requires only a constant number of
preparation calls, without querying the individual entries of $b$.

We prepare the branch label with known amplitudes and choose the
scale within each branch so that the compressed first moments add
to $H_{\rm ker}/\alpha$, while their second moments remain bounded.
Section~\ref{qls:sec:augmentation} gives these amplitudes, scales,
and moment bounds explicitly. Combining the branchwise Cayley
operations then gives an isometry $J$ from $\mathbb C^m$ into the
private space and a unitary $D_0$ on that space. One application
uses a constant number of matrix queries and source-preparation
calls. At the resulting positive scale $\alpha$, their compressed
resolvent has the same form as before:

\begin{equation}
 R_0:=(I-D_0)^{-1},\qquad
 J^\dagger R_0J=\frac{I+\ii H_{\rm ker}/\alpha}{2}.
\end{equation}
We form the unitary walk
$\mathcal W:=(I-2JJ^\dagger)D_0$. The reflection
$I-2JJ^\dagger$ reverses the sign on $\operatorname{ran}J$ and
fixes its orthogonal complement. Its purpose is to turn a vector
annihilated by $H_{\rm ker}$ into a vector fixed by the walk.
For a system vector $u$, the candidate vector on the larger space
is $R_0Ju$. Using $D_0R_0=R_0-I$ and the compressed resolvent
identity, we find

\begin{equation}\label{ov:eq:walk-design}
 \begin{aligned}
 (\mathcal W-I)R_0Ju
 &=Ju-2J(J^\dagger R_0J)u\\
 &=-\ii J(H_{\rm ker}/\alpha)u.
 \end{aligned}
\end{equation}
If $H_{\rm ker}u=0$, the right side vanishes, so $R_0Ju$ is
fixed by $\mathcal W$. Conversely, suppose $\mathcal Wv=v$.
Multiplying by the reflection gives
$D_0v=(I-2JJ^\dagger)v$, hence
$(I-D_0)v=2JJ^\dagger v$. Set $u:=2J^\dagger v$ and invert
$I-D_0$ to obtain $v=R_0Ju$. Inserting this vector into
\eqref{ov:eq:walk-design} forces $JH_{\rm ker}u=0$ and hence
$H_{\rm ker}u=0$, because $J$ is an isometry. The map
$u\mapsto R_0Ju$ is injective, as both $R_0$ and the embedding
are injective. It therefore gives a bijection between the matrix
kernel and the eigenvalue-one subspace of the walk.

This correspondence identifies the subspace we want, but an efficient
approximate projection also requires spectral separation. Write a walk
eigenvalue as $e^{\ii\theta}$, with $\theta\in(-\pi,\pi]$.
The desired fixed vectors have phase $\theta=0$; we need a lower
bound on $|\theta|$ for all other eigenvectors. The second compressed
moment controls the norm of the map $R_0J$ and, together with the
matrix spectral information, supplies this phase-gap bound.
Lemma~\ref{qls:lem:walk} proves it without assuming that
$H_{\rm ker}$ is invertible.

Under the normalisation $\norm A\le1$ and the supplied bound
$\norm{A^{-1}}\le\kappa$, the final regularisation scale gives
a phase gap $\Omega(1/(\kappa\sqrt{s}))$.
We select the fixed vectors using a Chebyshev polynomial filter
that equals one at phase zero and is small outside this gap.
Its implementation uses nonnegative combination weights, and its
degree gives
$\mathcal O(\kappa\sqrt{s}\log(1/\eps))$ walk steps to reach
the required accuracy. These weights multiply integer powers of
$\mathcal W$, with negative powers implemented using
$\mathcal W^\dagger$. Thus the reciprocal phase gap sets the
resolution cost, while the polynomial suppression gives logarithmic
precision dependence.

We must also find the balancing guess $\varrho$. The algorithm starts
with a regularised matrix for which the solution norm is known and
then halves the regularisation parameter at successive stages. Between
adjacent stages the norm changes by at most a factor of two, so a
constant number of nearby guesses suffices at each stage. Earlier
stages have larger gaps and are geometrically cheaper. By assigning
their confidence requirements according to these different costs,
we keep the sum of the tests within the final-stage order, without
an extra factor of $\log\kappa$. This norm-estimation argument is
essential to the total bound; the final filter degree alone does
not account for finding a good initial overlap.

Combining the norm-estimation stages with the final projection,
Theorem~\ref{qls:thm:main} gives
$\mathcal O(\kappa\sqrt{s}\log(1/\eps))$ matrix queries and at
most the same order of calls to the supplied preparation of $b$
and its inverse. For $0<\eps\le1/10$, it returns a success flag
with probability at least $2/3$, and the conditional output has
trace-distance error at most $\eps$. The Cayley representation is
the common ingredient in this solver and in Hamiltonian simulation,
but it supports different operations in the two algorithms: the
simulation balances and reuses a stationary private component,
whereas the linear-system solver projects onto the fixed subspace
of a reflected walk.

When a constant-factor estimate of $\norm{A^{-1}b}$ is supplied,
we can also use this walk construction to reduce the number of
preparation calls. We build the walk from $A$ alone, using its
Hermitian dilation when necessary, and block encode the enlarged
matrix $(I-\mathcal W)/2$. This gives a linear system whose matrix
operations are independent of $b$, allowing us to apply the
algorithm of Gao, Ji and Liu~\cite{GJL26}. The inverse-norm bound
and the solution norm of the enlarged system both increase by a
factor of order $\sqrt{s}$; their ratio therefore remains
$O(\kappa/\norm{A^{-1}b})$. Corollary~\ref{qls:cor:separate}
uses this cancellation to obtain that preparation-query bound
while retaining the matrix-query bound above. Its proof also
shows how to recover the original solution with constant success
probability and a constant number of further matrix queries.

\section{Access model and error guarantees}\label{sec:overview}

The algorithms above rely on two kinds of promises: what coherent
access to the matrix permits, and what approximation their output
must satisfy. We now state both precisely. This matters in particular
for the locator change of coordinates: it uses supplied inverses on
the full register space, not merely a way to enumerate nonzero
entries. It also matters for simulation accuracy: approximating a
selected block is an intermediate calculation, whereas the final
guarantee includes every work-register output.

All spaces are finite dimensional. We use the operator norm for
matrices, the Euclidean norm for vectors and $\norm{\cdot}_1$
for the trace norm. Logarithms are natural unless a base is displayed,
and asymptotic constants are absolute. The following conventions
apply to the matrix and state-preparation algorithms throughout
the paper unless a result explicitly specifies otherwise.

\subsection{Sparse data and coherent selection}\label{hs:sec:control}

For each pair of valid matrix indices $(i,j)$, the value oracle
stores a bit string $w_{ij}\in\{0,1\}^{b_{\rm val}}$. A known map
$\operatorname{val}:\{0,1\}^{b_{\rm val}}\to\mathbb C$ decodes it,
and the represented matrix has entry
$H_{ij}=\operatorname{val}(w_{ij})$. The same convention applies
to the other matrices in the paper. The oracle is supplied as a unitary on its full
register space together with its inverse: on a clean target it writes
the encoding, and its inverse erases it. For example, XOR-writing acts as
\begin{equation}\label{access:eq:value}
 \ket{i,j,z}\longmapsto\ket{i,j,z\oplus w_{ij}}.
\end{equation}
The encoding is a finite bit string; its decoded value is the
complex number on which the algorithm acts. The map
$\operatorname{val}$ may be any supplied known decoder of the
stated type, rather than specifically a fixed-point convention.
Thus the represented matrix is exact relative to that decoder.
If those values approximate entries of a different intended matrix,
the resulting matrix perturbation receives a separate error budget.
It is not absorbed into an unspecified precision of the quantum
algorithm. For an $N$-by-$N$ matrix the indices satisfy
$0\le i,j<N$, and $z$ in Eq.~\eqref{access:eq:value} is a
$b_{\rm val}$-bit target register.

The row locator replaces a list position by its column index:
\begin{equation}\label{access:eq:locator}
 \ket{i,a}\longmapsto\ket{i,\operatorname{row}(i,a)},
 \qquad 0\le i<N,\quad 0\le a<s.
\end{equation}
The column locator similarly replaces a list position by a row index
while retaining the column. Each list contains all nonzero entries
and has $s$ distinct indices after zero padding; the row and column
padding may differ. For valid listed indices, the supplied inverse recovers the list
position from the matrix index. In-place replacement is important
here. An out-of-place query could instead write
$\ket{i,a,0}\mapsto\ket{i,a,\operatorname{row}(i,a)}$.
Its inverse erases the last register when $a$ is still present,
but does not by itself recover $a$ when only the matrix index is
given. For an arbitrarily ordered list, we therefore do not identify
out-of-place access with our stronger in-place model unless inverse
location access is also supplied.

We use locators in coherent coordinate changes such as
$F^\dagger D_0F$. Their specified valid-input action therefore has
phase one, and each locator is supplied as a unitary on its full
register space together with its inverse. The extension outside the valid list inputs may be arbitrary; it
need not be a permutation. These requirements have separate roles.
Phase one ensures that applying a locator to a uniform list
superposition gives exactly the embedding amplitudes, without
unknown entry-dependent phases. Full unitary access and the supplied inverse make $F^\dagger D_0F$ a unitary on the entire private space, so the
coordinate change and the subsequent operator identities are defined on every private vector. Lemma~\ref{access:lem:control} also uses the supplied full-unitary locators and their inverses to implement controlled calls on arbitrary data states.

One primitive matrix query is a value or locator call, including a
supplied inverse call. Both locator directions are supplied for a
general matrix; Hermitian symmetry can be used to relate them.
An entry-only result uses no locators. A complete transducer or walk
step is a larger operation, and we explicitly reduce its cost to
these primitive queries.

Coherent selection of a direction or branch requires controls, but
we need not assume that an arbitrary unknown unitary can be controlled
for free. The following construction uses the specified basis action
of a locator to obtain the needed control.

\begin{lemma}[Controlled reversible access]\label{access:lem:control}
Let $F$ be a supplied unitary with inverse access, and suppose a known
valid basis input $x_0$ obeys $F\ket{x_0}=\ket{y_0}$ with phase one,
where the output basis label $y_0$ need not be known. A controlled
$F$ can be implemented using three calls to $F$ or $F^\dagger$,
with all added registers restored. A controlled XOR value query uses
two ordinary value queries.
\end{lemma}
\begin{proof}
We call the register on which we want to apply controlled $F$ the
data register. We need the inactive branch to undo the effect of an
unconditional query without applying $F$ to those data. The known valid input
$x_0$ supplies a way to do this. With one query, prepare a reference
register in $F\ket{x_0}=\ket{y_0}$. It is a basis state even though
we need not know its label classically. Also prepare a dummy register
in $\ket{x_0}$.

Use the enable bit to swap the data register with the dummy, apply $F$
unconditionally to the dummy, and undo that controlled swap. On a
data state $\phi$, this second oracle call produces the two updates
\begin{equation}\label{access:eq:control-map}
 \begin{array}{ll}
 \text{active:}&
 (\phi;x_0)\longmapsto(x_0;\phi)
 \longmapsto(x_0;F\phi)\longmapsto(F\phi;x_0),\\[2pt]
 \text{inactive:}&
 (\phi;x_0)\longmapsto(\phi;x_0)
 \longmapsto(\phi;y_0)\longmapsto(\phi;y_0).
 \end{array}
\end{equation}
Here the two entries denote tensor-product registers, namely the
data register and dummy; the unchanged reference is omitted. This differs
from the direct-sum grouping in the transducer examples. The active branch already has the desired data state $F\phi$ and a
clean dummy $x_0$. The inactive branch has left the data register alone,
but its dummy contains $y_0$. We can clean the latter using the
reference without measuring or classically learning $y_0$.
For a basis label $y$ in the reference register, let $T_y$ exchange
$x_0$ and $y$ in the dummy register and fix all other labels;
when $y=x_0$, take $T_y=I$. The controlled permutation
\begin{equation}
 \ket y\ket z\longmapsto\ket y\,T_y\ket z
\end{equation}
is a known reversible operation: comparisons of the two labels with
$x_0$ and with each other determine which transposition to apply.
It depends on a stored label, not on a classical description of $F$.
Apply it only on the inactive branch. It sends that dummy from
$y_0$ to $x_0$ and leaves the reference at $y_0$; the active
branch is unchanged. Finally, one call to $F^\dagger$ restores
the reference to $x_0$.

There have been three calls: preparation of the reference, the query
on the routed data or dummy, and unpreparation of the reference.
Both added registers end at $x_0$ on both branches. The construction
therefore acts coherently on the enable bit and on data states entangled
with other registers. The phase-one promise ensures that no relative
phase remains between the active and inactive branches. Reversing
the entire circuit controls $F^\dagger$ at the same cost; we do
not need to know a valid basis preimage for that inverse separately.
The data state may pass through any part of the supplied full unitary,
but the clean-up uses only $F\ket{x_0}=\ket{y_0}$ and its inverse.

For an XOR value query, we can obtain control with two calls.
Let $c\in\{0,1\}$ be the enable bit, $z$ the original value
target, and the last register clean. The sequence is
\begin{equation}
 \ket{c,i,j,z,0}
 \longmapsto\ket{c,i,j,z,w_{ij}}
 \longmapsto\ket{c,i,j,z\oplus(cw_{ij}),w_{ij}}
 \longmapsto\ket{c,i,j,z\oplus(cw_{ij}),0}.
\end{equation}
Here $cw_{ij}$ means the all-zero bit string when $c=0$ and
$w_{ij}$ when $c=1$. The first and last operations are ordinary
value queries; only the known bitwise XOR between them is controlled.
The queried indices and stored encoding are unchanged during that
XOR, so uncomputation works on every coherent branch.
\end{proof}

For a time-dependent matrix, the oracle additionally takes a
sample-time address and retains it while applying the corresponding
operation. If $q$ labels a time in a finite sampling set, the supplied
unitary has the coherent form
\begin{equation}
 \sum_q\ket q\bra q\otimes O_q,
\end{equation}
where $O_q$ is the value or locator operation at that sample.
A call on a superposition of addresses counts as one query. This
is an access promise: a collection of separate opaque oracles does
not automatically provide their coherent selection. The support lists
and padding may vary with time. When controlling this oracle by the
lemma above, the time address is part of the full queried register;
the reference preparation uses one known valid basis input including
its address.
We also require classically supplied bounds that let us choose a
finite sampling mesh with a guaranteed error, including the error
in the represented samples. We use the word \emph{certificate}
for such an a priori numerical error bound. It is information for
choosing the circuit parameters, not an extra quantum oracle or a
claim that the algorithm discovers the regularity of an unknown
function. A source-quadrature certificate has the same role for
approximating the source integral. The time-dependent theorems state
the uniform norm, regularity and preparation assumptions under which
these bounds are supplied.

\subsection{Resources and complete output states}

A known operation is independent of the unknown entries and uses
no matrix query. Arithmetic on an entry after it has been loaded
also uses no further matrix query. Both contribute to the gate cost
when we count gates. In those implementation results, we distinguish
arithmetic on public parameters and labels from arithmetic on loaded
matrix data and the work registers used for the Cayley rotation.
Initial and source states are supplied by preparation circuits, with
the inverse and control access specified locally. Their calls are
counted separately from matrix queries. A known circuit can be
controlled gate by gate; an opaque state-preparation oracle requires
its controlled form to be supplied when the algorithm needs it.

When a coherent branch selects a time, frequency or duration, we
choose register sizes, filter degree, query count and error budgets
that work for all branches. Shorter branches become inactive within
that common circuit. The resulting coherent direct sum has the
largest branch error and the common worst-case query count. It
does not inherit an average of separately optimised classical costs.
This is the convention used when a large frequency or time quadrature
is said not to multiply the number of queries.

For simulation, temporary registers start in a specified zero state
and should return to that state. Our notation
$\norm{\mathcal A-U}\le\eps$ abbreviates
\begin{equation}\label{access:eq:full-output}
 \sup_{\norm\psi=1}
 \norm{\mathcal A(\psi\otimes\ket0)
             -(U\psi)\otimes\ket0}\le\eps.
\end{equation}
The two maps being compared have the same initialised input
subspace and the same complete output space: one applies the actual
circuit, while the other applies $U$ and returns all work registers
to zero. Their difference therefore includes residual amplitude in
every work register. For a unitary circuit and a unitary target, a zero-work block approximation with operator-norm error at most $\delta$ already implies a complete-output error of at most $\sqrt{2\delta}$ by unitarity. Thus the block bound alone generally gives a weaker error guarantee. The amplification argument in Section~\ref{ov:sec:conversion} bounds the complete isometry difference directly. Tensoring that difference
with an identity preserves its norm, giving the same guarantee for
an input entangled with an untouched reference. Known changes of input or output
encoding can be absorbed into surrounding circuits; the isometries
in compressed-operator identities remain explicit.

For solution preparation, the algorithm instead returns a classical
success flag. Constant heralded success means a probability bounded
below by an absolute positive constant on every promised input.
The stated accuracy is conditional on success; failure outputs may
be arbitrary. For a possibly mixed conditional state $\rho$ and desired pure
state $\ket x$, trace-distance error at most $\eps$ means
$\tfrac12\norm{\rho-\ket x\bra x}_1\le\eps$.
For pure outputs we also use Euclidean error, minimised over their
global phase. The pure-state identity
\begin{equation}
 \tfrac12\norm{\ket\phi\bra\phi-\ket\psi\bra\psi}_1
 =\sqrt{1-|\langle\phi,\psi\rangle|^2}
 \le\min_\theta\norm{\phi-e^{\ii\theta}\psi}
\end{equation}
shows why a Euclidean error bound implies the same trace-distance
bound. It does not turn a conditional guarantee into an unconditional
one: the success probability remains a separate part of the result. Where
specified, resource bounds are worst-case bounds, including failed
attempts and classical random choices.

Throughout, $s$ bounds both row and column sparsity, while $\eta$
bounds their Euclidean norms and need not bound the operator norm.
A spectral-norm bound can always be used as a row and column
Euclidean-norm bound: each column norm is $\norm{Ae_j}\le\norm A$,
and each row norm is $\norm{A^\dagger e_i}\le\norm A$ for
standard basis vectors $e_i,e_j$. The converse need not hold, which
is why retaining the Euclidean-norm promise can improve the query
bound.
For an unnormalised linear system with supplied
$\norm A\le\Lambda$, the corresponding inverse-norm parameter is
$\Lambda\norm{A^{-1}}$; it agrees with the condition number up to
constants only when $\Lambda$ is comparably tight. The ODE ratio
$g=(\norm{x_0}+\int_0^T\norm{b(u)}\,du)/\norm{x(T)}$
measures solution normalisation and is distinct from a block-encoding
normalisation. We introduce precision logarithms locally where they
simplify a sustained calculation, retaining all dependence on other
parameters inside them. Numerical constants in explicit error budgets
and integer filter choices are kept as stated.

Our lower bounds permit the dimension to grow with the parameters;
none is asserted at an insufficient fixed dimension. Their precise
ranges and output models appear in the relevant sections. In
particular, the two-norm simulation match applies to the explicit
families of Proposition~\ref{hsl:prop:two-norm}, and the ODE results
do not establish a full multiplicative dissipative precision
logarithm. Appendix~\ref{app:formalisation} separately records which results are covered by the accompanying Lean source archive
and which require further hypotheses. In particular, the formal Dyson
comparison uses a narrower input class than Lemma~\ref{odel:lem:dyson}.

\section{Hamiltonian simulation beyond the baseline}\label{part:hs}

Our aim in this section is to extend the sparse Hamiltonian simulation
algorithm of Section~\ref{sec:methods} while retaining its use of the
row and column Euclidean norms. The baseline construction has already
separated two tasks: we first build a cheap unitary interaction whose
stationary public action is a product of Cayley steps, and then convert
that relation into a circuit starting with zero work registers. This
separation gives us several ways to improve or extend the algorithm.
For a time-dependent Hamiltonian, we keep the time steps in chronological
order and show that conversion depends on their total coupling strength,
rather than their number. For a fixed Hamiltonian, we can change the
representation and filter to improve the precision dependence as well.

We develop these extensions in that order. After treating time-dependent
Hamiltonians, we compare entrywise Cayley access with ordinary block
encoding; this comparison improves the additive precision term in both
the constant and time-dependent settings. We then explain the additional
cancellation that improves both precision terms for a fixed Hamiltonian,
and how a maximum-entry bound can sharpen that result. Finally, we
state the improved-precision gate guarantee, apply the simulation bounds
to implementing a unitary from its entries, and explain the compression
of the baseline circuit that also underlies the improved gate
implementation. Throughout, a query bound counts calls to the supplied
matrix oracles. A gate bound additionally requires us to implement the
known operations between those calls efficiently.

\subsection{Time-dependent Hamiltonian simulation}
\label{hs:sec:time}

We first consider $H(u)=H(u)^\dagger\in\mathbb C^{N\times N}$
that changes over physical time $u\in[0,T]$. The challenge is not simply
to approximate the path by constant matrices: a fine approximation may
require many intervals, and simulating each interval with its own
high-precision circuit would repeatedly pay the conversion cost. Instead,
we give each interval a separate private summand, compose the interactions,
and perform the conversion only once. Each interval may have a different
matrix and different nonzero positions. The public system nevertheless
visits the intervals in their prescribed order.

Two features of this construction need separate proofs. First, one use
of the combined interaction must have constant matrix-query cost even
when the mesh is fine. Second, the number of such uses needed for an
accurate, initialised circuit must remain controlled when the matrices
do not commute. Coherent time-addressed access gives the first feature;
the chronological resolvent estimate below gives the second.

To select the matrix for each interval within one query, we require
the oracles to retain a time address. For a known mesh
$0=t_0<\cdots<t_M=T$, let $\widetilde H_r$ be the fixed
Hermitian matrix represented for interval $r$. For example, its row
locator acts as
\begin{equation}\label{hs:eq:timeaccess}
 \ket{r,i,l}\longmapsto\ket{r,i,\operatorname{row}_r(i,l)},
 \qquad 0\le r<M,\quad0\le l<s.
\end{equation}
The column locator and entry-value oracle likewise retain $r$, so a
single query can act coherently on a superposition of sample times.
We assume the same reversible sparse access as in
Section~\ref{hs:sec:control} for each sample. The positions of the
nonzero entries may change between samples. The time address is part of
the input to a single supplied oracle: we do not assume that a collection
of unrelated opaque sample oracles can be selected coherently for free.

The matrix $\widetilde H_r$ encoded by the oracle may only be an
approximation to the true $H(t_r)$. The theorem below specifies the
required bound on $\norm{\widetilde H_r-H(t_r)}$.

\begin{theorem}[Time-dependent sparse simulation]\label{hs:thm:time}
Let $T\ge0$, $0<\eps\le1/2$, and let $H(u)=H(u)^\dagger$ on
$[0,T]$ have row and column sparsity at most $s\ge1$ and Euclidean
norms at most a supplied $\eta>0$. Suppose a known operator-norm
Lipschitz bound satisfies
\begin{equation}
 \norm{H(u)-H(v)}\le\beta|u-v|.
\end{equation}
On the mesh chosen in the proof, assume the coherent sampled-time
access above to Hermitian matrices $\widetilde H_r$ with the same
sparsity and norm promises and a known representation-error bound
\begin{equation}
 \norm{\widetilde H_r-H(t_r)}\le\delta_{\rm val},
 \qquad T\delta_{\rm val}\le\eps/12.
\end{equation}
A unitary circuit approximates the time-ordered propagator with
error
\begin{equation}\label{hs:eq:timeerror}
 \norm{\mathcal A-U_H(T,0)}\le\eps,
 \qquad
 U_H(T,0):=\mathcal T\exp\!\left(-\ii\int_0^T H(u)\,du\right),
\end{equation}
using
\begin{equation}\label{hs:eq:timebound}
 \mathcal O\!\left(\sqrt{s}\,\eta T+
       \sqrt{s\eta T\log(1/\eps)}+\log(1/\eps)\right)
\end{equation}
primitive sparse queries, including inverse calls. The error includes
all work registers and is unchanged by an untouched reference system.
Neither the mesh size nor $\beta$ multiplies the query count. For
zero evolution time or a Hamiltonian known to vanish throughout the
evolution, the evolution operator is the identity, so we can implement
it without any matrix queries.
\end{theorem}

\begin{proof}
To simulate the changing Hamiltonian, we divide the evolution into
short time intervals and use the fixed matrix $\widetilde H_r$ in
each interval to construct a Cayley step. These steps must be applied
in chronological order. Our strategy is to combine their transducers
before converting them into a quantum circuit. The main task is to
show that this conversion retains the desired query bound without
introducing a factor proportional to the number of intervals. We
first establish this bound for a finite sequence of fixed matrices.
We then choose sufficiently short intervals so that the ordered
product of Cayley steps approximates the intended continuous
evolution within the required error.

For the nontrivial case $T>0$, fix $0<\gamma\le1$ and set the
Hamiltonian scale
\begin{equation}\label{hs:eq:timeparams}
 \alpha:=\frac{\eta\sqrt{s}}\gamma.
\end{equation}
We will choose $\gamma$ after determining the conversion cost.
The scale $\alpha$ sets the strength of each boundary coupling. It is
not an additional promise on the Hamiltonian: the sparsity and Euclidean
norm bounds already give $\norm{\widetilde H_r}\le\sqrt{s}\eta$,
so $\norm{\widetilde H_r/\alpha}\le\gamma$.

\paragraph{A finite chronological circuit.}
We now construct the transducer for each time interval. To combine
these transducers inexpensively, we want the rotations coupling the
public and private spaces to require no matrix queries. As in the
baseline construction, we achieve this by using list coordinates:
we label a private basis state by a position in a row or column list,
rather than by the corresponding matrix-entry indices. Recall that
$J_{\rm free}$ is the known embedding that prepares a uniform
superposition of the two orientations and the $s$ list positions
for each system basis state. This preparation requires no matrix
queries and is the same for every time interval. The locators
subsequently convert the list positions into matrix-entry indices.
For interval $r$, let $V_r$ perform this conversion using the row
or column locator for $\widetilde H_r$, selected by the orientation
register. If $D_{{\rm edge},r}$ is its entrywise Cayley unitary,
then its action in list coordinates is
\begin{equation}
 D_{0,r}:=V_r^\dagger D_{{\rm edge},r}V_r.
\end{equation}
Thus all matrix dependence is in the private unitary, while the
coupling uses the same known embedding $J:=J_{\rm free}$.
To determine the induced public action, write
$h_r:=\widetilde H_r/\alpha$. The Cayley response identity
\eqref{ov:eq:response} gives
\begin{equation}
 R_{0,r}(1):=(I-D_{0,r})^{-1},\qquad
 J^\dagger R_{0,r}(1)J=\frac{I+\ii h_r}{2}.
\end{equation}
This is the same compressed response used in the baseline balance
calculation, now with the matrix for interval $r$.

For interval $r$, write $\Delta_r:=t_{r+1}-t_r$ and
$w_r:=\alpha\Delta_r$, and for now assume $w_r\le1/4$.
We choose the boundary rotation of
Section~\ref{ov:sec:circuit} with coefficients
\begin{equation}
 c_r:=\frac{1-w_r/2}{1+w_r/2},\qquad
 \beta_r:=\frac{\sqrt{2w_r}}{1+w_r/2}.
\end{equation}
Here $\beta_r$ is a rotation coefficient, distinct from the
Lipschitz bound $\beta$. The rotation $G_r$ couples the public
space to the image of $J$ in private slot $r$, and fixes the other
slots. Our choice makes $(1-c_r)/(1+c_r)=w_r/2$, so the stationary
balance calculation gives precisely the Cayley step for duration
$\Delta_r$:
\begin{equation}\label{hs:eq:timecayley}
 \begin{aligned}
 C_r&:=(I-\ii\Delta_r\widetilde H_r/2)
                  (I+\ii\Delta_r\widetilde H_r/2)^{-1},\\
 \Gamma_r&:=\sqrt{2w_r}(I-D_{0,r})^{-1}J
                          (I+\ii w_rh_r/2)^{-1}.
 \end{aligned}
\end{equation}
Here $C_r$ acts on the $N$-dimensional public system, while
$\Gamma_r$ maps that system into the private space for interval $r$.
To recall how the boundary coefficients produce this action, write
$S_r=G_r(I_{\rm pub}\oplus D_{0,r})$. Requiring the private output
to equal the private input gives
\begin{equation}
 [I-(I-(1-c_r)JJ^\dagger)D_{0,r}]\Gamma_r=\beta_rJ.
\end{equation}
We solve this equation using the displayed compression of
$R_{0,r}(1)$. Substitution into the public output gives
\begin{equation}
 c_rI-\beta_rJ^\dagger D_{0,r}\Gamma_r
 =\bigl[(1+c_r)I-\ii(1-c_r)h_r\bigr]
  \bigl[(1+c_r)I+\ii(1-c_r)h_r\bigr]^{-1}=C_r.
\end{equation}
The ratio $(1-c_r)/(1+c_r)=w_r/2$ was chosen precisely to turn
this rational expression into the Cayley step of duration $\Delta_r$.
Thus $S_r(\psi\oplus\Gamma_r\psi)=C_r\psi\oplus\Gamma_r\psi$.
As before, this is a stationary relation between direct-sum amplitudes;
it does not assume that we can prepare $\Gamma_r\psi$.

We now compose these steps before converting them to an initialised
circuit. Since a private operation in one slot commutes with every
boundary rotation on another slot, the same rearrangement as in
Eq.~\eqref{ov:eq:chronological-circuit} gives the complete interaction
and its public target:
\begin{equation}
 S_{\rm tot}=G_{M-1}\cdots G_0
       \left(I_{\rm pub}\oplus\bigoplus_{r=0}^{M-1}D_{0,r}\right),
 \qquad U_C:=C_{M-1}\cdots C_0.
\end{equation}
We have moved only operations on disjoint slots; the boundary
rotations retain their chronological order. Before slot $r$ is
used, its public input is $C_{r-1}\cdots C_0\psi$. Its stationary
component must therefore be $\Gamma_r$ applied to this input.
With the empty prefix equal to $I$, define
\begin{equation}\label{hs:eq:timecatalyst}
 \Gamma_{\rm tot}\psi
 :=\bigoplus_{r=0}^{M-1}\Gamma_rC_{r-1}\cdots C_0\psi.
\end{equation}
For example, with two intervals the input private components are
$\Gamma_0\psi$ and $\Gamma_1C_0\psi$. The first interaction changes
the public input to $C_0\psi$ and returns $\Gamma_0\psi$ unchanged.
The second therefore receives exactly the public input for which its
private component was chosen; it changes that input to $C_1C_0\psi$
and returns $\Gamma_1C_0\psi$ unchanged. The same reasoning applies
successively to every slot. Each interaction returns its own private
component unchanged and leaves all other private slots untouched.
Consequently,
\begin{equation}
 S_{\rm tot}(\psi\oplus\Gamma_{\rm tot}\psi)
       =U_C\psi\oplus\Gamma_{\rm tot}\psi.
\end{equation}
The public component has therefore undergone the desired ordered
product $U_C$, while the private component is returned unchanged.
To turn this relation into a simulation algorithm, we still need to
establish two things: that the combined interaction can be implemented
cheaply, and that we can obtain the desired evolution starting with
zero private amplitude. We address the implementation cost first.

A complete call to $S_{\rm tot}$, its inverse or its controlled
form has constant sparse-query cost, independently of $M$. We can
verify this from three implementation observations:
\begin{enumerate}
\item The direct sum of private operations uses the slot label as
the time address for the represented matrix. One coherent locator--entry-rotation--inverse-
locator sequence therefore applies every selected $D_{0,r}$, using
a constant number of primitive sparse queries. Reversing the sequence
gives its inverse with the same query cost.
\item Lemma~\ref{access:lem:control} supplies the required controls
on the entire oracle register, including the time address. A valid
basis input at any one sample provides the phase-one reference used
in that lemma, so the extra controls cost only a constant number
of queries.
\item The boundary rotations are known operations and require no
matrix queries. Their number can grow with $M$, but this does not
affect the query count of the complete call.
\end{enumerate}
We have shown that each use of the combined interaction requires
only a constant number of matrix queries. It remains to determine
how many uses are needed to obtain the desired evolution starting
with zero private amplitude. The conversion bound depends in part
on the size of the private component $\Gamma_{\rm tot}\psi$, so we
now bound its squared norm.

The compressed second-moment estimate
\eqref{ov:eq:compressed-response}, applied to the fixed matrix
$\widetilde H_r$, gives
\begin{equation}
 J^\dagger R_{0,r}(1)^\dagger R_{0,r}(1)J
       \preceq\frac{1+\gamma^2}{4}I.
\end{equation}
Since $h_r$ is Hermitian,
$\norm{(I+\ii w_rh_r/2)^{-1}}\le1$; substituting these two
bounds into Eq.~\eqref{hs:eq:timecayley} yields
$\Gamma_r^\dagger\Gamma_r\preceq w_r(1+\gamma^2)I/2$.
The private slots in Eq.~\eqref{hs:eq:timecatalyst} are orthogonal,
and every public prefix is unitary. For any public input $\psi$,
we therefore have
\begin{equation}
 \begin{aligned}
 \norm{\Gamma_{\rm tot}\psi}^2
 &=\sum_{r=0}^{M-1}
       \norm{\Gamma_rC_{r-1}\cdots C_0\psi}^2\\
 &\le\frac{1+\gamma^2}{2}\sum_rw_r\norm\psi^2
 =\frac{\alpha T(1+\gamma^2)}2\norm\psi^2.
 \end{aligned}
\end{equation}
The same estimate holds uniformly over all public inputs. We may thus choose
\begin{equation}
 W:=\frac{\alpha T(1+\gamma^2)}2
\end{equation}
as an upper bound on the squared stationary private amplitude for
every unit public input. Refining the mesh divides the same total
coupling $\sum_rw_r=\alpha T$ among more slots, so it does not
increase $W$. Nor did this estimate require the different $C_r$
to commute.

\paragraph{Controlling propagation between unequal slots.}
The bound on $W$ controls the stationary component we must replace,
but high-precision conversion also depends on how amplitude moves
through the private space during reuse. Let $D_{\rm priv}$ be the
private block of $S_{\rm tot}$. The averaging filter of
Eq.~\eqref{ov:eq:filter-guarantee} must make its polynomial in
$D_{\rm priv}$ small. Since $D_{\rm priv}$ need not be normal,
separation of its eigenvalues from one would not alone establish
this operator bound. The difficulty is that a nonnormal operator need not have an
orthonormal eigenbasis: small values of a polynomial on its eigenvalues
need not imply a small operator norm. We therefore use the same stronger
input as in the baseline conversion proof, a bound on the resolvent
$(zI-D_{\rm priv})^{-1}$ throughout a disk around one. This inverse
describes the private response to a forcing term in
$(zI-D_{\rm priv})x=y$; the contour estimate later uses its norm to
control the filter on the whole private space. We will first control the response inside one slot and
the public amplitude transmitted between slots, then use their
chronological order to assemble the inverse for the complete private
space. This will show why different matrices still incur a cost
controlled by total coupling, rather than by the number of intervals.

The bare Cayley unitaries supply the radius for this disk. The norm
promises imply $\norm{h_r}\le\gamma\le1$, and the same estimate
as in Eq.~\eqref{ov:eq:Q} gives the uniform integer bound
\begin{equation}
 \norm{R_{0,r}(1)}\le Q,\qquad
 Q:=\left\lceil\max\{1,\gamma\sqrt{s}\}\right\rceil.
\end{equation}
We work on $|z-1|\le1/(4Q)$, where the local inverse estimates
from Appendix~\ref{sec:conversion-proof} apply. To recall the objects
in those estimates, set
\begin{equation}
 \begin{aligned}
 R_{0,r}(z)&:=(zI-D_{0,r})^{-1},\qquad
 \Phi_r(z):=J^\dagger D_{0,r}R_{0,r}(z)J,\\
 R_r(z)&:=\bigl[zI-(I-(1-c_r)JJ^\dagger)D_{0,r}\bigr]^{-1}.
 \end{aligned}
\end{equation}
The variable $z$ is a complex spectral parameter, not a time label.
Both $R_{0,r}(z)$ and $R_r(z)$ act on the private space of a single
slot, whereas $\Phi_r(z)$ acts on the public system.
Here $R_r(z)$ includes the boundary rotation in the local private
block, whereas $R_{0,r}(z)$ is the response before that rotation.
The compression $\Phi_r(z)$ takes a public input into the private
space through $J$ and returns its response through $J^\dagger
D_{0,r}$. The calculation leading to Eq.~\eqref{hs:eq:locbounds}
uses only the compressed response identities, the bound $Q$,
$\norm{h_r}\le1$ and $w_r\le1/4$. Applying it separately to
each interval therefore gives
\begin{equation}\label{hs:eq:timelocal}
 \norm{R_r(z)}\le2Q,\qquad
 \norm{R_r(z)J},\ \norm{J^\dagger D_{0,r}R_r(z)}
       \le\frac{12}{7}\frac{\sqrt{1+\gamma^2}}2.
\end{equation}
The last two bounds are smaller than the full inverse bound because
they describe injection through $J$ and return through $J^\dagger$.
Keeping this distinction is what will make the sum depend on the
total coupling rather than the number of slots.

The endpoint estimates control entry into and return from a slot.
We also need to bound the public amplitude passed through the
intervening slots. Let $v$ be the incoming public amplitude, and
let $x$ be the private response in the resolvent equation when
there is no forcing within slot $r$. The local interaction gives
\begin{equation}
 \bigl[zI-(I-(1-c_r)JJ^\dagger)D_{0,r}\bigr]x
       =\beta_rJv,\qquad
 v'=c_rv-\beta_rJ^\dagger D_{0,r}x,
\end{equation}
where $v'$ is the outgoing public amplitude. Solving the first
equation gives $x=\beta_rR_r(z)Jv$. Substituting into the second
therefore gives $v'=u_r(z)v$, with
\begin{equation}
 \begin{aligned}
 u_r(z)&:=c_rI-\beta_r^2J^\dagger D_{0,r}R_r(z)J\\
 &=c_rI-\beta_r^2\Phi_r(z)
                         [I+(1-c_r)\Phi_r(z)]^{-1}.
 \end{aligned}
\end{equation}
This map passes the incoming public amplitude to the next slot in
the resolvent calculation. At $z=1$ it is the unitary $C_r$.
We bound its growth by subtracting this known unitary value:
\begin{equation}
 \begin{aligned}
 u_r(z)-u_r(1)
 ={}&-\beta_r^2[I+(1-c_r)\Phi_r(z)]^{-1}\\
 &\quad\times[\Phi_r(z)-\Phi_r(1)]
                   [I+(1-c_r)\Phi_r(1)]^{-1}.
 \end{aligned}
\end{equation}
This subtraction retains the unit norm at $z=1$ and bounds only
the growth away from it. The endpoint and inverse estimates used
in Eq.~\eqref{hs:eq:transfer} apply separately to each fixed
matrix and give
\begin{equation}\label{hs:eq:timeproducts}
 \begin{aligned}
 \norm{u_r(z)}&\le1+\tfrac32(1+\gamma^2)w_r|z-1|,\\
 \norm{u_b(z)\cdots u_a(z)}
 &\le\exp\!\left(\tfrac32(1+\gamma^2)|z-1|
                                      \sum_{r=a}^bw_r\right),
 \qquad 0\le a\le b<M.
 \end{aligned}
\end{equation}
For the second line we take the product of the individual norm
bounds and use $\prod_r(1+x_r)\le e^{\sum_rx_r}$ for $x_r\ge0$.
This argument controls the ordered product without interchanging
any of its generally noncommuting factors.

We now combine these local responses to obtain the complete
resolvent. For an arbitrary private vector $y$, we solve
$(zI-D_{\rm priv})x=y$; expressing the solution $x$ in terms of
$y$ will identify the inverse and allow us to bound its norm.
Write $x_r,y_r$ for the components in slot $r$, and let $v_r$
be the public amplitude leaving that slot while evaluating
$S_{\rm tot}$ on $0\oplus x$. There is no incoming public
amplitude, so $v_{-1}=0$.
The private equation and subsequent public output at slot $r$ give
\begin{equation}
 x_r=R_ry_r+\beta_rR_rJv_{r-1},\qquad
 v_r=u_rv_{r-1}-\beta_rJ^\dagger D_{0,r}R_ry_r,
\end{equation}
where the dependence on $z$ is suppressed in this recurrence.
Iteration yields
\begin{equation}
 v_{r-1}=-\sum_{j<r}u_{r-1}\cdots u_{j+1}
                      \beta_jJ^\dagger D_{0,j}R_jy_j.
\end{equation}
Thus a force in slot $j$ contributes only to that slot and later
ones. The recurrence determines a unique $x$ for each $y$, proving
invertibility throughout the disk. It also identifies the inverse's
diagonal blocks as $R_r$, its blocks above the diagonal as zero,
and its remaining blocks as
\begin{equation}\label{hs:eq:timeoffdiag}
 [(zI-D_{\rm priv})^{-1}]_{rj}
 =-\beta_r\beta_jR_rJ\,u_{r-1}\cdots u_{j+1}
                          J^\dagger D_{0,j}R_j,
 \qquad r>j.
\end{equation}
Read from right to left, this expression returns the forcing term
in slot $j$ to the public space, propagates it through the
intervening slots, and injects it into slot $r$. Empty transfer
products are the identity.

It remains to sum these contributions without losing a factor
depending on the mesh. Combining Eq.~\eqref{hs:eq:timelocal},
$\beta_r\le\sqrt{2w_r}$ and Eq.~\eqref{hs:eq:timeproducts}
bounds the norm of each off-diagonal block by
\begin{equation}
 \frac{72}{49}(1+\gamma^2)\sqrt{w_rw_j}\,e^{3W/(4Q)}.
\end{equation}
To see explicitly why summing these blocks costs total coupling rather
than the number of slots, let $K$ denote the strictly lower-triangular
part of the resolvent. The block estimate gives
\begin{equation}
 \norm{(Ky)_r}\le
 \frac{72}{49}(1+\gamma^2)e^{3W/(4Q)}
       \sqrt{w_r}\sum_{j<r}\sqrt{w_j}\norm{y_j}.
\end{equation}
All terms on the right are nonnegative. We can therefore enlarge the
sum to all $j$ and apply Cauchy--Schwarz:
\begin{equation}
 \left(\sum_j\sqrt{w_j}\norm{y_j}\right)^2
 \le\left(\sum_jw_j\right)\left(\sum_j\norm{y_j}^2\right).
\end{equation}
Squaring the output bounds and summing over $r$ supplies a second
factor $\sum_rw_r$. Taking the square root then gives
\begin{equation}
 \norm K\le\frac{72}{49}(1+\gamma^2)e^{3W/(4Q)}
                   \sum_rw_r
 =\frac{144}{49}We^{3W/(4Q)}.
\end{equation}
Equivalently, the scalar matrix dominating the block norms is a
constant times the rank-one matrix $(\sqrt{w_rw_j})_{r,j}$, whose
operator norm is $\sum_rw_r$. Retaining the two endpoint weights
$\sqrt{w_r}$ and $\sqrt{w_j}$ is the decisive step: replacing them
by separate uniform bounds would conceal the cancellation of the mesh
size. Adding the diagonal contribution gives
\begin{equation}\label{hs:eq:timeresolvent}
 \norm{(zI-D_{\rm priv})^{-1}}
 \le2Q+\frac{144}{49}We^{3W/(4Q)}
 \le3(Q+W)e^{3W/(4Q)}.
\end{equation}
We have therefore proved both that the complete private inverse
exists throughout the disk and that its norm has the same bound as
in the constant-Hamiltonian case. The estimate uses only the total
coupling and uniform local bounds, so it remains valid for changing
supports and noncommuting matrices.

This is the remaining input needed to convert the stationary
relation into an initialised circuit. Since $S_{\rm tot}$ is
unitary, its private block $D_{\rm priv}$ is also a contraction.
The contour estimate of Appendix~\ref{sec:conversion-proof} therefore
applies to the averaging filter from
Eq.~\eqref{ov:eq:filter-guarantee}: it suppresses
$D_{\rm priv}$ with degree
$d=\mathcal O(W+Q\log(1/\eps))$ and $d\ge2W$.
The finite-reuse construction expresses that filter through circuits
starting with zero private amplitude; its signed combination and
one amplification round give error at most $\eps/4$ relative to
$U_C$ when the filter tolerance is $\eps/8$. This error includes
all work-register outputs, as in Eq.~\eqref{access:eq:full-output}.
The complete transducer call count is the one in
Eq.~\eqref{hs:eq:explicitcount}, of order
$\mathcal O(W+Q\log(1/\eps))$. Thus we have an initialised circuit
for the ordered Cayley product, with no query multiplier from the
number of physical slots.

\paragraph{Sampling the intended evolution.}
We have obtained an initialised circuit for $U_C$ whose query cost
does not grow with the mesh. We now choose that mesh to make $U_C$
approximate the intended continuous evolution. There are two
approximations to control: we replace $H(u)$ on each interval by the
fixed represented matrix $\widetilde H_r$, and then replace that
matrix's exponential by its Cayley step. The first approximation
includes both variation of $H(u)$ within the interval and the error
in representing $H(t_r)$.

For the Cayley error we need a uniform operator-norm bound. Let
$\Lambda:=\sqrt{s}\eta$; the row and column
Euclidean-norm promises bound each absolute row and column sum by
$\Lambda$, so both the intended and represented matrices have
operator norm at most $\Lambda$. On interval $r$, Duhamel's
formula bounds the difference between the exact propagator and
$e^{-\ii\Delta_r\widetilde H_r}$ by
\begin{equation}
 \int_{t_r}^{t_{r+1}}\norm{H(u)-\widetilde H_r}\,du
 \le\frac{\beta\Delta_r^2}{2}+\delta_{\rm val}\Delta_r.
\end{equation}
The two terms account for motion along the path and error in its
represented left-endpoint value. Replacing the fixed matrix's
exponential by $C_r$ adds at most
$\Lambda^3\Delta_r^3/12$, by the Cayley phase estimate used in
Eq.~\eqref{ov:eq:discretisation}. Telescoping the unitary products
in their given order therefore gives
\begin{equation}\label{hs:eq:timedisc}
 \norm{U_C-U_H(T,0)}
 \le\sum_{r=0}^{M-1}\left[
 \frac{\Lambda^3\Delta_r^3}{12}
 +\frac{\beta\Delta_r^2}{2}+\delta_{\rm val}\Delta_r\right].
\end{equation}
For a uniform mesh, take any integer
\begin{equation}\label{hs:eq:timegrid}
 M\ge\max\left\{1,4\alpha T,
        \frac{6\beta T^2}{\eps},
        \frac{(\Lambda T)^{3/2}}{\sqrt\eps}\right\}.
\end{equation}
The second lower bound ensures $w_r\le1/4$. The last two make
the summed Lipschitz and Cayley errors at most $\eps/12$ each,
and the representation-error assumption supplies the same bound
for the third contribution. Sampling therefore costs at most
$\eps/4$. Together with conversion error, the ideal-gate error
is at most $\eps/2$, leaving $\eps/2$ for approximation of the
known gates in this finite circuit.

We have now separated the two choices. The mesh controls the error in
replacing the continuous path by Cayley steps, while $\gamma$ controls
the cost of converting their composition to an initialised circuit.
Neither refining the mesh nor improving the represented values changes
$W$ or $Q$. Substituting their definitions, the latter cost is
\begin{equation}
 \mathcal O\!\left(
 \frac{\sqrt{s}\eta T}{2}(\gamma+\gamma^{-1})
 +\left\lceil\max\{1,\gamma\sqrt{s}\}\right\rceil
       \log(1/\eps)\right).
\end{equation}
Reducing $\gamma$ increases the coupling contribution but decreases
the full private response until $Q$ reaches its constant floor. We can
therefore use exactly the balance in Eq.~\eqref{ov:eq:gamma}, with
$|t|$ replaced by $T$, to obtain Eq.~\eqref{hs:eq:timebound}. Finally, tensoring the full isometry
difference with an identity preserves its norm, proving the
untouched-reference guarantee.
\end{proof}

The proof also shows how to relax the regularity assumption. For
a piecewise Lipschitz path, we include its known breakpoints in
the mesh and use the appropriate one-sided sample and local bound
on each interval. More generally, for an integrable Hermitian path
with the same uniform sparsity and norm promises, it suffices to
supply a known error bound on
\begin{equation}
 \sum_r\int_{t_r}^{t_{r+1}}
                  \norm{H(u)-\widetilde H_r}\,du.
\end{equation}
This directly replaces the Lipschitz and representation terms in
Eq.~\eqref{hs:eq:timedisc}; the Cayley error and conversion analysis
are unchanged. Regularity is therefore used to choose an accurate
finite mesh, while the circuit allows coherent changes of support
without any smoothness assumption on the locators.

We may also select known start times, durations or other parameters
in superposition, provided the corresponding coherent sampled-time
access is supplied. For all branches we choose a common schedule
using the largest total coupling $W$ and uniform response bounds.
The smaller branch values obey the same resolvent estimate and
degree condition $d\ge2W$, while the boundary rotations and
time-address arithmetic depend coherently on the selected parameters.
The mesh size can still increase the number of known gates and
the workspace: the theorem bounds queries. Finally, this input
class contains the constant hard instances of
Theorem~\ref{hsl:thm:main}, with $\beta=\delta_{\rm val}=0$, so
the lower bounds of that theorem continue to apply.

\subsection{Hybrid Cayley and block-encoding bounds}

The time-dependent proof separates the source of the Hamiltonian from
the conversion of its stationary action into a circuit. Once a private
unitary has the required compressed response, the same boundary
construction produces the Cayley steps. Its query cost is then determined
by two quantities: the number of private calls needed for conversion,
and the matrix-query cost of each private call. We can exploit this
separation by representing the same Hamiltonian in two different ways.

Entrywise Cayley access retains the small Euclidean norms of the rows
and columns and gives the favourable $\sqrt{s}\eta T$ time term.
An ordinary block encoding has the larger normalisation $s\eta$, but
its particularly simple private spectrum admits a stronger precision
filter. We will implement both representations from the same sparse
oracles and choose the cheaper simulation from the supplied parameters.
The resulting bound improves the additive precision term even for a
changing Hamiltonian.

To make the comparison precise, suppose we have coherent controlled
access to Hermitian unitaries $O_r$ on a larger space and a known
isometry $J$ satisfying
\begin{equation}
 J^\dagger O_rJ=\widetilde H_r/\alpha_{\rm BE},
 \qquad \alpha_{\rm BE}>0.
\end{equation}
This is a Hermitian block encoding: after entering through $J$,
applying $O_r$ and selecting the original embedded subspace, we
obtain the represented matrix divided by the known normalisation
$\alpha_{\rm BE}$. The selected block need not be unitary, even
though $O_r$ is. Since $O_r$ is both Hermitian and unitary, $O_r^2=I$ and its
eigenvalues are $+1$ and $-1$. Taking $D_{0,r}:=\ii O_r$ therefore
gives only the two private phases $+\ii$ and $-\ii$. Neither is close
to one. More explicitly,
$(I-\ii O_r)(I+\ii O_r)=2I$, so we can compute the response
needed by the stationary balance exactly:
\begin{equation}\label{hs:eq:blockspecial}
 R_{0,r}(1)=\frac{I+\ii O_r}{2},\qquad
 J^\dagger R_{0,r}(1)J
       =\frac{I+\ii\widetilde H_r/\alpha_{\rm BE}}2,
 \qquad R_{0,r}(1)^\dagger R_{0,r}(1)=\frac I2.
\end{equation}
Thus the response conditions of Lemma~\ref{hs:lem:abstract}
hold with $Q=1$ and $\alpha=\alpha_{\rm BE}$. The compressed
squared response $I/2$ bounds the total squared stationary component
by $\alpha_{\rm BE}T$, so the preceding chronological argument gives
\begin{equation}\label{hs:eq:blockadditive}
 \mathcal O\!\left(\alpha_{\rm BE}T+\log(1/\eps)\right)
\end{equation}
queries to the block encoding. This recovers the additive-logarithmic
part of the time-indexed transducer result of Ref.~\cite{CGWZ}.

The preceding calculation used only bounds on the resolvent. It did
not yet use all the information in the two-phase spectrum. That extra
structure permits cancellation with factorial, rather than merely
exponential, error decay. For comparison, a geometric bound of the form
$e^{-c d}$ requires a degree proportional to $\log(1/\eps)$; a
factorial-type bound of the form $(Ca/d)^d$ can reach the same accuracy
at a smaller degree when $\log(1/\eps)$ is large relative to the
evolution scale $a$. Ref.~\cite{CGWZ} establishes this stronger
cancellation for private operations obtained from involutions. For $\alpha_{\rm BE}T>0$,
Ref.~\cite{CGWZ} proves the bound
\begin{equation}\label{hs:eq:prioroptimal}
 \mathcal O\!\left(\alpha_{\rm BE}T+
 \frac{\log(1/\eps)}
 {\log(e+\log(1/\eps)/(\alpha_{\rm BE}T))}\right)
\end{equation}
by obtaining factorial error decay for these involutions, meaning
unitaries whose square is the identity. Our exponential filter
establishes Eq.~\eqref{hs:eq:blockadditive}; it does not establish
the denominator in Eq.~\eqref{hs:eq:prioroptimal}. When the private
operation comes from an involution, we can instead use the filter
of Ref.~\cite{CGWZ} and retain that stronger dependence.

In the sparse model, the entry bound $|H_{ij}|\le\eta$ gives a
block encoding with normalisation $\alpha_{\rm BE}=s\eta$.
Its leading time cost is therefore $s\eta T$, whereas the Cayley
construction gives $\sqrt{s}\eta T$ together with its mixed
time--precision term. The factorial dependence can nevertheless
make the block-encoding route cheaper at high precision. Choosing
between the two gives the following improvement for both constant
and changing Hamiltonians.

\begin{corollary}[Improved additive precision dependence]
\label{hs:cor:hybrid}
Under the assumptions of Theorem~\ref{hs:thm:main} or
Theorem~\ref{hs:thm:time}, let $\tau$ denote the normalised duration
$\eta|t|$ or $\eta T$, respectively. For $\tau>0$, the same
error guarantee on all output registers is achievable using
\begin{equation}\label{hs:eq:hybridbound}
 \mathcal O\!\left(\sqrt{s}\,\tau+
       \sqrt{s\tau\log(1/\eps)}
       +\frac{\log(1/\eps)}
       {\log(e+\log(1/\eps)/(s\tau))}\right)
\end{equation}
primitive sparse queries. The time-dependent case retains coherent
sampled-time access and its known discretisation-error bound. For
$\tau=0$, the identity uses no queries.
\end{corollary}

\begin{proof}
We prove the corollary in three steps. First we construct the alternative
Hermitian block encoding using the given sparse oracles. We then check
its time discretisation and complete-output error, so that it can be
compared with the Cayley algorithm on equal terms. Finally, we choose
between their query bounds.

\emph{Constructing the alternative private operation.}
We first construct the Hermitian block encoding required by the
alternative route. We describe a single represented Hermitian
matrix $H$; retaining the sample address gives the time-dependent
version. Recall the row and column list maps $L,R$ and the diagonal
entry operator $Y$ from Eq.~\eqref{ov:eq:frames}, with
$Y_{(i,j),(i,j)}=H_{ij}/\eta$. Their uniform edge embedding and
orientation operator are
\begin{equation}
 J_{\rm edge}:=\frac1{\sqrt{2s}}\binom{L}{R},
 \qquad E:=\begin{pmatrix}0&Y\\Y^\dagger&0\end{pmatrix}.
\end{equation}
Every nonzero entry occurs in both its row and column lists, so
$L^\dagger YR=H/\eta$ even when the zero padding differs. The
overlap through the two orientations is therefore
\begin{equation}
 J_{\rm edge}^\dagger EJ_{\rm edge}
 =\frac{L^\dagger YR+R^\dagger Y^\dagger L}{2s}
 =\frac{H}{s\eta}.
\end{equation}
This identifies the desired normalisation. We still need to make
$E$ the selected block of a Hermitian unitary.

The entry bound $|H_{ij}|\le\eta$ lets us do this locally. On a
fixed edge, write $z:=H_{ij}/\eta$ and introduce one additional
flag with ordered states $\ket0,\ket1$. We apply the block
\begin{equation}
 E(z):=\begin{pmatrix}0&z\\\overline z&0\end{pmatrix},
 \qquad
 O(z):=\begin{pmatrix}
 E(z)&\sqrt{1-|z|^2}\,I_2\\
 \sqrt{1-|z|^2}\,I_2&-E(z)
 \end{pmatrix}
\end{equation}
on the flag and the two orientations. The additional flag supplies the amplitude missing from the contraction
$E(z)$. In the ordered flag basis, $O(z)$ acts on a flag-zero input as
\begin{equation}
 \ket0\otimes\chi\longmapsto
 \ket0\otimes E(z)\chi+
 \ket1\otimes\sqrt{1-|z|^2}\,\chi,
 \qquad \chi\in\mathbb C^2.
\end{equation}
The two squared output norms sum to $\norm\chi^2$, because
$E(z)^2=|z|^2I_2$. The negative sign in the flag-one diagonal block
also cancels the off-diagonal blocks of $O(z)^2$, giving
\begin{equation}
 O(z)^2=
 \begin{pmatrix}
 E(z)^2+(1-|z|^2)I_2&0\\
 0&E(z)^2+(1-|z|^2)I_2
 \end{pmatrix}=I_4.
\end{equation}
Thus $O(z)^\dagger=O(z)$ and $O(z)$ is unitary, with selected
flag-zero block exactly $E(z)$. Taking the direct sum over edge labels gives a
Hermitian unitary $O_{\rm edge}$ whose compression through
$\ket0\otimes J_{\rm edge}$ is $H/(s\eta)$.

To make the embedding known, use the same orientation-selected
locator $V$ as in the baseline construction, so that
$J_{\rm edge}=VJ_{\rm free}$. In list coordinates we implement
\begin{equation}
 O:=(I_{\rm flag}\otimes V^\dagger)O_{\rm edge}
                         (I_{\rm flag}\otimes V),
 \qquad J\psi:=\ket0\otimes J_{\rm free}\psi.
\end{equation}
Then $J$ is known, $O$ is a Hermitian unitary and
$J^\dagger OJ=H/(s\eta)$. Its query cost is constant. After
applying $V$, let $w_{ij}$ encode $H_{ij}$ and let $\ket\chi$
be any state of the additional flag and the orientation. The value
query, local gate and inverse query act as
\begin{equation}
 \begin{aligned}
 \ket{i,j}\ket\chi\ket0
 &\longmapsto\ket{i,j}\ket\chi\ket{w_{ij}}\\
 &\longmapsto\ket{i,j}(O(z)\ket\chi)\ket{w_{ij}}\\
 &\longmapsto\ket{i,j}(O(z)\ket\chi)\ket0.
 \end{aligned}
\end{equation}
The local gate changes neither the edge label nor the stored
encoding, so the inverse value query erases the work register
coherently. We then apply $V^\dagger$ to return to list coordinates.

We must also check that the locator selection, controlled block
encoding and adjoint retain constant query cost.
\begin{enumerate}
\item \textbf{Selecting the appropriate locator.}
The operation $V$ uses the orientation register to select the row
or column locator. Lemma~\ref{access:lem:control} lets us implement
this selection with constant query overhead.
\item \textbf{Controlling the complete block encoding.}
To apply $O$ conditionally on an additional control qubit, we
control only the local gate $O(z)$. When the control is enabled,
the circuit performs the block encoding described above. When it
is disabled, the local gate is the identity, so the value query
and locator operations are undone by their inverses. The complete
circuit therefore acts as the identity on that branch, as required.
\item \textbf{Implementing the adjoint.}
We reverse the circuit and replace each operation by its inverse.
This implements $O^\dagger$ with the same number of queries.
\end{enumerate}
Thus we obtain the block encoding, its controlled version and its
adjoint with constant sparse-query cost, including when the time
address is in superposition.

For a constant Hamiltonian we now use optimal block-encoding
simulation~\cite{LowChuang,Qubitization}, with the joint
time--precision bound of Corollary~60 in the \href{https://arxiv.org/abs/1806.01838}{full version} of Ref.~\cite{QSVT}; for a
changing Hamiltonian we use Eq.~\eqref{hs:eq:prioroptimal}.
These supply the alternative query bound
\begin{equation}\label{hs:eq:ordinarybound}
 \mathcal O\!\left(s\tau+
       \frac{\log(1/\eps)}
       {\log(e+\log(1/\eps)/(s\tau))}\right).
\end{equation}
\emph{Matching the access and error guarantees.}
Before selecting between this and the Cayley bound, we verify that
the alternative applies to our represented piecewise-constant path
and gives the required error guarantee on all output registers.

For the time-dependent case, we can refine the represented path
without increasing the number of block-encoding queries. The
finite-sequence construction in Propositions~6--7 and Lemma~8 of
Ref.~\cite{CGWZ} acts on a sequence of Hermitian block encodings,
and its query count is independent of the sequence length. Although
Theorem~9 of that work assumes a Lipschitz path, the following
observations allow us to use its finite-sequence construction for
our represented path
$\widetilde H(u)=\widetilde H_r$ on $[t_r,t_{r+1})$:
\begin{enumerate}
\item Refine to a uniform power-of-two mesh with $N_{\rm ref}$
intervals, and on each new interval choose the original sample
containing its left endpoint. Let $H_{\rm ref}$ denote the
resulting path. Only refined cells containing one of the $M-1$
original breakpoints can differ from $\widetilde H$. Each such
cell has length $T/N_{\rm ref}$, and the difference of two
represented matrices has norm at most $2s\eta$. Consequently,
\begin{equation}
 \int_0^T\norm{H_{\rm ref}(u)-\widetilde H(u)}\,du
 \le\frac{2s\tau(M-1)}{N_{\rm ref}}.
\end{equation}
\item The finite construction uses Cayley steps for the refined
samples. Their product differs from the evolution of $H_{\rm ref}$
by at most $(s\tau)^3/(12N_{\rm ref}^2)$, by the same phase
estimate and ordered telescoping used above.
\item All original breakpoints are known, so we can compute the
original sample label from a refined label without a matrix query.
One coherent original query then supplies the selected sample,
with controls obtained as before. Increasing $N_{\rm ref}$
therefore changes known arithmetic and register sizes, but does
not change the cost per block-encoding query.
\end{enumerate}
We choose $N_{\rm ref}$ at least as large as the filter degree
and large enough that the two new errors, together with the
original error bound relating $\widetilde H$ to $H$, fit within
$\eps/2$. The finite-sequence query count remains of the order in
Eq.~\eqref{hs:eq:ordinarybound}.

We next convert a selected-block guarantee into the required
guarantee on all output registers. Let $\mathsf E$ append zero work registers to the system input, let $\mathcal U$ be the resulting
unitary circuit, and write
$F_{\rm sim}:=\mathsf E^\dagger\mathcal U\mathsf E$ for its block
between initial and final zero work states. If
$\norm{F_{\rm sim}-U}\le\delta$ for a desired unitary $U$,
then unitarity gives
\begin{equation}
 \begin{aligned}
 (\mathcal U\mathsf E-\mathsf E U)^\dagger
       (\mathcal U\mathsf E-\mathsf E U)
       &=2I-U^\dagger F_{\rm sim}-F_{\rm sim}^\dagger U,\\
 \norm{\mathcal U-U}^2&\le2\delta.
 \end{aligned}
\end{equation}
The second line uses the initialisation and output convention of
Eq.~\eqref{access:eq:full-output}; it follows by writing the
first line as
$U^\dagger(U-F_{\rm sim})+(U^\dagger-F_{\rm sim}^\dagger)U$.
Thus a block-error target $\delta:=\eps^2/8$ makes the
simulation error at most $\eps/2$, including all work-register outputs. With the
time-discretisation allowance above, this gives total error at
most $\eps$. Replacing the block tolerance by $\eps^2/8$
changes its logarithm only by an absolute factor and preserves
the order in Eq.~\eqref{hs:eq:ordinarybound}.

We also retain the query bound when the continuous degree estimate
must be rounded to an integer. If $s\tau\le\eps$, the identity
already achieves the required error, since $\norm{U-I}\le s\tau$.
Otherwise the precision quotient in
Eq.~\eqref{hs:eq:ordinarybound} is bounded below by a positive
absolute constant, so it absorbs the at most one additional query
from rounding.

\emph{Choosing the cheaper simulation.}
The two constructions are alternatives, so we need not combine their
circuits or split the evolution between them. The final choice depends
only on the supplied parameters. If
$\log(1/\eps)\le s\tau$, we use the Cayley bound and absorb
its additive $\log(1/\eps)$ into
$\sqrt{s\tau\log(1/\eps)}$. If $\log(1/\eps)>s\tau$, we
use Eq.~\eqref{hs:eq:ordinarybound} and absorb $s\tau$ into
the same mixed term. Either choice gives
Eq.~\eqref{hs:eq:hybridbound}. All circuit errors bound the isometry
difference on the complete output space, so the untouched-reference
guarantee is preserved as well.
\end{proof}

This choice improves the additive precision term but does not close
the time-dependent precision gap. For $s\ge32$, $\tau\ge1$ and
$0<\eps\le1/64$, the constant-Hamiltonian lower bound in
Theorem~\ref{hsl:thm:main} applies to the time-dependent input class.
Its mixed term is smaller than the hybrid upper bound by a factor
$\sqrt{\log(e+\log(1/\eps)/\tau)}$. The additive terms also
have different denominators: the upper bound uses $s\tau$ inside
the logarithm, whereas the lower bound uses $\tau$. We next
improve both terms for the time-independent problem.

\subsection{Optimal precision with one or two norm bounds}
\label{hs:sec:optimal}

For a fixed Hamiltonian, we can improve both precision terms while
retaining the $\sqrt{s}\eta|t|$ time term. The hybrid bound above
chooses between a representation with favourable sparsity dependence
and a representation with favourable precision dependence. We now
construct representations that retain both advantages. The extra work
is to cancel unwanted contributions from repeated private excursions
without paying a large cost to implement the cancellation.

There are two useful constructions. We first express the entire
Hamiltonian through a preparation whose successful amplitudes carry
positive and negative signs. This lets us identify and cancel the
unwanted terms in a weighted combination of reuse circuits. When the
sparsity is too large for that construction to give the desired bound,
we reduce the entries used in the Cayley operation and represent only
the exact correction by signed amplitudes. We explain both mechanisms
before choosing between them in the proof of
Theorem~\ref{hs:thm:optimal}. Their finite-circuit analyses are in
Appendix~\ref{opt:sec:proof}.

\paragraph{Separating signed action from preparation weight.}
Let the system space be $\mathbb C^N$, and let $\mathcal K$ be the
larger space on which a supplied preparation unitary $U_B$ acts. Write
$E:\mathbb C^N\to\mathcal K$ for the known isometry that appends
zero work registers, and let $\Pi$ project onto a designated success
subspace $\mathcal K_{\rm succ}$. The rectangular map
$B:=\Pi U_BE$ is the successful part of the preparation. It need
not be unitary or have norm one. The full operation remains the unitary
$U_B$, and its inverse is supplied; we do not assume that $B$ itself
can be applied deterministically.

Suppose we can represent the Hamiltonian as
$H=\lambda B^\dagger ZB$, where $\lambda>0$ is known and $Z$ is
a known Hermitian unitary on $\mathcal K_{\rm succ}$. This is the
signed factor used below: $Z$ assigns a plus or minus sign to each
successful component. If $\Pi_+$ and $\Pi_-$ are its two eigenspace
projections, then for every unit system input $\psi$,
\begin{equation}
 \begin{aligned}
 \langle\psi,B^\dagger ZB\psi\rangle
 &=\norm{\Pi_+B\psi}^2-\norm{\Pi_-B\psi}^2,\\
 \langle\psi,B^\dagger B\psi\rangle
 &=\norm{\Pi_+B\psi}^2+\norm{\Pi_-B\psi}^2.
 \end{aligned}
\end{equation}
The difference is the signed action we want, while the sum measures
the total successful preparation weight before the signs cancel.
We assume a supplied bound $B^\dagger B\preceq\delta I$ with
$0<\delta\le1$. Since
$-B^\dagger B\preceq B^\dagger ZB\preceq B^\dagger B$,
we have $\norm H\le\lambda\delta$. We retain the evolution scale
\begin{equation}
 b:=\lambda|t|\delta
\end{equation}
because it governs the uncancelled remainder throughout the following
argument. It is distinct from the access scale $\lambda|t|$:
a small successful weight can reduce $b$ without reducing the
normalisation $\lambda$ in the supplied representation.

We turn this factor into a private unitary by assigning known phases
to the two success signs and to failure. For a chosen real $q\ge1$,
let $K$ act as $qZ$ on success and as zero on failure. In the output
coordinates of $U_B$, its Cayley transform acts by
\begin{equation}
 (K+\ii I)(K-\ii I)^{-1}=
 \zeta_+\Pi_++\zeta_-\Pi_--(I-\Pi),
 \qquad \zeta_\pm:=\frac{\pm q+\ii}{\pm q-\ii}.
\end{equation}
Here the success projections are extended by zero on failure. We implement this
operation by applying the indicated phases to the three known
subspaces, not by running a matrix-inversion algorithm. Applying
$U_B$, then these phases, then $U_B^\dagger$ gives the private unitary
\begin{equation}
 P=U_B^\dagger\bigl[\zeta_+\Pi_++\zeta_-\Pi_--(I-\Pi)\bigr]U_B.
\end{equation}
For example, the component $\Pi_+U_BE\psi$ acquires phase
$\zeta_+$, the component $\Pi_-U_BE\psi$ acquires phase
$\zeta_-$, and the failure component acquires phase $-1$ before
we undo the preparation. The successful components have compressed
weights
\begin{equation}
 E^\dagger U_B^\dagger\Pi_\pm U_BE
       =\frac{B^\dagger B\pm B^\dagger ZB}{2},
 \qquad
 E^\dagger U_B^\dagger(I-\Pi)U_BE=I-B^\dagger B.
\end{equation}
These three positive operators sum to the identity. They are the
connection between the implementable phase operation and the two
system operators that enter its error analysis.

An external control need act only on the known middle phase operation.
On its inactive branch, the unconditional $U_B$ and $U_B^\dagger$
cancel. Replacing the middle phases by their inverses implements the
adjoint. Thus each controlled private call uses two factor queries,
and we do not need controlled access to an otherwise opaque preparation
oracle. To connect this private operation to the Hamiltonian, we now compute
the compressed response used by the boundary construction. The Cayley
identity gives
\begin{equation}
 (I-P)^{-1}=\frac{I+\ii U_B^\dagger K U_B}{2},\qquad
 E^\dagger(I-P)^{-1}E
 =\frac{I+\ii qB^\dagger ZB}{2}
 =\frac{I+\ii H/(\lambda/q)}{2}.
\end{equation}
The failure space contributes zero to $K$, while its success space
contributes $qZ$, which explains the middle equality. This is exactly
the response form needed to obtain a Hamiltonian Cayley step from the
stationary balance equation. Its scale is now $\lambda/q$: increasing
$q$ reduces the total boundary coupling $\lambda|t|/q$, but also
moves the two nontrivial private phases towards one. The appendix
balances these effects with the cancellation order. We use the response
identity here to identify the Hamiltonian encoded by the interaction;
the sharper error analysis below supplies its improved conversion cost.

\paragraph{Cancelling the unwanted terms of finite reuse.}
With this response, the known boundary rotations couple $P$ to the
system. We then use the same general circuit operation as before: prepare several public labels with equal amplitudes, let
them interact successively with a shared private component initially
zero, and unprepare the labels. Different finite reuse lengths give
different selected public blocks. We combine these blocks with known
signed scalar weights.

The desired expansion contains powers of $B^\dagger ZB$, since these
are proportional to powers of $H$. The reuse expansion also contains
$B^\dagger B$, which records successful excursions without their
sign cancellation. The useful algebraic structure is that, at each
order, the scalar coefficient of an ordered product depends on how
many factors of each kind it contains, rather than their order.
For example,
\begin{equation}
 (B^\dagger ZB)(B^\dagger B)
 \quad\hbox{and}\quad
 (B^\dagger B)(B^\dagger ZB)
\end{equation}
can be different operators, but have the same scalar multiplier.
One scalar cancellation condition can therefore remove both.
Appendix~\ref{opt:sec:weights} constructs finitely many weights that
retain the required pure powers of $B^\dagger ZB$ and cancel every
unwanted ordered product below a chosen order. This uses no commutation
between $B^\dagger ZB$ and $B^\dagger B$.

The weights must satisfy an implementation constraint as well as these
algebraic conditions. Preparing a coefficient label, selecting the
corresponding reuse circuit and its sign, and unpreparing the label
implements the weighted sum divided by the sum of the absolute weights.
If that absolute sum grew with the cancellation order, amplification
would introduce an additional cost. The appendix proves an absolute
constant bound on the sum, so the weighted cancellation has constant
amplification overhead.

Cancelling finitely many terms is not yet an error bound. We also need
to control the terms above the cancellation order. To do this,
Appendix~\ref{opt:sec:tail} first studies an analytic approximation to
the weighted finite-reuse blocks. It scales both
$B^\dagger ZB/\delta$ and $B^\dagger B/\delta$ by the same real
parameter $x\in[0,1]$; evaluating at $x=\delta$ returns the physical
operators. The three compressed weights become
\begin{equation}
 I-\frac{x}{\delta}B^\dagger B,\qquad
 \frac{x}{2\delta}(B^\dagger B+B^\dagger ZB),\qquad
 \frac{x}{2\delta}(B^\dagger B-B^\dagger ZB).
\end{equation}
They remain positive and sum to $I$. This positivity bounds the
exponential used to approximate the reuse dynamics by one in norm;
the coefficient averages defining reuse preserve that bound. In
particular, the argument bounds the exponential itself, without
requiring a small norm for its generator.

The use of an exponential in this comparison does not mean that we
must bound its entire infinite series coefficient by coefficient.
A reuse circuit of a fixed length selects only finitely many return
delays: the delay is the number of public reuse positions from the
input position to the output position. Delay zero describes the
immediate public output; delay one describes injection into the
private space followed by return at the next public position. In the analytic return-amplitude
generating function, the coefficient for delay $k$ is a polynomial
of degree at most $k$ in the two scaled system operators. This is the
finite-coefficient property established in Appendix~\ref{opt:sec:tail}.
If $n$ is the largest reuse length in the weighted combination, only
delays $k=0,\ldots,n-1$ enter, so the weighted approximation is a
polynomial in $x$ of degree at most $n-1$. The underlying exponential
comparison supplies its interval bound; finite coefficient extraction
is what makes this particular weighted approximation polynomial.
The appendix extends the bound from $[0,1]$ to a complex circle and
then applies Cauchy's coefficient estimate. Because the low-order coefficients
have already been matched, the resulting estimate concerns only the
remaining high-order terms. Lemma~\ref{opt:lem:tail} bounds that
remainder by a fixed multiple of $(Cb/m)^m$ once the cancellation
order $m$ is a sufficiently large multiple of $b$, with $C$ an
absolute constant. The comparison with the actual finite physical
mesh is estimated separately, so the analytic approximation does not
replace the finite circuit in the final guarantee.

This estimate explains the improved precision denominator. To make
$(Cb/m)^m$ small, we require $m\log(m/(Cb))$ to be a sufficiently
large multiple of $\log(1/\eps)$. For $0<\eps\le1/64$ and the
nontrivial case $b>\eps$, an integer order of size
\begin{equation}
 m=\mathcal O\!\left(b+
       \frac{\log(1/\eps)}{\log(e+\log(1/\eps)/b)}\right)
\end{equation}
suffices. The finite-mesh and amplification estimates require squared
internal error tolerances; these change the logarithms, and hence this
order, only by absolute factors.

After balancing the phase parameter and reuse length, the factor-query
cost is
\begin{equation}\label{eq:methods-signed-cost}
 \mathcal O\!\left(\sqrt{\lambda|t|\,b}
             +\sqrt{\lambda|t|\,r_b}+r_b\right),
 \qquad
 r_b:=\frac{\log(1/\eps)}{\log(e+\log(1/\eps)/b)},
\end{equation}
with $r_0=0$ by continuity. The first two terms couple the access
scale $\lambda|t|$ to the evolution and precision scales $b$ and
$r_b$; the last is the remaining precision cost. Each use of $U_B$
or its supplied inverse counts as one factor query, including the two
calls in each controlled private operation. The guarantee restores
all work registers within the stated error. If $t=0$ or $b\le\eps$,
the identity suffices because
$\norm{e^{-\ii tH}-I}\le|t|\norm H\le b$.

\paragraph{Reducing the entries while retaining their correction.}
Applying the signed factor to all of $H$ will cover one sparsity
regime. At larger sparsity, its whole-matrix remainder scale is too
large for the desired bound. We then keep the Cayley representation,
but first reduce its large entries. Large eigenvalues of the Hermitian
edge operator give Cayley phases close to one; indeed,
\begin{equation}
 \left|\frac{k+\ii}{k-\ii}-1\right|
       =\frac{2}{\sqrt{1+k^2}},\qquad k\in\mathbb R.
\end{equation}
Reducing the largest entry therefore makes the corresponding bare
phases easier to separate from one. We must, however, represent the
removed part as well: discarding it would change the Hamiltonian.

For nonzero $t$, choose the reduction strength
\begin{equation}\label{hs:eq:optimal-parameters}
 h:=\log\!\left(e+\frac{\log(1/\eps)}{\eta|t|}\right),
\end{equation}
and define
\begin{equation}
 (H_c)_{ij}:=\frac{H_{ij}}{1+h|H_{ij}|/\eta},
 \qquad H=H_c+(H-H_c).
\end{equation}
This smooth reduction preserves Hermiticity and the original support.
It gives two estimates with different purposes:
\begin{equation}
 |(H_c)_{ij}|\le\frac\eta h,\qquad
 |(H-H_c)_{ij}|=\frac{h|H_{ij}|^2}{\eta+h|H_{ij}|}
                       \le\frac h\eta|H_{ij}|^2.
\end{equation}
The first improves the Cayley phase separation. Summing the second
along a row and using the Euclidean-norm promise gives
\begin{equation}
 \sum_j|(H-H_c)_{ij}|\le\frac h\eta\sum_j|H_{ij}|^2\le h\eta.
\end{equation}
Thus a stronger entry reduction permits a larger correction, but the
correction's absolute row sum remains controlled.

Let $B_c$ be the entrywise signed factor for this correction, constructed
in Appendix~\ref{opt:sec:core}, with known sign reflection $Z$.
Its successful preparation weight is precisely the absolute row sum
divided by $s\eta$:
\begin{equation}
 H-H_c=s\eta B_c^\dagger ZB_c,\qquad
 B_c^\dagger B_c
 =\frac1{s\eta}\diag_i\!\left(\sum_j|(H-H_c)_{ij}|\right)
 \preceq\frac hs I.
\end{equation}
This identity is an analysis of the coherent entrywise preparation;
it does not require the circuit to calculate a row sum.

To represent $H_c$ and the correction together, let $J_0$ be the uniform row/column embedding into the oriented edge
space, and let $E_{H_c}$ be the Hermitian edge operator with block
\begin{equation}
 \begin{pmatrix}
 0&(H_c)_{ij}/\eta\\
 \overline{(H_c)_{ij}}/\eta&0
 \end{pmatrix}
\end{equation}
on the two orientations of edge $(i,j)$, for $0\le i,j<N$.
Its compression is $J_0^\dagger E_{H_c}J_0=H_c/(s\eta)$.
Let $J_c$ be the full preparation isometry for the correction, with
success component $B_c$ and its failure component retained. Unlike
$B_c$, the full map $J_c$ preserves the input norm.

For this joint construction we choose a phase parameter $q>0$ anew.
On the direct sum of the edge space and the correction preparation
space, use
\begin{equation}
 J\psi:=\frac1{\sqrt2}(J_0\psi\oplus J_c\psi),\qquad
 K:=\operatorname{diag}\bigl(2qE_{H_c},\,
                  2qZ\text{ on success},\,0\text{ on failure}\bigr).
\end{equation}
The direct-sum components specify two orthogonal branches of the
embedding; they are not two separately prepared copies of $\psi$.
Since $J_0$ and $J_c$ are isometries, so is $J$. The factor $2q$
in each block compensates for the probability $1/2$ assigned to that
branch. Consequently,
\begin{equation}
 \begin{aligned}
 J^\dagger KJ
 &=qJ_0^\dagger E_{H_c}J_0+qB_c^\dagger ZB_c\\
 &=\frac{qH_c}{s\eta}+\frac{q(H-H_c)}{s\eta}
  =\frac{H}{s\eta/q}.
 \end{aligned}
\end{equation}
We take $D_0=(K+\ii I)(K-\ii I)^{-1}$. Its compressed
resolvent is therefore $J^\dagger(I-D_0)^{-1}J
=[I+\ii H/(s\eta/q)]/2$, the response for the original $H$.
The two contributions have been added before the boundary interaction
is formed. No product of separately simulated summands is involved,
so their noncommutation causes no product-formula error.

The embedding as written includes the entry-dependent correction
preparation. As with the locator coordinate change, we conjugate its
preparation unitary, and the edge preparation, into the private
operation. The boundary embedding then becomes known, while one
coherent private call performs the two possible entrywise operations
on their respective branches. The constructions in
Appendix~\ref{opt:sec:core} give constant sparse-query cost for that
call, its inverse and its controls.

The joint representation also provides the structure needed for a
short filter. The reduced entries supply phases separated from one;
the correction supplies two known phases whose compressed weights
contain its unknown magnitudes. We choose an even integer $p$ and set $2q=\cot(\pi/p)$, which
makes the two success phases $e^{2\pi\ii/p}$ and
$e^{-2\pi\ii/p}$. The appendix chooses $p$ within the required
range and balances the resulting filter length. Thus the correction
phases are nontrivial roots of unity. If $\zeta^p=1$ and $\zeta\ne1$, then
\begin{equation}
 \frac1p\sum_{j=0}^{p-1}\zeta^j
       =\frac{1-\zeta^p}{p(1-\zeta)}=0.
\end{equation}
A power of this averaging polynomial has a high-order zero at each
correction phase. A second averaging factor suppresses the remaining
phases. These scalar observations determine the filter design; they
do not, on their own, bound the polynomial of the composed private
block, which can be nonnormal. Appendix~\ref{opt:sec:core} proves
that bound using the resolvent and implements the resulting polynomial
by finite reuse. The choice of $h$ balances the entry reduction and
correction weight at the precision scale required below.

\begin{theorem}[Optimal time-independent sparse simulation]
\label{hs:thm:optimal}
Assume the sparse access and row/column norm promises of
Theorem~\ref{hs:thm:main}. For $\eta|t|\ge1$ and
$0<\eps\le1/64$, a unitary circuit achieves error at most
$\eps$, including all work-register outputs, using
\begin{equation}\label{hs:eq:optimal-bound}
 \mathcal O\!\left(\sqrt{s}\,\eta|t|
 +\sqrt{\frac{s\eta|t|\log(1/\eps)}
                   {\log(e+\log(1/\eps)/(\eta|t|))}}
 +\frac{\log(1/\eps)}{\log(e+\log(1/\eps)/(\eta|t|))}\right)
\end{equation}
primitive sparse queries. The guarantee includes an untouched
reference. For $s\ge32$, allowing the dimension to grow, the
worst-case complexity is
\begin{equation}\label{hs:eq:optimal-theta}
 \Theta\!\left(\sqrt{s}\,\eta|t|
 +\sqrt{\frac{s\eta|t|\log(1/\eps)}
                   {\log(e+\log(1/\eps)/(\eta|t|))}}
 +\frac{\log(1/\eps)}{\log(e+\log(1/\eps)/(\eta|t|))}\right).
\end{equation}
\end{theorem}

\begin{proof}
The desired bound replaces the baseline precision logarithm by
$\log(1/\eps)/h$, where $h$ is defined in
Eq.~\eqref{hs:eq:optimal-parameters}. We choose among the baseline,
the whole-matrix signed factor, and the joint reduced-entry
representation according to which achieves that scale.

\emph{When the precision denominator is bounded.}
If $h$ is at most a fixed absolute threshold, division by $h$ changes
the precision terms only by absolute factors. The baseline theorem
already has the required query order, so no additional cancellation
is needed in this regime.

\emph{When a signed factor of the entire matrix suffices.} For larger $h$, the comparison depends on whether a signed factor
of the whole matrix already has a sufficiently small remainder.
The native factor construction in Appendix~\ref{opt:sec:proof}
gives
\begin{equation}
 H=s\eta B^\dagger ZB,\qquad
 B^\dagger B\preceq s^{-1/2}I,
\end{equation}
with constant sparse-query cost per factor query. Here the successful preparation weight is the absolute row sum divided
by $s\eta$. Cauchy--Schwarz gives
$\sum_j|H_{ij}|\le\sqrt{s}(\sum_j|H_{ij}|^2)^{1/2}
\le\sqrt{s}\eta$, which explains the bound $s^{-1/2}$.
The native preparation and its inverse load individual entries, apply
the entry-controlled local factor, and erase the value registers;
Lemma~\ref{access:lem:control} supplies the locator controls.
The controlled private calls again enable only the known middle
phase, so they retain constant sparse-query cost without a new
preparation-oracle assumption. The access scale is $s\eta|t|$
and the remainder scale is
$b=\sqrt{s}\eta|t|$. We use this construction when
$s\eta|t|\le h\log(1/\eps)$. To verify that this condition makes
the remainder small enough, put $x:=\log(1/\eps)/(\eta|t|)$
for this comparison, so $h=\log(e+x)$ and $s\le hx$. Since
$b=\sqrt{s}\eta|t|$, we obtain
\begin{equation}
 \frac{\log(1/\eps)}b=\frac{x}{\sqrt{s}}
       \ge\sqrt{\frac{x}{h}},\qquad
 \frac{b}{\log(1/\eps)/h}=\frac{h\sqrt{s}}x
       \le\frac{h^{3/2}}{\sqrt{x}}.
\end{equation}
For sufficiently large $h$, the logarithm of $e+\sqrt{x/h}$ has
order $h$, while $h^{3/2}/\sqrt{x}$ is bounded. Also
$x/\sqrt{s}\le x$ because $s\ge1$. These upper and lower bounds
on the logarithm give
\begin{equation}
 r_b=\Theta\!\left(\frac{\log(1/\eps)}h\right),\qquad
 b=\mathcal O\!\left(\frac{\log(1/\eps)}h\right).
\end{equation}
Thus both terms under the square roots in the signed-factor cost are
controlled by the desired precision scale.
Substituting into Eq.~\eqref{eq:methods-signed-cost} gives
\begin{equation}
 \mathcal O\!\left(
       \sqrt{\frac{s\eta|t|\log(1/\eps)}h}
       +\frac{\log(1/\eps)}h\right).
\end{equation}
\emph{When we must reduce the entries first.}
For the complementary regime
$s\eta|t|\ge h\log(1/\eps)$, we use the joint representation of
the reduced entries and their exact correction. Proposition~\ref{opt:prop:core}
in Appendix~\ref{opt:sec:core} supplies the finite filter and proves
query cost $\mathcal O(\sqrt{s\eta|t|\log(1/\eps)/h})$.
The absolute row weight of the correction, rather than that of the
whole Hamiltonian, is what makes this construction useful here.
Both appendix constructions are actual finite unitary circuits:
their error estimates include coefficient preparation, the finite
physical mesh and amplification. Combining their bounds with the
bounded-$h$ baseline proves Eq.~\eqref{hs:eq:optimal-bound}.

Theorem~\ref{hsl:thm:main} supplies the matching lower bound in
the stated range. Its hard matrices have maximum row and column Euclidean
norms exactly one, so rescaling them by $\eta$ gives the present
promise at normalised time $\eta|t|$. Their stronger spectral-norm
normalisation is not an assumption needed by the upper bound.
\end{proof}

Relative to the baseline, the mixed precision term improves by
$\sqrt h$ and the additive term by $h$.
Corollary~\ref{hs:cor:hybrid} had already improved the additive
logarithm, but with $s\eta|t|$ inside its denominator; the new
theorem replaces that scale by $\eta|t|$ and improves the mixed
term at the same time. This is a joint-parameter improvement:
at fixed sparsity and positive time, the two additive terms have
the same asymptotic precision order. To see a regime where the
overall costs differ, take $\eta|t|=1$ and
$s=\lceil\log(1/\eps)\rceil$. As $\eps\to0$, the hybrid
bound is $\mathcal O(\log(1/\eps))$, whereas the new bound is
\begin{equation}
 \mathcal O\!\left(
       \frac{\log(1/\eps)}{\sqrt{\log\log(1/\eps)}}\right).
\end{equation}
The finite query constructions allow large known operations. We give
their gate implementation below, retaining this query order under
explicit scalar-arithmetic assumptions. Extending the improved joint
precision dependence to changing Hamiltonians remains open.

We can further improve the time-independent query bound if we also
know a maximum-entry bound. This supplies information different
from the Euclidean row and column norms: the entry bound controls
the full edge operator, while the Euclidean norms control the
compressed second moment seen by an embedded system input. To
remove redundant promises before applying this distinction, let
the supplied positive bounds be $\eta_0$ and $\mu_0$, and take
\begin{equation}
 \eta:=\min\{\eta_0,\sqrt{s}\mu_0\},\qquad
 \mu:=\min\{\mu_0,\eta\}.
\end{equation}
We can tighten the Euclidean bound using $s$ entries of magnitude at
most $\mu_0$, and tighten the entry bound using the Euclidean bound
on its row. These operations use no matrix query. To check compatibility
after both changes, if $\mu=\mu_0$ then
$\eta\le\sqrt{s}\mu_0=\sqrt{s}\mu$ by construction; if
$\mu=\eta$ then $\eta\le\sqrt{s}\mu$ since $s\ge1$.
The resulting valid bounds therefore satisfy
\begin{equation}\label{hs:eq:two-norm-promises}
 \norm H_{\max}\le\mu\le\eta\le\sqrt{s}\mu,
\end{equation}
where $\norm H_{\max}:=\max_{i,j}|H_{ij}|$. The sparse lists
and access model are unchanged.

\begin{theorem}[Two-norm time-independent sparse simulation]
\label{hs:thm:two-norm}
In addition to the assumptions of Theorem~\ref{hs:thm:main}, supply
an entry bound $\mu$ satisfying Eq.~\eqref{hs:eq:two-norm-promises}.
For $\eta^2|t|/\mu\ge1$ and $0<\eps\le1/64$, there is a
finite unitary circuit with error at most $\eps$, including all
work-register outputs, using
\begin{equation}\label{hs:eq:two-norm-bound}
 \mathcal O\!\left(\sqrt{s}\,\eta|t|
 +\sqrt{\frac{s\mu|t|\log(1/\eps)}
                   {\log(e+\mu\log(1/\eps)/(\eta^2|t|))}}
 +\frac{\log(1/\eps)}
       {\log(e+\mu\log(1/\eps)/(\eta^2|t|))}\right)
\end{equation}
primitive sparse queries, with absolute constants; the guarantee
includes an untouched reference.
\end{theorem}

\begin{proof}
The two bounds affect different parts of the construction. An individual
entry controls the largest eigenvalue of its two-orientation edge
block; the squared entries summed along a row or column control the
second moment on an embedded system state. We retain both estimates
on the actual sparse-oracle registers before applying the precision
analysis. This is essential: simply replacing $s$ by a smaller number
in an oracle specification would not give a valid algorithm. Define the time scale
\begin{equation}\label{hs:eq:two-norm-parameters}
 \theta:=\frac{\eta^2|t|}{\mu}
\end{equation}
and the ratio of the entry and Euclidean-norm information
\begin{equation}\label{hs:eq:two-norm-rescaling}
 \sigma:=\frac{s\mu^2}{\eta^2}\in[1,s],\qquad
 \sqrt\sigma\,\theta=\sqrt{s}\,\eta|t|,\qquad
 \sigma\theta=s\mu|t|.
\end{equation}
To see why these are the relevant parameters, normalise each
entry by $\mu$ rather than $\eta$. Let $E_H$ be the Hermitian
edge operator whose orientation block at $(i,j)$ is
\begin{equation}
 \begin{pmatrix}0&H_{ij}/\mu\\\overline{H_{ij}}/\mu&0\end{pmatrix},
\end{equation}
and let $J_0$ be the uniform embedding into those edge coordinates
through the original row and column lists. The embedding still uses the original $s$ list positions; we do not
remove or search for any entries. On a system basis state $i$, the
second moment is the average of the squared entries in row $i$ and
column $i$, each normalised by $\mu^2$. In particular,
\begin{equation}
 \bra{i}J_0^\dagger E_H^2J_0\ket{i}
 =\frac{\sum_j|H_{ij}|^2+\sum_j|H_{ji}|^2}{2s\mu^2}
 \le\frac{\eta^2}{s\mu^2}.
\end{equation}
The squared edge operator is diagonal in the edge labels, so the
compression has no off-diagonal system entries. Combining this with
the first-moment calculation of Eq.~\eqref{hs:eq:moments} gives
\begin{equation}
 \norm{E_H}\le1,\qquad
 J_0^\dagger E_HJ_0=\frac{H}{s\mu},\qquad
 J_0^\dagger E_H^2J_0\preceq
       \frac{\eta^2}{s\mu^2}I=\frac1\sigma I.
\end{equation}
Rescaling to $\sqrt\sigma E_H$ now gives
\begin{equation}
 J_0^\dagger(\sqrt\sigma E_H)J_0=\frac{H}{\sqrt{s}\eta},
 \qquad
 J_0^\dagger(\sqrt\sigma E_H)^2J_0\preceq I,
 \qquad \norm{\sqrt\sigma E_H}\le\sqrt\sigma.
\end{equation}
The first two are exactly the baseline moment relations: the desired
public action and the bound on its stationary private weight are
unchanged. The full-norm bound is now $\sqrt\sigma$ instead of
$\sqrt{s}$. This improves the private spectral control entering the
precision analysis.

Appendix~\ref{twonorm:sec:proof} verifies the signed-factor and
joint reduced-entry constructions on these same native registers,
including the correction's phase weights. Their precision analyses
then use $\sigma$ and $\theta$ in place of the single-bound
parameters and give
\begin{equation}
 \mathcal O\!\left(
 \sqrt\sigma\,\theta+
 \sqrt{\frac{\sigma\theta\log(1/\eps)}
                  {\log(e+\log(1/\eps)/\theta)}}
 +\frac{\log(1/\eps)}{\log(e+\log(1/\eps)/\theta)}\right).
\end{equation}
Substitution of Eq.~\eqref{hs:eq:two-norm-rescaling} is exactly
Eq.~\eqref{hs:eq:two-norm-bound}. The parameter $\sigma$ may
be real: it describes the moment estimates, not the number of
entries returned by a locator. The appendix checks the finite
circuits with the original $s$-position lists before making this
analytic change of parameters.
\end{proof}

The endpoints show how the two promises interact. When $\mu=\eta$,
we recover Theorem~\ref{hs:thm:optimal}. When
$\eta=\sqrt{s}\mu$, we have $\sigma=1$. In the rescaled bound,
the mixed term is the square root of the product of the time term and
the precision quotient. The inequality $2\sqrt{ab}\le a+b$ for
$a,b\ge0$ absorbs it into their sum, giving
\begin{equation}\label{hs:eq:two-norm-flat}
 \mathcal O\!\left(\theta+
       \frac{\log(1/\eps)}{\log(e+\log(1/\eps)/\theta)}\right).
\end{equation}
For fixed $\eta$ and $t$, decreasing $\mu$ increases $\theta$
and changes the precision denominator as well as the coefficient
in the mixed term. Replacing $\eta$ by $\mu$ only in that
coefficient would therefore give the wrong bound. Since
$\mu\le\eta$, the condition $\eta|t|\ge1$ implies
$\theta\ge1$, but the two-norm theorem covers its own wider
time range. Its matching lower bound is established on the
aligned parameter families of Proposition~\ref{hsl:prop:two-norm};
we do not assert a match for every real pair of norm promises.

Finally, the improved precision and the additional entry promise
can be useful together. For example, take $\eta=t=1$ and
\begin{equation}
 s=\left\lceil[\log(1/\eps)]^2\right\rceil,
 \qquad \mu=[\log(1/\eps)]^{-1/2},\qquad \eps\to0.
\end{equation}
Then $\theta=\sqrt{\log(1/\eps)}$, so the precision quotient
is $\Theta(\log(1/\eps)/\log\log(1/\eps))$. The baseline
construction with both norm bounds costs
$\mathcal O([\log(1/\eps)]^{5/4})$, while the optimal
single-bound theorem costs
$\mathcal O([\log(1/\eps)]^{3/2}/\sqrt{\log\log(1/\eps)})$.
Combining both refinements gives
\begin{equation}\label{hs:eq:two-norm-example}
 \mathcal O\!\left(
       \frac{[\log(1/\eps)]^{5/4}}{\sqrt{\log\log(1/\eps)}}\right).
\end{equation}
This example distinguishes the two sources of improvement: the entry
promise reduces the mixed-term scale, and the refined filter supplies
the additional precision denominator. Both are time-independent
improvements. We next state the gate guarantee for these algorithms,
then apply the simulation bounds to implementing a unitary from its
matrix entries.

\subsection{Gate implementation of improved-precision simulation}
\label{gate:sec:result}

The two preceding theorems reduce the number of matrix queries, but
this alone does not bound the cost of their known operations. Their
chronological constructions can use many more physical time intervals
than queries, and their signed combinations require nontrivial
coefficient states. We now show that both constructions retain their
query order with a polynomial overhead in the relevant bit lengths.
The proof in Appendix~\ref{gate:sec:proof} has two implementation
tasks. It compresses the chronological boundary network using a binary
hierarchy, so the number of physical intervals appears through its
register width rather than as a multiplicative gate cost. It also
prepares the signed coefficient state once per linear-combination
circuit and shares one coherent sequence of private calls across all
reuse lengths. We explain the common boundary and reuse mechanisms
for the baseline in Section~\ref{hs:sec:gates}; the appendix verifies
that both improved-precision constructions admit the same compilation.

We count gates outside the supplied value and locator oracles. Assume
that the bit strings encoding the entries and the public parameters have
finite descriptions of at
most $b_{\rm data}$ bits. The decoder and reversible scalar arithmetic
must evaluate the normalised entries, elementary real functions and
local rotation parameters specified in Section~\ref{gate:sec:native}
to accuracy $\xi$ using gates and clean work polynomial in
$b_{\rm data}$, the index widths and $\log(1/\xi)$. The bit lengths
of the scale parameters are included, and computed arithmetic is
uncomputed. These assumptions concern calculations on a single loaded entry or a
known scalar parameter. They do not supply a simulation oracle: the
boundary, reuse, coefficient-selection and control circuits still have
to be built, and their costs are included in the theorem. The
represented matrix is the target, so any discrepancy from an intended
physical Hamiltonian requires a separate representation-error budget.

For all finite-gate bounds below, fix a finite universal set with the
following properties. It implements NOT, CNOT and Toffoli exactly at
constant cost, and the adjoint of each gate is available at constant
cost, so reversing a compiled circuit costs at most a constant factor.
An arbitrary one-qubit unitary can be synthesised to error $\zeta$
with $\mathcal O(\log(1/\zeta))$ gates for $0<\zeta\le1/2$~\cite{KMM2013}.
A one-qubit unitary with at most three controls (on specified bits
equal to zero or one) has an exact decomposition into a bounded
number of gates from the set and uncontrolled one-qubit unitaries~\cite{Barenco1995}.
Both synthesis and decomposition use at most a constant number of
clean ancillas. These assumptions let us decompose known controls
before allocating synthesis error to the remaining rotations.

\begin{theorem}[Gate-efficient improved-precision simulation]
\label{gate:thm:main}
Assume the access and compatible norm bounds of
Theorem~\ref{hs:thm:two-norm}, the computational assumptions above,
$\eta^2|t|/\mu\ge1$ and $0<\eps\le1/64$. Define the query scale
\begin{equation}\label{gate:eq:query-scale}
 \mathcal Q_\mu:=\sqrt{s}\,\eta|t|
 +\sqrt{\frac{s\mu|t|\log(1/\eps)}
 {\log(e+\mu\log(1/\eps)/(\eta^2|t|))}}
 +\frac{\log(1/\eps)}{\log(e+\mu\log(1/\eps)/(\eta^2|t|))}.
\end{equation}
Under the finite-gate assumptions above, a unitary circuit satisfies
\begin{equation}\label{gate:eq:full}
 \sup_{\norm\psi=1}
 \norm{\mathcal A(\psi\otimes\ket0)
       -(e^{-\ii tH}\psi)\otimes\ket0}\le\eps
\end{equation}
with $\mathcal O(\mathcal Q_\mu)$ primitive sparse queries and
\begin{equation}\label{gate:eq:main-gates}
 G_{\rm additional}
 =\mathcal Q_\mu\,\operatorname{poly}\!\left(
 b_{\rm data},\log(N+1),\log(s+1),
 \log\frac{\mathcal Q_\mu+1}{\eps}\right)
\end{equation}
additional quantum gates. Its auxiliary-qubit count is polynomial in
the same arguments, without the leading factor $\mathcal Q_\mu$.
The error includes every work register and an untouched reference.
For injective padded lists, $s\le N$, so the $\log(s+1)$ argument
can be absorbed into $\log(N+1)$.
\end{theorem}

The leading factor $\mathcal Q_\mu$ is the number of private calls up
to an absolute constant. The polynomial overhead accounts for the bits
needed to describe and process a call, and for the accuracy needed to
keep the accumulated implementation error below $\eps$. Workspace
has no leading $\mathcal Q_\mu$ factor because these registers can
be reused between calls.

Taking $\mu=\eta$ gives the gate implementation of
Theorem~\ref{hs:thm:optimal} for $\eta|t|\ge1$. Thus both
precision-dependent improvements survive the passage from query
complexity to quantum gates. The polynomial depends on the chosen
decoder and arithmetic; we do not assert optimal logarithmic factors.
Classical compilation is a separate resource: the signed coefficient
tables have finite approximations, but we prove no well-conditioned
polynomial-time classical algorithm for producing them. These
improved-precision gate guarantees are time independent.

The accompanying Lean source constructs the circuit for general $N$ and
$s\le N$, including its locator controls and stable signed factors, and
bounds its error on all output registers. For a public decoder the
construction takes reversible scalar arithmetic and local gate synthesis
as inputs. Appendix~\ref{app:formalisation} specifies these assumptions,
including the phase accuracy required for small nonzero entries.

\subsection{Implementing a unitary from its entries}\label{hs:sec:unitary}

We can use Hamiltonian simulation to apply an unknown unitary when its
entries, rather than a circuit for the unitary itself, are supplied.
The key observation is that every row and column of an $N\times N$
unitary has Euclidean norm one, even when all $N$ entries are nonzero.
Thus the Euclidean-norm promise remains favourable in the dense case.
To exploit it, we place the unitary and its adjoint in the off-diagonal
blocks of a Hermitian matrix. Evolution under this matrix for a known
constant time transfers the system between two flag sectors and applies
the desired unitary during that transfer. This reduction was identified
by Berry and Childs~\cite[Section~8]{BerryChilds}; the improvement
here comes from the simulation cost under the Euclidean-norm promise.
The independent construction of Bravo-Prieto, Harrow and Kothari
also gives constant-error cost $\mathcal O(\sqrt N)$ and general
precision cost $\mathcal O(\sqrt N\log(1/\eps))$
\cite[Corollary~3]{BPHK26}.

\begin{corollary}[Black-box unitary implementation]\label{hs:cor:unitary}
Let $U\in\mathbb C^{N\times N}$ be unitary, $N\ge2$, and
$0<\eps\le1/2$. Suppose a reversible entry-value oracle supplies
its entries exactly or with entrywise absolute error at most
$\eps/(\pi N)$. A unitary circuit achieves error on all output registers
\begin{equation}\label{hs:eq:unitaryerror}
 \norm{\mathcal A_U-U}\le\eps
\end{equation}
using
\begin{equation}\label{hs:eq:unitarybound}
 \mathcal O\!\left(\sqrt{N\log(1/\eps)}+
       \frac{\log(1/\eps)}{\log(e+\log(1/\eps)/N)}\right)
\end{equation}
entry-oracle calls, including inverses. No location oracle is needed.
The guarantee includes an untouched reference, and the constant-error
worst-case complexity is $\Theta(\sqrt N)$.
\end{corollary}
\begin{proof}
We first treat exact entries, so that the represented matrix is itself
unitary. Introduce a two-dimensional flag and the Hermitian operator
on $\mathbb C^2\otimes\mathbb C^N$

\begin{equation}\label{hs:eq:unitarydilation}
 H_U:=\begin{pmatrix}0&U\\U^\dagger&0\end{pmatrix}.
\end{equation}
Multiplying the two blocks gives
$H_U^2=\operatorname{diag}(UU^\dagger,U^\dagger U)=I$.
The even powers in the exponential therefore sum to $\cos(t)I$
and the odd powers to $-\ii\sin(t)H_U$. At time $\pi/2$ this
becomes
\begin{equation}\label{hs:eq:unitarytransfer}
 e^{-\ii(\pi/2)H_U}=-\ii H_U,\qquad
 e^{-\ii(\pi/2)H_U}(\ket1\otimes\psi)
                  =-\ii\ket0\otimes U\psi.
\end{equation}
Thus known flag and phase corrections yield a clean implementation
of $U$. We can simulate the dilation using the entry oracle for $U$
alone. Unitarity makes each row and column of $H_U$ have Euclidean
norm one. We give each row and column a list of all $N$ indices in
the opposite block, including any zero entries. These complete lists
and their inverse permutations are known independently of $U$, so no
support-dependent locator is needed.

For a value query, the two diagonal blocks return zero. In the
upper-right block we query the encoding of $U_{jk}$; in the lower-left
block we query the encoding of $U_{kj}$ and conjugate its decoded value.
Concretely, we retain the block flag together with the queried bit
string and let the known decoder return zero, the original decoded
value, or its complex conjugate as appropriate. The bit string need
not use fixed-point encoding: the decoding map in
Section~\ref{hs:sec:control} can be any supplied map of the stated
kind. Conjugation is performed on its decoded value, so no oracle for
a separately conjugated encoding is required. Known routing and coherent controls give
constant entry-query cost, with $s=N$ and $\eta=1$. Applying Corollary~\ref{hs:cor:hybrid} with
$s=N$, $\eta=1$ and $t=\pi/2$ gives
Eq.~\eqref{hs:eq:unitarybound}. Its separate $\sqrt N$ term is
absorbed by $\sqrt{N\log(1/\eps)}$, since
$\log(1/\eps)\ge\log2$.

For approximate entries, the represented matrix $\widetilde U$ need
not be unitary. We therefore compare its Hermitian evolution with that
of the exact dilation, rather than use the identity
$\widetilde H_U^2=I$, which need not hold. Let
$\xi:=\max_{i,j}|\widetilde U_{ij}-U_{ij}|$.
Each absolute row and column sum of $\widetilde U-U$ is at most
$N\xi$, so $\norm{\widetilde U-U}\le N\xi$. Hermitian dilation
preserves this operator norm. Its sparsity remains at most $N$, giving
\begin{equation}
 \norm{\widetilde H_U-H_U}
 =\norm{\widetilde U-U}\le N\xi\le\eps/\pi,
 \qquad \norm{\widetilde H_U}\le2.
\end{equation}
Every row or column Euclidean norm is at most the operator norm of
the dilation. We may therefore use the valid bound $\eta=2$ for
the represented problem, even though its matrix need not be exactly
unitary.
Duhamel's estimate at the same physical time gives
\begin{equation}
 \norm{e^{-\ii(\pi/2)\widetilde H_U}
             -e^{-\ii(\pi/2)H_U}}
 \le(\pi/2)N\xi\le\eps/2.
\end{equation}
Simulation error $\eps/2$ accounts for the remaining budget and
changes the query order only by a constant factor. Both errors are
operator-norm bounds and hence persist with an untouched reference.

For the matching constant-error lower bound, let $m$ be the unique
marked element of $\mathbb Z_N$ and set
$(U_m)_{jk}=\mathbf1_{j+k=m\ ({\rm mod}\ N)}$.
Each column has its unique one at the row $m-k\pmod N$, so this
is a permutation matrix and $U_m\ket0=\ket m$. To answer an entry
query at $(j,k)$, we compute $j+k\pmod N$ and query whether that
label is marked; the resulting Boolean value is precisely the matrix
entry. One matrix-entry query therefore uses one search query. A constant-error
implementation would solve unstructured search, requiring
$\Omega(\sqrt N)$ queries~\cite[Section~9]{BerryChilds}.
\end{proof}

The optimal time-independent theorem improves the precision dependence
of the same reduction without changing its access requirements.

\begin{corollary}[Refined precision for unitary implementation]
\label{hs:cor:unitary-optimal}
Under the assumptions of Corollary~\ref{hs:cor:unitary}, the same
guarantee on all output registers is achievable with
\begin{equation}\label{hs:eq:unitary-optimal-bound}
 \mathcal O\!\left(
 \sqrt{\frac{N\log(1/\eps)}{\log(e+\log(1/\eps))}}
 +\frac{\log(1/\eps)}{\log(e+\log(1/\eps))}\right)
\end{equation}
entry-oracle calls, including inverses and no locators. The
constant-error worst-case complexity remains $\Theta(\sqrt N)$.
\end{corollary}
\begin{proof}
We apply Theorem~\ref{hs:thm:optimal} to $H_U$, or to
$\widetilde H_U$ for inexact entries, using the valid bound
$\eta=1$ or $2$, respectively, and time $\pi/2$.
Choose simulation error $\delta:=\min\{\eps/2,1/64\}$.
The preceding representation-error estimate contributes at most
$\eps/2$ when needed. In both cases $\eta|t|$ is a fixed constant
at least one, and
\begin{equation}\label{hs:eq:unitary-optimal-parameters}
 \log(1/\delta)
 =\max\{\log(1/\eps)+\log2,\log64\}
 =\Theta(\log(1/\eps)).
\end{equation}
The theorem's precision quotient therefore has order
$\log(1/\eps)/\log(e+\log(1/\eps))$. This quotient is bounded
below by a positive absolute constant for $\eps\le1/2$, so the
separate $\sqrt N$ term is absorbed by the first term of
Eq.~\eqref{hs:eq:unitary-optimal-bound}. The flag and phase corrections
and the constant-error search lower bound are unchanged.
\end{proof}

\subsection{Baseline circuit compression and gate complexity}
\label{hs:sec:gates}

We return to the baseline circuit to explain how its small query count
can be accompanied by an efficient implementation of the known gates.
Although a time slot costs no additional query when the private operations are selected coherently, a literal
implementation still applies one known boundary rotation per slot.
Thus one complete transducer contains $M$ boundary rotations. A signed reuse circuit makes at most $2d$
transducer calls, and one amplification round uses that circuit, its
inverse and the circuit again. The direct construction therefore uses
at most $6dM$ boundary rotations. Since the mesh can grow polynomially
with inverse error, this gate count can do so as well, even though
the time-slot register needs only $\mathcal O(\log M)$ qubits.
We now compress the ordered boundary network using the identity of
Chen, Gao, Ji, Li, Wang and Zhou~\cite{CGJLWZ}, and implement the
reuse labels without expanding the physical mesh.

Under efficient binary decoding, polynomial-cost reversible
arithmetic and the gate-set assumptions of Section~\ref{gate:sec:result},
the baseline time-independent simulation retains its query order and
uses
\begin{equation}\label{hs:eq:gate-headline}
 \begin{split}
 &\left[\sqrt{s}\,\eta|t|
       +\sqrt{s\eta|t|\log(1/\eps)}+\log(1/\eps)\right]\!\times\!{}\\
 &\hspace{1cm}\operatorname{poly}\!\left(
       b_{\rm data},\log(N+1),
       \log\frac{s\eta|t|+2}{\eps}\right)
 \end{split}
\end{equation}
additional gates. Here $b_{\rm data}$ bounds the bit lengths of
the entries and known parameters, whose decoding and arithmetic
costs are polynomial in these displayed arguments. The guarantee
has error on all output registers $\eps$, including an untouched reference;
it excludes gates internal to the supplied matrix oracles.
Theorem~\ref{hs:thm:gates} below gives the precise accounting in
terms of the decoder, arithmetic and synthesis costs, including
the case of a decoder without this polynomial-cost guarantee.
This section proves the baseline implementation in terms of its private
arithmetic costs. Appendix~\ref{gate:sec:proof} reuses the boundary
compression and finite-error argument to establish
Theorem~\ref{gate:thm:main} at the improved-precision query bound.

Our implementation has three parts: compress the boundary network
inside one transducer call, implement the selected reuse lengths with
a separate clock, and allocate finite gate errors across the resulting
circuit. We describe the operations first and then count their cost.
Write $T=|t|$ and take a power-of-two number $M=2^{m_{\rm slot}}$ of equal
physical time slots. In this subsection, $m_{\rm slot}$ is the width of that
binary slot register; it is distinct from the cancellation order used
in Section~\ref{hs:sec:optimal}. In list coordinates the known boundary preparation has the
form $J\psi=U_J(\ket{0^a}\otimes\psi)$; let $C_J$ be the gate
count of $U_J$. For now, known single-qubit rotations count as single
elementary gates. The precise mesh and synthesis accuracies will be
chosen after the exact circuit is described. Negative time uses the
adjoint of the constructed circuit.

\paragraph{Compressing consecutive boundary rotations.}
Conjugating private coordinates by $U_J^\dagger$ leaves a public
mode $p$ and private boundary modes $0,\ldots,M-1$, each carrying
the same system space. The remaining private directions are fixed.
With $w:=\alpha T/M$ and $c:=(1-w/2)/(1+w/2)$, the chronological
boundary network is
\begin{equation}
 \mathcal G_M:=G_{M-1}\cdots G_0,\qquad
 G_j\big|_{(p,j)}=R(c):=
 \begin{pmatrix}c&-\sqrt{1-c^2}\\\sqrt{1-c^2}&c\end{pmatrix}.
\end{equation}
The $M$ rotations share the public mode, so they cannot simply be
performed simultaneously. The useful structure is instead that a
consecutive group communicates with that mode through only one private
input direction and one private output direction. These directions
differ because the rotations are ordered. We can align each with a
single coordinate using separate input and output changes of basis.

Consider the reduction for two consecutive private modes. In this
paragraph write $\beta_{\rm bdy}:=\sqrt{1-c^2}$ for their boundary rotation
coefficient. In the ordered basis $(p,0,1)$, the product is
\begin{equation}
 G_1G_0=
 \begin{pmatrix}
 c^2&-c\beta_{\rm bdy}&-\beta_{\rm bdy}\\
 \beta_{\rm bdy}&c&0\\
 c\beta_{\rm bdy}&-\beta_{\rm bdy}^2&c
 \end{pmatrix}.
\end{equation}
A public input leaves private amplitude proportional to $(1,c)$,
whereas a private input returns public amplitude through the row
$(c,1)$. Both have squared length $1+c^2$. Choose the private basis
changes
\begin{equation}
 B_{\rm in}(c):=\frac1{\sqrt{1+c^2}}
       \begin{pmatrix}c&-1\\1&c\end{pmatrix},\qquad
 B_{\rm out}(c):=\frac1{\sqrt{1+c^2}}
       \begin{pmatrix}1&-c\\c&1\end{pmatrix}.
\end{equation}
Their first columns are the normalised return and injection directions.
Multiplication then gives
\begin{equation}
 (1\oplus B_{\rm out}(c)^\dagger)G_1G_0
       (1\oplus B_{\rm in}(c))
 =\begin{pmatrix}
 c^2&-\sqrt{1-c^4}&0\\
 \sqrt{1-c^4}&c^2&0\\
 0&0&1
 \end{pmatrix}.
\end{equation}
One private direction is now inactive, and the other couples to the
public mode by a single rotation with attenuation $c^2$. We have
preserved the ordered product exactly. In particular, using different
bases on the two sides is essential; this is a factorisation of the
network, not a claim that its individual rotations commute.

The same reduction combines two already compressed groups. A group
of $n$ slots has public attenuation $c^n$, so the next level uses
$B_{\rm in}(c^n)$ and $B_{\rm out}(c^n)$ and leaves one rotation
with attenuation $c^{2n}$. For $M=2^{m_{\rm slot}}$, a binary hierarchy leaves
one public/private rotation with attenuation $c^M$, surrounded by
the two changes of private basis. On the binary slot register, the
identical operations for all groups at a given level are implemented
coherently by one controlled rotation on the relevant bit. They do not
require a separate gate sequence for each group. The controls test
that the lower bits select the retained coordinate; their predicates
can be computed and erased with shared clean work.

Lemma~\ref{hs:lem:boundary} proves this factorisation and its linear
cost in the register width, including the inverse and externally
controlled networks. We use it here with boundary preparation $U_J$ and $M=2^{m_{\rm slot}}$ physical
slots. Its proof in Appendix~\ref{gate:sec:compiler} supplies the
recursive circuit and shared predicate construction. Including the
change to boundary coordinates and back, the cost is
$\mathcal O(C_J+a+m_{\rm slot})$ gates and $\mathcal O(a+m_{\rm slot})$ clean work qubits.
The reduction concerns the boundary network alone and imposes no
involution assumption on the private unitary. This is why the same
compression also serves both improved-precision constructions.

In the original private coordinates, a complete transducer is now
\begin{equation}\label{hs:eq:gate-transducer}
 S_{\rm tot}=\mathcal G_M^{\rm original}
               \bigl(I_{\rm pub}\oplus(I_M\otimes D_0)\bigr).
\end{equation}
The factorisation follows from commuting a private operation past
boundary rotations on other slots, exactly as in
Eq.~\eqref{ov:eq:chronological-circuit}. Thus each complete call uses
one coherent private Cayley sequence and the compressed known network.
Its inverse reverses the actual circuit.

\paragraph{Implementing finite reuse with a separate clock.}
The boundary network is now inexpensive even for a fine physical
mesh. Our next task is to implement the weighted reuse circuits used
by the conversion filter. We distinguish two labels throughout this
construction. The physical slot labels the short Cayley steps inside
one complete interaction; the reuse label identifies which public
input component meets that complete interaction next. Refining the
physical evolution and increasing the number of reuses are different
operations, so these labels need separate registers. Let
$f(z)=g_{b_{\rm av}}(z)^k=\sum_{j=0}^d a_jz^j$ be the baseline filter, with
$b_{\rm av}=32Q$ and $d=k(b_{\rm av}-1)$. Its nonnegative coefficients sum to one and
satisfy $a_j=a_{d-j}$; set them to zero outside $0\le j\le d$.
The signed reuse weights
$\lambda_n=(n/d)(a_{n-d}-a_n)$ have absolute sum $2-2a_0$.
We prepare
\begin{equation}\label{hs:eq:gate-filter-state}
 \sum_{n=1}^{2d}\sqrt{|\lambda_n|/2}\ket n
                 +\sqrt{a_0}\ket\perp.
\end{equation}
The ordinary branches have total squared amplitude
$\sum_n|\lambda_n|/2=1-a_0$, and the dummy branch carries the
remaining probability $a_0$. This is therefore a unit vector. Its
$\mathcal O(d)$ amplitudes can be prepared
with $\mathcal O(d)$ rotations and CNOTs and
$\mathcal O(\log d)$ qubits~\cite{MottonenStatePrep}.
Finite convolution computes the rational $a_j$, whose common
denominator is $b_{\rm av}^k$; the classical table has $\mathcal O(d)$ entries
independently of $M$. Classical preprocessing is not counted as a
quantum-gate cost.

For an ordinary label $n$, we select its reuse circuit and known sign.
The unused coefficient mass must have zero selected public block. We
obtain this by encoding $\perp$ as zero and, on that branch, flipping
the public/private flag and skipping reuse. On branch $n$, coefficient preparation and unpreparation together
contribute $|\lambda_n|/2$ to the selected block; its known sign
changes that coefficient to $\lambda_n/2$. The dummy branch contributes
zero. Thus preparation, selected reuse and unpreparation of
Eq.~\eqref{hs:eq:gate-filter-state} give selected block
$\widetilde P/2$, where $\widetilde P=\sum_n\lambda_nP_n$.
The factor of two is the fixed normalisation subsequently removed by
one amplification round.
We include every temporary register in the eventual success test.

We use a separate register $R_{\rm reuse}$ because the input and
output labels of successive reuses are distinct from the physical time
slots inside one interaction. This register has
$r_{\rm reuse}:=\lceil\log_2(2d+2)\rceil$ qubits. On the public
branch, conditional on length $n$, we need its initial state to be
\begin{equation}
 \ket{u_n}:=\frac1{\sqrt n}\sum_{j=0}^{n-1}\ket j.
\end{equation}
For a fixed $n$, this takes $\mathcal O(\log(n+1))$ gates.
For an exact preparation before finite-set synthesis, we first prepare
uniform labels in $[0,2^{p_n})$, where
$p_n:=\lceil\log_2n\rceil$, and a flag with accepted amplitude
$\sqrt{2^{p_n}}/(2\sqrt n)$. We accept when the flag is accepted
and the label is below $n$. Every accepted label then has
amplitude
\begin{equation}
 \frac1{\sqrt{2^{p_n}}}\frac{\sqrt{2^{p_n}}}{2\sqrt n}
 =\frac1{2\sqrt n}.
\end{equation}
The entire accepted vector is $\ket{u_n}/2$, so its amplitude is
$1/2$ and its probability is $1/4$. This is the norm of the whole
accepted branch, not the amplitude of an individual label. Writing
$1/2=\sin(\pi/6)$, one amplitude-amplification step changes that
branch amplitude to $\sin(3\pi/6)=1$ while preserving its direction.
The same exact amplification angle works for every $n$.
We then erase the comparison work and reset the successful flag.
Comparisons and reflections cost $\mathcal O(p_n+1)$ gates and
clean work bits. Looping over the $2d$ known choices with equality
controls gives coherent preparation for all $n$ in
$\mathcal O(d r_{\rm reuse})$ gates and
$\mathcal O(r_{\rm reuse})$ reusable work qubits.

For each $j=0,\ldots,2d-1$, enable the following operations when the
ordinary length label satisfies $n>j$: apply the private Cayley
operation across both clocks, apply the compressed boundary network
on $R_{\rm reuse}=j$, then increment $R_{\rm reuse}$ modulo
$2^{r_{\rm reuse}}$ on the private branch. Before interaction $j$
the private amplitude has reuse label $j$, and afterwards it has
label $j+1$. Completed public outputs keep their earlier labels;
future public inputs keep theirs. Thus the clock changes which
public component couples to the same private amplitude without
resetting that amplitude. The comparisons and increments cost
$\mathcal O(r_{\rm reuse})$ per interaction.

More explicitly, write the complete transducer in this paragraph as
$S_{\rm tot}=\left(\begin{smallmatrix}A_{\rm tot}&B_{\rm tot}\\C_{\rm tot}&D_{\rm tot}\end{smallmatrix}\right)$,
with stationary map $\Gamma_{\rm tot}$ and public action $U$.
The stationary identities are
$C_{\rm tot}+D_{\rm tot}\Gamma_{\rm tot}=\Gamma_{\rm tot}$ and
$A_{\rm tot}+B_{\rm tot}\Gamma_{\rm tot}=U$. They let us verify the reuse operation
step by step, rather than assume that an unknown stationary component
has been prepared. Initially the private vector is zero, which is
$(I-D_{\rm tot}^0)\Gamma_{\rm tot}\psi/\sqrt n$. If before interaction $j$
it equals $(I-D_{\rm tot}^j)\Gamma_{\rm tot}\psi/\sqrt n$, then the new
public and private vectors are respectively
\begin{equation}
 \begin{aligned}
 \frac1{\sqrt n}\bigl[A_{\rm tot}+B_{\rm tot}(I-D_{\rm tot}^j)\Gamma_{\rm tot}\bigr]\psi
 &=\frac1{\sqrt n}(U-B_{\rm tot}D_{\rm tot}^j\Gamma_{\rm tot})\psi,\\
 \frac1{\sqrt n}\bigl[C_{\rm tot}+D_{\rm tot}(I-D_{\rm tot}^j)\Gamma_{\rm tot}\bigr]\psi
 &=\frac1{\sqrt n}(I-D_{\rm tot}^{j+1})\Gamma_{\rm tot}\psi.
 \end{aligned}
\end{equation}
This proves the recurrence for every $j$. The private amplitude
persists and changes between interactions; the clock only moves its
label so that it meets the next public input. Unpreparing the uniform
public labels averages their output components and gives
\begin{equation}\label{hs:eq:gate-reuse-block}
 P_n=U-B_{\rm tot}\left(\frac1n\sum_{j=0}^{n-1}D_{\rm tot}^j\right)\Gamma_{\rm tot}.
\end{equation}
This verifies the required reuse block on superpositions of lengths
and on inputs with an untouched reference. One LCU call uses at most
$2d$ private operations and
$\mathcal O(d(C_J+a+m_{\rm slot}+r_{\rm reuse}))$ other elementary gates.

\paragraph{Entry arithmetic and the precise resource bound.}
The private circuit has the form
\begin{equation}
 D_0=F^\dagger O_{\rm val}^\dagger R_{\rm val}O_{\rm val}F,
\end{equation}
where $R_{\rm val}$ is the full entry Cayley rotation in
Eq.~\eqref{hs:eq:edge}. An external enable bit controls only that
middle rotation; computation and uncomputation cancel when it is
disabled. Decoding the loaded entry and computing this rotation can
be costly, even though they require no further matrix queries.
Let $G_C(\xi)$ and $A_C(\xi)$ bound the additional gates and clean
work qubits for this controlled operation to error on all output registers at
most $\xi$. These costs include decoding, rotation arithmetic,
validity tests and routing, but exclude gates inside the supplied
matrix oracles. The inverse is the adjoint of the implemented
circuit, and the primitive query cost remains a fixed constant
independent of $\xi$.

We also specify the action on unused basis states, since this is
needed for a full finite circuit. Uniform list-position and orientation
preparation has $a,C_J=\mathcal O(\log(N+1))$; the uniform-range
construction above handles non-power-of-two $s$. Include any
attenuation work in the private space and tensor $D_0$ with the
identity on those added work registers. Then
$(I-D_0\otimes I)^{-1}=R_0(1)\otimes I$, and the spectrum and
compressed bounds remain unchanged because $J$ uses zero work.
For index-register labels outside $[0,N)$, compute the validity
predicate $[i<N]\wedge[j<N]$ before querying. On an invalid pair,
select the zero-entry rotation $-I$ independently of the returned
encoding; erase the predicate after the inverse query. After
conjugating by $U_J^\dagger$, enable boundary rotations only on
zero boundary work and a valid system index. These tests use
$\mathcal O(\log(N+1))$ gates and leave the other public states
inactive. The compressed identities still depend only on valid
list inputs and impose no additional restriction on the arbitrary
locator extensions.

We can now state the resource guarantee, with the physical and reuse
clocks already accounted for. For $T>0$, choose $\gamma$ as in
Eq.~\eqref{ov:eq:gamma}, and use $\alpha=\eta\sqrt{s}/\gamma$
and $Q$ from Eq.~\eqref{ov:eq:Q}. Set
\begin{equation}\label{hs:eq:gate-schedule}
 W:=\frac{\alpha T(1+\gamma^2)}2,\qquad b_{\rm av}:=32Q,\qquad
 k:=\left\lceil4W/Q+\log_2(4096/\eps)\right\rceil,
 \qquad d:=k(b_{\rm av}-1).
\end{equation}
Let $M$ be the least power of two at least
\begin{equation}\label{hs:eq:gate-clocks}
 \max\left\{2,4\alpha T,
            \frac{(\alpha T)^{3/2}}{\sqrt{3\eps}}\right\},
 \qquad m_{\rm slot}:=\log_2M,\qquad
 r_{\rm reuse}:=\lceil\log_2(2d+2)\rceil.
\end{equation}
The first lower bound on $M$ supplies a nontrivial binary mesh, the
second ensures the local coupling $\alpha T/M\le1/4$, and the
third bounds the accumulated Cayley approximation error. In contrast,
$d$ was chosen to suppress the private evolution during reuse.
The reuse count places no lower bound on $M$: these registers encode
different parts of the circuit.

\begin{theorem}[Gate-efficient sparse simulation]\label{hs:thm:gates}
Under the assumptions of Theorem~\ref{hs:thm:main} and the
implementations just specified, the baseline simulation circuit uses at most
$6d$ private Cayley operations and
\begin{equation}\label{hs:eq:gate-known}
 \mathcal O\!\left(d(C_J+a+m_{\rm slot}+r_{\rm reuse})\right)
\end{equation}
other elementary gates before finite-set synthesis. For a sufficiently
large absolute $C_0$, define
\begin{equation}\label{hs:eq:gate-Gzero}
 \xi:=\frac{\eps}{24d},\qquad
 G_{\rm syn}:=C_0[d(C_J+a+m_{\rm slot}+r_{\rm reuse})+A_C(\xi)].
\end{equation}
Under the finite-gate assumptions of Section~\ref{gate:sec:result},
the additional gate count and work-qubit count are respectively
\begin{equation}\label{hs:eq:gate-total}
 \mathcal O\!\left(dG_C(\xi)+
       G_{\rm syn}\log\frac{G_{\rm syn}+1}{\eps}\right),
 \qquad
 \mathcal O\!\left(A_C(\xi)+a+m_{\rm slot}+r_{\rm reuse}\right).
\end{equation}
The error on all output registers is at most $\eps$, including an untouched
reference, and the primitive matrix-query order remains that of
Theorem~\ref{hs:thm:main}.

In particular, put
$\mathcal Q_{\rm base}:=\sqrt{s}\eta T+\sqrt{s\eta T\log(1/\eps)}+\log(1/\eps)$.
Suppose the entries use an efficiently decoded binary encoding,
the known parameters are efficiently computable, and their bit
lengths are bounded by $b_{\rm data}$. If decoding and reversible
arithmetic have polynomial gate and workspace costs in
$b_{\rm data}$, $\log(N+1)$ and $\log(1/\xi)$, then
\begin{equation}\label{hs:eq:gate-binary}
 G_{\rm additional}
 =\mathcal Q_{\rm base}\,\operatorname{poly}\!\left(b_{\rm data},\log(N+1),
                             \log\frac{\mathcal Q_{\rm base}+1}{\eps}\right).
\end{equation}
This is a time-independent result; at $t=0$ the identity suffices.
Theorem~\ref{gate:thm:main} gives the corresponding implementation
for the improved-precision refinements.
\end{theorem}
\begin{proof}
We first count the ideal operations, then allocate the remaining error
to their finite implementation. The boundary compression and coherent
reuse construction above give one linear-combination circuit with at
most $2d$ private calls. Equation
\eqref{hs:eq:gate-reuse-block} is exactly the block used in the
baseline conversion proof. One oblivious amplification round uses
three LCU calls, hence at most $6d$ private operations, and its
error on all output registers is bounded by Eq.~\eqref{hs:eq:amplified-full-output}.
The mesh in Eq.~\eqref{hs:eq:gate-clocks} gives Cayley bias at most
$\eps/4$. The filter schedule gives tolerance $\eps/8$ and amplified
error at most $\eps/4$. Thus the ideal-gate circuit has error at
most $\eps/2$.

Success reflections also test the arithmetic work, adding
$\mathcal O(A_C(\xi)+a+m_{\rm slot}+r_{\rm reuse})$ gates. This explains the
work-size term in $G_{\rm syn}$. Give each private call error
$\xi=\eps/(24d)$, so telescoping contributes at most $\eps/4$.
We must justify the same allocation for actual inverse calls, since
we reverse the implemented circuit, including its approximation errors.
Let $\mathsf E$ append zero arithmetic work, let $V$ be the ideal
logical private operation, and let $\widetilde V$ be its implemented
unitary on the enlarged registers. The promised forward error is
$\norm{\widetilde V\mathsf E-\mathsf E V}\le\xi$. Because the
ideal operation preserves the clean-work subspace,
\begin{equation}
 \widetilde V^\dagger\mathsf E-\mathsf E V^\dagger
 =\widetilde V^\dagger
       (\mathsf E V-\widetilde V\mathsf E)V^\dagger.
\end{equation}
The unitaries on the right preserve the norm, so the actual adjoint
has the same clean-input error bound. In the product telescoping
argument we evaluate each local error on an ideal intermediate state
with clean work and multiply it afterwards by unitary operations.
We therefore do not assume that an imperfect intermediate state has
exactly zero work. Residual work-register amplitude is included in
the same total error.

Approximate the remaining known rotations to total operator error
$\eps/4$, using accuracy $\eps/(4G_{\rm syn})$ per rotation.
Logarithmic-accuracy synthesis, for example the construction using at
most two clean ancillas in Ref.~\cite{KMM2013}, gives
Eq.~\eqref{hs:eq:gate-total}. Decompose controlled known operations
before synthesis; comparisons, routing and reflections can use exact
reversible gates. All work is reused, and residual synthesis amplitude
is included in the error on all output registers. We have allocated $\eps/4$ to each of four contributions: physical
Cayley discretisation, amplified filtering, private-call implementation,
and synthesis of the remaining known rotations. Their sum is at most
$\eps$. The construction counts the resulting quantum circuit; it
does not assert a classical compilation-time bound.

Finally, the chosen rescaling gives $d=\mathcal O(\mathcal Q_{\rm base})$.
Since $(1+\gamma^2)/2\ge1/2$ and $d\ge2W$, we have
$\alpha T\le d$, which implies
\begin{equation}
 m_{\rm slot}+r_{\rm reuse}=\mathcal O\!\left(\log\frac{d+1}{\eps}\right).
\end{equation}
Together with $a,C_J=\mathcal O(\log(N+1))$ and the arithmetic
assumptions, this proves Eq.~\eqref{hs:eq:gate-binary}. An arbitrary
known decoding map can instead have a large gate cost; in that case
the explicit $G_C,A_C$ formula remains the applicable guarantee.
\end{proof}

The query theorem itself needs none of these efficient-arithmetic
assumptions. For any fixed finite circuit, its known operations can
be approximated to an allocated total error without extra matrix
queries. The exact decoded-value construction therefore introduces
no precision-dependent amplitude-multiplication query overhead.
The resulting separation of resources is useful. Loading an entry and
applying its full Cayley rotation gives constant matrix-query cost;
the contour estimate and rescaling determine how many times that
operation must be called. Compressing the ordered boundary network
and sharing the reuse schedule then make the intervening gates
efficient under the stated arithmetic assumptions. These are distinct
parts of the argument, and none requires a multivariate
quantum-signal-processing synthesis theorem.

\section{Hamiltonian simulation lower bounds}\label{part:hsl}

We now show why Hamiltonian simulation must pay for sparsity, evolution
time and precision together. Our reduction places a hidden parity in
the sign of an endpoint amplitude. The initial state and the measurement
are known independently of the hidden input. A simulation accurate
enough to preserve that amplitude therefore supplies information about
the parity, and a Boolean query lower bound tells us how many matrix
queries are needed to obtain that information.

The role of precision is particularly useful here. A long chain can
have an extremely small endpoint amplitude after a short evolution.
Measuring the endpoint alone would reveal only its squared magnitude
and would lose the parity sign. We instead interfere the endpoint with
an isolated reference state. The resulting binary output has a small
positive correlation with parity whenever the simulation error is
sufficiently smaller than the endpoint amplitude. We do not need to
recover parity with constant success advantage. This lets a demand for
high precision force many queries even at fixed physical time.

To make sparsity contribute as well, we replace the scalar transitions
along the chain by unitary blocks built from independent search
instances. Each matrix entry can be evaluated using a constant number
of Boolean membership queries, while the action of a whole block on a
known vector reveals the parity of its hidden marked position. A
long search-decorated chain then yields the mixed time--sparsity--precision
term. Two related constructions isolate the other terms: a shorter,
rescaled decorated chain gives the evolution-time bound, and a chain
of scalar signs gives the precision-only bound.

We first state the joint guarantee and prove the Boolean lower bound
for small parity correlation. We then construct the chains and their
phase-sensitive readout, choose their lengths to obtain the three
terms, and finally enlarge the hard matrices to accommodate a separate
maximum-entry promise. All supports and locator extensions in these
families are known independently of the hidden input. The only source
of hidden information is the entry-value oracle.

\subsection{Joint Hamiltonian lower bound and small-bias parity}
\label{hsl:sec:main}

We use the native sparse access of Section~\ref{sec:overview}, allowing
all location operations, known gates and state preparation for free.
Only value queries carry the hidden input. The dimension may grow
with the parameters. We construct hard instances with $\norm H=1$ and evolution time
$\tau$. Their maximum row and column Euclidean norms are also exactly
one, so their normalised time agrees with the upper-bound parameter
$\tau=\eta|t|$ when $\eta=1$. Multiplying such a Hamiltonian by
$\eta$ and changing the physical time to $\tau/\eta$ gives the
same evolution. The operator-norm normalisation is a property of our
hard instances, not an extra assumption on the inputs of the upper
bound.

It suffices to demand accuracy on one known initial state. For the
hidden input $X$, the simulation algorithm may return a mixed system state
$\rho_X$, provided
\begin{equation}\label{hsl:eq:accuracy}
 \tfrac12\norm{\rho_X-\ket{\psi_X}\bra{\psi_X}}_1\le\eps,
 \qquad \ket{\psi_X}:=e^{-\ii\tau H_X}\ket{\psi_{\rm in}}.
\end{equation}
The complete-output norm guarantee of our upper bounds implies this
weaker requirement after the work registers are discarded. It is
therefore enough to prove a lower bound for this known-input task:
we do not require the lower-bound algorithm to approximate the entire
unitary on every initial state. We allow mixed outputs, but initially
allow no postselection, and the query bound must hold on every
execution. Lemma~\ref{odel:lem:herald} later gives the constant-success
heralded extension.

\begin{theorem}[Joint simulation lower bound]\label{hsl:thm:main}
Let $\tau\ge1$, let the integer sparsity satisfy $s\ge32$, and let
$0<\eps\le1/64$. Define
\begin{equation}\label{hsl:eq:parameters}
 L:=\log(1/\eps),\qquad h:=\log(e+L/\tau),\qquad r:=L/h.
\end{equation}
The worst-case sparse-query complexity obeys
\begin{equation}\label{hsl:eq:main}
 Q\ge c_0(\sqrt{s}\,\tau+\sqrt{s\tau r}+r),
 \qquad c_0:=\frac{\log2}{4096\pi}.
\end{equation}
The hard matrices can be real symmetric with dyadic entries,
operator norm and maximum row and column Euclidean norms exactly
one, and maximum sparsity exactly $s$. Their location lists and full
in-place locator extensions are input independent. Dimension
$\mathcal O(s(\tau+r+1))$ suffices, with further known padding allowed.
\end{theorem}

We need not find one matrix on which all three mechanisms are
simultaneously difficult. An algorithm that works on the whole promise
class must work on each of our three families. Its worst-case query
count is consequently at least the largest of their separate lower
bounds. For three nonnegative numbers, that maximum is at least one
third of their sum, which is enough for the stated joint bound. Theorem
\ref{hs:thm:optimal} matches it term by term. The baseline retains
\begin{equation}\label{hsl:eq:upper}
 \mathcal O(\sqrt{s}\tau+\sqrt{s\tau L}+L),
\end{equation}
with a mixed term larger by $\sqrt h$; the hybrid bound improves its
additive term to $L/\log(e+L/(s\tau))$. These latter constructions
also support time-dependent inputs. The present optimal upper bound
does not establish the same precision dependence for that larger class.

For the underlying Boolean problem, take $n$ independent search
instances of even size $m$, with hidden positions
$j_1,\ldots,j_n\in\{0,\ldots,m-1\}$. A membership query acts as
\begin{equation}
 O_X\ket{k,a,z}=\ket{k,a,z\mathbin\oplus[a=j_k]},
 \qquad 1\le k\le n,\quad0\le a<m,\quad z\in\{0,1\}.
\end{equation}
Here $[a=j_k]$ is one when the queried label $a$ is marked in
instance $k$, and zero otherwise. The hidden input is the tuple
$X=(j_1,\ldots,j_n)$; each of its $m^n$ possible values is equally
likely in the correlation calculation. We want correlation with
$\sigma(X):=(-1)^{j_1+\cdots+j_n}$, the product of the signs
indicating whether the marked positions are even or odd. For a binary output
$Y\in\{-1,1\}$, its average correlation is
$\beta:=\mathbb E_X[\sigma(X)\mathbb E(Y\mid X)]$, with hidden
positions uniform. For a fixed input, the indicator of a correct sign is
$(1+\sigma(X)Y)/2$, so averaging gives success probability
$(1+\beta)/2$. We will obtain $\beta>0$ from a small endpoint
amplitude, and it may be exponentially small in the chain length.
A bound formulated only for constant success advantage would therefore
not suffice. The next lemma retains the dependence on the correlation
explicitly.

\begin{lemma}[Small-correlation search parity]\label{hsl:lem:search}
If a $q$-query Boolean algorithm has correlation at least
$\beta_0>0$ and $m\ge4\beta_0^{-1/n}$, then, for every integer
$n\ge1$,
\begin{equation}\label{hsl:eq:searchbound}
 q\ge\frac{\log2}{24}\,n\sqrt m\,\beta_0^{1/(2n)}.
\end{equation}
\end{lemma}
\begin{proof}
We compare two requirements on a weighted norm of the algorithm's
state. Each query can increase that norm only slightly, whereas a
positive correlation with the parity of all $n$ positions forces it
to become large. The parameter in the weight will let us make this
comparison at the small correlation $\beta_0$ of interest. This is
a multiplicative-adversary argument in the sense of Ambainis and
\v{S}palek~\cite{Ambainis,Spalek}; we give the weight and the
per-query estimate explicitly.

\emph{Keeping track of the unknown inputs.}
For the analysis, introduce $n$ registers of dimension $m$ containing
a uniform superposition of the hidden positions. These registers are
not made available to the algorithm. We also purify any randomness
and defer measurements, so the joint state after $q$ queries can be
written as
\begin{equation}
 \ket{\Psi_q}=m^{-n/2}\sum_{j_1,\ldots,j_n=0}^{m-1}
       \ket{j_1,\ldots,j_n}\ket{\phi_{X,q}}.
\end{equation}
The second factor includes all algorithm registers and the registers
used to purify its operations. Initially it is the same state
$\phi_0$ for every $X$. Input-independent operations act only on
this second factor; the query is the operation through which the
algorithm becomes correlated with the hidden positions.

\emph{Choosing a weight that detects parity information.}
On one hidden-position register define
\begin{equation}
 u:=m^{-1/2}\sum_{j=0}^{m-1}\ket j,\qquad
 Z\ket j:=(-1)^j\ket j,\qquad v:=Zu,\qquad
 R:=I-P_u-P_v,
\end{equation}
where $P_u=\ket u\bra u$ and $P_v=\ket v\bra v$.
Since $m$ is even, the alternating sum of $m$ signs is zero, so
$\langle u,v\rangle=0$. The operator $R$ is the orthogonal
projection onto the complement of their span. In this proof, let
$a$ denote a real weight parameter, rather than a queried position.
For $1\le a\le\sqrt m$, set
\begin{equation}
 w:=P_u+a^2P_v+aR,\qquad W:=w^{\otimes n}.
\end{equation}
We let $W$ act as the identity on the algorithm registers and use
$\norm{W\Psi_q}^2$ as our progress measure. Initially the hidden
registers are in $u^{\otimes n}$, so this measure is one. Every
query-free unitary commutes with $W$ and leaves the measure unchanged.

The unequal weights have a specific purpose. The initial direction
$u$ has weight one; its parity transform $v$ has weight $a^2$;
the remaining directions have the intermediate weight $a$. The
operator $Z$ exchanges $u$ and $v$ and preserves their orthogonal
complement. Consequently, on each of these three subspaces,
\begin{equation}
 w^{-1}Zw^{-1}=a^{-2}Z.
\end{equation}
For example, applying the left side to $u$ first leaves $u$
unchanged, then sends it to $v$, and finally multiplies it by
$a^{-2}$. On the complement, the two inverse weights each contribute
$a^{-1}$. Tensoring this identity across the $n$ inputs will bound
the parity correlation by $a^{-2n}$ times the progress measure.
We next show that queries cannot increase that measure too quickly.

\emph{Bounding the change caused by one query.}
In the Hadamard basis of the answer qubit, an XOR query is the identity
on the $+$ branch. On the $-$ branch, a query to position $j$ in
one instance acts on its hidden register as
$D_j:=I-2\ket j\bra j$: only the hidden input whose marked
position equals $j$ receives a minus sign. Thus a coherently selected
query is a direct sum of identities and such phase flips.

The weighted amplification of a phase flip is controlled by
$wD_jw^{-1}$. To estimate its norm, define
$e:=(w/a-I)\ket j$ and $f:=(aw^{-1}-I)\ket j$.
Both vectors lie in the span of $u,v$, since the two parenthesised
operators vanish on the orthogonal complement. Their components are
\begin{equation}
 \begin{aligned}
 e&=(a^{-1}-1)u\langle u|j\rangle
                    +(a-1)v\langle v|j\rangle,\\
 f&=(a-1)u\langle u|j\rangle
                    +(a^{-1}-1)v\langle v|j\rangle.
 \end{aligned}
\end{equation}
The two overlaps each have magnitude $m^{-1/2}$, and $u,v$ are
orthonormal. Hence
\begin{equation}
 \norm e^2=\norm f^2
 =\frac{(a-1)^2+(1-a^{-1})^2}{m}\le\frac{2a^2}{m}.
\end{equation}
Using $w\ket j=a(\ket j+e)$ and
$\bra j w^{-1}=a^{-1}(\bra j+f^\dagger)$, we can now expand
\begin{equation}
 wD_jw^{-1}=D_j-2(e\bra j+\ket j f^\dagger+ef^\dagger).
\end{equation}
Each correction is rank one, with norm equal to the product of the
two vector norms. The triangle inequality and $a/\sqrt m\le1$
therefore give
\begin{equation}
 \norm{wD_jw^{-1}}
 \le1+\frac{4\sqrt2a}{\sqrt m}+\frac{4a^2}{m}
 \le1+\frac{12a}{\sqrt m}.
\end{equation}
On the other hidden registers, the conjugating weights cancel.
The same estimate holds for every instance, position and answer-basis
branch, so it also bounds the direct sum describing a coherent query.
After $q$ queries, starting from weighted norm one, we obtain
\begin{equation}
 \norm{W\Psi_q}^2
 \le\left(1+\frac{12a}{\sqrt m}\right)^{2q}
 \le e^{24aq/\sqrt m}.
\end{equation}

\emph{Relating the final weight to the required correlation.}
Represent the binary output by a Hermitian observable $M$ on the
algorithm registers with $\norm M\le1$. The average correlation
is $\langle\Psi_q,(Z^{\otimes n}\otimes M)\Psi_q\rangle$:
the hidden-register observable supplies $\sigma(X)$ on each input,
and $M$ supplies the algorithm's expected output. Inserting the
inverse weights on both sides gives
\begin{equation}
 \begin{aligned}
 \beta_0
 &\le(W\Psi_q)^\dagger
    \bigl[(W^{-1}Z^{\otimes n}W^{-1})\otimes M\bigr]
           (W\Psi_q)\\
 &=a^{-2n}(W\Psi_q)^\dagger(Z^{\otimes n}\otimes M)
           (W\Psi_q)\\
 &\le a^{-2n}\norm{W\Psi_q}^2
 \le a^{-2n}e^{24aq/\sqrt m}.
 \end{aligned}
\end{equation}
The observable in the second line has norm at most one. This
inequality exhibits the tradeoff: increasing $a$ suppresses the
allowed parity correlation, but permits faster weighted growth per
query.

Choose $a:=2\beta_0^{-1/(2n)}$. A nonvacuous correlation promise
has $\beta_0\le1$, so $a\ge1$, and the hypothesis on $m$ gives
$a^2=4\beta_0^{-1/n}\le m$. Taking logarithms of the last
inequality then yields
\begin{equation}
 \frac{24aq}{\sqrt m}\ge\log(\beta_0a^{2n})
       =2n\log2.
\end{equation}
Substituting this choice of $a$ proves
Eq.~\eqref{hsl:eq:searchbound}.
\end{proof}

\subsection{Parity-chain construction and output distinguishability}

We now need a Hamiltonian whose endpoint amplitude carries the parity
of the hidden positions, but whose individual entries reveal no more
than a membership query. We build it in two stages. First we construct
a unitary block for one search instance; its entries are easy to query,
while its action on a known vector supplies the desired parity sign.
We then place independent blocks on successive edges of a weighted
chain, so the signs multiply along the path.

\paragraph{A unitary block carrying one hidden sign.}
For a single search instance, write $x_a=[a=j]$ for the bit string
with its unique marked position at $j\in\mathbb Z_m$. Here
$\mathbb Z_m$ labels the integers $0,\ldots,m-1$ with addition
modulo $m$. Let $S\ket a=\ket{a+1\bmod m}$ be the known cyclic
shift, and retain the uniform vector $u$ and alternating vector $v$
from the preceding proof. Define
\begin{equation}\label{hsl:eq:gadget}
 U_x:=S^{j+1}(2P_u-I),\qquad
 (U_x)_{ab}=2/m-x_{a-b-1\bmod m}.
\end{equation}
The first expression explains the action; the second explains the
access cost. Both factors in the first expression are unitary, so
$U_x$ is unitary. Multiplying the uniform-reflection matrix by the
shift gives the displayed entry formula. We do not implement the
unknown shift $S^{j+1}$ in the reduction: to answer an entry query
at $(a,b)$, it suffices to query the single membership bit at
$a-b-1\pmod m$.

The alternating vector is a known eigenvector with an unknown sign.
Because $m$ is even, $Sv=-v$ and $P_uv=0$, giving
\begin{equation}
 U_xv=S^{j+1}(-v)=(-1)^{j+2}v=(-1)^jv.
\end{equation}
For $m\ge4$, the two possible entry values $2/m$ and $2/m-1$
are both nonzero. Thus the block has full support independent of $j$:
a location oracle cannot reveal where the marked position is hidden.
This is the sign-adjusted version of Patel's unitary gadget
\cite[Lemma~4]{Patel} used throughout the construction.

\paragraph{A chain with a calculable endpoint amplitude.}
The chain has basis states $\ket0,\ldots,\ket n$ for an integer
$n\ge1$, so its dimension is $n+1$. We choose the positive edge
weights
\begin{equation}\label{hsl:eq:spin}
 b_k:=\frac{\sqrt{(k+1)(n-k)}}n,\qquad
 B_n:=\sum_{k=0}^{n-1}b_k
       (\ket{k+1}\bra k+\ket k\bra{k+1}).
\end{equation}
These weights allow us to calculate the entire end-to-end transition
without estimating a sum over paths. For that calculation only,
identify $\ket k$ with the normalised symmetric state of $n$ qubits
having Hamming weight $k$, meaning exactly $k$ ones:
\begin{equation}
 \ket{S_k}:=\binom nk^{-1/2}
       \sum_{z\in\{0,1\}^n:\,|z|=k}\ket z.
\end{equation}
If $X_j$ is Pauli $X$ on qubit $j$, there are
$\binom nk(n-k)$ single-bit flips from weight $k$ to weight $k+1$.
Including the two state normalisations therefore gives
\begin{equation}
 \bra{S_{k+1}}\sum_{j=1}^nX_j\ket{S_k}
 =\frac{\binom nk(n-k)}{\sqrt{\binom nk\binom n{k+1}}}
 =\sqrt{(k+1)(n-k)}.
\end{equation}
A single flip changes the weight by one, so these are precisely the
nonzero off-diagonal entries of the restriction. Thus $B_n$ is the
restriction of $n^{-1}\sum_jX_j$ to this symmetric subspace.
This identification is an analytic device; the hard matrix still
uses only the $n+1$ chain coordinates.

The Pauli terms commute and each has eigenvalues $\pm1$, so the
norm of their average is at most one. The symmetric vector
$\ket+^{\otimes n}$ has eigenvalue one and lies in the restricted
subspace, proving that $\norm{B_n}=1$. The endpoints correspond to
$\ket{0}^{\otimes n}$ and $\ket{1}^{\otimes n}$. Each qubit
contributes transition amplitude $-\ii\sin(\tau/n)$, giving
\begin{equation}\label{hsl:eq:spinamp}
 \norm{B_n}=1,\qquad
 \bra n e^{-\ii\tau B_n}\ket0=(-\ii)^n\sin(\tau/n)^n.
\end{equation}
In particular, when the time per edge $\tau/n$ is small, the endpoint
amplitude is small to the $n$th power. This is the dependence that
will convert a precision requirement into a lower bound on $n$.

We now attach the independent search blocks to this chain. The
Hamiltonian acts on $\mathbb C^m\otimes\mathbb C^{n+1}$:
\begin{equation}\label{hsl:eq:dressed}
 H_X^{\rm path}:=\sum_{k=0}^{n-1}b_k
 \left(U_{x^{(k+1)}}\otimes\ket{k+1}\bra k
 +U_{x^{(k+1)}}^\dagger\otimes\ket k\bra{k+1}\right).
\end{equation}
Here $x^{(k+1)}$ is the search instance with marked position
$j_{k+1}$. Using a block and its adjoint on the two directions makes
the matrix Hermitian. Each clock position is connected to at most
two neighbouring positions, each through an $m\times m$ dense
block, so every row and column has at most $2m$ nonzero entries.
The matrix is real symmetric because the gadgets and weights are real.

To compute its evolution, define the cumulative products
$V_0:=I$, $V_{k+1}:=U_{x^{(k+1)}}V_k$ and the unitary change of
basis $G:=\sum_{k=0}^nV_k\otimes\ket k\bra k$.
On a forward edge,
\begin{equation}
 V_{k+1}^\dagger U_{x^{(k+1)}}V_k=I_m;
\end{equation}
the reverse edge gives the adjoint identity. Consequently,
\begin{equation}
 G^\dagger H_X^{\rm path}G=I_m\otimes B_n.
\end{equation}
This clock-dependent change of basis is often called a gauge change.
It removes the unitary blocks from the calculation of the spectrum,
so the decorated Hamiltonian also has norm one. The change of basis
is used only in the proof; it is not an algorithmic operation, and
we never assume access to the unknown products $V_k$.

At the starting clock position $V_0=I$. At the endpoint,
$V_nv=\prod_{k=1}^n(-1)^{j_k}v=\sigma(X)v$. Combining the gauge
change with the scalar-chain amplitude therefore gives
\begin{equation}\label{hsl:eq:parityamp}
 \bra{v,n}e^{-\ii\tau H_X^{\rm path}}\ket{v,0}
 =\sigma(X)(-\ii)^n A,
 \qquad A:=\sin(\tau/n)^n.
\end{equation}
We use $\ket{v,k}:=v\otimes\ket k$. After removing the known phase $(-\ii)^n$, the endpoint amplitude
is $\sigma(X)A$: its magnitude $A$ is known from $n$ and $\tau$,
and its sign $\sigma(X)$ is the hidden parity. The regimes below
ensure $A>0$.

\paragraph{Exact encoded values and a phase-sensitive readout.}
The endpoint identity has the desired information, but we still have
two obligations. The oracle model must supply exact finite encodings
of the hard entries, and the output measurement must recover their
parity sign despite rounding and simulation error. We handle these
together by making the rounding error a small fraction of $A$.

A dyadic rational is an integer divided by a power of two, and hence
has an exact finite binary representation. The known square-root
weights $b_k$ need not be dyadic, so, when $A>0$, choose nearby
positive dyadic weights $b'_k\le b_k$ satisfying
\begin{equation}\label{hsl:eq:round}
 0\le b_k-b'_k\le\delta,
 \qquad
 \delta\le\min\left\{\tfrac12\min_kb_k,\frac{A}{16\tau}\right\}.
\end{equation}
The first condition on $\delta$ keeps every edge nonzero; the
second will bound the evolution error. Arbitrarily fine known binary
approximations are available, and the query lower bound places no
restriction on the classical work used to choose them.

Rounding down also preserves the operator-norm promise. Let $B'_n$
be the chain with weights $b'_k$. It is real symmetric and has
nonnegative entries. For every complex vector $z$, writing $|z|$
for the vector of component magnitudes gives
\begin{equation}
 |z^\dagger B'_nz|
 \le|z|^{\mathsf T}B'_n|z|
 \le|z|^{\mathsf T}B_n|z|\le\norm z^2.
\end{equation}
The middle inequality uses entrywise ordering on nonnegative vectors;
the last uses $\norm{B_n}=1$. Since $B'_n$ is Hermitian, bounding
all its quadratic forms proves $\norm{B'_n}\le1$.
Moreover, $B_n-B'_n$ has at most two entries of magnitude $\delta$
in each row and column, so its operator norm is at most $2\delta$.
Duhamel's bound for the two unitary evolutions then gives
\begin{equation}\label{hsl:eq:roundevolution}
 \norm{e^{-\ii\tau B_n}-e^{-\ii\tau B'_n}}
 \le2\tau\delta\le A/8.
\end{equation}
The same $G$ conjugates both decorated chains to their scalar chains,
so this estimate also bounds the change in the decorated endpoint
amplitude. If $m$ is a power of two, both $2/m$ and every $b'_k$
are dyadic, making each entry $b'_k(2/m-x_a)$ dyadic as well.
We can therefore encode the represented matrix exactly. Its entry
calculation uses only the queried membership bit and known constants;
it does not require an input-dependent norm or other global information.

We add two known isolated vertices to the rounded path: a zero-energy
vertex $r_0$ for interference, and an energy-one vertex to fix the
normalisation. The latter is a one-dimensional block $[1]$. Since
the path has norm at most one, this block makes the full operator norm
exactly one. Every row and column Euclidean norm is at most the
operator norm, and the isolated block attains one, so their maxima
are also exactly one. The input and readout below have no component
on this energy-one vertex.

All edges remain nonzero after rounding, so the supports are known
independently of $X$. We can supply known in-place permutation
locators, their inverses and their controls, including their action
outside the valid lists. A value query requires at most two Boolean
queries: compute the relevant membership bit into a work register,
write the corresponding known encoded matrix value, then uncompute
the bit. Known blocks require no Boolean query. Reversing this
procedure or controlling the value-writing step has the same constant
query cost. Thus a simulation making at most $Q$ sparse queries
induces a Boolean algorithm with membership-query count $q\le2Q$.

We now read the sign rather than the endpoint probability. Prepare
\begin{equation}\label{hsl:eq:readout}
 \ket{\psi_{\rm in}}:=\frac{\ket{r_0}+\ket{v,0}}{\sqrt2},\qquad
 M:=\ii^n\ket{r_0}\bra{v,n}
            +(-\ii)^n\ket{v,n}\bra{r_0}.
\end{equation}
These states and the observable depend only on known parameters.
The vector $\ket{r_0}$ is an additional basis state orthogonal to
the whole path space; it is not a tensor-factor control for a separate
simulation call. Before rounding, the evolved input is
\begin{equation}
 \frac1{\sqrt2}\left(\ket{r_0}
              +e^{-\ii\tau H_X^{\rm path}}\ket{v,0}\right).
\end{equation}
The reference component stays fixed because its energy is zero.
The two cross terms in the expectation of $M$ therefore give
\begin{equation}
 \operatorname{Re}\!\left(
 \ii^n\bra{v,n}e^{-\ii\tau H_X^{\rm path}}\ket{v,0}\right)
       =\sigma(X)A.
\end{equation}
Measuring only the endpoint would instead give a probability
proportional to $A^2$ and erase the sign. Interference retains the
amplitude linearly, and the known phase $\ii^n$ removes its clock
phase.

The operator $M$ has eigenvalues $+1$ and $-1$ on the span of
$\ket{r_0},\ket{v,n}$ and is zero on the orthogonal complement,
so $\norm M=1$. By Eq.~\eqref{hsl:eq:roundevolution}, rounding changes its expectation by at most $A/8$,
because the expectation is the real part of the endpoint amplitude
just bounded. If $\rho_X$ is a simulation output satisfying
Eq.~\eqref{hsl:eq:accuracy} for the rounded Hamiltonian, trace-norm
duality gives
\begin{equation}
 \left|\operatorname{Tr}\!\left[
 M\bigl(\rho_X-\ket{\psi_X}\bra{\psi_X}\bigr)\right]\right|
 \le\norm M\,\norm{\rho_X-\ket{\psi_X}\bra{\psi_X}}_1
 \le2\eps.
\end{equation}
The parity-corrected expectation is thus at least $7A/8-2\eps$
on every hidden input. To obtain a binary output, measure $M$,
return its sign on the two nonzero eigenspaces, and return an independent
fair sign on the zero eigenspace. This preserves the expectation.
The resulting average parity correlation has the same lower bound.
All of this uses an ordinary simulation of the enlarged Hamiltonian.
We need neither controlled simulation nor a convention for its absolute
global phase, because the reference and path are parts of the same
simulated state.

\begin{lemma}[Mixed lower bound]\label{hsl:lem:mixed}
Let $n\ge2$, $0<\tau\le n$, and let $m$ be a power of two with
$m\ge16n/\tau$. If
\begin{equation}\label{hsl:eq:mixepsilon}
 \eps\le\frac18\left(\frac{\tau}{2n}\right)^n,
\end{equation}
the decorated family requires
\begin{equation}\label{hsl:eq:mixedlower}
 Q\ge\frac{\log2}{96}\sqrt{mn\tau}
\end{equation}
sparse-oracle queries.
\end{lemma}
\begin{proof}
Our purpose is to substitute the endpoint correlation into
Lemma~\ref{hsl:lem:search}. For $0<\tau/n\le1$, the elementary
bound $\sin x\ge x/2$ on $0\le x\le1$ gives
$A\ge(\tau/(2n))^n$. The error hypothesis therefore implies
$\eps\le A/8$. The preceding readout has correlation at least
$7A/8-2\eps\ge5A/8$, so we can use the simpler lower bound
$\beta_0:=A/2$. The search-parity lemma applies,
since
\begin{equation}
 4\beta_0^{-1/n}
 =\frac{4\,2^{1/n}}{\sin(\tau/n)}\le\frac{16n}{\tau}\le m,
 \qquad
 \beta_0^{1/(2n)}
 =2^{-1/(2n)}\sqrt{\sin(\tau/n)}\ge\tfrac12\sqrt{\tau/n}.
\end{equation}
The first inequality verifies the search-parity lemma's gadget-size
hypothesis; the second determines the resulting query scale.
Using $q\le2Q$ in that lemma gives
\begin{equation}
 2Q\ge\frac{\log2}{24}n\sqrt m\,
            \beta_0^{1/(2n)}
 \ge\frac{\log2}{48}\sqrt{mn\tau},
\end{equation}
which proves Eq.~\eqref{hsl:eq:mixedlower}. This is the mixed
mechanism: the chain supplies $n$ independent hidden signs, but its
small endpoint correlation reduces the Boolean lower bound by a
factor of order $\sqrt{\tau/n}$.
\end{proof}

We next isolate the precision obstruction by removing the search
blocks. In this construction $x_1,\ldots,x_n\in\{0,1\}$ are
independent bits, rather than marked-position strings. Replace each
unitary block on edge $k-1\to k$ by the scalar sign $(-1)^{x_k}$.
Every interior chain row then has only two nonzero entries. The
clock-dependent diagonal change of basis has entry
$(-1)^{x_1+\cdots+x_k}$ at position $k$, so it removes the edge
signs in the same way as $G$ did above. The endpoint amplitude remains
$(-\ii)^n A$, multiplied now by
$\sigma(X)=(-1)^{x_1+\cdots+x_n}$. Rounding, the isolated vertices
and the interference readout all work as before.

For this Boolean problem, any strictly positive correlation with
$n$-bit parity already requires a linear number of queries. We can
prove that fact directly from polynomial degree, without a search-size
condition.

\begin{lemma}[Precision-only lower bound]\label{hsl:lem:precision}
For $n\ge2$, $0<\tau\le n$ and Eq.~\eqref{hsl:eq:mixepsilon},
the signed scalar chain requires
\begin{equation}\label{hsl:eq:precisionlower}
 Q\ge n/4.
\end{equation}
\end{lemma}
\begin{proof}
The same amplitude and error estimates as in the mixed lemma give
strictly positive correlation with full $n$-bit parity. We show that
fewer than $n/2$ Boolean queries would force that correlation to be
exactly zero.

Purify the Boolean algorithm and defer its measurements. Initially
its amplitudes do not depend on $X$. A Boolean XOR query has matrix
elements linear in the queried bit, so each query increases the
polynomial degree of every amplitude by at most one. Known unitaries
take input-independent linear combinations and do not increase the
degree. After $q$ queries every amplitude is therefore a polynomial
of degree at most $q$ in $x_1,\ldots,x_n$. On Boolean inputs we
may replace $x_k^2$ by $x_k$, obtaining multilinear polynomials
without increasing degree.

An expectation is a sum of products of an amplitude and a conjugate
amplitude, with coefficients from the fixed output observable. Its
degree is consequently at most $2q$. For any monomial
$\prod_{k\in S}x_k$ with $|S|<n$, choose a missing bit
$\ell\notin S$. Averaging over that bit first gives
\begin{equation}
 \mathbb E_X\!\left[
 (-1)^{x_1+\cdots+x_n}\prod_{k\in S}x_k\right]=0,
\end{equation}
since the rest of the expression is independent of $x_\ell$ and
$\mathbb E_{x_\ell}(-1)^{x_\ell}=0$. Thus every polynomial of
degree below $n$ has zero correlation with parity under the uniform
input distribution. The positive correlation obtained by the readout
requires $2q\ge n$. The value-query reduction has $q\le2Q$, so
$Q\ge n/4$ as claimed. Purification and deferred measurement include
adaptive and randomised algorithms in this argument.
\end{proof}

This familiar factorial-precision obstruction uses arbitrarily small
positive correlation, rather than bounded-error parity computation;
see also Ref.~\cite[Eq.~(18)]{LowChuang}. We complete the collection
of families with constant-error evolution time, obtaining the same
$\sqrt{s}\tau$ dependence as Low's lower bound~\cite{Low} while
retaining the public-support and exact-encoding conventions above.

\begin{lemma}[Evolution-time lower bound]\label{hsl:lem:time}
For $s\ge32$, $\tau\ge1$ and $0<\eps\le1/64$, the preceding
access and output model admits hard instances requiring
\begin{equation}\label{hsl:eq:timelower}
 Q\ge c_T\sqrt{s}\tau,
 \qquad c_T:=\frac{\log2}{192\pi}.
\end{equation}
\end{lemma}
\begin{proof}
Here we choose the chain to have a constant endpoint amplitude, so
that the same search-parity argument pays for elapsed evolution time.
Choose $m$ to be the largest power of two at most $s/2$, so
\begin{equation}\label{hsl:eq:mrange}
 m\ge16,\qquad s/4<m\le s/2.
\end{equation}
The endpoint formula becomes $\sin(c\tau/n)^n$ if we scale all
chain weights by $c$. We want $c\tau/n=\pi/2$, which gives
perfect endpoint transfer, while keeping $c\le1$ to retain the norm
promise. Choose
\begin{equation}
 n:=\max\{1,\lfloor2\tau/\pi\rfloor\},\qquad
 c:=\min\{1,\pi n/(2\tau)\},
\end{equation}
and use weights $cb_k$. If $\tau\ge\pi/2$, then
$n=\lfloor2\tau/\pi\rfloor\ge1$ and the choice of $c$ gives
the desired angle $\pi/2$, hence amplitude one. If
$1\le\tau<\pi/2$, we instead have $n=1$, $c=1$ and endpoint
amplitude $\sin\tau\ge\sin1>1/2$. Thus both cases retain a
constant signal, and the unrounded path norm is at most one.

Round the positive scaled weights $cb_k$ down to dyadic values,
using their minimum in the positivity condition and the new endpoint
amplitude in the error condition of Eq.~\eqref{hsl:eq:round}.
The earlier norm and Duhamel arguments still apply. With the same
reference readout, the final parity correlation is at least
\begin{equation}
 \frac78\cdot\frac12-\frac2{64}=\frac{13}{32}>\frac14.
\end{equation}
Lemma~\ref{hsl:lem:search} applies with $\beta_0=1/4$ because
$4\beta_0^{-1/n}\le16\le m$. Also
$\beta_0^{1/(2n)}\ge1/2$ and $q\le2Q$, giving
$Q\ge(\log2)n\sqrt m/96$. For $\tau\ge\pi/2$, the inequality $\lfloor x\rfloor\ge x/2$
for $x\ge1$ gives $n\ge\tau/\pi$; for smaller $\tau$ the same
bound follows from $n=1$. Also $\sqrt m>\sqrt{s}/2$ by
Eq.~\eqref{hsl:eq:mrange}. Substituting these two estimates gives
Eq.~\eqref{hsl:eq:timelower}.
\end{proof}

\subsection{Chain-length selection and the joint bound}

We have established lower bounds in terms of a chosen chain length
$n$. To prove Theorem~\ref{hsl:thm:main}, we must choose that integer
from the desired simulation parameters and check that the endpoint
signal survives error $\eps$.
Recall the theorem's parameters $L=\log(1/\eps)$,
$h=\log(e+L/\tau)$ and $r=L/h$. The sufficient signal condition
in Eq.~\eqref{hsl:eq:mixepsilon} can be written as
\begin{equation}
 n\log(2n/\tau)\le L-\log8.
\end{equation}
This explains the scale $r$: the exponential smallness of the
endpoint amplitude permits a chain of length proportional to
$L/\log(e+L/\tau)$. We choose a conservative constant multiple
of that length so that the signal and integer constraints hold
together. The remaining issue is whether the sparsity allows a large
enough search block. If it does not, the precision-only family will
already give the required mixed-term lower bound.

\begin{proof}[Proof of Theorem~\ref{hsl:thm:main}]
\emph{Choosing a detectable long chain.}
First consider $L\ge64\tau$, where the required precision is high
relative to the evolution time, and set
\begin{equation}\label{hsl:eq:nchoice}
 n:=\lfloor r/4\rfloor.
\end{equation}
We first verify that this chain is long enough for both precision
lemmas. Put $z:=L/\tau\ge64$ for this parameter comparison.
The inequality $\log(e+z)\le z/8$ holds at $z=64$; for larger
$z$, the derivative $1/(e+z)$ of its left side is at most the
right side's derivative $1/8$. Hence
$r/\tau=z/\log(e+z)\ge8$. Since $\tau\ge1$, this also gives
$r/4\ge2$, and rounding $r/4$ down reduces it by at most a factor
of two. We therefore have
\begin{equation}\label{hsl:eq:nrange}
 n\ge2,\qquad n\ge r/8\ge\tau,\qquad n\le r/4.
\end{equation}
Moreover $h\ge1$ and $2n/\tau\le z/(2h)\le e+z$, so
\begin{equation}\label{hsl:eq:precisionchoice}
 n\log(2n/\tau)\le nh\le L/4.
\end{equation}
The detectable-amplitude condition follows:
\begin{equation}
 \left(\frac{\tau}{2n}\right)^n\ge e^{-L/4}
                                    \ge8e^{-L}=8\eps,
\end{equation}
where $L\ge64$. Both precision lemmas therefore have their required error range.
In particular, we can always use the scalar chain, which has no
search-size condition.

\emph{Obtaining the precision and mixed terms.}
The scalar chain gives
\begin{equation}\label{hsl:eq:rbound}
 Q\ge n/4\ge r/32.
\end{equation}
To obtain the mixed term, choose the largest power-of-two search
size that fits the sparsity, namely $m$ from
Eq.~\eqref{hsl:eq:mrange}. Lemma~\ref{hsl:lem:mixed} also requires
$m\ge16n/\tau$, so we distinguish whether the supplied sparsity
is large enough. If $s\ge64n/\tau$, then
$m>s/4\ge16n/\tau$ and that lemma applies. Substituting
$m>s/4$ and $n\ge r/8$ gives
\begin{equation}\label{hsl:eq:mixrbound}
 Q\ge\frac{\log2}{96}\sqrt{mn\tau}
 \ge\frac{\log2}{384\sqrt2}\sqrt{s\tau r}.
\end{equation}
If $s<64n/\tau$, we do not claim that the decorated lemma applies.
Instead, the prospective mixed term is itself bounded by a constant
times the scalar-chain length. Using $r\le8n$ gives
$\sqrt{s\tau r}<\sqrt{512}\,n=16\sqrt2\,n$, hence
\begin{equation}\label{hsl:eq:failedmix}
 Q\ge\frac1{64\sqrt2}\sqrt{s\tau r}.
\end{equation}
Thus the failure of the search-size condition does not leave a gap:
the precision-only lower bound pays for the mixed term in that regime.
Both constants obtained for the mixed term, as well as the precision
constant $1/32$, exceed the evolution-time constant $c_T$ of
Lemma~\ref{hsl:lem:time}. The worst-case query count must dominate
all three witnessing family bounds, so taking their maximum gives
\begin{equation}
 Q\ge\frac{c_T}{3}
             (\sqrt{s}\tau+\sqrt{s\tau r}+r)
 \qquad(L\ge64\tau).
\end{equation}

\emph{When evolution time dominates precision.}
For $L<64\tau$, we need no long-chain signal estimate. Since
$h\ge1$, we have $r\le L<64\tau$, and therefore
\begin{equation}
 \sqrt{s\tau r}\le8\sqrt{s}\tau,\qquad
 r\le64\tau\le12\sqrt{s}\tau\quad(s\ge32).
\end{equation}
The three-term target is consequently at most $21\sqrt{s}\tau$.
The evolution-time lower bound alone is at least $c_T/21$ times
that target, so it establishes the joint dependence in this remaining
regime. The stated $c_0=\log2/(4096\pi)$ is smaller
than both $c_T/3$ and $c_T/21$.

\emph{Completing the exact promises and dimension bound.}
The preceding arguments gave the query bounds using matrices with
sparsity at most $s$. We now make the maximum sparsity exactly $s$
and check that the other promises can hold simultaneously.
The decorated family has dimension $(n+1)m+2$ and the scalar family
has dimension $n+3$. Their clock lengths are
$\mathcal O(\tau+r+1)$ and $m=\mathcal O(s)$.
Append, if needed, the known $s\times s$ block
$2^{-\lceil\log_2s\rceil}\mathbf1\mathbf1^{\mathsf T}$.
Here $\mathbf1$ is the all-ones vector of length $s$. The block
has one nonzero eigenvalue $s\,2^{-\lceil\log_2s\rceil}\le1$,
dyadic entries, and exactly $s$ nonzeros in every row and column.
Appending it as a direct summand therefore makes the maximum sparsity
exactly $s$ without increasing the norm above one. The input has zero
amplitude in this block, so it changes neither the simulated state
used by the reduction nor its readout, and its known values cost no
Boolean queries. The isolated energy-one vertex still makes the
operator and maximum row/column norms exactly one. Further known
zero padding equalises dimensions or accommodates a qubit register.
The total dimension is still $\mathcal O(s(\tau+r+1))$, including
this additional $s$-dimensional block. All supports are known and all
lists can be padded with known distinct zero positions; their in-place
locator extensions remain public. This completes the access, exact
normalisation and dimension claims in the theorem.
\end{proof}

We can now read off what the joint dependence means at fixed time.
Taking $\tau=1$ in the matched bound gives
\begin{equation}\label{hsl:eq:fixedtime}
 \Omega\!\left(\sqrt{s}+
       \sqrt{\frac{sL}{\log(e+L)}}+\frac{L}{\log(e+L)}\right).
\end{equation}
The improvement from $\sqrt{sL}$ to
$\sqrt{sL/\log(e+L)}$ closes the former
$\sqrt{\log\log(1/\eps)}$ gap at fixed positive time when the
mixed term dominates. It does not permit a smaller fixed power of
the precision logarithm uniformly in sparsity.

\begin{corollary}[The mixed precision exponent]\label{hsl:cor:power}
For any fixed $p<1/2$, a uniform bound
$\mathcal O(\sqrt{s}\,L^p+L)$ at $\tau=1$ is impossible, even
with a fixed additional polylogarithmic factor in $s$.
\end{corollary}
\begin{proof}
The proposed bound is uniform in sparsity, so we may let sparsity
grow with precision. Set $s=\lceil L^2\rceil$ and let
$L\to\infty$. The mixed term in Eq.~\eqref{hsl:eq:fixedtime}
then gives lower bound $\Omega(L^{3/2}/\sqrt{\log L})$.
The proposed cost, however, is $\mathcal O(L^{1+p}+L)$.
The largest power of $L$ in the latter expression is
$\max\{1+p,1\}<3/2$, so the ratio of the lower bound to the
proposed cost diverges as a positive power of $L$, up to logarithms.
Multiplying the proposed cost by any fixed power of
$\log s=\Theta(\log L)$ cannot remove that separation.
\end{proof}

The optimal mixed term still contains a divisor
$\sqrt{\log(e+L/\tau)}$, so the stronger lower bound
$\Omega(\sqrt{s\tau L})$ cannot hold uniformly: our upper bound
is smaller in the relevant joint regimes. There is also a direct reason why factorial decay appears in this
particular parity-amplitude construction. The first $n-1$ powers
of the Hamiltonian cannot involve all $n$ independent input blocks,
so they cannot contribute to correlation with their full parity. If $\norm{H_X}\le1$
and every entry is linear in the membership bits of independent
blocks, then
\begin{equation}
 \mathbb E_X[\sigma(X)H_X^k]=0\qquad(k<n).
\end{equation}
To see this, expand a matrix entry of $H_X^k$ as a sum of products
of $k$ entries of $H_X$. Since each entry is linear in the membership
bits, each resulting monomial depends on at most $k$ of the independent
input blocks. For $k<n$, at least one block is absent. Averaging its
balanced marked-position parity first makes the term zero. Thus the
Taylor expansion of the evolution contributes to the parity average
only from order $n$ onwards, giving
\begin{equation}
 \left\|\mathbb E_X[\sigma(X)e^{-\ii\tau H_X}]\right\|
 \le\sum_{k\ge n}\frac{\tau^k}{k!}
 \le e^\tau\frac{\tau^n}{n!}.
\end{equation}
For completeness, the last inequality follows by writing $k=n+j$
and using $(n+j)!\ge n!j!$:
\begin{equation}
 \sum_{k\ge n}\frac{\tau^k}{k!}
 \le\frac{\tau^n}{n!}\sum_{j\ge0}\frac{\tau^j}{j!}
 =e^\tau\frac{\tau^n}{n!}.
\end{equation}
This explains why the detectable order has the scale $r$ in these
linear-entry parity-amplitude encodings. The argument concerns this
specific structure of a reduction; it is not a restriction on all
possible Hamiltonian-simulation lower-bound methods.

\subsection{Two-norm lower bounds by enlargement}
\label{twonorml:sec:lower}

We finally impose a maximum-entry promise in addition to the row
and column Euclidean-norm promise. The hard family above has both
scales equal to one. We will replace each old entry by a constant
$m\times m$ block and multiply its value by a known factor $\mu$.
In this subsection, $m$ denotes this enlargement factor, not the
search-instance size used earlier; it can be any positive integer.

Repeating an entry $m$ times along a row multiplies that row's squared
Euclidean norm by $m$, while coherent action on a uniform state of
the added register multiplies the Hamiltonian by $m$. These different
scalings are exactly what we need: the new row norm scales as
$\mu\sqrt m$, but the evolution time seen by the old hard problem
scales as $\mu m|t|=\eta^2|t|/\mu$. We will make that reduction
explicit and obtain a matching lower bound for the aligned parameter
family in the proposition. We do not assert a match for every real
pair of norm promises.

\begin{proposition}[Aligned two-norm lower bound]\label{hsl:prop:two-norm}
Let $m\ge1$ and $s_0\ge32$ be integers, let $\mu>0$ be dyadic,
and set
\begin{equation}\label{twonorml:eq:parameters}
 \begin{aligned}
 s&:=ms_0,\qquad \eta:=\mu\sqrt m,\qquad
 \theta:=\mu m|t|=\eta^2|t|/\mu\ge1,\\
 L&:=\log(1/\eps),\qquad
 r_\mu:=L/\log(e+L/\theta).
 \end{aligned}
\end{equation}
For $0<\eps\le1/64$, there are real symmetric dyadic Hamiltonians
with maximum sparsity exactly $s$, maximum entry magnitude exactly
$\mu$, and maximum row and column Euclidean norms exactly $\eta$,
requiring
\begin{equation}\label{twonorml:eq:lower}
 \Omega\!\left(\sqrt{s}\,\eta|t|
             +\sqrt{s\mu|t|r_\mu}+r_\mu\right)
\end{equation}
queries. Their location lists and full in-place extensions are input
independent. The output need only satisfy the known-input system
trace-distance guarantee~\eqref{hsl:eq:accuracy}, with physical time
$t$ replacing $\tau$ in the target state, permitting mixed states
but no postselection.
\end{proposition}
\begin{proof}
We first take $t>0$; for negative time we negate the Hamiltonian,
which preserves all sparsity, magnitude and norm promises. Let
$H_X^{(0)}$ be a hard family from Theorem~\ref{hsl:thm:main} at
sparsity $s_0$ and time $\theta$, including the energy-one vertex
and exact-sparsity padding. If its dimension is $N_0$, the enlarged
matrix will act on $\mathbb C^{N_0}\otimes\mathbb C^m$.
Define
\begin{equation}\label{twonorml:eq:enlargement}
 \ket{u_m}:=\frac1{\sqrt m}\sum_{a=1}^m\ket a,\qquad
 H_X:=\mu H_X^{(0)}\otimes\mathbf1_m\mathbf1_m^{\mathsf T}
     =\mu mH_X^{(0)}\otimes\ket{u_m}\bra{u_m}.
\end{equation}
The two expressions agree because
$\mathbf1_m\mathbf1_m^{\mathsf T}=m\ket{u_m}\bra{u_m}$.
In the basis labelled by $(i,a)$, its entries are simply
\begin{equation}
 (H_X)_{(i,a),(j,b)}=\mu(H_X^{(0)})_{ij},
 \qquad 1\le a,b\le m.
\end{equation}
Every original nonzero is repeated $m$ times in each enlarged row,
so the maximum sparsity becomes $ms_0=s$. Similarly,
\begin{equation}
 \sum_{j,b}|(H_X)_{(i,a),(j,b)}|^2
 =\mu^2m\sum_j|(H_X^{(0)})_{ij}|^2,
\end{equation}
which multiplies every row Euclidean norm by $\mu\sqrt m$;
the same calculation applies to columns. Since the original operator
norm is one, all its entries have magnitude at most one. Its isolated
energy-one block attains both new promises after enlargement: it
becomes an $m\times m$ block with every entry $\mu$ and every
row and column norm $\mu\sqrt m=\eta$. Thus both maxima are
exact, not merely upper bounds. The entries remain dyadic because
the encoded values are multiplied by the dyadic number $\mu$;
there is no division by $m$ in this entry formula.

We now connect accurate evolution of the enlarged system to the old
hard problem. Let $V:\mathbb C^{N_0}\to\mathbb C^{N_0}\otimes
\mathbb C^m$ be the known isometry $V\psi:=\psi\otimes u_m$.
The all-ones block acts on $u_m$ with eigenvalue $m$, so
\begin{equation}\label{twonorml:eq:embedded-evolution}
 H_XV=V(\mu mH_X^{(0)}),\qquad
 e^{-\ii tH_X}V=Ve^{-\ii\theta H_X^{(0)}}.
\end{equation}
The first identity shows that the image of $V$ is invariant under
$H_X$; applying the identity to each power in the exponential gives
the second. Thus, on the known input
$V\ket{\psi_{\rm in}^{(0)}}$, the desired enlarged evolution is
the old evolution at time $\theta$, tensored with the unchanged
state $u_m$. We can prepare that added state without matrix queries.
After running an enlarged-system simulation, we discard the added
register. Trace distance contracts under this operation, so an
$\eps$-accurate enlarged output yields an $\eps$-accurate old
output, even when the simulation output is mixed. This gives a
reduction between the stated output tasks, not just an identity
between the ideal unitaries.
A value query to $H_X$ costs at most two old value queries:
load the old entry, compute its product with the known dyadic $\mu$
in the output encoding, and uncompute. The enlarged supports and
locator extensions are known. Thus the reduction has constant query
overhead, including inverses.

The old family is being evolved for time $\theta$, so its precision
parameter is $r_\mu=L/\log(e+L/\theta)$. Its lower bound is
therefore $\Omega(\sqrt{s_0}\theta+\sqrt{s_0\theta r_\mu}+r_\mu)$.
The constant overhead in simulating each enlarged value query
preserves this asymptotic bound.
Using
\begin{equation}\label{twonorml:eq:scaling}
 \sqrt{s_0}\,\theta=\sqrt{s}\,\eta|t|,\qquad
 s_0\theta=s\mu|t|
\end{equation}
gives Eq.~\eqref{twonorml:eq:lower}. As before, the maximum of the
three witnessing family bounds supplies their sum up to a constant.
\end{proof}

The enlarged operator norm is exactly $\mu m$, which can exceed
$\eta=\mu\sqrt m$. This is consistent with the promise: a bound
on every row and column Euclidean norm is not the same bound on the
operator norm. The uniform added-register state accumulates the $m$
identical entries coherently, which is why the old evolution runs for
time $\theta=\mu m|t|$. In particular, its precision denominator
must contain this time scale, rather than $\eta|t|$. To make that necessity
explicit, enlarge the two-sparse signed chain of
Lemma~\ref{hsl:lem:precision}, let $m$ tend to infinity through powers
of four, and choose
\begin{equation}\label{twonorml:eq:counterexample-parameters}
 s=2m,\qquad\eta=t=1,\qquad\mu=m^{-1/2},\qquad
 L=\sqrt m\log m.
\end{equation}
The choice of powers of four makes $\sqrt m$ a power of two,
so $\mu=m^{-1/2}$ is dyadic and the construction retains exact
dyadic entries. We set $\eps=e^{-L}$; the required error range holds
for all sufficiently large $m$. If we changed the coefficient of the
mixed term from $\eta$ to $\mu$ but retained the old precision
parameter based on $\eta|t|$, the resulting expression would be
\begin{equation}\label{twonorml:eq:wrong-denominator}
 \widetilde r:=\frac{L}{\log(e+L/(\eta|t|))}=\Theta(\sqrt m),
 \qquad
 \sqrt{s}\,\eta|t|+\sqrt{s\mu|t|\widetilde r}+\widetilde r
 =\mathcal O(\sqrt m).
\end{equation}
To check these orders, the denominator of $\widetilde r$ is
$\log(e+\sqrt m\log m)=\Theta(\log m)$, so
$\widetilde r=\Theta(\sqrt m)$. The time term is of order
$\sqrt m$, and the mixed term is
$\sqrt{(2m)m^{-1/2}\widetilde r}=\Theta(\sqrt m)$ as well.
Thus the proposed expression really would give cost
$\mathcal O(\sqrt m)$.

On the uniform added-register state, however, the original scalar
chain evolves for time $\theta=\sqrt m$. Here
$L/\theta=\log m$, so the correct precision quotient is
\begin{equation}
 r_\mu=\frac{\sqrt m\log m}{\log(e+\log m)}.
\end{equation}
For sufficiently large $m$, $L\ge64\theta$, and the chain
length $n=\lfloor r_\mu/4\rfloor$ satisfies the scalar-chain
hypotheses by Eqs.~\eqref{hsl:eq:nrange}--\eqref{hsl:eq:precisionchoice}.
The scalar-chain construction and its enlargement therefore require
\begin{equation}\label{twonorml:eq:necessary-denominator}
 \Omega(r_\mu)
 =\Omega\!\left(\frac{\sqrt m\log m}{\log\log m}\right)
\end{equation}
queries. The ratio to $\sqrt m$ grows as
$\log m/\log\log m$, contradicting the proposed uniform bound. Hence changing only the mixed-term coefficient from $\eta$ to $\mu$, while keeping the old precision parameter throughout, cannot be a uniform upper bound.

\section{Linear differential equations}\label{part:ode}

We now use Hamiltonian simulation to prepare a quantum state
proportional to the solution of a sparse linear differential equation.
Two operations are needed. We must propagate an initial vector and
the vectors supplied by a source, and then add the resulting vectors
with their correct relative amplitudes. The second operation can
produce substantial cancellation. Its cost therefore depends on the
norm of the final solution relative to the total input and source
norms, even when each propagation is easy.

We begin with dissipative generators, whose homogeneous propagators
cannot increase a vector's norm. A linear combination of Hamiltonian
simulations (LCHS) represents each such propagator as a weighted sum
of unitary evolutions. We implement all terms coherently, combine the
propagated initial and source states in a selected output component,
and amplify that component. This separates the cost of propagation
from the cost of preparing the normalised answer: sparse matrix
queries implement the Hamiltonian evolutions, while their coefficient
sum and the amount of cancellation determine the success amplitude.

We develop this construction first for constant coefficients and then
for time-dependent inputs. For a general constant generator, which
may amplify intermediate states, we give a separate construction:
a linear system stores the trajectory, and its final-time components
yield the answer. Finally, we return to the dissipative algorithm and
show how to implement its coherent selections and amplification with
finite gate circuits. Section~\ref{part:odel} will use the same
normalisation parameter to show that propagation and cancellation
can force a joint query cost.

\subsection{Dissipative propagation and the forced solution}
\label{ode:sec:main}

Consider
\begin{equation}\label{ode:eq:ode}
 \dot x(t)=A(t)x(t)+b(t),\qquad x(0)=x_0,\qquad0\le t\le T,
\end{equation}
where $A(t)\in\mathbb C^{N\times N}$ and $x(t),b(t)\in\mathbb C^N$.
The matrix has row and column sparsity at most $s$, with the native
value and in-place locator access of Section~\ref{sec:overview}.
The value oracle returns an encoded entry, while the row and column
locators reversibly convert a list position into the corresponding
matrix index. Both locator directions and inverses are supplied.
We begin with constant $A,b$; the later extension specifies coherent
time access and a finite quadrature for the source integral. The
target is the state $x(T)/\norm{x(T)}$, for $x(T)\ne0$, accompanied
by a classical flag identifying successful preparation. We do not
ask the algorithm to output the entries of the solution as classical
data.

For constant inputs, supplied circuits prepare $x_0/\norm{x_0}$
and $b/\norm b$ whenever those vectors are nonzero. Their inverses
and coherent controls on known branch-register predicates are also
supplied. We state the enabled subspaces explicitly when constructing
the source sum below; we do not infer control of an opaque preparation
oracle. Each controlled preparation or inverse counts as one
input/source call. The norms $r_0:=\norm{x_0}$ and $r_b:=\norm b$
are supplied exactly, along with the bounds in the theorem.

For dissipativity, decompose the negative generator as
\begin{equation}\label{ode:eq:split}
 G(t):=-A(t)=D(t)+\ii H(t),\qquad
 D(t):=\frac{G(t)+G(t)^\dagger}{2}\succeq0,\qquad
 H(t):=\frac{G(t)-G(t)^\dagger}{2\ii}.
\end{equation}
Here $D$ denotes the dissipative part, called $L$ in Ref.~\cite{LS},
and has no connection to a private transducer block. Positivity
implies contraction of the homogeneous evolution, because
\begin{equation}
 \frac{d}{dt}\norm{x(t)}^2
       =-2\langle x(t),D(t)x(t)\rangle\le0
 \qquad(b=0).
\end{equation}
Let $\eta_D,\eta_H$ be supplied uniform row and column Euclidean-norm
bounds for $D,H$. A common bound $\eta$ for $A$ permits
$\eta_D=\eta_H=\eta$ by the triangle inequality, using both row
and column bounds. We keep these bounds separate because the unitary evolutions in
LCHS have Hermitian generators $H+kD$, with a real frequency label
$k$. Increasing the frequency range enlarges the contribution of
$D$, while the contribution of $H$ remains unchanged. The logarithm
in the dissipative time term of the next theorem comes from precisely
this distinction.

\begin{theorem}[Dissipative constant-coefficient equations]
\label{ode:thm:main}
Let $A,b$ be constant in Eq.~\eqref{ode:eq:ode}, with $T\ge0$,
$x(T)\ne0$, the sparse and state-preparation access above, and
Eq.~\eqref{ode:eq:split}. Supply $\bar g\ge g$, where
\begin{equation}
 B:=\norm{x_0}+T\norm b,\qquad g:=B/\norm{x(T)}.
\end{equation}
For $0<\eps\le1/2$, set
\begin{equation}\label{ode:eq:theta}
 \Theta:=T\eta_H+T\eta_D\log(10^4\bar g/\eps).
\end{equation}
An algorithm prepares the normalised solution to Euclidean error
at most $\eps$, up to global phase, conditional on a classical
success flag of constant probability. It uses
\begin{equation}\label{ode:eq:main}
 \mathcal O\!\left(\bar g[
       \sqrt{s}\,\Theta+\sqrt{s\Theta\log(\bar g/\eps)}+\log(\bar g/\eps)]\right)
\end{equation}
matrix queries and $\mathcal O(\bar g)$ input/source calls, with
absolute constants and no suppressed logarithm in $s$ or $T$.
\end{theorem}

The quantity $B$ is an upper bound on the final norm. Indeed,
variation of constants writes the solution as the propagated initial
vector plus an integral of propagated source vectors, and contractivity
bounds their norms by $\norm{x_0}$ and $T\norm b$. Thus $g\ge1$.
Decay or cancellation can make $\norm{x(T)}$ much smaller than $B$;
our preparation then has success amplitude of order $1/g$, and
amplification costs order $\bar g$. The supplied upper bound
$\bar g$ lets us choose that amplification schedule without knowing
the final norm. The other scale, $\Theta$, comes from the largest
Hamiltonian selected in the LCHS sum. To compare the two
contributions, write $\tau_H=T\eta_H$ and $\tau_D=T\eta_D$.
Expanding the square root gives the equivalent order
\begin{equation}\label{ode:eq:expanded}
 \begin{split}
 \mathcal O\!\Bigl(\bar g\bigl[
 &\sqrt{s}\tau_H+\sqrt{s\tau_H\log(\bar g/\eps)}\\
 &+\{\sqrt{s}(\tau_D+\sqrt{\tau_D})+1\}
       \log(\bar g/\eps)\bigr]\Bigr).
 \end{split}
\end{equation}
For a common bound $\eta$ and $\tau:=T\eta$, this becomes
\begin{equation}\label{ode:eq:simple}
 \mathcal O\!\left(\bar g[
             \sqrt{s}(\tau+\sqrt\tau)+1]\log(\bar g/\eps)\right),
\end{equation}
which is $\mathcal O(\bar g\sqrt{s}\tau\log(\bar g/\eps))$ when $\tau\ge1$.
With $\eta_D=0$, the separated expression retains exactly the
baseline Hamiltonian precision structure, multiplied by $\bar g$.
At zero time or for a known zero generator, the solver assembles
the initial and source terms directly. It makes no matrix queries
outside the supplied input and source preparations.

We next build the Hamiltonian simulations used by the solver.
Although $H+kD$ is specified by $A$, its nonzero entries may occupy
both the support of $A$ and its transpose. Given indices $i,j$, we
can compute the required entry from two value queries. A reversible
list of the union is a different requirement: overlapping entries
must be identified, and a matrix index must be mapped back to its
position in that new list. Our oracle assumptions do not supply this
union locator. We instead keep the two original lists separate and
attach a bit identifying which list supplied an entry. Where the
lists overlap, we divide the entry's weight between the two copies.
This preserves the original reversible access and will also preserve
the row-norm bound needed by the simulator.

The simulation construction of Section~\ref{sec:methods} uses an
isometry $J$ from the $N$-dimensional data space into a larger private
space, together with a Hermitian operator $K_0$ on that space, as in
Eq.~\eqref{hs:eq:JK}; Lemma~\ref{hs:lem:abstract} then turns the
resulting resolvent bounds into a simulation. The image of $J$ specifies where the data enter and leave the private
operation. We need to make this boundary a known, query-free
preparation after changing private coordinates. For a branch
Hamiltonian $K$, let $\rho>0$ be a common row and column
Euclidean-norm bound and let $p$ count the list states per input.
The tagged construction below will use $p=2s$. The three required
estimates are
\begin{equation}\label{ode:eq:moments}
 J^\dagger K_0J=K/(\rho\sqrt p),\qquad
 J^\dagger K_0^2J\preceq I,\qquad\norm{K_0}\le\sqrt p,
\end{equation}
We may choose a rescaling
$0<\gamma\le1$; it balances the evolution and filter costs below.
The associated private Cayley unitary
\begin{equation}\label{ode:eq:interface}
 C_0:=(\gamma K_0+\ii I)(\gamma K_0-\ii I)^{-1},\qquad
 \alpha:=\rho\sqrt p/\gamma,\qquad
 Q:=\lceil\max\{1,\gamma\sqrt p\}\rceil
\end{equation}
has the response required by that construction. To see how the
moment conditions enter, use
$(I-C_0)^{-1}=(I+\ii\gamma K_0)/2$. Compressing between the
boundary maps gives
\begin{equation}
 J^\dagger(I-C_0)^{-1}J=(I+\ii K/\alpha)/2.
\end{equation}
The norm that controls the private work is the norm of the response
applied to the boundary, before compressing its output. Its square is
bounded through the second moment:
\begin{equation}
 J^\dagger[(I-C_0)^{-1}]^\dagger(I-C_0)^{-1}J
 =\frac{I+\gamma^2J^\dagger K_0^2J}{4}
 \preceq\frac{1+\gamma^2}{4}I.
\end{equation}
Thus the first moment supplies the desired Hamiltonian, and the
second bounds the work needed to realise its evolution. The last
condition in \eqref{ode:eq:moments} supplies the scale $Q$ for the
filter. If $C_0$ has constant native-query cost, the baseline
construction gives
\begin{equation}\label{ode:eq:hs}
 \mathcal O\!\left(\sqrt p\,\rho t+
       \sqrt{p\rho t\log(1/\delta)}+\log(1/\delta)\right)
\end{equation}
queries for error on all output registers $0<\delta\le1/2$.
The chronological theorem gives the same result under uniform
sampled-time promises. We will choose the common scale and accuracy
after forming the solver. In particular, we do not assume that the
optimal time-independent refinement automatically supplies a uniform
coherently selected ODE circuit; that would require a further
implementation and cost analysis.

\begin{lemma}[Two tagged copies of the original support]
\label{ode:lem:tags}
Let $G=-A$ have the supplied sparse access. Suppose $K=K^\dagger$
is supported on
$\operatorname{supp}(G)\cup\operatorname{supp}(G^\dagger)$, where
$\operatorname{supp}(G):=\{(i,j):G_{ij}\ne0\}$, and its entries
are computable from $G_{ij},G_{ji}$ and known data. If its row and
column Euclidean norms are at most $\rho>0$, then
Eq.~\eqref{ode:eq:moments} holds with $p=2s$. In suitable private
coordinates, $J$ is known and $C_0$ has constant sparse-query cost,
without a locator for the support union.
\end{lemma}
\begin{proof}
We first construct the embedding and check its moments, and then
describe the unitary that implements its Cayley response. Write the
original row and column lists as $r_i(l)$ and $c_j(l)$, for
$1\le l\le s$. The label $r_i(l)$ is a column index in row $i$;
$c_j(l)$ is a row index in column $j$. A tag $a\in\{0,1\}$
records whether an edge came from the original support or its
transpose. The edge space has basis $\ket{i,j,a}$ and dimension
$2N^2$. On this space define the unnormalised list embeddings
\begin{align}
 E_L\ket i&:=\sum_{l=1}^s
       (\ket{i,r_i(l),0}+\ket{i,c_i(l),1}),\\
 E_R\ket j&:=\sum_{l=1}^s
       (\ket{c_j(l),j,0}+\ket{r_j(l),j,1}).
\end{align}
For a left input $i$, the first embedding lists all edges leaving
row $i$ in either support. For a right input $j$, the second lists
the same edges by their column endpoint. Each input has $2s$ list
states, including the permitted padding. Within a list the supplied
reversible locator has distinct outputs, and the two tags remain
orthogonal even when their matrix indices coincide. Distinct left
inputs also have distinct first indices; distinct right inputs have
distinct second indices. Hence
$E_L^\dagger E_L=E_R^\dagger E_R=pI$ for $p=2s$.

An edge that belongs to both supports now appears twice. Assigning
the full value $K_{ij}$ to both copies would double that entry.
We instead divide it equally among its active copies. This recovers
the entry in the first moment and can only decrease the sum of
squared magnitudes used in the second moment. Set
\begin{equation}
 \chi_0(i,j):=\boldsymbol1_{G_{ij}\ne0},\qquad
 \chi_1(i,j):=\boldsymbol1_{G_{ji}\ne0},\qquad
 m_{ij}:=\chi_0(i,j)+\chi_1(i,j),
\end{equation}
and let $V$ be diagonal on tagged edges with entry
\begin{equation}\label{ode:eq:tagvalue}
 v_{ija}:=\begin{cases}
       \chi_a(i,j)K_{ij}/m_{ij},&m_{ij}>0,\\
       0,&m_{ij}=0.
       \end{cases}
\end{equation}
With one nonzero support the active tag gets all of $K_{ij}$;
with two, each gets half. Thus
\begin{equation}
 E_L^\dagger VE_R=K,\qquad
 \sum_av_{ija}=K_{ij},\qquad
 \sum_a|v_{ija}|^2=|K_{ij}|^2/m_{ij}\quad(m_{ij}>0).
\end{equation}
Every active tagged entry occurs in both relevant lists. The
indicators kill zero-padding entries even when only one list contains
them, and cancellation to $K_{ij}=0$ is retained exactly.
The matrices $E_L^\dagger VV^\dagger E_L$ and
$E_R^\dagger V^\dagger VE_R$ are consequently diagonal and at most
$\rho^2I$; also $|v_{ija}|\le\rho$.
We introduce one more two-dimensional register to distinguish the
left and right orientations. It is different from the support tag:
the tag distinguishes the two copies of an edge, while the orientation
is the register on which the Cayley rotation acts. On the direct sum
of the two oriented edge spaces, set
\begin{equation}
 J:=\frac1{\sqrt{2p}}\binom{E_L}{E_R},\qquad
 K_0:=\frac{\sqrt p}{\rho}
             \begin{pmatrix}0&V\\V^\dagger&0\end{pmatrix}.
\end{equation}
The normalisation gives $J^\dagger J=I$. For the first moment,
the two off-diagonal blocks contribute
\begin{equation}
 J^\dagger K_0J
 =\frac{E_L^\dagger VE_R+E_R^\dagger V^\dagger E_L}
        {2\rho\sqrt p}
 =\frac{K+K^\dagger}{2\rho\sqrt p}
 =\frac K{\rho\sqrt p}.
\end{equation}
Squaring $K_0$ makes it block diagonal, so the second moment is
\begin{equation}
 J^\dagger K_0^2J
 =\frac{E_L^\dagger VV^\dagger E_L+
          E_R^\dagger V^\dagger VE_R}{2\rho^2}
 \preceq I.
\end{equation}
Finally, $V$ is diagonal and $|v_{ija}|\le\rho$, which gives
$\norm{K_0}=\sqrt p\,\norm V/\rho\le\sqrt p$.
These are exactly the three estimates in \eqref{ode:eq:moments}.

For the circuit, let $F$ select an original locator using the
orientation and tag. In the $E_L$ orientation, tag zero uses the
row locator and tag one the column locator; the choices reverse
in the $E_R$ orientation, with known index routing. Then
$J=FJ_{\rm free}$, where the known preparation makes uniform
orientation, tag and list labels. Full unitary extensions make this
a valid coordinate change on the entire private space.
In edge coordinates, put $z:=\gamma\sqrt p\,v_{ija}/\rho$.
The full Cayley unitary on the two orientations is
\begin{equation}
 \frac{(|z|^2-1)I+
       2\ii\begin{pmatrix}0&z\\\overline z&0\end{pmatrix}}
       {1+|z|^2}.
\end{equation}
For each fixed $i,j,a$, this formula is a two-by-two unitary on
the orientation register. We query $G_{ij}$ and $G_{ji}$, compute
$v_{ija}$ and hence $z$, and use the resulting number to rotate
that register. The indices $i,j$, the support tag and both loaded
entry encodings remain unchanged during the rotation. We can
therefore reverse the arithmetic and the two value queries even
though the orientation is now in a different state. More explicitly,
let $\ket{\xi}$ be an arbitrary orientation state, let
$w_{ij}$ denote the pair of entry encodings, and let $C_{ija}$
denote the displayed two-by-two unitary. After uncomputing the
temporary arithmetic, the query--rotation--inverse-query sequence is
\begin{equation}
\begin{aligned}
 \ket{i,j,a}\ket\xi\ket0
 &\longmapsto\ket{i,j,a}\ket\xi\ket{w_{ij}}\\
 &\longmapsto\ket{i,j,a}(C_{ija}\ket\xi)\ket{w_{ij}}\\
 &\longmapsto\ket{i,j,a}(C_{ija}\ket\xi)\ket0.
\end{aligned}
\end{equation}
By linearity this implements the direct sum over all edge labels,
including on coherent superpositions of them. The fixed labels make
the uncomputation valid on every branch.

In the original list coordinates the operation is $F^\dagger C_0F$,
and its boundary is $J_{\rm free}$. Thus all the boundary preparation
is known, while a constant number of original locator calls moves
into and out of edge coordinates. Lemma~\ref{access:lem:control} gives
the same order for coherent controls and inverses.
\end{proof}

Apply the lemma to $K_k:=H+kD$, whose entry is
\begin{equation}\label{ode:eq:Kentry}
 (K_k)_{ij}
 =\frac{k(G_{ij}+\overline{G_{ji}})
                 +(G_{ij}-\overline{G_{ji}})/\ii}{2}.
\end{equation}
The identity
$H+kD=\tfrac12[(k-\ii)G+((k-\ii)G)^\dagger]$ is also used for
LCHS block encodings in Refs.~\cite{PJKW25,LS}.
Our tagged construction retains the original locators and the
separate norm bounds. For $|k|\le R$, it uses the common bound
\begin{equation}\label{ode:eq:rho}
 \rho:=\eta_H+R\eta_D.
\end{equation}
The value of $k$ changes known arithmetic and rotations but not the
oracle sequence. In particular, cancellations in $H$ or $D$, including
$D=0$, survive the representation. We can now use this uniform
family to assemble the propagator and source.

\begin{proof}[Proof of Theorem~\ref{ode:thm:main}]
The tagged construction has supplied the unitary branches
$e^{-\ii u(H+kD)}$ using only the original sparse oracles. We now
turn them into a preparation of the forced solution. We first combine
frequencies to obtain one propagator block for every duration
$u\in[0,T]$. We then combine the propagated initial vector and a
quadrature of the source, and amplify the component containing their
sum. A numerator error of order $\nu B$ becomes an error of order
$\nu g$ after normalisation. This is why we choose an internal
tolerance proportional to $\eps/\bar g$:

\begin{equation}\label{ode:eq:nu}
 \nu:=\frac{\eps}{10^4\bar g},\qquad
 \beta^2:=1+\log\frac{1+1/(2\pi)}\nu,\qquad R:=2\beta^2.
\end{equation}
Here $\nu$ is an approximation tolerance and $\beta$ is the kernel
width of Low--Somma's formula; the transducer rescaling is still
$\gamma$. Theorems~2 and~3 of Ref.~\cite{LS}, with their parameter
$c=1$, give frequency-truncation and quadrature errors at most $\nu$
each. Uniformly for $0\le u\le T$, they provide
\begin{equation}\label{ode:eq:lchs}
 \left\|e^{-Gu}-\sum_jc_je^{-\ii u(H+k_jD)}\right\|\le2\nu,
 \qquad |k_j|\le R,\quad\sum_j|c_j|<3.
\end{equation}
We specify a common grid and verify its coefficient bound below.

\paragraph{The propagator block and the source vector.}
The coefficient bound in \eqref{ode:eq:lchs} lets us use the same
normalisation for every duration. If $C_{\rm LCU}=\sum_j|c_j|$, add two identity terms with coefficients $(3-C_{\rm LCU})/2$ and $-(3-C_{\rm LCU})/2$. Their sum is zero, while their absolute weights bring the total
to three. Include these terms in the frequency label and prepare
\begin{equation}
 P_f\ket0=\sum_j\sqrt{|c_j|/3}\ket j.
\end{equation}
On label $j$, apply the phase $c_j/|c_j|$ and the corresponding
unitary; zero coefficients are omitted. If those unitaries were
exact, preparation, selection and $P_f^\dagger$ would give the
selected data block
\begin{equation}
 (\bra0 P_f^\dagger\otimes I)
 \left(\sum_j\ket j\!\bra j\otimes
       \frac{c_j}{|c_j|}e^{-\ii u(H+k_jD)}\right)
 (P_f\ket0\otimes I)
 =\frac13\sum_jc_je^{-\ii u(H+k_jD)}.
\end{equation}
For the two added identity labels the selected unitary in this
notation is simply $I$. The two square-root weights multiply to
the coefficient magnitude, and the separate phase supplies its sign
or complex phase.

We implement each genuine evolution to complete-output error at most
$\nu$, using the tagged representation and one common sequence of
queries. Select the frequency register and simulation work registers
in their initial states, and denote the actual data block by
$\widetilde U(u)/3$. A controlled direct sum has error equal to
the largest branch error, even on superpositions of labels. Restoring
the normalisation three therefore contributes at most $3\nu$;
the analytic LCHS approximation contributes $2\nu$. Hence
\begin{equation}\label{ode:eq:properror}
 \norm{\widetilde U(u)-e^{Au}}\le5\nu.
\end{equation}
The frequency and duration change coherently inside the same schedule;
there is no factor from the number of quadrature nodes. Negative
durations, if selected, use controlled adjoints.

Variation of constants gives
\begin{equation}
 x(T)=e^{AT}x_0+\int_0^Te^{A(T-v)}b\,dv.
\end{equation}
For the source integral we use a left-endpoint rule at
$v_j=jT/m$. The derivative of its integrand is
$-Ae^{A(T-v)}b$, whose norm is at most $\norm A\norm b$ by
contractivity. On a cell of width $h=T/m$, replacing the integrand
by its left endpoint therefore costs at most
$\int_0^h\norm A\norm b\,u\,du=\norm A\norm b\,h^2/2$.
Summing the $m$ cells bounds the quadrature error by
\begin{equation}\label{ode:eq:sourcequad}
 \frac{T^2\norm A\norm b}{2m}.
\end{equation}
Choose $m$ large enough to make this at most $\nu B$.
A supplied common row/column bound gives $\norm A\le\sqrt{s}\eta$;
if only $\eta_D,\eta_H$ are supplied, take
$\eta=\eta_D+\eta_H$ for this mesh estimate. Increasing $m$ adds
source labels and known preparation work, but all labels select
propagators from the same circuit and do not multiply matrix queries.

We must now add vectors, rather than form an incoherent mixture of
their states. We give each term a source label: $\mathrm{init}$
for the initial vector and $j=0,\ldots,m-1$ for the quadrature
nodes. We prepare the square roots of the terms' norm weights,
because unpreparing that label will supply a second copy of each
square root. The weight preparation alone is
\begin{equation}\label{ode:eq:constant-weights}
 P_w\ket0:=\sqrt{r_0/B}\ket{\mathrm{init}}
      +\sum_{j=0}^{m-1}\sqrt{Tr_b/(mB)}\ket j.
\end{equation}
On the subspace with label $\mathrm{init}$, load the state
$x_0/r_0$; on the source-label subspace spanned by the $\ket j$,
load $b/r_b$. Thus the loading operation is the direct sum of the
supplied initial-state circuit and the supplied forcing circuit,
selected by those labels. If a larger circuit enables this preparation
with a bit $e$, we use
$\ket1\bra1_e\otimes U_{\rm load}+\ket0\bra0_e\otimes I$;
this specifies the additional control requirement operationally.
Zero-norm branches are omitted. The preparation phases are the phases
of the stated vectors, since unrelated global phases on branches
would alter their sum.

After loading and applying the selected propagator block, the
successful component is
\begin{equation}
 \frac13\left[
 \sqrt{r_0/B}\ket{\mathrm{init}}\widetilde U(T)\frac{x_0}{r_0}
 +\sum_{j=0}^{m-1}\sqrt{Tr_b/(mB)}\ket j
                     \widetilde U(T-v_j)\frac b{r_b}\right].
\end{equation}
Apply $P_w^\dagger$ to the label alone and select its initial state.
Each square-root coefficient is then supplied a second time, giving
\begin{equation}\label{ode:eq:constant-output}
 \frac y{3B}:=\frac1{3B}\left[
 \widetilde U(T)x_0+\frac Tm\sum_{j=0}^{m-1}
                                  \widetilde U(T-v_j)b\right].
\end{equation}
For example, the initial contribution acquires the weight
$(r_0/B)(x_0/r_0)=x_0/B$, while a source contribution acquires
$(Tr_b/(mB))(b/r_b)=Tb/(mB)$. The label projection has therefore
added the intended vectors. The data register is the desired output;
we do not undo its loading at this point. Inverting the loading
circuit would act on the propagated data and would not perform the
required vector sum.

We can now bound the numerator error before normalising it.
Each propagator approximation contributes its operator error times
the norm weight of its input. Those weights sum to $B$, so the
propagator errors contribute at most $5\nu B$. Adding the source
quadrature error gives
\begin{equation}\label{ode:eq:vectorerror}
 \norm{y-x(T)}\le6\nu B.
\end{equation}
The normalisation step is where cancellation affects accuracy.
For nonzero vectors $y,x$, the triangle inequality and
$|\norm y-\norm x|\le\norm{y-x}$ give
\begin{equation}
 \left\|\frac y{\norm y}-\frac x{\norm x}\right\|
 \le\left\|\frac y{\norm y}-\frac y{\norm x}\right\|
       +\frac{\norm{y-x}}{\norm x}
 \le\frac{2\norm{y-x}}{\norm x}.
\end{equation}
Our error bound is smaller than $\norm{x(T)}$, so $y\ne0$.
Applying the inequality to the actual prepared numerator gives
\begin{equation}
 \left\|\frac y{\norm y}-\frac{x(T)}{\norm{x(T)}}\right\|
 \le\frac{2\norm{y-x(T)}}{\norm{x(T)}}
 \le12\nu g<\eps/2.
\end{equation}
Its successful amplitude satisfies
$\norm y/(3B)\ge1/(4g)\ge1/(4\bar g)$ and is at most
$(1+6\nu)/3<1/2$. Randomised amplitude amplification with iteration
count bounded by $\mathcal O(\bar g)$, followed by a constant number
of heralded attempts, gives constant success probability. It uses
$\mathcal O(\bar g)$ calls to the preparation and its inverse, and
therefore the same order of input/source calls.

To see why amplification does not require a second error budget,
write the actual preparation output as
$a\ket{\mathrm{good}}+\sqrt{1-a^2}\ket{\mathrm{bad}}$, where
the good component has all selected work registers zero and data
direction $y/\norm y$. The Grover reflections act within the
two-dimensional span of these actual vectors. They change their
coefficients and leave the direction inside the good subspace
unchanged. The conditional error is therefore the error already
bounded for $y/\norm y$. We need not compare every Grover iterate
with an ideal one and pay another factor of $\bar g$ in accuracy.
This avoids an unnecessary fixed-point-amplification precision
logarithm; remaining finite-gate errors fit inside the unused budget.

\paragraph{A common grid and the total matrix cost.}
We now verify that the coefficients and branch schedule can be chosen
uniformly. Choose a frequency spacing $h_k$ with $R/h_k$ an integer
and
\begin{equation}\label{ode:eq:grid}
 h_k\le\frac{\pi}{\mathcal D/2+
                    \log(64e^{3/2}/(15\nu))},
 \qquad \mathcal D\ge\int_0^T\norm{D(v)}\,dv.
\end{equation}
For the constant problem the integral is $T\norm D$, and
$\mathcal D=T\sqrt{2s}\eta_D$ is sufficient. Every source
subinterval has no larger integrated norm, so this one grid and its
coefficients work for all selected durations. With
\begin{equation}
 \operatorname{erfc}(x):=\frac2{\sqrt\pi}
                             \int_x^\infty e^{-y^2}\,dy
 \qquad(x\ge0),
\end{equation}
the coefficient estimate of Ref.~\cite{LS} reads
\begin{equation}
 \begin{aligned}
 \sum_j|c_j|
 &\le e\operatorname{erfc}(1/(2\beta))+\frac{\nu}{1+2\pi}\\
 &\quad+\nu\exp\!\left[-\frac12
          \left(1+\int_0^T\norm{D(v)}\,dv\right)\right]<3.
 \end{aligned}
\end{equation}
Indeed $\operatorname{erfc}(x)\le1$, the integral is nonnegative,
and $\nu\le1/(2\cdot10^4)$, so the right side is at most
$e+2\nu<3$. This supplies the normalisation used in the circuit.

For the moment representation with parameters $p,\rho$, a duration
$u$ and branch error $\delta$, Lemma~\ref{hs:lem:abstract} gives
\begin{equation}\label{ode:eq:tradeoff}
 \mathcal O(W+Q\log(2/\delta)),\qquad
 W:=\frac{u\alpha(1+\gamma^2)}2
   =\frac{u\rho\sqrt p}{2}(\gamma+\gamma^{-1}).
\end{equation}
The rescaling trades two costs against each other. A smaller
$\gamma$ lowers $Q$, and hence the filtering cost, but increases
the $\gamma^{-1}$ contribution to $W$. Away from the endpoints,
balancing $\rho T/\gamma$ against
$\gamma\log(2/\delta)$ suggests
$\gamma=\sqrt{\rho T/\log(2/\delta)}$. Below
$p^{-1/2}$ the ceiling in $Q$ prevents a further asymptotic saving;
above one the $\gamma$ contribution to $W$ is already increasing.
We therefore take
\begin{equation}
 \gamma:=\min\left\{1,
    \max\left\{p^{-1/2},
           \sqrt{\frac{\rho T}{\log(2/\delta)}}\right\}\right\}
\end{equation}
This optimises the common rescaling for the maximum duration $T$
and gives Eq.~\eqref{ode:eq:hs} at that duration. Every smaller branch uses
the same registers, degree and query count. Its error is controlled
because, for the baseline filter
$f(z)=[(1+z+\cdots+z^{32Q-1})/(32Q)]^{k_{\rm filt}}$,
with integer power $k_{\rm filt}\ge2$,
\begin{equation}\label{ode:eq:filter}
 \norm{f(D_{\rm priv})}
 \le65(1+W/Q)e^{3W/(4Q)}2^{-k_{\rm filt}}
\end{equation}
has a right-hand side increasing in $W$. Here $D_{\rm priv}$ denotes the composed
private contraction, not the dissipative matrix $D$.
The separate gate theorem implements this common schedule efficiently
under additional arithmetic assumptions.

For $p=2s$, $\rho=\eta_H+R\eta_D$ and $\delta=\nu$, the branch
query count is
\begin{equation}\label{ode:eq:selectedcost}
 \mathcal O\!\left(\sqrt{s}\,T\rho+
       \sqrt{sT\rho\log(2/\nu)}+\log(2/\nu)\right).
\end{equation}
Since $R=\Theta(\log(10^4\bar g/\eps))$ and
$\log(2/\nu)=\Theta(\log(10^4\bar g/\eps))$, we have
$T\rho=\Theta(\tau_H+\tau_D\log(10^4\bar g/\eps))$
whenever this scale is nonzero.
Multiplying by the $\mathcal O(\bar g)$ amplification calls proves
Eq.~\eqref{ode:eq:main}, with $\mathcal O(\bar g)$ source calls.
The direct zero-time and zero-generator constructions cover the
remaining cases.
\end{proof}

\subsection{Coherent time dependence and source quadrature}
\label{ode:sec:time}

The constant-coefficient construction used a source label to choose
how long a fixed vector should evolve. With time-dependent inputs,
the same label must also choose which vector $b(v)$ to load and
which interval of the generator to simulate. The propagator from
$v$ to $T$ is now time ordered: matrices at different times need
not commute. The chronological Hamiltonian construction handles this
ordering. Our additional task is to specify source access that
really prepares a quadrature of the vector solution, including its
relative phases and normalisation.

The query bound can accommodate many quadrature nodes because their
labels are processed in superposition. This does not make an
arbitrary source cheaply preparable. We explicitly assume coherent
weight and data preparations, count their calls separately, and
postpone their gate costs to Section~\ref{ode:sec:gates}.

Let $U(T,v)$ denote the dissipative propagator and
$B:=\norm{x_0}+\int_0^T\norm{b(v)}\,dv$.
At any requested tolerance $\zeta>0$, suppose the input supplies
finitely many nodes $v_i\in[0,T]$, coefficients
$\omega_i\in\mathbb C$ for $1\le i\le m_{\rm src}$, and coherent
node and coefficient-phase evaluation. Put $r_{b,i}:=\norm{b(v_i)}$
and $r_0:=\norm{x_0}$. Their known absolute coefficient mass and
source-quadrature certificate must satisfy
\begin{equation}\label{ode:eq:source-interface-mass}
 \begin{aligned}
 0<B_q&:=r_0+\sum_{i=1}^{m_{\rm src}}|\omega_i|r_{b,i}\le c_qB,\\
 \left\|\int_0^TU(T,v)b(v)\,dv
       -\sum_{i=1}^{m_{\rm src}}\omega_iU(T,v_i)b(v_i)\right\|
 &\le\zeta B,
 \end{aligned}
\end{equation}
where $c_q$ is an absolute constant independent of accuracy.
These two requirements do different jobs. The first bounds the
total magnitude of the coefficients used in coherent assembly;
otherwise a quadrature with large cancelling weights could make
preparation arbitrarily unlikely. The second certifies the actual
vector integral after propagation. Accurate integration of
$\norm{b(v)}$ alone would not control that vector, whose direction
can vary with $v$. Both certificates are available at each accuracy
the algorithm chooses; neither is inferred merely from the number
of nodes.

Use label $0$ for the initial state and $i\ge1$ for the source nodes.
The supplied preparations act as
\begin{equation}\label{ode:eq:source-interface-preparation}
 \begin{aligned}
 P_w\ket0&=\sqrt{\frac{r_0}{B_q}}\ket0
   +\sum_{i=1}^{m_{\rm src}}
            \sqrt{\frac{|\omega_i|r_{b,i}}{B_q}}\ket i,\\
 U_{\rm load}\ket0\ket{0_{\rm data}}&=\ket0\frac{x_0}{r_0},\\
 U_{\rm load}\ket i\ket{0_{\rm data}}&=\ket i\frac{b(v_i)}{r_{b,i}}
                                           \quad(i\ge1).
 \end{aligned}
\end{equation}
Omit zero-weight branches and allow arbitrary unitary completions
on unused labels. The displayed vectors fix the loading phases:
independent phase choices for the source states would change the
solution sum. Apply the separate phase $\omega_i/|\omega_i|$ on
source label $i$, with phase one on the initial branch.
The preparations, their actual adjoints and coherent controls by
known branch and enable predicates are supplied, as for the constant
source layout. A constant number of these operations counts as a
constant number of source-preparation calls. The query theorem does
not impose a small gate cost on them. For finite implementations,
the sum of their clean-input operator errors is chosen within the
allocated tolerance.

Prepare the weights, load the states, and apply the selected
normalisation-three propagator block. Unprepare only the label
register with $P_w^\dagger$, retaining the data output. The two
occurrences of each square-root weight and the coefficient phase
give the selected vector
\begin{equation}\label{ode:eq:source-interface-output}
 \frac1{3B_q}\left[U(T,0)x_0+
      \sum_{i=1}^{m_{\rm src}}\omega_iU(T,v_i)b(v_i)\right],
\end{equation}
up to propagator and preparation errors. Selection means label zero
and the LCHS success registers in their prescribed states.
Uniform branch and preparation errors of order $\zeta$ change the
numerator by $\mathcal O(\zeta B_q)$. Thus both the coefficient
mass and the quadrature certificate are needed to bound the final
vector error.

\begin{theorem}[Time-dependent dissipative equations]
\label{ode:thm:time}
The guarantee and query bound of Theorem~\ref{ode:thm:main} also
hold for time-dependent dissipative $A(t)$ under uniform bounds
$s,\eta_D,\eta_H$, with
$B=\norm{x_0}+\int_0^T\norm{b(t)}\,dt$ and
$\bar g\ge B/\norm{x(T)}$. Assume coherent time-indexed sparse
oracles and known bounds that, at each requested tolerance, specify
a finite mesh and an attainable precision for its represented sparse
samples. The certified integrated sampling errors below must meet
that tolerance uniformly over the selected frequencies and source intervals,
accounting for both variation within cells and error in the decoded
values. The represented branch matrices must retain the common
moment and norm bounds. The Hamiltonian branches at a given
address use the same decoded generator. Assume the source access in
Eqs.~\eqref{ode:eq:source-interface-mass} and
\eqref{ode:eq:source-interface-preparation}. Source preparations
remain separately counted, at $\mathcal O(\bar g)$ calls.
\end{theorem}
\begin{proof}
The LCHS formula and common grid of Ref.~\cite{LS} already permit
time ordering on any source subinterval. For a path that is only bounded
and measurable, the formula passes from smooth paths by $L^1$
approximation preserving $D\succeq0$, since the propagators are stable
under $L^1$ changes of the generator; $U$ is then the Carath\'eodory
propagator. We need to implement each
of those unitary propagators with the same uniform sparse-query
bound. For $K_k(u):=H(u)+kD(u)$ and $|k|\le R$, the tagged
construction has common row/column bound $\rho=\eta_H+R\eta_D$.
Here $\Lip(H)$ is the least $L$ satisfying
$\norm{H(u)-H(v)}\le L|u-v|$ on the interval considered,
and similarly for $D$ and $b$.
The norm is the operator norm for matrix paths and the Euclidean norm
for vector paths. Thus
\begin{equation}
 \norm{K_k(u)}\le a_*:=\sqrt{2s}\rho,
 \qquad \Lip(K_k)\le\Lip(H)+R\Lip(D).
\end{equation}
On a mesh $t_0<\cdots<t_n$ of the selected source interval,
freeze the Hamiltonian at the left endpoint of each cell and use
its Cayley step. Here $S_r$ denotes that system time step, not the
private entrywise Cayley unitary $C_0$ of
\eqref{ode:eq:interface}. With $h_r=t_{r+1}-t_r$, there are two
local errors: replacing a varying Hamiltonian by its sample costs
at most $\Lip(K_k)h_r^2/2$, and replacing the frozen exponential
by the Cayley step costs at most $(a_*h_r)^3/12$. Since the
factors are unitary, telescoping their ordered product adds these
errors without norm growth. Thus, for exact samples,
\begin{equation}\label{ode:eq:timedisc}
 \left\|S_{n-1}\cdots S_0-
          \mathcal T e^{-\ii\int K_k(u)\,du}\right\|
 \le\sum_r\left[\frac{(a_*h_r)^3}{12}
                   +\frac{\Lip(K_k)h_r^2}{2}\right].
\end{equation}
For represented samples add the certified contribution
$\sum_rh_r\norm{\widetilde K_{k,r}-K_k(t_r)}$.
For a common represented generator $\widetilde G_r$, the identity
$K_k=\Reop((k-\ii)G)$ gives
$\norm{\widetilde K_{k,r}-K_k(t_r)}
 \le\sqrt{1+k^2}\,\norm{\widetilde G_r-G(t_r)}$.
Thus a certificate for the decoded generator suffices for every
$|k|\le R$ after allowing the factor $\sqrt{1+R^2}$.
The represented matrices must retain the uniform moment and norm
promises. Choose one mesh making this bound fit the branch error
budget and satisfying $w_r:=\alpha h_r\le1/4$. Known breakpoints
are included for piecewise Lipschitz data. More general certificates
for the same integrated sample error may replace Lipschitz bounds.

At every sample the tagged lists give the same known boundary
$J_{\rm free}$, although their private unitaries and locator changes
may differ. This common boundary is the interface needed for chronological
composition: every sample accepts the same data embedding, even
when its sparse support changes. The common $\rho,\gamma,Q$ and
response parameter $(1+\gamma^2)/2$ therefore satisfy the
hypotheses of the chronological proof; the latter parameter is
twice the squared-response bound in the display following
\eqref{ode:eq:interface}. Its stationary-map and resolvent estimates apply with
\begin{equation}
 W=\frac{1+\gamma^2}{2}\sum_rw_r
          \le\frac{T\alpha(1+\gamma^2)}2.
\end{equation}
The bound depends on the sum of the step lengths, which is at most
$T$, rather than on the number of steps. Choosing the filter for
this largest $W$ therefore gives \eqref{ode:eq:selectedcost} on
every source branch. Coherent source start times and durations use
the supplied time address and known arithmetic. Refining the mesh affects known gates
but not this query count, and requires neither commuting samples
nor fixed supports.

Choose branch, source-quadrature and preparation tolerances to be a
sufficiently small constant times $\eps/\bar g$, accounting for
$c_q$ in that constant. Equation~\eqref{ode:eq:source-interface-output}
then has numerator error $\mathcal O(\eps B/\bar g)$.
Since $\norm{x(T)}\ge B/\bar g$ and $B_q\le c_qB$, its actual
successful amplitude is $\Omega(1/\bar g)$.
The same normalisation and actual-direction amplification used for
constant sources gives the claimed error, matrix queries and source
calls. No additional time-dependence logarithm enters.
\end{proof}

\subsection{General constant generators}\label{ode:sec:general}

Dissipativity gave us contractive propagators and a bounded-weight
unitary representation. A general constant generator need not have
either property: the trajectory can grow before ending in a much
smaller vector. We use a different route in this case. We discretise
the trajectory and make each time update one row of a linear system.
Solving that system prepares a state whose time-labelled components
are the solution vectors. Additional copies of the last component
then make the final time accessible by amplitude amplification.

This history method needs a block encoding of the generator, rather
than its sparse locators directly. We first construct that encoding
from Hamiltonian evolutions at a cost proportional to $\sqrt s$.
We then explain how to use the actual implemented block consistently
throughout the history system, how to prepare its forcing term, and
how the propagation and trajectory bounds control the resulting
condition number and final-time probability.

The required conversion uses the Hermitian dilation
$H_A=\left(\begin{smallmatrix}0&A\\A^\dagger&0\end{smallmatrix}\right)$.
The dilation is Hermitian even when $A$ is neither normal nor
dissipative, and its upper-right block is the desired $A$.
Let $\Lambda$ be a supplied spectral-norm bound. On the spectrum
of $X=H_A/(2\Lambda)$, we have
$\arcsin(\sin X)=X$. We can approximate the arcsine by a polynomial
and express each power of $\sin X$ as a linear combination of
forward and backward evolutions. We expand the entire polynomial into evolutions before implementing it: this lets us pay for the longest selected evolution once, rather than for an independently accurate evolution at every
polynomial factor. The common coherent simulation schedule therefore incurs the additive logarithmic precision cost only once. The simulation error is also controlled by the total absolute coefficient sum, with no additional factor from the polynomial degree.

\begin{lemma}[A matrix block from sparse Hamiltonian evolutions]
\label{ode:lem:block}
Let $A$ have the original sparse access, and let $\Lambda>0$ be a
known bound on $\norm A$. For $0<\delta\le1/2$, we can implement a
unitary with a selected block within operator norm $\delta$ of
$A/(2\pi\Lambda)$ using
\begin{equation}\label{ode:eq:blockcost}
 \mathcal O\!\left(\sqrt{s}\log(2/\delta)\right)
\end{equation}
matrix queries. Neither normality nor dissipativity is assumed.
\end{lemma}

\begin{proof}
Set
\begin{equation}
 H_A=\begin{pmatrix}0&A\\A^\dagger&0\end{pmatrix},
 \qquad X=\frac{H_A}{2\Lambda}.
\end{equation}
The two dilation sectors each have dimension $N$. An upper-sector
row contains a row of $A$; a lower-sector row contains a conjugated
column of $A$. Thus the given row and column locators also locate
the nonzeros of $H_A$, with one known sector bit and conjugation
of the decoded value. Its sparsity is $s$. Every row and column
Euclidean norm is at most $\norm A\le\Lambda$, and
$\norm{H_A}=\norm A$, so $\norm X\le1/2$. No new access
model is required. The arcsine coefficients
and their truncation are
\begin{equation}
 a_n=\frac{\binom{2n}{n}}{4^n(2n+1)},\qquad
 p_d(y)=\frac1\pi\sum_{n=0}^d a_ny^{2n+1}.
\end{equation}
They satisfy $0<a_n\le1$ and $\sum_{n\ge0}a_n=\pi/2$. Functional
calculus gives $\arcsin(\sin X)=X$ and $\norm{\sin X}\le1/2$.
Consequently the truncation has error
\begin{equation}\label{ode:eq:arcsintail}
 \norm{p_d(\sin X)-X/\pi}
 \le\frac1\pi\sum_{n=d+1}^{\infty}2^{-(2n+1)}
 =\frac{4^{-d}}{6\pi}.
\end{equation}

To implement this polynomial, we use
\begin{equation}\label{ode:eq:fourier}
 \sin^mX=(2\ii)^{-m}\sum_{r=0}^{m}
            (-1)^r\binom mr e^{\ii(m-2r)X}.
\end{equation}
This is the binomial expansion of
$(e^{\ii X}-e^{-\ii X})^m/(2\ii)^m$. The absolute coefficients
sum to $2^{-m}\sum_r\binom mr=1$, independent of $m$.
Keeping a separate label for each polynomial degree and binomial
term, rather than collecting equal frequencies, therefore gives
total coefficient magnitude
\begin{equation}
 C_d=\frac1\pi\sum_{n=0}^d a_n\le\frac12.
\end{equation}
We add two identity terms with opposite coefficients, each of
magnitude $(1-C_d)/2$. This changes neither the target polynomial
nor its implementation error and makes the LCU normalisation one.
For the resulting coefficients $c_\xi$, preparation of
$\sum_\xi\sqrt{|c_\xi|}\ket{\xi}$, selection of the coefficient phase
and the corresponding evolution, and unpreparation select exactly
the Fourier sum.

Choose $d\ge1$ of order $\log(2/\delta)$ so that
\eqref{ode:eq:arcsintail} is at most $\delta/2$. The selected
evolutions are $e^{\ii jH_A/(2\Lambda)}$, with $|j|\le2d+1$.
A common coherent simulation schedule implements every branch to
error $\delta/2$; negative durations use the corresponding adjoints.
The largest duration multiplied by the row/column norm bound is
$d+1/2$. Equation~\eqref{ode:eq:hs} therefore gives
\begin{equation}
 \mathcal O\!\left(
 \sqrt{s}\,d+\sqrt{sd\log(2/\delta)}+\log(2/\delta)\right)
 =\mathcal O\!\left(\sqrt{s}\log(2/\delta)\right)
\end{equation}
queries. Since the coefficient magnitudes sum to one, their weighted
simulation errors also sum to at most $\delta/2$; there is no further
factor of $d$. Together with the polynomial error, this gives a
selected block within $\delta$ of $H_A/(2\pi\Lambda)$.
Preparing the lower dilation sector and selecting the upper output
sector extracts $A/(2\pi\Lambda)$.
\end{proof}

For the history solver, the useful content of this lemma is the
accuracy of the selected matrix block. It does not assert that the
implemented unitary is close to an arbitrary exact unitary extension.
We will account for this distinction by regarding the implemented
block as an exact encoding of a slightly perturbed generator.

Let $A$ and $b$ now be constant. Besides $B=\norm{x_0}+T\norm b$,
we need two measures of the evolution:
\begin{equation}
 C_A=\max_{0\le t\le T}\norm{e^{At}},\qquad
 g_{\rm path}=\frac{\max_{0\le t\le T}\norm{x(t)}}{\norm{x(T)}}.
\end{equation}
The first bounds propagation inside the history linear system; the
second controls the probability of selecting the final state from
its solution. We assume throughout that $x(T)\ne0$.
For the following general-generator theorem, error means trace
distance between the conditional data output and the pure state
$x(T)/\norm{x(T)}$. The history solver may return a mixed data
state, so we use this convention rather than the Euclidean vector
error specified for the preceding dissipative preparations.

\begin{theorem}[General constant-coefficient equations]
\label{ode:thm:general}
Suppose $\norm A\le\Lambda$, with $\Lambda>0$, and the matrix and
input-preparation oracles satisfy the access assumptions above.
Supply $\Lambda,s,T,\norm{x_0},\norm b$ and upper bounds
$\bar C\ge C_A$, $\bar g_{\rm path}\ge g_{\rm path}$ and
$\bar g_B\ge B/\norm{x(T)}$.
For $0<\eps\le1/2$, the normalised final state can be prepared within
error $\eps$, with constant heralded success probability, using
\begin{equation}\label{ode:eq:general}
 \widetilde{\mathcal O}\!\left(
 \sqrt{s}\,\bar g_{\rm path}\bar C(1+\Lambda T)\right)
\end{equation}
matrix queries. The suppressed factors are fixed polynomials in
logarithms of $2s$, $2\bar C$, $2\bar g_{\rm path}$,
$2+\Lambda T$ and $2+\bar g_B/\eps$.
\end{theorem}

The bound $\bar C$ supplies the history solver's condition-number
promise, and $\bar g_{\rm path}$ fixes the final-state amplification
schedule. The additional bound $\bar g_B$ determines accuracy and
appears only logarithmically. Actual values may of course be used
when supplied. This theorem states the matrix-query count; a separate
bound on the original input and source preparation calls is not included.
It combines the new sparse-access
conversion with the propagator-dependent stability analysis
of~\cite{Krovi}; related block-history formulations appear
in~\cite{BC}.

\begin{proof}
There are three approximations to control: the implemented generator,
the finite-degree history updates and forcing term, and the
linear-system solver. We separate them in that order. First we
interpret the actual selected matrix block as an exact encoding of
a nearby generator. We then build both the history matrix and its
right-hand side from this one generator. This prevents later
products from silently mixing ideal and implemented matrix blocks.
Finally, we choose a solver accuracy that survives selection of
the final time.

The case $T=0$ is direct preparation of the initial state. For $T>0$,
let $B_U$ be the actual selected block in Lemma~\ref{ode:lem:block}
and put $\widehat A:=2\pi\Lambda B_U$. We request block error at most
$\delta_A/(2\pi\Lambda)$, where
\begin{equation}\label{ode:eq:generalperturb}
 \delta_A\le\min\left\{\Lambda,\frac1{2\bar C T},
              \frac{\eps}{32\bar C^2T\bar g_B}\right\}.
\end{equation}
Then $\norm{\widehat A-A}\le\delta_A$ and
$\norm{\widehat A}\le2\Lambda$. Variation of constants and
Gr\"onwall's inequality give
\begin{equation}\label{ode:eq:generalstability}
 \begin{aligned}
 \norm{e^{\widehat A t}}&\le\bar C e^{\bar C\delta_A t},\\
 \norm{\widehat x(t)-x(t)}
  &\le\delta_A\bar C^2T e^{\bar C\delta_A T}B
   \le\frac{\eps}{16}\norm{x(T)}
       \qquad(0\le t\le T).
 \end{aligned}
\end{equation}
For the first estimate, Duhamel's formula writes the perturbed
propagator as the original propagator plus an integral containing
$\widehat A-A$; Gr\"onwall's inequality then contributes the
factor $e^{\bar C\delta_A t}$. For the forced trajectory, the
two equations have the same initial vector and source, so subtracting
them gives the exact identity
\begin{equation}
 \widehat x(t)-x(t)=\int_0^t
 e^{\widehat A(t-u)}(\widehat A-A)x(u)\,du.
\end{equation}
The original trajectory obeys $\norm{x(u)}\le\bar C B$.
Combining this with the propagator estimate gives the second line
of \eqref{ode:eq:generalstability}. The first two choices in
\eqref{ode:eq:generalperturb} control the represented norm and
the exponential factor; the last makes the trajectory error small
relative to $\norm{x(T)}$, using $B/\norm{x(T)}\le\bar g_B$. These bounds preserve the propagator and
path ratios up to absolute factors and cost at most $\eps/8$ in
normalised-state error. From now on in this proof, we write $A$ for
$\widehat A$ and enlarge $\bar C$ and $\bar g_{\rm path}$ by absolute
factors so that they also bound the represented propagator norm and
$\max_{0\le t\le T}\norm{\widehat x(t)}/\norm{\widehat x(T)}$, respectively.
We retain $x(t)$ for the original trajectory and $\widehat x(t)$
for the represented trajectory; $\bar g_B$ continues to bound
$B/\norm{x(T)}$ for the original problem.
The implemented unitary is an exact block encoding of this matrix,
whether or not its other blocks resemble any chosen ideal extension.

\emph{Encoding the trajectory.}
We choose $m=\Theta(1+\Lambda T)$ intervals of duration $h=T/m$,
small enough that the known block normalisation times $h$ is at
most one. Let $V$ be the degree-$K$ Taylor approximation to $e^{Ah}$.
A lower-bidiagonal linear system enforces the successive updates:
its diagonal blocks are $I$, its evolving subdiagonal blocks are
$-V$, and a further $m$ subdiagonal blocks $-I$ copy the final state.
The repeated final labels will later provide enough probability to
extract that state. Concretely, for a history vector
$Y_V=\sum_{j=0}^{2m}\ket j\otimes y_j$, the system
$\mathcal L_VY_V=f$ reads
\begin{equation}
 y_0=f_0,\qquad
 y_j-Vy_{j-1}=f_j\quad(1\le j\le m),\qquad
 y_j-y_{j-1}=0\quad(m<j\le2m).
\end{equation}
The first $m$ rows evolve the trajectory, and the remaining $m$
rows repeat its last value. Let $\mathcal L_E$ be the exact
history matrix, with $E=e^{Ah}$ in place of $V$. Their difference
is a single block shift with blocks $E-V$. A block shift has norm
at most the largest norm of its blocks, so the $m$ rows do not
introduce a factor $m$ in this local matrix error:
\begin{equation}\label{ode:eq:history-matrix-error}
 \norm{\mathcal L_V-\mathcal L_E}\le\delta_V,
 \qquad \delta_V:=\norm{V-e^{Ah}}.
\end{equation}

We need a bound on the inverse of this system, rather than merely
on one update. Since $V$ is a polynomial in $A$,
\begin{equation}
 V^j=e^{Ajh}\left(I+e^{-Ah}(V-e^{Ah})\right)^j.
\end{equation}
This identity uses commutation: $V$, $e^{Ah}$ and $e^{-Ah}$
are all functions of the same matrix $A$. Since
$\norm{Ah}\le1$, we have $\norm{e^{-Ah}}\le e$. If
$\delta_V\le1/(2em)$, it follows that
\begin{equation}
 \norm{V^j}\le\norm{e^{Ajh}}(1+e\delta_V)^j
 \le\bar C e^{em\delta_V}\le\bar C e^{1/2}<2\bar C
 \qquad(0\le j\le m).
\end{equation}
Thus a sufficiently accurate single update controls every product
that occurs in the history inverse, even for a nonnormal generator.
Factorial Taylor convergence achieves this with $K$ logarithmic
in the accuracy and stability parameters.

To see what this controls, inject a source vector $u_{\rm inj}$ at history
label $i$. Solving $y_j-Vy_{j-1}=f_j$ gives
$y_j=V^{j-i}u_{\rm inj}$ for $j\ge i$ on the evolving segment and zero for
$j<i$; the copy segment then propagates by the identity. These are
precisely the blocks of the history inverse. Each block has norm
$O(\bar C)$ and each block row and column contains $O(m)$ blocks.
For any block matrix $(X_{ij})$, the scalar matrix
$(\norm{X_{ij}})$ dominates its action on the vector of block norms.
The scalar row/column-sum inequality consequently bounds
$\norm X$ by $\sqrt{rc}$ when its maximum block row and column sums
are $r$ and $c$. Applying this to the history inverse gives
$\norm{\mathcal L_V^{-1}}\le2(2m+1)\bar C\le6m\bar C$.
Also $\norm{\mathcal L_E}\le1+e<4$ and
$\norm{\mathcal L_V}<5$. Thus the inverse estimate supplies a
condition-number bound of $30m\bar C$. We will also use it to
convert the local error in \eqref{ode:eq:history-matrix-error}
into an error in the complete trajectory.
To encode the history matrix, use controlled shifts of the history
label for its subdiagonal blocks and an LCU of the Taylor powers
for $V$. If $\alpha_A$ is the fixed block normalisation of $A$,
our choice of $h$ gives $\alpha_Ah\le1$. Hence the sum of the
Taylor coefficient weights is at most $e$, and adding the diagonal
and copy shifts preserves a constant block normalisation. The
largest Taylor power uses $K$ matrix-block calls, so one history
encoding uses $O(K)$ calls to the fixed encoding. The matrix it
encodes may be dense: no sparse locators for $V$, or for the
implemented nearby $A$, are required.

\emph{Supplying the forcing term.}
The exact update over one interval contains
$h\phi_1(Ah)b$, where
\begin{equation}
 \phi_1(z)=\sum_{l\ge0}\frac{z^l}{(l+1)!}
          =\frac{e^z-1}{z}
\end{equation}
uses its analytic value at zero. Indeed, integrating the constant source over a single step gives
\begin{equation}
 \int_0^h e^{A(h-v)}b\,dv
 =\sum_{l\ge0}\frac{A^lh^{l+1}}{(l+1)!}b
 =h\phi_1(Ah)b.
\end{equation}
Thus the right-hand side inserts $x_0$ at label zero and the same
integrated source at each evolving label:
\begin{equation}
 f=\ket0\otimes x_0+
       h\sum_{j=1}^m\ket j\otimes\phi_1(Ah)b,
\end{equation}
with zero entries on the additional copy labels. We must prepare
both terms with their correct relative amplitudes. Since
$\norm{Ah}\le1$, we have
$\norm{\phi_1(Ah)-I}\le e-2$, whence
$\norm{\phi_1(Ah)b}\ge(3-e)\norm b$. This lower bound will permit
coherent preparation without knowing the norm of the transformed
source.

Let $\alpha_A$ denote the known normalisation of the fixed encoding.
For
$p_\phi(z)=\sum_{l=0}^K z^l/(l+1)!$, the usual product-and-LCU
construction selects $p_\phi(Ah)/a_\phi$, with
\begin{equation}
 a_\phi=\sum_{l=0}^K\frac{(\alpha_Ah)^l}{(l+1)!}\le e-1<2.
\end{equation}
Each factor in a power uses its own encoding work register.
If $U_A$ has selected block $A/\alpha_A$, then, with identities
on all untouched work registers understood,
\begin{equation}
 (\bra{0_{a_1}\cdots0_{a_l}}\otimes I)
 U_A^{(a_l)}\cdots U_A^{(a_1)}
 (\ket{0_{a_1}\cdots0_{a_l}}\otimes I)
 =(A/\alpha_A)^l.
\end{equation}
A component that fails an earlier factor's selection cannot return
through a later factor, because later factors do not act on that
earlier work register. Selecting the degree and padding all products
to the same maximum length therefore implements the polynomial
with $O(K)$ calls to the fixed encoding and its adjoint.

Write $r_0=\norm{x_0}$, $r_b=\norm b$ and
$R_{\rm src}=(r_0^2+mh^2r_b^2)^{1/2}$. Using the supplied preparations, we
first prepare
\begin{equation}\label{ode:eq:history-source-input}
 \ket{\chi}=
 \frac{r_0}{R_{\rm src}}\ket0\otimes\frac{x_0}{r_0}
 +\frac{\sqrt m\,hr_b}{R_{\rm src}}
    \left(\frac1{\sqrt m}\sum_{j=1}^m\ket j\right)
       \otimes\frac b{r_b}.
\end{equation}
We omit any zero-norm branch. The history labels are orthogonal, so the normalisation here is
the Euclidean norm $R_{\rm src}$ of their weights, rather than the sum $B$
used to add vectors in the dissipative algorithm. Each source label
has amplitude $hr_b/R_{\rm src}$. We will retain these labels as part of
the history right-hand side. Apply the polynomial circuit on the
source labels. On label zero, apply a known work-register
rotation with success amplitude $1/a_\phi$: this gives the initial
term the same normalisation as the transformed source. The complete
preparation therefore satisfies
\begin{equation}\label{ode:eq:history-source-output}
 \begin{aligned}
 U_{\rm src}(\ket{0_{\rm anc}}\otimes\ket{\chi})
  &=\ket{0_{\rm anc}}\otimes
       \frac{\widehat f}{a_\phi R_{\rm src}}+v_{\rm fail},\\
 \widehat f&=\ket0\otimes x_0+
       h\sum_{j=1}^m\ket j\otimes p_\phi(Ah)b,
 \end{aligned}
\end{equation}
where $v_{\rm fail}$ is an unnormalised vector orthogonal to the selected work-register
state. In particular, selection preserves the relative initial and
forcing terms.

Choose
$\delta_\phi=\norm{p_\phi(Ah)-\phi_1(Ah)}\le(3-e)/2$.
Then $\norm{p_\phi(Ah)b}\ge(3-e)r_b/2$, and orthogonality of the
history labels gives
\begin{equation}
 \frac{\norm{\widehat f}^2}{a_\phi^2R_{\rm src}^2}
   \ge\frac{(3-e)^2}{16},\qquad
 \norm{\widehat f-f}\le\delta_\phi R_{\rm src}.
\end{equation}
We thus have a known constant lower bound on preparation success.
Fixed-point amplification~\cite{YLC2014} of the complete preparation and its actual
inverse prepares $\widehat f/\norm{\widehat f}$ coherently within
error $\zeta$, using $O(\log(1/\zeta))$ polynomial and preparation
calls. We retain its work registers and let the history matrix act
as the identity on them. This allows us to analyse the actual
coherently prepared vector, including its small preparation error,
without assuming an exact source-preparation oracle.

\emph{Allocating the history errors.}
Let $Y=\mathcal L_E^{-1}f$ be the exact padded trajectory of the
represented generator. Let $\widetilde f/\norm{\widehat f}$ be the
actual unit vector prepared above, with its retained work registers;
the ideal $f$ and $\widehat f$ have those registers in their initial
state. Then
\begin{equation}
 \norm{\widetilde f-\widehat f}\le\zeta\norm{\widehat f},
 \qquad \widetilde Y:=\mathcal L_V^{-1}\widetilde f.
\end{equation}
Scaling the prepared right-hand side by $\norm{\widehat f}$ is
only an analytical convention; the linear-system algorithm needs
its normalised state, not this unknown norm. With identity factors
on the retained work understood, the residual identity is
\begin{equation}\label{ode:eq:history-residual}
 \widetilde Y-Y
 =\mathcal L_V^{-1}\bigl[
       (\widetilde f-f)-(\mathcal L_V-\mathcal L_E)Y\bigr].
\end{equation}
To obtain the identity, substitute $f=\mathcal L_EY$ into
$\mathcal L_V(\widetilde Y-Y)=\widetilde f-\mathcal L_VY$
and then multiply by $\mathcal L_V^{-1}$. It accounts for both
the imperfect right-hand side and the imperfect update matrix.
This is the point at which the inverse bound amplifies local
approximation errors into an error in the whole trajectory. Indeed, $\norm f\ge(3-e)R_{\rm src}$ and
$\delta_\phi\le(3-e)/2$ imply
$\norm{\widehat f}\le3\norm f/2$. Since
$\norm f\le4\norm Y$, we obtain
\begin{equation}\label{ode:eq:history-relative-error}
 \frac{\norm{\widetilde Y-Y}}{\norm Y}
 \le6m\bar C\left(
       \delta_V+\frac{4\delta_\phi}{3-e}+6\zeta\right).
\end{equation}
Final-time selection has amplitude at least a constant times
$1/\bar g_{\rm path}$, so it can enlarge an error in the
normalised history by at most a constant times $\bar g_{\rm path}$. We therefore
need the relative trajectory error above to be of order
$\eps/\bar g_{\rm path}$. The preceding estimate makes that
requirement concrete: divide the error budget by both the history
inverse scale $m\bar C$ and the final-selection scale
$\bar g_{\rm path}$. We choose
\begin{equation}\label{ode:eq:history-error-allocation}
 \delta:=\frac{\eps}{2^{15}m\bar C\bar g_{\rm path}},
 \qquad
 \delta_V\le\delta,\quad
 \delta_\phi\le(3-e)\delta,\quad \zeta\le\delta.
\end{equation}
Because $\bar C,\bar g_{\rm path}\ge1$ and $\eps\le1/2$,
this also satisfies the earlier stability requirement
$\delta_V\le1/(2em)$ and the source lower-bound requirement.
Equation~\eqref{ode:eq:history-relative-error} is then at most
$\eps/(256\bar g_{\rm path})$. Normalising two nonzero vectors
increases this relative error by at most a factor of two, giving
\begin{equation}\label{ode:eq:history-normalised-error}
 \norm{\frac{\widetilde Y}{\norm{\widetilde Y}}
            -\frac{Y}{\norm Y}}
 \le\frac{\eps}{128\bar g_{\rm path}}.
\end{equation}
For $\norm{Ah}\le1$, both Taylor tails are at most
$e/(K+1)!$. Thus one degree
$K=O(\log(2m\bar C\bar g_{\rm path}/\eps))$ suffices for
both requirements, and the coherent source preparation has only
logarithmic overhead as well.

\emph{Solving and selecting the final time.}
A near-optimal block-input linear-system algorithm~\cite{Costa,Dalzell}
uses $m\bar C$ times polylogarithmic factors for this history matrix.
This cost includes making the solver's own heralded success
probability a constant. We keep its implementation and success
records coherent so that the complete preparation and its actual
inverse can later be used in final-time amplification. The latter
amplification is a separate cost: even after the solver succeeds,
we must select the final-copy labels from its history output.
To choose the solver's output accuracy, first examine that selection. The exact padded history $Y$ of the represented generator
$\widehat A$ is
\begin{equation}
 \sum_{j=0}^{m-1}\ket j\otimes \widehat x(jh)
       +\sum_{j=m}^{2m}\ket j\otimes \widehat x(T).
\end{equation}
Let $P$ select labels $m,\ldots,2m$. This final-copy sector has squared norm at least
$m\norm{\widehat x(T)}^2$, while the entire vector has squared norm at most
$(2m+1)\bar g_{\rm path}^2\norm{\widehat x(T)}^2$. Selection in that sector leaves
the represented final-data direction $\widehat x(T)/\norm{\widehat x(T)}$, and
$\norm{PY}/\norm Y\ge1/(\sqrt3\bar g_{\rm path})$.
Equation~\eqref{ode:eq:history-normalised-error} shows that the
normalised $\widetilde Y$ has final-copy probability at least
$1/(12\bar g_{\rm path}^2)$. After selection, its normalised
vector differs from $PY/\norm{PY}$ by at most
$\sqrt3\eps/64<\eps/32$. This follows by applying the same
normalisation inequality to the projected vectors; projection
cannot increase their unnormalised difference.

Request conditional output trace distance
\begin{equation}\label{ode:eq:history-solver-tolerance}
 \delta_{\rm LS}:=\frac{\eps}{2^{12}\bar g_{\rm path}^2}
\end{equation}
from the normalised $\widetilde Y$ in the block-input solver.
The extra square of $\bar g_{\rm path}$ has a different origin
from the previous vector-error budget. We now allow a mixed solver
output and use trace distance. Conditioning a density operator on an
event of probability $p$ can enlarge trace error by order $1/p$;
here $p$ is bounded below by a constant times
$1/\bar g_{\rm path}^2$. This choice
therefore accommodates a possibly mixed solver output. If two
states have trace distance at most $\delta_{\rm LS}$ and the
ideal state's selection probability is $p$, their selection
probabilities differ by at most $\delta_{\rm LS}$, and their
normalised selected states have trace distance at most
$2\delta_{\rm LS}/p$. Here $p\ge1/(12\bar g_{\rm path}^2)$,
so the actual final-copy probability is at least
$1/(24\bar g_{\rm path}^2)$ and the additional selected-state
error is at most $24\bar g_{\rm path}^2\delta_{\rm LS}<\eps/128$.

Keep the solver's work and any measurement records coherently,
and amplify the joint event that its success flag and $P$ both
succeed. The complete preparation and its actual inverse require
$O(\bar g_{\rm path})$ uses for constant heralded success.
Amplification preserves the actual vector in that selected subspace,
and hence its conditional reduced state; it introduces no new
normalisation error. Discarding the history and work registers also
cannot increase trace distance. The errors from the represented
generator, the history approximation and the solver therefore sum
to less than $\eps/8+\eps/32+\eps/128<\eps$.

All chosen tolerances are inverse-polynomial in the supplied
parameters. In particular, \eqref{ode:eq:history-solver-tolerance}
adds only logarithms of $\bar g_{\rm path}/\eps$ to the solver
cost, and Lemma~\ref{ode:lem:block} contributes $\sqrt{s}$ times
logarithmic factors. For the original trajectory, $\norm{x(T)}\le C_A B$, so the supplied
bounds satisfy $\bar C\bar g_B\ge1$; hence $\log(1/\eps)$ is
also bounded by the logarithms listed in the theorem, even when
$\bar g_B<1$. Combining the history condition bound, its
polynomial encoding cost and the final amplification proves
\eqref{ode:eq:general}, with the stated assumptions and suppressed
logarithms.
\end{proof}

The general-generator theorem depends on a spectral bound $\Lambda$.
It does not justify replacing $\Lambda$ by the row/column bound
$\eta$. Using only $\norm A\le\sqrt{s}\eta$ may consequently lose a
factor of $\sqrt{s}$ in the time term. The direct dissipative
construction in \eqref{ode:eq:main} retains $\eta$ and avoids that
loss.

\subsection{Gates for the complete dissipative solver}\label{ode:sec:gates}

We now return to the dissipative solver and account for the gates
needed to execute it. The query proof processes time and source
labels in superposition. To retain its advantage in a finite circuit,
we must also compute the operations selected by those labels without
storing or traversing a separate circuit for every mesh point.
The crucial case is a branch that specifies both an LCHS frequency
and a source starting time: it must choose the matrix samples and
the boundary rotations coherently.

The boundary factorisation from Section~\ref{hs:sec:gates} compresses
a uniform sequence of time-step rotations into operations on binary
labels. We show that it still works when its rotation parameter is
computed from the selected source time. We then include the frequency
coefficients, the source weights and loading circuits, and the
reflections used for amplification. This establishes a gate bound
for the entire dissipative preparation, including its success test.

The additional assumptions concern computation on known data.
Let $N$ be the matrix dimension, and let $b_{\rm data}$ bound the bit
lengths of the encoded entries, public parameters, source labels and
time addresses. We assume that entry decoding, the exact zero tests in
Eq.~\eqref{ode:eq:tagvalue} and the required
arithmetic admit reversible circuits of fixed polynomial gate and
work cost in these lengths and the requested accuracy bits. Matrix
oracle internals are excluded from the gate count, while the cost of
the chosen source-loading circuits is included. We use the finite-gate
assumptions of Section~\ref{gate:sec:result} for both the selected
propagators and the full preparation. These assumptions supplement the query theorems; they are not needed for their
query-only conclusions.

The resulting gate bound is the matrix-query count times polynomial
overhead in the relevant bit lengths, together with the explicit
source-preparation cost. We first establish the common selected
propagator circuit, including its error budget, and then assemble
the full algorithm. This analysis covers the constant and
time-dependent dissipative theorems. It makes no gate claim for
the general-generator history algorithm. Appendix~\ref{app:formalisation} specifies which source nodes, sampling
certificates and circuit costs the formal gate results cover.

\begin{lemma}[Gate-efficient coherent selection]\label{ode:lem:gate-select}
Fix $p=2s$, the frequency interval $[-R,R]$, and the common bound
$\rho=\eta_H+R\eta_D>0$. Let $0<\delta\le1/2$, and choose
$\gamma,\alpha,Q$ as in \eqref{ode:eq:interface}, optimising
$\gamma$ once using $\rho T$ and $\log(2/\delta)$. Put
\begin{equation}\label{ode:eq:gate-select-schedule}
 W:=\alpha T\frac{1+\gamma^2}{2},\qquad
 k_{\rm filt}:=\left\lceil4W/Q+\log_2(2^{18}/\delta)\right\rceil,
 \qquad d:=k_{\rm filt}(32Q-1).
\end{equation}
Suppose a label $z$ specifies $|k_z|\le R$ and $0\le v_z\le T$ by
known reversible arithmetic, and write $u_z:=T-v_z$. In the
time-dependent case assume Lipschitz bounds with
$L_*\ge\Lip(H)+R\Lip(D)$; set $L_*:=0$ in the constant case.
Let $M$ be a power of two satisfying
\begin{equation}\label{ode:eq:gate-select-mesh}
 M\ge\max\left\{2,\,4\alpha T,\,
        \frac{2(\alpha T)^{3/2}}{\sqrt\delta},\,
        \frac{8L_*T^2}{\delta}\right\}.
\end{equation}
At the addresses $v_z+r u_z/M$, the represented Hermitian samples
must satisfy the common moment bounds, and their certified total
operator-norm error, integrated over the selected interval, is at most
$\delta/16$.

Then the controlled direct sum of propagators
\begin{equation}\label{ode:eq:gate-select-target}
 \sum_z\ket z\!\bra z\otimes
 \mathcal T\exp\left[-\ii\int_{v_z}^T
                 (H(t)+k_zD(t))\,dt\right]
\end{equation}
has a clean-input implementation with error on all output registers at most
$\delta$, at most $6d$ private Cayley calls, and no factor equal to
the number of source labels or time slots in its gate count.
For efficiently decoded binary entries and polynomial-cost reversible
arithmetic for the public labels, times and parameters, its additional
gate count is
\begin{equation}\label{ode:eq:gate-select-gates}
 G_{\rm sel}(\delta)
 =d\,\operatorname{poly}\!\left(
    b_{\rm data},\log(N+1),\log M,\log\frac{d+1}{\delta}\right),
\end{equation}
where $b_{\rm data}$ includes the bit lengths of public parameters and labels.
Its reusable work space $A_{\rm sel}(\delta)$ is polynomial
in the same bit lengths, without the leading factor $d$. Gates internal to
supplied sparse oracles are excluded. The inverse is the adjoint of
the actual circuit and has the same resource bounds.
\end{lemma}

\begin{proof}
We use one circuit whose arithmetic reads the branch label $z$.
Its rescaling, filter degree and sequence of oracle calls are the
same on every branch; only its rotation parameters and sample
addresses depend on $z$. This is what makes it a coherent selection
of propagators. In particular, if its branch errors are $E_z$,
then for an arbitrary input $\sum_z\ket z\otimes\psi_z$,
\begin{equation}
 \left\|\sum_z\ket z\otimes E_z\psi_z\right\|^2
 =\sum_z\norm{E_z\psi_z}^2
 \le\bigl(\max_z\norm{E_z}^2\bigr)
                \sum_z\norm{\psi_z}^2.
\end{equation}
The vectors $\psi_z$ need not be normalised and may include an untouched reference register.
Thus a uniform complete-output bound on each branch suffices; there
is no sum of errors over the number of labels.

We first explain how the common circuit implements a branch, then
bound its discretisation, filtering and finite-synthesis errors.

For fixed $z$, choose $M$ equal time steps of length $u_z/M$ and set
\begin{equation}
 w_z=\alpha u_z/M,\qquad
 c_z=\frac{1-w_z/2}{1+w_z/2}.
\end{equation}
Recall the boundary network being compressed. There is one public
mode carrying the data and one private boundary mode for each of
the $M$ physical time slots. At slot $r$, a rotation with cosine
$c_z$ couples the public mode to that slot's boundary mode; other
slots are untouched. All $M$ rotations in a fixed branch have the
same angle because its steps have equal duration. The exact dyadic
factorisation \eqref{gate:eq:dyadic} groups consecutive slots in
powers of two. A group of $2^j$ slots has public transmission
$c_z^{2^j}$, and known input and output coordinate changes align
the private directions through which that group exchanges amplitude.
Consequently the entire ordered network uses
$\mathcal{O}(\log M)$ rotations of the forms
$B_{\rm in}(c_z^{2^j})$, $B_{\rm out}(c_z^{2^j})$ and
$R(c_z^M)$, together with the prefix predicates and the two boundary-coordinate
changes described in Theorem~\ref{hs:thm:gates}.
Here $B_{\rm in}$ and $B_{\rm out}$ are the input and output
changes of private basis within that factorisation, while
$R$ is the remaining two-mode rotation. Their arguments are powers
of one scalar $c_z$, so we can compute them coherently from the
branch label. The number of source times therefore does not become
the number of stored boundary circuits.

The powers required by this factorisation can be computed reversibly
as $c_z,c_z^2,c_z^4,\ldots,c_z^M$. We keep their approximations in
$[0,1]$. For numbers in $[0,1]$, one squaring amplifies an existing
absolute error by at most two, because
$|x^2-y^2|=|x-y||x+y|\le2|x-y|$. If $e_j$ is the error
after $j$ squarings and each rounding adds at most $\zeta$,
then $e_{j+1}\le2e_j+\zeta$. Iterating to $2^j=M$ gives
$e_j=O(M\zeta)$, including the initial error in $c_z$.
Thus an exponentially long time mesh requires only linearly more
accuracy bits in $\log M$. The functions defining $B_{\rm in}$ and $B_{\rm out}$
are uniformly Lipschitz on this interval. The central rotation obeys
\begin{equation}
 \norm{R(x)-R(y)}\le C\sqrt{|x-y|}
 \qquad(0\le x,y\le1)
\end{equation}
for an absolute constant $C$. The square-root bound also covers
the endpoint where the selected interval has length zero.
It is therefore enough to compute each power within
$O((\delta/[d(1+\log M)])^2)$. This requires
$O(\log M+\log((d+1)/\delta))$ fractional bits, apart from a
further logarithm in $1+\log M$. Certified square roots, angle
computation and rotations controlled by the angle bits all have
polynomial cost at this precision. We then uncompute the arithmetic.
We may also disable the zero-duration branch explicitly, making its
action exactly the identity.

In the query-free coordinates of Lemma~\ref{ode:lem:tags}, write
$C_{0,\rm free}(k,t):=F(t)^\dagger C_0(k,t)F(t)$, where
$F(t)$ is the locator coordinate change at time $t$. One private
operation then has the direct-sum form
\begin{equation}
 \bigoplus_{z,r}C_{0,\rm free}(k_z,v_z+r u_z/M).
\end{equation}
Compute the time address and then the two original entries in
\eqref{ode:eq:Kentry}; form the tagged value \eqref{ode:eq:tagvalue},
apply its two-dimensional Cayley rotation, and uncompute. This uses
only a constant number of time-indexed locator and value calls, with
polynomial-cost decoding, arithmetic and routing. For a power-of-two source mesh with $v_i=iT/m$, the address has the
explicit form
\begin{equation}
 v_i+r(T-v_i)/M
   =T\frac{iM+(m-i)r}{mM}.
\end{equation}
Thus it is an integer address on a uniform grid of $mM$ points and is
computed by multiplication and addition on
$\mathcal{O}(\log m+\log M)$ bits. There is no source-time table.
For other efficiently computed nodes, any time-address rounding is
included in the stated represented-sample error.
The locator controls are implemented using the phase-fixed dummy
input construction of Section~\ref{hs:sec:control}, on the complete
query register including its time address. This construction uses
only a constant number of supplied oracle calls and polynomially many known
gates. An external enable bit need control only the middle known
rotation in a compute-rotate-uncompute circuit. Invalid basis labels
are masked as in Theorem~\ref{hs:thm:gates}. Hence no free control of an
arbitrary opaque unitary is assumed.

The sampled private operations remain compatible with the boundary
factorisation even though their matrices vary with the slot.
A private operation acts only within its own slot's private space;
a boundary rotation at a different slot acts on the public mode
and that other slot. These operations therefore have disjoint
nontrivial actions and commute. Moving the private operations past
the other slots' rotations groups them into one slot-controlled
direct sum, exactly as in \eqref{hs:eq:gate-transducer}.
The factor $I_M\otimes D_0$ there is consequently replaced by
the slot-dependent sum above, with the same compressed boundary
network. No commutation between different sampled Hamiltonians is
being assumed. The reuse recurrence
\eqref{hs:eq:gate-reuse-block} likewise acts on the entire
chronological transducer in each branch. The physical mesh label and the reuse clock remain separate:
one identifies a time slot inside a complete transducer call, while
the other counts repeated calls used to suppress its private block.
The filter coefficients depend on the common $Q$ and degree $d$,
not on the source label or the number $M$ of physical slots.
Filter preparation and reuse require $O(d\log(d+1))$ known gates and,
after the one oblivious-amplification round, at most $6d$ private
Cayley calls.

We now check why the common mesh in the statement is sufficient.
Its condition $M\ge4\alpha T$ ensures $w_z\le1/4$.
Also $\gamma\le1$ gives $\norm{K_{k_z}}\le\sqrt p\rho\le\alpha$,
so we may use $\alpha$ as the norm bound in the Cayley error.
Summing the $M$ equal steps, the freezing and Cayley errors in
\eqref{ode:eq:timedisc} obey
\begin{equation}
 \frac{L_*u_z^2}{2M}\le\frac\delta{16},\qquad
 \frac{(\alpha u_z)^3}{12M^2}\le\frac\delta{48}.
\end{equation}
The represented-sample contribution is at most $\delta/16$.
For $x=W/Q\ge0$,
$ (1+x)e^{3x/4}2^{-4x}\le1$. Consequently the filter estimate
\eqref{ode:eq:filter} is at most $65\delta/2^{18}<\delta/2048$,
and the finite-reuse error including all work registers is at most $\delta/1024$
by the amplification analysis in Section~\ref{part:hs}. The sum
of these ideal errors is less than $\delta/4$. We reserve a further
$\delta/4$ for private Cayley synthesis and $\delta/4$ for the other
known operations. The first allocation permits error $\delta/(24d)$
per private call. Dividing the second allocation among the
polynomially many arithmetic and rotation operations leaves
polynomial cost in the stated bit lengths; logarithmic single-qubit
synthesis preserves that bound. As in the Hamiltonian gate theorem,
we telescope errors on ideal intermediate states with clean work
registers. The argument also applies to the actual adjoint circuits,
without assuming that an imperfect computation has exactly clean
work at every step. The total error is below $\delta$, uniformly
over the branch label. This proves the direct-sum guarantee and
\eqref{ode:eq:gate-select-gates}.
\end{proof}

The selection lemma supplies the propagator circuit; we now place
it inside the complete solution preparation. To express the resulting
cost, retain the internal tolerance
$\nu:=\eps/(10^4\bar g)$ and put $\ell=\log(1/\nu)$.
Let $q$ denote the scale of one selected propagator's query count:

\begin{equation}\label{ode:eq:gate-q}
 \Theta=T\eta_H+T\eta_D\ell,\qquad
 q=\sqrt{s}\,\Theta+\sqrt{s\Theta\ell}+\ell.
\end{equation}
For the full solver we now specify the source operations that enter
the gate count. In the time-dependent case let
$B=\norm{x_0}+\int_0^T\norm{b(v)}\,dv$ and suppose
$g:=B/\norm{x(T)}\le\bar g$. We specialise the query interface to a
positive quadrature: its nodes $v_i$ and weights $w_i\ge0$ satisfy,
with
\begin{equation}
 B_m=\norm{x_0}+\sum_{i=0}^{m-1}w_i\norm{b(v_i)},
\end{equation}
the certificates
\begin{equation}\label{ode:eq:gate-source-certificates}
 \left\|x(T)-U(T,0)x_0-
       \sum_{i=0}^{m-1}w_iU(T,v_i)b(v_i)\right\|\le\nu B,
 \qquad 0<B_m\le2B.
\end{equation}
Here $U(T,v)$ is the dissipative propagator, the normalisation
$B_m$ is known, and the nodes can be computed coherently from their
labels at polynomial cost. The earlier query theorem still permits
signed and complex quadrature coefficients.

We require a weight-only preparation $P_w$ with amplitudes
$\sqrt{\norm{x_0}/B_m}$ and
$\sqrt{w_i\norm{b(v_i)}/B_m}$, together with phase-consistent
controlled loading of the corresponding normalised data states.
Terms of zero weight are omitted. Let $G_{\rm src}(\zeta)$ count
the gates for $P_w$, this loading and $P_w^\dagger$, with their
combined implementation error at most $\zeta$ in the clean-work
convention. This includes the non-matrix-oracle gates inside the chosen
loading circuits; $A_{\rm src}(\zeta)$ bounds their work qubits.
We count a constant number of supplied input/source calls per assembly.
If their circuits themselves call a primitive matrix oracle, let
$Q_{\rm src}(\zeta)$ bound the total number of those calls in one
assembly or its actual adjoint. Reflections
about zero work registers will be counted when we amplify the
completed circuit.

\begin{theorem}[Gate-efficient dissipative equations]\label{ode:thm:gates}
Consider the constant-coefficient problem of
Theorem~\ref{ode:thm:main}, or the time-dependent dissipative problem
of Theorem~\ref{ode:thm:time} with the source certificates and circuits
specified above. Assume efficient binary entry decoding,
polynomial-cost reversible public arithmetic and the finite-gate
assumptions of Section~\ref{gate:sec:result}. For time-dependent
generators, require the Lipschitz bounds and represented-sample
certificates of Lemma~\ref{ode:lem:gate-select}; more generally, a
certificate for the same uniform meshes suffices.

For $T>0$ and $\rho>0$, choose the common degree $d$ and a
power-of-two time mesh $M$ from that lemma with $\delta=\nu$, and set
\begin{equation}\label{ode:eq:gate-word-budget}
 \mathfrak b=b_{\rm data}+\log(N+1)+\log(M+1)+\log(m+1)
                    +\log\frac{q\bar g+1}{\eps}.
\end{equation}
There is a fixed polynomial $P$, determined by the arithmetic and
decoding implementations, such that the complete solver uses
\begin{align}
 G_{\rm ODE}
 &\le C\bar g\left[qP(\mathfrak b)+G_{\rm src}(\nu)
                                     +A_{\rm src}(\nu)\right],
      \label{ode:eq:gate-main}\\
 A_{\rm ODE}
 &\le P(\mathfrak b)+C A_{\rm src}(\nu)
      \label{ode:eq:gate-space}
\end{align}
non-matrix-oracle gates and work qubits, respectively, for an
absolute constant $C$. The solver makes $\mathcal{O}(\bar gq)$
primitive matrix queries outside its $\mathcal{O}(\bar g)$
input/source calls. Counting also the matrix queries inside supplied
preparation circuits gives $\mathcal O(\bar g[q+Q_{\rm src}(\nu)])$
primitive matrix queries in total.
The number of calls and gates is bounded in the worst case over its
random choices. Heralded success has probability at least $2/3$, and
the conditional mixed output has trace distance at most $\eps$ from
the normalised solution. In fact, every successful trial has
Euclidean error at most $\eps$, up to a global phase.

If the source gate and work costs are polynomial in $\mathfrak b$,
the total gate bound is $\bar gq\,\operatorname{poly}(\mathfrak b)$.
At fixed $s,N,T,\eta_H,\eta_D,\bar g$ and regularity bounds,
this gives polylogarithmic dependence on $1/\eps$ if the data width
and both register widths $\log(M+1),\log(m+1)$ grow at most
polynomially in $\log(1/\eps)$. The quadrature certificate alone
does not bound $m(\nu)$. The Lipschitz source rule in the proof
meets the width condition when the least sufficient binary mesh
is used; efficient weight preparation and source loading are still
required. For a constant generator set $L_*:=0$.
At $T=0$ only the initial preparation is needed. For a known zero
generator the solver assembles the source quadrature directly.
Neither case needs a propagator query; any matrix queries inside
the supplied preparations still count.
\end{theorem}

\begin{proof}
We first construct a single finite circuit with selected data vector
$y/(3B_m)$ and bound $\norm{y-x(T)}$. This incorporates coefficient
synthesis, selected simulation and source loading in one actual
preparation. We then amplify this circuit using its actual inverse.
Because amplification preserves its successful data direction, the
finite-gate errors need only be charged in the first construction.
Finally, we count the gates, including the work-register reflections
used by the amplification.

\emph{Preparing the LCHS combination.}
In \eqref{ode:eq:grid} choose $\mathcal D=T\sqrt{2s}\eta_D$ and let
\begin{equation}
 H_k=\frac{\pi}{\mathcal D/2+
                 \log(64e^{3/2}/(15\nu))},\qquad
 n_{\rm freq}=\max\{1,\lceil R/H_k\rceil\},\qquad h_k=R/n_{\rm freq}.
\end{equation}
There are $N_{\rm freq}=2n_{\rm freq}+1$ grid points.  Since $R=\mathcal{O}(\ell)$,
\begin{equation}\label{ode:eq:gate-grid-size}
 N_{\rm freq}=\mathcal{O}\!\left(\sqrt{s}\,T\eta_D\ell+\ell^2+1\right)
     =\mathcal{O}(q\ell).
\end{equation}
This choice is important for the gate count: the quadrature condition
requires a sufficiently small stepsize, but does not require an
arbitrarily fine grid.

Let $a_{\rm coeff}=\sum_j|c_j|<3$. Add opposite identity
coefficients of magnitude $(3-a_{\rm coeff})/2$ and prepare label
amplitudes equal to the square roots of the absolute coefficients
divided by three. A known diagonal circuit supplies their phases.
Generic state preparation~\cite{MottonenStatePrep} and diagonal
phase synthesis on these $N_{\rm freq}+2$ labels require
$O(N_{\rm freq})$ arbitrary rotations and
$O(\log(N_{\rm freq}+2))$ work qubits. Using the single-qubit
synthesis of Ref.~\cite{KMM2013}, we approximate each complete
PREP or phase operation within $\nu/32$ at gate cost
\begin{equation}
 O\!\left(N_{\rm freq}
        \log\frac{N_{\rm freq}+2}{\nu}\right).
\end{equation}
This table is over frequency nodes, whose number has just been
bounded by $O(q\ell)$; it is not a table over all physical time
slots. The coefficients and angles are classically precomputed,
and no quantum oracle for their evaluation is assumed. Their
finite rotation count can therefore be absorbed into the stated
query scale times a polynomial in accuracy and parameter bit lengths.

PREP, the selected evolutions, the phase operation and the actual
inverse of PREP produce the propagator block divided by three.
After restoring that normalisation, branch simulation error $\nu$
and Low--Somma error $2\nu$ contribute at most $5\nu$.
The two PREP operations and the phases contribute at most
$3(3\nu/32)<\nu$. Thus the implemented propagator block
$\widehat U(v)/3$ obeys the uniform bound
\begin{equation}\label{ode:eq:gate-propagator-error}
 \norm{\widehat U(v)-U(T,v)}\le6\nu.
\end{equation}
Every inverse used in the rest of the proof is the adjoint of the
circuit actually implemented.

\emph{Assembling the source.}
For source assembly we use the same coherent vector sum as in the
query proof, now with finite circuits. The success test includes
the frequency-selection work, simulation work and source label,
so the selected component leaves only the data register free.
The weight preparation gives the initial/source label amplitudes
\begin{equation}
 \sqrt{\norm{x_0}/B_m},\qquad
 \sqrt{w_i\norm{b(v_i)}/B_m}.
\end{equation}
Load the normalised data state conditionally on that label and apply
its selected propagator. Applying $P_w^\dagger$ to the label alone
then selects $y/(3B_m)$, with all output work registers zero.
We leave the data-loading circuit uninverted at this stage: its
loaded and evolved states are the vectors being added.

For constant forcing, take $m$ to be a power of two large enough for
\eqref{ode:eq:sourcequad} to be at most $\nu B$, and let $w_i=T/m$.
Then $B_m=B$, the label circuit consists of one rotation between the
initial and forcing sectors and $\log_2m$ Hadamard gates, and the
forcing state is loaded by the same controlled circuit on every
forcing label.  Consequently,
\begin{equation}\label{ode:eq:gate-constant-source}
 G_{\rm src}(\nu)=G_{x_0}(c\nu)+
 G_b(c\nu)+
 \operatorname{poly}\!\left(b_{\rm data},\log(m+1),\log(1/\nu)\right),
\end{equation}
where $c>0$ is a sufficiently small absolute constant and
$G_{x_0},G_b$ count the supplied controlled loading circuits.
There are a constant number of calls to each supplied preparation circuit.
Here an \emph{admissible mesh} means one satisfying the source
quadrature error and positivity certificate in
\eqref{ode:eq:gate-source-certificates}. A sufficient choice obeys
\begin{equation}
 m\le2\max\left\{1,
   \frac{T^2\sqrt{s}(\eta_D+\eta_H)\norm b}{2\nu B}\right\},
\end{equation}
up to replacing the right-hand side by the next larger integer.
In particular, its inverse-precision dependence costs only logarithmically
many label gates and qubits.

A sufficient positive quadrature for time-dependent forcing can also
be specified explicitly. If $\norm{b(v)}\le b_*$ and $\Lip(b)\le L_b$, contractivity and
$\norm{A(v)}\le\sqrt{s}(\eta_D+\eta_H)$ give Lipschitz bound
$\sqrt{s}(\eta_D+\eta_H)b_*+L_b$ for
$v\mapsto U(T,v)b(v)$. A power-of-two left-endpoint mesh with
\begin{equation}
 m\ge\max\left\{1,\frac{T^2[\sqrt{s}(\eta_D+\eta_H)b_*+L_b]}{2\nu B}\right\}
\end{equation}
has vector quadrature error
at most $\nu B$. Since $v\mapsto\norm{b(v)}$ is $L_b$-Lipschitz,
the same quadrature estimate gives
\begin{equation}
 |B_m-B|\le\frac{T^2L_b}{2m}\le\nu B,
 \qquad (1-\nu)B\le B_m\le(1+\nu)B.
\end{equation}
As $B>0$ and $0<\nu<1$, this proves both $B_m>0$ and $B_m\le2B$,
including when $x_0=0$.
This specifies a sufficient mesh; the theorem also allows other
efficiently addressable positive quadratures satisfying the certificates.  For time-dependent forcing, efficient
weight preparation and coherent source loading are additional input
assumptions; their gate cost is retained as $G_{\rm src}$ rather than
being hidden in the matrix-query bound.

We next bound the successful vector of the complete circuit.
Before source-synthesis error, \eqref{ode:eq:gate-propagator-error}
and the quadrature certificate give
$\norm{y-x(T)}\le6\nu B_m+\nu B$. Source-circuit error $\nu$
changes its selected, normalised block by at most $\nu$ and hence
changes $y$ by at most $3\nu B_m$. Therefore
\begin{equation}\label{ode:eq:gate-vector-error}
 \norm{y-x(T)}\le9\nu B_m+\nu B\le19\nu B,\qquad
 \left\|\frac y{\norm y}-\frac{x(T)}{\norm{x(T)}}\right\|
      \le38\nu g<\eps.
\end{equation}
The actual success amplitude $a_{\rm succ}=\norm y/(3B_m)$
satisfies
\begin{equation}
 \frac1{8\bar g}\le a_{\rm succ}\le\frac13+3\nu<\frac12.
\end{equation}
For the lower bound we use
$\norm y\ge(1-19\nu g)\norm{x(T)}$ and $B_m\le2B$.
For the upper bound, first consider the exact propagators at the
quadrature nodes. Their weighted sum has norm at most $B_m$ by
contractivity, independently of the quadrature error relative to
the continuous solution. Propagator and source synthesis add at
most $9\nu B_m$. Dividing by $3B_m$ gives $1/3+3\nu$. The direction is already
accurate enough, so the remaining operation need only increase
its success probability.

\emph{Amplifying the implemented circuit.}
Let $M_{\rm amp}$ be the least power of two at least $16\bar g$.
Choose a uniformly random integer $j$ from
$0,\ldots,M_{\rm amp}-1$ with fair classical bits, and perform
$j$ Grover iterations using the actual preparation and its actual
inverse. The two reflections test the known initial state and the selected
zero-output-work subspace. They use reversible Boolean gates and a
phase flip, with cost linear in the number of tested qubits. For
the actual preparation write $a_{\rm succ}=\sin\theta$; after
$j$ Grover iterations the successful amplitude is
$\sin((2j+1)\theta)$, with the same successful data direction.
Averaging over the bounded random iteration count gives
\begin{equation}
 \frac1{M_{\rm amp}}\sum_{j=0}^{M_{\rm amp}-1}
        \sin^2((2j+1)\theta)
 =\frac12-\frac{\sin(4M_{\rm amp}\theta)}
                   {4M_{\rm amp}\sin(2\theta)}
 \ge\frac12-\frac1{8\sqrt3}>\frac14.
\end{equation}
Four heralded attempts succeed with probability at least
$1-(3/4)^4=175/256>2/3$, with $O(\bar g)$ preparation calls
in every possible run. Every successful trial has data direction
$y/\norm y$, up to phase. Thus conditioning and mixing over the
random iteration counts preserve \eqref{ode:eq:gate-vector-error}.
We have charged gate approximation once, in constructing $y$;
there is no need to compare this procedure with amplification of
a different ideal preparation and impose another factor $\bar g$
on its synthesis accuracy.

If $G_{\rm sel}(\nu)$ and $A_{\rm sel}(\nu)$ denote the gate and work
costs of the common selected Hamiltonian simulation, these constructions
give the explicit bounds
\begin{equation}
\begin{aligned}
 G_{\rm ODE}
 &=\mathcal{O}\!\Biggl(\bar g\Biggl[
 G_{\rm sel}(\nu)+N_{\rm freq}\log\frac{N_{\rm freq}+2}{\nu}+G_{\rm src}(\nu)
 +A_{\rm sel}(\nu)+A_{\rm src}(\nu)\\
 &\hspace{38mm}{}+\log(N+1)+\log(m+1)+\log(\bar g+1)
 \Biggr]\Biggr),\\
 A_{\rm ODE}
 &=\mathcal{O}\!\left(A_{\rm sel}(\nu)+A_{\rm src}(\nu)
       +\log(N_{\rm freq}+2)+\log(m+1)+\log(\bar g+1)\right).
\end{aligned}
\label{ode:eq:gate-assembly-cost}
\end{equation}
The $A_{\rm sel}$ and $A_{\rm src}$ terms account for reflections
on simulation and source work registers.  By Lemma~\ref{ode:lem:gate-select}, $d=\mathcal{O}(q)$.
Together with \eqref{ode:eq:gate-grid-size}, this gives
\begin{equation}
 \mathcal{O}\!\left(\bar g\left[
 q\,\operatorname{poly}\!\left(b_{\rm data},\log(N+1),
       \log(M+1),\log(m+1),\log\frac{q\bar g+1}{\eps}\right)
 +G_{\rm src}(\nu)+A_{\rm src}(\nu)\right]\right),
\end{equation}
with $\mathcal{O}(\bar g q)$ matrix queries for propagation and
$\mathcal{O}(\bar g)$ input/source calls. Expanding the latter circuits
adds at most $\mathcal O(\bar g Q_{\rm src}(\nu))$ matrix queries.

The work-space estimate and the definition of $\mathfrak b$ give
\eqref{ode:eq:gate-space}, and the last gate estimate gives
\eqref{ode:eq:gate-main}. Source costs polynomial in $\mathfrak b$
are absorbed into $q\,\operatorname{poly}(\mathfrak b)$ since
$q\ge\ell>1$.

\end{proof}

\emph{Scope and limitations.}
For the least admissible power of two in
\eqref{ode:eq:gate-select-mesh}, the inequalities
$\frac{1+\gamma^2}{2}\ge1/2$ and $d\ge2W$ imply $\alpha T\le d$ and
\begin{equation}
 \log M=\mathcal{O}\!\left(
       \log\frac{d+1}{\nu}+\log(2+L_*T^2)\right).
\end{equation}
The regularity bound therefore affects the gate cost through a
logarithm, while the matrix-query bound is unchanged. A time-dependent
source may still have a large preparation cost; this remains visible
in $G_{\rm src}$ and $A_{\rm src}$. The earlier query theorem permits
more general certified meshes and source preparations. In particular,
an arbitrary nonuniform mesh obtained by inserting breakpoints does
not automatically supply the equal-angle boundary circuit used here.

The gate accounting retains the mechanisms responsible for the
query improvement. Each private Cayley operation is obtained from
individual entries, the contour filter reduces the number of reuse
calls, and the scalar rescaling balances evolution against filtering.
The ideal entrywise operations are then synthesised to the finite
accuracies allocated above. The tagged-support representation also
separates the Hamiltonian and dissipative row-norm contributions.
Nevertheless, this LCHS construction still uses maximum frequency
$R=\Theta(\log(\bar g/\eps))$, which places a logarithm on the
dissipative time term.

We do not prove that this logarithm is necessary. The general
multiplicative time--precision lower bound discussed in Section~7
of Low--Somma is left unproved there; its polynomial calculations
give evidence for that claim rather than a theorem~\cite{LS}.
Our Hamiltonian lower bound applies to the anti-Hermitian subclass.
The next section adds a joint normalisation--time--sparsity lower
bound and a positive-semidefinite precision lower bound, while
leaving the full dissipative precision gap open. Nor does decay of
the LCHS coefficients by itself allow the largest coherent branch
cost to be replaced by an average frequency: that would require
a different circuit and an accompanying error analysis.

A direct dissipative transducer would have to represent the
amplitude lost by the evolving system. An unchanged private
component cannot absorb it: if a unitary obeys
$S(\psi\oplus\Gamma\psi)=U\psi\oplus\Gamma\psi$ for every $\psi$,
norm conservation gives
\begin{equation}
 \norm\psi^2+\norm{\Gamma\psi}^2
 =\norm{U\psi}^2+\norm{\Gamma\psi}^2.
\end{equation}
The private terms cancel, forcing $\norm{U\psi}=\norm\psi$
for every input and hence $U^\dagger U=I$. A genuinely
contractive public action must therefore send the missing norm
to an additional output component.
For the dissipative Cayley map, with $G=D+\ii H$ and step length $h\ge0$,
\begin{equation}
 C_h=(I-hG/2)(I+hG/2)^{-1},
\end{equation}
we can identify the required extra output explicitly. The products $(I+hG/2)^\dagger(I+hG/2)$ and $(I-hG/2)^\dagger(I-hG/2)$ differ by $h(G+G^\dagger)=2hD$. Multiplying this difference on the left by $(I+hG^\dagger/2)^{-1}$ and on the right by $(I+hG/2)^{-1}$ gives the norm defect
\begin{equation}
 I-C_h^\dagger C_h
 =2h(I+hG^\dagger/2)^{-1}D(I+hG/2)^{-1}.
\end{equation}
If a factor $B_D^\dagger B_D=2D$ were supplied, the additional
output $\sqrt h\,B_D(I+hG/2)^{-1}$ would carry exactly this
amplitude. The recent Lindbladian transducer
construction~\cite{Lindblad} exploits factor or square-root access
of this kind and obtains stronger ODE results under that access.
Our sparse oracle for $A$ does not supply such a positive factor
at the required cost. In particular, the signed factorisation
$H=\lambda B_{\rm fac}^\dagger ZB_{\rm fac}$ in
Theorem~\ref{opt:thm:signed} does not yield a factor
$B_D^\dagger B_D=2D$. It therefore does not remove this obstruction
or improve \eqref{ode:eq:main} unconditionally.

The dissipative gate theorem accounts for the unitary decomposition,
coherent source assembly and amplification. It does not cover the
general-generator history construction or extraction of a classical
solution vector.

\section{Linear differential-equation lower bounds}\label{part:odel}

The preceding algorithms pay for three different tasks: propagating
vectors, resolving small amplitudes accurately, and normalising a
solution that may be much smaller than the input and source terms.
We now ask which of these costs are forced by the oracle model.
Our main reduction makes propagation and normalisation difficult
together. We compute a parity for each of many candidates, and then
use a known source to cancel all but one candidate. The surviving
solution has small norm, but its normalised state identifies the
candidate. Solving the ODE therefore performs a search whose
candidates themselves require parity computations.

We first prove the query lower bound for this combined search and
parity problem and realise it by a time-dependent ODE. We then
study accuracy. Hamiltonian examples already give precision lower
bounds inside the dissipative class; a second construction shows
precision hardness even when the generator has only a positive
semidefinite dissipative part. Finally, we give a comparison
algorithm with poor dependence on time and sparsity but better
asymptotic dependence on precision at fixed parameters. It identifies
a limitation of any claim that the full logarithm in our upper
bound must always be necessary.

\subsection{Joint normalisation, time and sparsity lower bound}
\label{odel:sec:main}

We use the same dissipative input class,
\begin{equation}\label{odel:eq:ode}
 \dot x(t)=A(t)x(t)+b(t),\qquad x(0)=x_0,\qquad 0\le t\le T,
 \qquad -A(t)=D(t)+\ii H(t),\quad D(t)\succeq0.
\end{equation}
Each row and column has at most $s$ nonzeros and Euclidean norm at
most a known $\eta$. Define
\begin{equation}\label{odel:eq:params}
 \tau:=T\eta,\qquad
 B:=\norm{x_0}+\int_0^T\norm{b(t)}\,dt,\qquad
 g:=B/\norm{x(T)},\qquad g\le\bar g.
\end{equation}
As in the upper bounds, $x(T)\ne0$ is required so that the output
state is defined. The parameter $\tau$ measures evolution time
in units of the supplied row/column norm bound. The numerator $B$
is the total initial and source norm before propagation, and $g$
measures how much smaller the final solution can become. The promise
supplies $\bar g$, not the unknown final norm. The value oracle
and in-place row and column locators have the canonical reversible
access, including inverses, controls and coherent time addresses
where needed.
Their list actions have full unitary extensions. Known gates,
arithmetic and the specified known state preparations are free;
the number of matrix queries is bounded in the worst case.

A solver must herald success with probability at least $2/3$ on
every input. Its conditional output density operator must be within
trace distance $\eps$ of the normalised solution. A pure output
with Euclidean error $\eps$ satisfies this requirement, so the
lower bounds apply to the upper bounds of the preceding section.
We allow dimension to grow with the displayed parameters; an
insufficient fixed dimension is outside the claim.

In our reductions, the support lists and their full unitary
extensions are public. Locators may consequently be given for
free: only the exactly represented matrix values carry the hidden
input. Initial and source preparations contain no hidden
information. The time-dependent examples are piecewise constant
with known breakpoints, as permitted by
Theorem~\ref{ode:thm:time}.

We will use two precision scales. The upper bound contains the
logarithm $\ell$ that includes both output accuracy and amplification.
The precision lower bounds use a chain length of order $r$: the
factorial decay of a long chain's endpoint amplitude replaces a full
logarithm by a logarithm divided by another logarithm. We introduce
the following notation for those comparisons:
\begin{equation}\label{odel:eq:precision}
 L:=\log(1/\eps),\qquad {h}:=\log(e+L/\tau),\qquad r:=L/{h},
 \qquad \ell:=\log(C\bar g/\eps),
\end{equation}
where $C>1$ is an absolute constant and logarithms are natural.

\begin{theorem}[Combined worst-case lower bound]\label{odel:thm:combined}
For $s\ge32$, $\tau\ge4$, $\bar g\ge2$, and
$0<\eps\le1/64$, every solver for the time-dependent class above has
worst-case matrix-query complexity
\begin{equation}\label{odel:eq:combined}
 Q_A=\Omega\!\left(
       \bar g\sqrt{s}\tau+\sqrt{s\tau r}+r\right).
\end{equation}
The constants are absolute. The first term is witnessed by an
inhomogeneous time-dependent family with $D=0$ and
$g=\Theta(\bar g)$. The last two terms can be witnessed by homogeneous
constant-coefficient families with $D=0$ and $g=1$.
\end{theorem}

The general dissipative upper bound \eqref{ode:eq:simple} is
\begin{equation}\label{odel:eq:upper}
 Q_A=\mathcal{O}\!\left(\bar g[\sqrt{s}(\tau+\sqrt\tau)+1]\ell\right).
\end{equation}
For $\tau\ge4$, the terms $\sqrt\tau$ and one in this
upper bound are absorbed by its $\sqrt{s}\tau$ term. The first
term of Theorem~\ref{odel:thm:combined} therefore already places
the lower bound within a factor $\ell$ of the upper bound.
The other two lower-bound terms give additional information about
precision. This comparison concerns the whole time-dependent input
class. It does not prove the same joint
$\bar g\sqrt{s}\tau$ lower bound for constant-coefficient equations,
nor does it justify multiplying the two precision terms by $\bar g$.

The refined upper bound, with $\tau_D=T\eta_D$ and
$\tau_H=T\eta_H$, is
\begin{equation}\label{odel:eq:refined}
 \mathcal{O}\!\left(\bar g[
 \sqrt{s}\Theta+\sqrt{s\Theta\ell}+\ell]\right),
 \qquad \Theta=\tau_H+\tau_D\ell.
\end{equation}
The first family in Theorem~\ref{odel:thm:combined} has $\tau_D=0$.
It matches the time term of \eqref{odel:eq:refined}; on that subclass,
the whole upper bound matches when $\tau\ge\ell$.

To prove the joint term, we need a query problem in which search
and parity occur in the same adversary argument. Separate lower
bounds would not justify their product.

Let $M$ be even, let $m\ge4$ be even, and let $n\ge1$.
An input consists of hidden positions $j_{ak}\in\mathbb Z_m$,
for $a\in\{1,\ldots,M\}$ and $k\in\{1,\ldots,n\}$.
The membership oracle is
\begin{equation}
 O_X\ket{a,k,z,c}=\ket{a,k,z,c\mathbin\oplus[z=j_{ak}]}.
\end{equation}
For each candidate define $\sigma_a:=(-1)^{\sum_k j_{ak}}$.
Promise that exactly one candidate $a_*$ has $\sigma_{a_*}=-1$.
We ask whether $a_*$ lies in the first or second half of the candidates.

A membership query checks one possible value of one hidden position.
Even learning the parity of a single position is a search problem
over its $m$ possible values. A candidate's sign multiplies the
parities of $n$ such positions, and the promise leaves one odd
candidate among $M$ candidates. The following argument combines
these requirements within one query lower bound. In particular,
its factor $\sqrt M$ will become the ODE's normalisation cost,
its factor $n$ the evolution time, and its factor $\sqrt m$ the
square root of sparsity.

\begin{lemma}\label{odel:lem:outer}
Bounded-error quantum decision of this promise problem requires
$\Omega(n\sqrt{Mm})$ membership queries.
\end{lemma}
\begin{proof}
We apply the positive spectral-adversary theorem~\cite{BSS}.
For a symmetric nonnegative matrix $\Gamma$ indexed by promised
inputs, supported only on pairs with different answers, the
bounded-error query complexity is at least a constant times
\begin{equation}
 \frac{\norm\Gamma}{\max_q\norm{\Gamma\circ\Delta_q}}.
\end{equation}
Here $\Delta_q(X,Y)$ is one precisely when the inputs give
different answers to query $q$, and $\circ$ denotes entrywise
multiplication. We construct $\Gamma$ so that its numerator
combines the parity changes with movement of the marked candidate,
while the denominator records the change visible at one
membership address.

The progress quantity behind this ratio is a weighted overlap of
the algorithm's states on different inputs. Take a nonnegative
unit eigenvector $\xi$ of $\Gamma$ with eigenvalue $\norm\Gamma$,
and let $\ket{\psi_X^{(j)}}$ denote the complete state after $j$
queries on input $X$, retaining any measurement records. Then
\begin{equation}\label{odel:eq:adversary-progress}
 W_j=\sum_{X,Y}\Gamma_{XY}\xi_X\xi_Y
              \langle\psi_X^{(j)}|\psi_Y^{(j)}\rangle
\end{equation}
starts at $W_0=\norm\Gamma$, since all initial states agree.
Input-independent operations leave these overlaps unchanged.
A query can change the overlap only on addresses where $X$ and
$Y$ differ. To make the per-query estimate explicit, express the
answer qubit in its $+$ and $-$ basis. The query is the identity
on the $+$ component and multiplies the $-$ component at address
$q$ by the sign of that membership answer. Write
$\phi_{X,q,-}$ for the generally unnormalised vector in the remaining
registers of this component immediately before the query, and set
\begin{equation}
 \Phi_q:=\sum_X\xi_X\ket X\otimes\phi_{X,q,-}.
\end{equation}
The input label here is an analysis register, not information
available to the algorithm. Normalisation gives
$\sum_q\norm{\Phi_q}^2\le1$. For a pair with different answers
at address $q$, the overlap changes sign; otherwise it is unchanged.
Consequently
\begin{equation}
\begin{aligned}
 |W_{j+1}-W_j|
 &=2\left|\sum_q\Phi_q^\dagger
            \bigl[(\Gamma\circ\Delta_q)\otimes I\bigr]\Phi_q\right|\\
 &\le2\max_q\norm{\Gamma\circ\Delta_q}
                    \sum_q\norm{\Phi_q}^2
 \le2\max_q\norm{\Gamma\circ\Delta_q}.
\end{aligned}
\end{equation}

At the end, inputs joined by $\Gamma$ must give different answers.
Let $P_0,P_1$ project onto the two final decision outcomes, including
any measurement dilation in the work registers. Cauchy--Schwarz in
these two orthogonal subspaces gives
\begin{equation}
 |\langle\psi_X|\psi_Y\rangle|
 \le\sum_{b=0}^1\norm{P_b\psi_X}\norm{P_b\psi_Y}
 \le2\sqrt2/3.
\end{equation}
The last inequality uses error at most $1/3$ on each of two inputs
with opposite correct answers. The largest overlap occurs when both
errors equal $1/3$. Nonnegativity of the weights then gives
$|W_{\mathrm{final}}|\le(2\sqrt2/3)\norm\Gamma$.
A constant fraction of the initial weighted overlap must therefore
be removed, which yields the stated norm ratio.

We first build the distinction required within one candidate. On tuples
$\mathbf j\in\mathbb Z_m^n$, join two tuples when exactly one
coordinate changes, and its parity changes. This is a bipartite
regular graph between even and odd tuple parities. Its degree and
adjacency norm are
\begin{equation}
 {d}=nm/2,\qquad \norm{G}={d}.
\end{equation}
Each of the $n$ coordinates can change to any of the $m/2$
positions of opposite parity, giving the degree $d$. A regular
nonnegative adjacency matrix has the all-ones vector as an eigenvector
of eigenvalue $d$, and its maximum row sum also bounds its norm by
$d$. Since the graph is bipartite, its adjacency matrix has the
form $\left(\begin{smallmatrix}0&F\\F^\dagger&0\end{smallmatrix}\right)$;
its norm is the operator norm (largest singular value) of the
rectangular even-to-odd block $F$, so that block also has norm $d$.

For a fixed membership query $(k,z)$, retain an edge only if its
endpoints give different answers to that query. An edge changing
a coordinate other than $k$ is removed, and an edge changing
coordinate $k$ survives only when one of its endpoints equals $z$. The remaining graph is a
direct sum of stars with $m/2$ leaves: the other coordinates are
fixed, and coordinate $k$ changes between $z$ and a position of
opposite parity. Consequently
\begin{equation}\label{odel:eq:innermask}
 \norm{G\circ\Delta_{k,z}}=\sqrt{m/2}.
\end{equation}

We now let the odd candidate move between the two answer classes.
Define $\Gamma$ on promised full inputs by joining
$X$ and $Y$ when their unique odd candidates move from $a$ to $b$
across the two halves, candidate blocks $a$ and $b$ each change along
one inner graph edge, and all other blocks agree. Starting from one promised input, there are $M/2$ choices for the
new odd candidate in the opposite half. There are $d$ choices to
change the old odd candidate to even and $d$ to change the new
candidate from even to odd. All other candidates remain even and
unchanged. Every resulting input is still promised and has the
opposite answer. Thus $\Gamma$ is symmetric, nonnegative and
regular, with
\begin{equation}\label{odel:eq:outerdegree}
 \norm\Gamma=(M/2){d}^2.
\end{equation}
For a fixed query $(c,k,z)$, a surviving edge must move the odd
candidate to or from $c$, because all other candidate blocks are
unchanged. Thus $\Gamma\circ\Delta_{c,k,z}$ is a block star between
the sector whose unique odd candidate is $c$ and the $M/2$
opposite-half sectors. Each
off-diagonal block is a tensor product of a masked inner block, an
ordinary inner block, and identities. The query masks the change in candidate $c$, giving the factor
$\sqrt{m/2}$ from \eqref{odel:eq:innermask}; it imposes no
restriction on the other changed candidate, whose block has norm
$d$. The identities act on every unchanged candidate. Norms
multiply under tensor products, so each off-diagonal block has
norm at most $d\sqrt{m/2}$.

To bound the whole star, write its nonzero block row as
$R=(R_1,\ldots,R_{M/2})$. The symmetric star has norm $\norm R$,
and
\begin{equation}
 \norm R^2=\norm{RR^\dagger}
 =\left\|\sum_b R_bR_b^\dagger\right\|
 \le\sum_b\norm{R_b}^2.
\end{equation}
Hence
\begin{equation}\label{odel:eq:outermask}
 \norm{\Gamma\circ\Delta_{c,k,z}}
 \le\sqrt{M/2}\,{d}\sqrt{m/2}.
\end{equation}
The adversary ratio is therefore at least
\begin{equation}
 \frac{(M/2){d}^2}{\sqrt{M/2}\,{d}\sqrt{m/2}}
 =\frac{n\sqrt{Mm}}2.
\end{equation}
The bounded-error spectral-adversary theorem proves the claim.
This argument treats the joint promise directly and does not assume
that separate lower bounds multiply.
\end{proof}

The adversary argument supplies a query problem costing
$\Omega(n\sqrt{Mm})$. We now realise its answer as a normalised
ODE solution. The construction has two stages with different roles.
Unitary evolution multiplies candidate $a$ by its parity sign
$\sigma_a$. A subsequent known source subtracts the equal
superposition that all even candidates share. The sole odd
candidate then survives with amplitude of order $1/\sqrt M$.
This will identify the normalisation parameter exactly as
$g=\sqrt M$, while the parity computation uses time of order
$n$ and sparsity of order $m$.

Take $m$ to be a power of two. On $\mathbb C^m$ let $S$ be the
cyclic shift and let
\begin{equation}
 u_{\rm unif}=m^{-1/2}\sum_z\ket z,\qquad
 v=m^{-1/2}\sum_z(-1)^z\ket z.
\end{equation}
For a marked position $j$, define the unitary
\begin{equation}\label{odel:eq:gadget}
 U_j=S^{j+1}(2\ket{u_{\rm unif}}\bra{u_{\rm unif}}-I)
     =2\ket{u_{\rm unif}}\bra{u_{\rm unif}}-S^{j+1},
 \qquad (U_j)_{bc}=2/m-[b-c-1=j].
\end{equation}
The shift $S$ sends $\ket z$ to $\ket{z+1\bmod m}$.
It fixes the uniform vector, which gives the second expression for
$U_j$. This sign-adjusted version of the gadget in
\eqref{hsl:eq:gadget} follows Patel~\cite[Lemma 4]{Patel}.
Its entry is either $2/m$ or $2/m-1$, both nonzero for $m\ge4$.
To evaluate the entry at $(b,c)$, we need only ask whether the
hidden position equals $b-c-1$ modulo $m$. Thus a dense block
has public support even though its values contain hidden information.

The alternating vector $v$ is orthogonal to the uniform vector.
The reflection in $U_j$ therefore multiplies $v$ by $-1$, while
$S^{j+1}$ multiplies it by $(-1)^{j+1}$. Together they give
\begin{equation}
 U_jv=-(-1)^{j+1}v=(-1)^jv.
\end{equation}
The eigenvalue is real and has modulus one, so $U_j^\dagger$
has the same action on $v$. This known vector lets us read the
parity of the hidden position through one block evolution.
Define
\begin{equation}\label{odel:eq:dilation}
 S_j=\begin{pmatrix}0&U_j\\U_j^\dagger&0\end{pmatrix},
 \qquad \omega=\ket+\otimes v.
\end{equation}
The added qubit in $\omega$ is the dilation sector; it is
separate from the $m$-dimensional position register. The off-diagonal
blocks make $S_j$ Hermitian, and unitarity of $U_j$ gives
$S_j^2=I$. Acting on both components of
$\omega=(\ket0\otimes v+\ket1\otimes v)/\sqrt2$ therefore
gives $S_j\omega=(-1)^j\omega$. Each row and column contains
one full row or column of a unitary block, so it has exactly $m$
nonzeros and Euclidean norm one. Its support is public.

For the $k$th interval of length $\pi/2$, set
\begin{equation}\label{odel:eq:stages}
 H(t)=\bigoplus_{a=1}^M S_{j_{ak}},\qquad A(t)=-\ii H(t),
 \qquad b(t)=0,
\end{equation}
and start from $x_0=u_M\otimes\omega$, where
$u_M=M^{-1/2}\sum_a\ket a$. Since
$e^{-\ii(\pi/2)S_j}=-\ii S_j$, after $n$ stages the solution is
\begin{equation}\label{odel:eq:beforeforce}
 x_{\mathrm{comp}}=
 \frac{(-\ii)^n}{\sqrt M}\sum_a\sigma_a\ket{a,\omega}.
\end{equation}
Here $\ket{a,\omega}:=\ket a\otimes\omega$ and
$\ket{u_M,\omega}:=u_M\otimes\omega$ denote product states
of the candidate register and the position/dilation registers.
On a final interval of length one, set
\begin{equation}\label{odel:eq:force}
 A(t)=0,\qquad b(t)=-(-\ii)^n\ket{u_M,\omega}.
\end{equation}
Over this last interval the generator is zero, so the source is
simply added to the vector. The coefficient of candidate $a$
becomes $(-\ii)^n(\sigma_a-1)/\sqrt M$. It is zero when
$\sigma_a=1$ and equals $-2(-\ii)^n/\sqrt M$ for the unique
odd candidate. The source is known independently of all hidden
positions; it does not itself reveal which candidate survives.
The resulting solution and its norm parameters are
\begin{equation}\label{odel:eq:cancel}
 x(T)=-\frac{2(-\ii)^n}{\sqrt M}\ket{a_*,\omega},\qquad
 B=2,\qquad \norm{x(T)}=2/\sqrt M,\qquad g=\sqrt M.
\end{equation}
Thus the normalised final state identifies the unique odd candidate.
The generator is anti-Hermitian at every time, including the final
zero interval, so it is within the dissipative promise with $D=0$.

This construction also has the required oracle reduction.
All entries are complex dyadic numbers with public support,
public time dependence of that support and public full locator
extensions. A value query computes one relevant membership bit,
writes the known entry and uncomputes the bit. It therefore uses
at most two membership queries, also for inverse and controlled
calls. The initial and forcing states are known, so their
preparations need no membership queries.

\begin{theorem}[Joint time, sparsity, and normalisation]\label{odel:thm:joint}
For $s\ge4$, $\tau\ge4$, $\bar g\ge2$, and
$0<\eps\le1/64$, the ODE promise $g\le\bar g$ requires
\begin{equation}\label{odel:eq:joint}
 Q_A=\Omega(\bar g\sqrt{s}\tau).
\end{equation}
Hard instances can be chosen as above with $D=0$, $\eta=1$,
and $g\ge\bar g/\sqrt2$. Dimension $\mathcal{O}(\bar g^2s)$ suffices.
\end{theorem}
\begin{proof}
Choose
\begin{equation}
 M=2\lfloor\bar g^2/2\rfloor,\qquad
 m=2^{\lfloor\log_2s\rfloor},\qquad
 n=\left\lfloor\frac{2(\tau-1)}\pi\right\rfloor.
\end{equation}
Then $\bar g/\sqrt2\le\sqrt M\le\bar g$ and $s/2<m\le s$.
The elementary bound $\lfloor x\rfloor\ge x/2$ for $x\ge1$ gives
\begin{equation}
 n\ge\frac{\tau-1}{\pi}\ge\frac{3\tau}{4\pi}>\frac\tau5,
 \qquad n\le\frac{2\tau}{\pi}\qquad(\tau\ge4).
\end{equation}
Thus $n=\Theta(\tau)$ with an explicit safe lower constant $1/5$.
Insert a known idle interval before the final source if necessary to
make the total duration exactly $T=\tau$. The uniform row/column
norm bound is $\eta=1$.

These choices place the hard ODE inside the claimed promise:
its dimension is $2Mm=O(\bar g^2s)$, its sparsity is at most
$s$, and its cancellation ratio is at most $\bar g$ but within
a constant factor of it. We now turn a solver into a decision
algorithm for the original membership problem. On success, measure
the candidate register and answer which half it belongs to; on
failure, return a fair bit.
Conditional trace distance $\eps$ gives correct probability at least
$1-\eps$ on success. Since success has probability at least $2/3$,
overall correctness is at least
\begin{equation}
 \frac23(1-\eps)+\frac16>\frac23.
\end{equation}
If the solver uses $Q_A$ matrix queries, the resulting decision
algorithm uses at most $2Q_A$ membership queries. Apply
Lemma~\ref{odel:lem:outer} and the parameter choices above.

If maximum sparsity exactly $s$ is wanted, append an unoccupied known
block with Hamiltonian
$2^{-\lceil\log_2s\rceil}\mathbf1\mathbf1^{\mathsf T}$ of size
$s\times s$. Its norm is at most one and its entries are dyadic.
It requires no membership queries and leaves the solution unchanged.
Known zero padding permits any sufficiently large qubit register.
\end{proof}

\begin{proposition}[Homogeneous time-dependent variant]
\label{odel:rem:homogeneous}
Under the sparsity, normalisation and error assumptions of
Theorem~\ref{odel:thm:joint}, when $\tau\ge20(1+\log\bar g)$,
the joint lower bound
\eqref{odel:eq:joint} also has a homogeneous time-dependent witness
with $g=\Theta(\bar g)$.
\end{proposition}
\begin{proof}
To remove the forcing term, we must suppress the unwanted
candidates through homogeneous evolution. We first store each
candidate's parity in a flag and then damp the flag value associated
with even candidates. Choose $M,m$ as in
Theorem~\ref{odel:thm:joint}, set $n=\lfloor\tau/\pi\rfloor$,
and add a flag initially in $\ket+$. Perform the $n$ computation
stages in \eqref{odel:eq:stages} only on flag one. Conditional
on candidate $a$, this leaves the flag state
\begin{equation}
 \frac{\ket0+(-\ii)^n\sigma_a\ket1}{\sqrt2},
\end{equation}
while the position and dilation registers remain in $\omega$.
The factor $(-\ii)^n$ is known; removing it and applying a
Hadamard will convert the relative parity sign into a flag value. On the flag,
apply the known Hamiltonian $-\ket1\bra1$ for
$\delta=(n\bmod4)\pi/2$ time, followed by Pauli $Z$ for $\pi/2$
time and Pauli $Y$ for $\pi/4$ time. The first pulse cancels the
relative phase $(-\ii)^n$, and the latter two implement a Hadamard
up to the global phase $-\ii$:
\begin{equation}
 e^{-\ii\pi Y/4}e^{-\ii\pi Z/2}
 =-\frac{\ii}{\sqrt2}\begin{pmatrix}1&1\\1&-1\end{pmatrix}.
\end{equation}
After the phase correction, the flag is
$(\ket0+\sigma_a\ket1)/\sqrt2$. The displayed Hadamard
maps this to $\ket0$ for an even candidate and $\ket1$ for
the odd candidate, with the same irrelevant global phase on every
candidate. The flag now permits selective damping without knowing
$a_*$. Each pulse Hamiltonian is one-sparse, has norm one and
complex dyadic entries; their total duration is at most $9\pi/4$.

Apply the known dissipative generator $-P_{\mathrm{reject}}$, where
$P_{\mathrm{reject}}=\ket0\bra0\otimes I$, for
${t_{\rm damp}}=\log(4\sqrt M)$ time. This generator is diagonal and has norm one.
This multiplies each even-candidate amplitude by
$e^{-t_{\rm damp}}=1/(4\sqrt M)$ and leaves the odd candidate
unchanged. Up to the common phase, the resulting vector is
\begin{equation}
 \frac1{\sqrt M}\left[
 \ket{a_*}\ket1\otimes\omega+
 e^{-t_{\rm damp}}\sum_{a\ne a_*}\ket a\ket0\otimes\omega
 \right].
\end{equation}
All stages obey the row/column Euclidean-norm bound $\eta=1$
and sparsity bound $m\le s$. Orthogonality of the candidate
labels determines the final norm. Since $b=0$ and
$\norm{x_0}=1$, we obtain
\begin{equation}
 \norm{x(T)}^2=\frac{1+(M-1)e^{-2{t_{\rm damp}}}}M,
 \qquad \frac4{\sqrt{17}}\sqrt M\le g\le\sqrt M\le\bar g.
\end{equation}
The normalised solution has probability at least $16/17$ on $a_*$.
Including heralding and trace error, the decision succeeds with
probability at least $\frac23(16/17-\eps)+\frac16>2/3$.

The unknown stages take $n\pi/2\le\tau/2$ time. The known pulses
and dissipation take at most
\begin{equation}
 \frac{9\pi}{4}+{t_{\rm damp}}
 \le\frac{9\pi}{4}+\log4+\log\bar g
 <10(1+\log\bar g)\le\frac\tau2.
\end{equation}
Known idle intervals make $T=\tau$ exactly. Moreover
$n\ge\tau/\pi-1\ge\tau/5$ in this range, so the same oracle
simulation and adversary reduction prove \eqref{odel:eq:joint},
with $g=\Theta(\bar g)$. This variant
is homogeneous but still time dependent; it is not a
constant-coefficient reduction.
\end{proof}

\subsection{Precision in unitary and dissipative evolution}

The Hamiltonian lower bounds already provide precision hardness
within the dissipative class. For $\tau\ge1$, $s\ge32$ and
$\eps\le1/64$, Theorem~\ref{hsl:thm:main} gives
\begin{equation}\label{odel:eq:hsbound}
 Q_A=\Omega\!\left(\sqrt{s}\tau+\sqrt{s\tau r}+r\right)
\end{equation}
for the evolution of a known input under matrices of norm one
with dyadic entries and public supports. Taking $A=-\ii H_X$
and $b=0$ produces a constant-coefficient ODE with $g=1$.
The Hamiltonian lower-bound proof detects a small signed endpoint
amplitude through a known observable. Our ODE solver is allowed
to fail with a flagged outcome, and its success probability may
depend on the hidden matrix. We must therefore show that the
observable still has the required unconditional parity correlation.
A lower bound applied only after conditioning on success would not
justify the claimed worst-case query count.

\begin{lemma}[Heralded transfer]\label{odel:lem:herald}
Bound \eqref{odel:eq:hsbound} remains valid when the simulation algorithm succeeds
with probability at least $2/3$ on each input and has conditional
trace error at most $\eps$.
\end{lemma}
\begin{proof}
For each hidden input $X$, let $\sigma(X)$ be its parity sign.
Each reduction in Section~\ref{part:hsl} provides a known
Hermitian norm-one observable $O$ whose ideal expectation obeys
$\sigma(X)\langle O\rangle=a>0$. One may turn this expectation
into a binary output by the two-outcome measurement
$(I\pm O)/2$. This formulation includes eigenvalues strictly
between $-1$ and $1$. Dyadic rounding leaves parity-corrected expectation at least
$7a/8$. Conditional trace distance at most $\eps$ changes a
norm-one observable's expectation by at most $2\eps$. On success
we use that observable; on failure we output an independent fair
sign, whose expected correlation is zero. This defines an ordinary
unconditional output of the bounded-query algorithm.
For input-dependent success probability $p_X\ge2/3$, the
unconditional parity correlation is consequently at least
\begin{equation}\label{odel:eq:heraldcorr}
 \E_X[p_X(7a/8-2\eps)].
\end{equation}
This follows pointwise before averaging over hidden inputs.
In particular, we never apply a query lower bound to a ratio
obtained by postselection.

For the time lower-bound family, $a\ge1/2$. Expression
\eqref{odel:eq:heraldcorr} is at least
$(2/3)(13/32)>1/4$, the correlation used in the original proof.
For the high-precision families, $L\ge64\tau$ and
$n:=\lfloor r/4\rfloor$ give
\begin{equation}
 a\ge(\tau/(2n))^n\ge e^{-L/4}\ge16\eps.
\end{equation}
Consequently \eqref{odel:eq:heraldcorr} is at least $a/2$, again the
correlation used there. The mixed small-correlation adversary and
the scalar parity polynomial argument apply unchanged. In the
remaining range $L<64\tau$, the time bound absorbs the other terms.
Thus the parameter proof of Theorem~\ref{hsl:thm:main} and its constants survive
heralding; in particular its constant $\log2/(4096\pi)$ is valid.
\end{proof}

\begin{proof}[Proof of Theorem~\ref{odel:thm:combined}]
An algorithm for the promise class $g\le\bar g$ must handle
both the cancellation family in Theorem~\ref{odel:thm:joint}
and the $g=1$ precision families covered by
Lemma~\ref{odel:lem:herald}. Its worst-case query count is therefore
at least the larger of their two bounds. For nonnegative $a,b$,
$\max\{a,b\}\ge(a+b)/2$, so this larger bound is at least a
fixed fraction of the sum in \eqref{odel:eq:combined}. Pad dimensions to a common
known size if desired. This proves \eqref{odel:eq:combined} as a
worst-case statement. It does not assert that one instance witnesses
all terms, or that independent lower bounds may be multiplied.
\end{proof}

We next construct precision hardness using a positive semidefinite (PSD)
damping matrix and no Hamiltonian part. The parity chain used above has a signed endpoint amplitude under the real
exponential $e^{uH'}$, where trigonometric propagation becomes
hyperbolic propagation. We turn this into dissipative evolution by
setting $D=(I-H')/2\succeq0$, so the differential-equation generator
is $A=-D$ and $e^{-TD}=e^{-T/2}e^{TH'/2}$. The scalar decay retains
the hidden endpoint sign. A separate reference component with $D=0$
remains undamped and keeps the solution norm bounded below. Thus the
small signed amplitude creates precision hardness without a large
normalisation penalty.

\begin{theorem}[Constant-coefficient PSD precision lower bound]
\label{odel:thm:psd}
Let $s\ge32$, $\tau\ge2$, and $L=\log(1/\eps)\ge64\tau$.
For the homogeneous equation $\dot x=-D_Xx$, there are hard families
with
\begin{equation}
 D_X\succeq0,\qquad \norm{D_X}=\eta_D=1,\qquad T=\tau,
 \qquad 1\le g\le\sqrt2,
\end{equation}
public supports and dyadic entries, for which
\begin{equation}\label{odel:eq:psd}
 Q_A=\Omega(\sqrt{s\tau r}+r).
\end{equation}
The statement allows the same heralded output as before.
\end{theorem}
\begin{proof}
We build a PSD damping matrix from the same parity-carrying chain,
but now the propagation is a real exponential. There are three
requirements to meet together. The endpoint amplitude must retain
its hidden sign, it must remain larger than the requested output
error after damping, and normalisation must not introduce a large
$g$. A shift by the identity makes the path generator PSD; an
isolated undamped reference keeps the total norm bounded below
and supplies the interference needed to read the sign.

Let $u=\tau/2$ and $n:=\lfloor r/4\rfloor$. With
$h=\log(e+L/\tau)$, the inequality $h\le L/(8\tau)$ for
$L/\tau\ge64$ gives
\begin{equation}\label{odel:eq:nrange}
 n\ge r/8\ge\tau\ge {u},\qquad n\ge2,\qquad n\le r/4.
\end{equation}
Use the spin chain
\begin{equation}
 B_n:=\sum_{k=0}^{n-1}b_k(\ket{k+1}\bra k+\ket k\bra{k+1}),
 \qquad b_k:=\frac{\sqrt{(k+1)(n-k)}}n.
\end{equation}
To obtain its endpoint amplitude, view $\ket k$ as the normalised
symmetric $n$-qubit state of Hamming weight $k$. The restriction
of $n^{-1}\sum_{j=1}^n X_j$ to these states has exactly the
coefficients $b_k$ above. It has norm one, attained by the
all-$+$ state. Because the single-qubit operators commute,
\begin{equation}
 e^{(u/n)\sum_jX_j}=\bigotimes_{j=1}^n e^{(u/n)X_j},
 \qquad \bra1e^{(u/n)X}\ket0=\sinh(u/n).
\end{equation}
The endpoints correspond to all zeroes and all ones, so their
matrix element is the product of these $n$ amplitudes. In
particular, $\norm{B_n}=1$ and
\begin{equation}\label{odel:eq:imaginaryamp}
 e^{-{u}}\bra n e^{{u}B_n}\ket0=e^{-{u}}\sinh({u}/n)^n=:a.
\end{equation}
Decorate consecutive edges with independent gadgets $U_{j_k}$ from
\eqref{odel:eq:gadget}, putting their adjoints on reverse edges. To see the resulting change of basis, define
$W_0=I$ and $W_k=U_{j_k}\cdots U_{j_1}$ for
$1\le k\le n$. At chain site $k$, change the position-register
basis by $W_k$. The edge from $k-1$ to $k$ then has block
$W_kW_{k-1}^\dagger=U_{j_k}$, so the decorated Hermitian
matrix $H_X$ is a unitary conjugate of $I_m\otimes B_n$.
This is the gauge equivalence: a site-dependent change of basis
removes the edge unitaries. It preserves the norm, and each site
has at most two neighbouring blocks of width $m$, giving sparsity
at most $2m$.

Since $U_{j_k}v=(-1)^{j_k}v$, the endpoint change of basis
acts as $W_nv=\sigma(X)v$, where
$\sigma(X)=(-1)^{\sum_kj_k}$. Thus the damped endpoint
amplitude in \eqref{odel:eq:imaginaryamp} becomes
$\sigma(X)a$ on the known alternating vector $v$.

Since $n\ge {u}$ and $\sinh z\ge z$,
\begin{equation}\label{odel:eq:alower}
 a\ge({u}/(en))^n\ge e^{-L/4}.
\end{equation}
For the last inequality, put $z:=L/\tau\ge64$. Then
$en/{u}\le ez/(2{h})\le e+z$ (here $h\ge\log(e+64)>e/2$), so
$n\log(en/{u})\le n{h}\le L/4$. In particular $a\ge8\eps$.

Round every $b_k$ down to a positive dyadic $b'_k$ with error at most
\begin{equation}
 \delta\le\min\{\tfrac12\min_k b_k,\ a/(16{u})\}.
\end{equation}
Entrywise comparison for the nonnegative path gives
$\norm{B'_n}\le1$, and $\norm{B_n-B'_n}\le2\delta$.
The same site-dependent change of basis decorates both the original
and rounded chains, so it preserves their norm difference and gives
$\norm{H'_X}\le1$. A Hermitian matrix with norm at most one
has spectrum in $[-1,1]$. We can therefore make a PSD contraction
by shifting and rescaling it. On the path block, set
\begin{equation}\label{odel:eq:Dpath}
 D_X=(I-H'_X)/2.
\end{equation}
Append an isolated reference of eigenvalue zero and an isolated
eigenvalue one. The path eigenvalues of $D_X$ now lie in $[0,1]$. The isolated
eigenvalue one makes both the operator norm and the maximum
row/column Euclidean norm exactly one; that component is unoccupied
by the initial state. The zero-eigenvalue reference will remain
undamped. Thus $D_X\succeq0$, $\norm{D_X}=\eta_D=1$.
The added identity contributes one diagonal nonzero per path row,
giving sparsity at most $2m+1$. The rounded edge weights and the
power-of-two gadget entries are dyadic, as are all diagonal entries. Contractive Duhamel comparison
between the unrounded and rounded PSD generators gives
\begin{equation}
 \norm{e^{-{u}(I-H_X)}-e^{-{u}(I-H'_X)}}
 \le2{u}\delta\le a/8.
\end{equation}
Thus the rounded endpoint amplitude is $\sigma(X)a'$ with
$a'\ge7a/8$.

Start from $(\ket{\mathrm{ref}}+\ket{v,0})/\sqrt2$.
The evolved reference amplitude remains $1/\sqrt2$, while
contractivity bounds the norm of the path component by
$1/\sqrt2$. Their orthogonality gives a total squared norm
$P\in[1/2,1]$. Since the input norm is one and $b=0$,
$g=1/\sqrt P$ lies in $[1,\sqrt2]$.

The final path endpoint has amplitude $\sigma(X)a'/\sqrt2$.
Measuring its population would lose the sign, so use the known
interference observable
\begin{equation}
 M_0=\ket{\mathrm{ref}}\bra{v,n}
       +\ket{v,n}\bra{\mathrm{ref}}
\end{equation}
on the normalised solution. Its two off-diagonal terms multiply
the reference and endpoint amplitudes in both orders. Hence
\begin{equation}
 \langle M_0\rangle
 =\frac{2(1/\sqrt2)(\sigma(X)a'/\sqrt2)}P
 =\frac{\sigma(X)a'}P.
\end{equation}
The parity-corrected expectation is at least $7a/8$ because
$P\le1$. As in Lemma~\ref{odel:lem:herald}, conditional trace
error costs at most $2\eps$ in expectation and a fair sign on
failure contributes zero. The unconditional correlation is
therefore at least $(2/3)(7a/8-2\eps)\ge a/4$.

Apply the small-correlation search-parity bound,
Lemma~\ref{hsl:lem:search}. For $n$ independent hidden positions
in lists of length $m$, correlation $\beta_0>0$ requires
\begin{equation}\label{odel:eq:smallbias}
 q\ge\frac{\log2}{24}n\sqrt m\,\beta_0^{1/(2n)},
 \qquad m\ge4\beta_0^{-1/n}.
\end{equation}
Our correlation is $\beta_0=a/4$. The amplitude estimate gives
$4\beta_0^{-1/n}\le8en/u<32n/u$, so $m\ge32n/u$
is sufficient. It also gives
$\beta_0^{1/(2n)}\ge\frac14\sqrt{u/n}$.
Since each matrix query costs at most two membership queries,
\begin{equation}\label{odel:eq:psdmixedlocal}
 Q_A\ge\frac{\log2}{192}\sqrt{mnu}.
\end{equation}

Choose $m$ to be the largest power of two at most $(s-1)/2$.
Then $m>s/8$ and $2m+1\le s$. If $s\ge512n/\tau$,
the sufficient condition $m\ge32n/{u}$ holds, and
\eqref{odel:eq:nrange} and \eqref{odel:eq:psdmixedlocal} yield
\begin{equation}
 Q_A\ge\frac{\log2}{1536\sqrt2}\sqrt{s\tau r}.
\end{equation}

To cover the precision term even when the list length is too
small for this adversary estimate, replace each edge gadget by
a scalar sign $(-1)^{x_k}$. The resulting matrix is three-sparse
and gives the same positive endpoint correlation. The expectation of an unconditional observable after $q$ Boolean
queries is a multilinear polynomial of degree at most $2q$ in
the input bits: amplitudes have degree at most $q$, and an
expectation is quadratic in amplitudes. Every monomial of degree
less than $n$ omits at least one bit. Averaging its product with
full $n$-bit parity over that missing bit gives zero. A positive
uniform parity correlation therefore requires $2q\ge n$. With $q\le2Q_A$, this gives
$Q_A\ge n/4\ge r/32$. The observable includes the success flag,
as above; the polynomial argument is applied to its unconditional
expectation.

If $s<512n/\tau$, then $r\le8n$ implies
$\sqrt{s\tau r}<64n$, so this scalar family also provides
$Q_A>\sqrt{s\tau r}/256$. Taking the maximum of the two families
proves \eqref{odel:eq:psd}. Known positive semidefinite dense padding,
as in Theorem~\ref{odel:thm:joint}, makes maximum sparsity exactly $s$
when required.
\end{proof}

This theorem establishes precision-dependent hardness inside the
purely dissipative, constant-coefficient subclass. It does not prove
a constant-error $\Omega(\sqrt{s}\tau)$ bound for that subclass or
necessity of the factor $\ell$ multiplying $\tau_D$ in
\eqref{odel:eq:refined}. In particular, a linear power of $L$ in the
sparsity-dependent dissipative cost has not been established.

\subsection{Fixed-parameter precision comparison}

The lower bounds above leave open how precision combines with
large time, sparsity and normalisation. We can nevertheless rule
out one stronger assertion: the entire logarithm $\ell$ in the
upper bound cannot be necessary uniformly across all parameter
regimes. To test that assertion, hold $s,\tau,\bar g$ fixed
and ask how the query count grows as $\eps$ tends to zero.
If an algorithm valid for the same input class has growth smaller
than $\log(1/\eps)$ in this limit, a uniform lower bound
$\Omega(\bar g\sqrt{s}\tau\ell)$ is impossible.

We construct such a comparison algorithm by truncating the Dyson
series. The factorial in its tail gives sublogarithmic precision
dependence, but we pay the full coefficient sum $e^{s\tau}$
when amplifying the solution. This exponential cost makes the
algorithm unattractive when time or sparsity grows. At fixed
parameters, however, it is a constant independent of precision,
which is all that the present comparison requires.

The comparison must cover the whole time-dependent input class;
a faster constant-coefficient algorithm would not suffice.
Accordingly, the following lemma uses the coherent access,
certified discretisation and source-preparation assumptions of
Theorem~\ref{ode:thm:time}.

\begin{lemma}[Fixed-parameter precision obstruction]\label{odel:lem:dyson}
Assume the coherent time queries, attainable sample accuracy and
source access of Theorem~\ref{ode:thm:time}, including coherent
matrix access on every finite mesh selected from the public error
certificate. Assume a supplied common
row and column Euclidean-norm bound $\eta$ for $A(t)$.
Put $z:=s\eta T=s\tau>0$ and
$L_g=\log(C\bar g/\eps)$. There is a generic sparse-oracle solver
with
\begin{equation}\label{odel:eq:dyson}
 Q_A=\mathcal{O}\!\left(\bar g e^z
       \left[1+z+\frac{L_g}{\log(e+L_g/z)}\right]\right).
\end{equation}
Only matrix queries are counted. Gate costs and time-register sizes
may depend on the certified discretisation.
\end{lemma}
\begin{proof}
We first represent a finite-mesh propagator as an LCU of ordered
matrix products. Its largest retained degree will determine the
query cost of one preparation; its coefficient sum will determine
the amplification cost. We then combine these propagators with
the given source quadrature, just as in the dissipative solver.

First fix a propagator error budget $\delta>0$. On a public mesh
$0=t_0<\cdots<t_n=T$ containing the known breakpoints, let
$\widetilde A_j$ denote the generator decoded from the actual
oracle on cell $j$ and write $h_j=t_{j+1}-t_j$. Its integrated
generator error $E_0=\sum_j\int_{t_j}^{t_{j+1}}
\norm{\widetilde A_j-A(t)}\,dt$ can be certified directly or by
\begin{equation}
 E_0\le\sum_j\int_{t_j}^{t_{j+1}}
           \norm{A(t)-A(t_j)}\,dt
       +\sum_j h_j\norm{\widetilde A_j-A(t_j)}.
\end{equation}
At a breakpoint use the endpoint value from inside the cell.
For a Lipschitz bound $L_j$ there, the first integral on cell $j$
is at most $L_jh_j^2/2$. The second sum uses the supplied value accuracy;
a finer mesh alone need not reduce it.

The same control follows if the input gives these bounds for
$K_{\pm k}=H\pm kD$ on a common mesh, for two opposite frequencies
with $k\ge1$ (for example the endpoints of the common frequency grid).
Indeed, let $F_+$ and $F_-$ be the differences between the decoded
and true branch generators at time $t$, with the decoded values
held constant on each cell. The corresponding differences of $H$ and $D$ are
$(F_++F_-)/2$ and $(F_+-F_-)/(2k)$. Since $A=-D-\ii H$, its
difference has norm at most
$(1+1/k)(\norm{F_+}+\norm{F_-})/2$. Integrating this bound over the cells gives the required bound for $E_0$.

If the supplied bound already gives $|\widetilde A_{j,ab}|\le\eta$
for every entry, set
$A_{\rm mesh}(t)=\widetilde A_j$ on cell $j$. Otherwise radially
project each decoded entry onto the disk $|w|\le\eta$ and use this
clipped matrix on that cell. Every true entry lies in the disk.
Projection cannot increase its entrywise error, and the union of
the supports of $\widetilde A_j$ and $A(t)$ has at most $2s$
positions per row and column. Before clipping, each entry error is at most
$\norm{\widetilde A_j-A(t)}$. The maximum absolute
row and column sums therefore bound the clipped operator error by
$2s\norm{\widetilde A_j-A(t)}$ at every $t$ in cell $j$. In either case we can
certify
\begin{equation}
 E:=\int_0^T\norm{A_{\rm mesh}(v)-A(v)}\,dv\le\log(1+\delta)
\end{equation}
by requesting $E_0\le\log(1+\delta)/(2s)$; without clipping the
factor $2s$ can be omitted. Clipping the decoded value uses no
additional matrix queries and does not enlarge its support.

Let $U$ and $V$ be the propagators of $A$ and $A_{\rm mesh}$,
respectively, and set $R=A_{\rm mesh}-A$. Only $U$ is known to
be contractive. For $u\le t$, variation of constants gives
\begin{equation}
 V(t,u)=U(t,u)+\int_u^t U(t,v)R(v)V(v,u)\,dv.
\end{equation}
Gronwall's inequality gives
$\norm{V(t,u)}\le\exp(\int_u^t\norm{R(v)}\,dv)$.
Substitution in the same identity bounds the propagator error by
\begin{equation}
 \norm{V(T,u)-U(T,u)}
 \le\exp\!\left(\int_u^T\norm{R(v)}\,dv\right)-1
 \le e^E-1\le\delta,
\end{equation}
uniformly over source start times. This permits a represented
sample whose Hermitian part is not negative semidefinite.

We use the ordinary sparse block encoding of the matrix actually
implemented, with normalisation $\alpha=s\eta$. Fix a cell,
write $M$ for its implemented matrix, and denote the supplied row
and column lists by $r_i(l)$ and $c_j(l)$. On two matrix-index registers and a
flag qubit, define the isometries
\begin{align}
 V_L\ket i
 &=\frac1{\sqrt s}\sum_{l=1}^s\ket{i,r_i(l)}
 \left[\frac{\overline{M_{i,r_i(l)}}}\eta\ket0
 +\sqrt{1-\frac{|M_{i,r_i(l)}|^2}{\eta^2}}\ket1\right],\\
 V_R\ket j
 &=\frac1{\sqrt s}\sum_{l=1}^s\ket{c_j(l),j}\ket0.
\end{align}
Each bracketed flag state has norm one, and the reversible lists
give distinct basis labels, so both maps preserve norms. Their
overlap selects the flag-zero amplitude and matches the two matrix
indices. For every entry in the support, that match occurs once
in each list; for a zero entry the flag-zero amplitude vanishes,
including on padding labels. Thus
\begin{equation}
 \bra iV_L^\dagger V_R\ket j
 =\frac1s\frac{M_{ij}}\eta.
\end{equation}
The conjugate in $V_L$ becomes the desired entry when taking
its adjoint; no Hermiticity is used.

Let $U_L$ and $U_R$ be unitary preparations of these maps
from clean work registers. They prepare a uniform list position,
use the original locator to obtain the second matrix index, and,
for $U_L$, query the entry, rotate the flag and erase the entry
encoding. The indices stay fixed during the rotation, so that
erasure is valid coherently. Known routing places the input in
the first or second index register as appropriate. Their actual
unitary extensions give
\begin{equation}
 (\bra{0_{\rm work}}\otimes I)U_L^\dagger U_R
 (\ket{0_{\rm work}}\otimes I)
 =V_L^\dagger V_R=M/(s\eta).
\end{equation}
Each preparation and its inverse uses a constant number of original
queries. The public mesh maps a cell label to its address reversibly,
so one coherent selection implements the corresponding block at
constant query cost, even when the mesh has many cells. Below,
$U_L(t_j)$ and $U_R(t_j)$ denote the preparations for cell $j$.

For a branch beginning at $u$, the degree-$k$ Dyson term is
\begin{equation}
 \int_{u\le t_1\le\cdots\le t_k\le T}
 A_{\rm mesh}(t_k)\cdots A_{\rm mesh}(t_1)
 \,dt_1\cdots dt_k.
\end{equation}
The ordered integration region has volume $(T-u)^k/k!$.
After replacing each factor by its normalised block
$A_{\rm mesh}/\alpha$, the coefficient weight is therefore
$[\alpha(T-u)]^k/k!\le z^k/k!$. The remaining operator is
an average of chronological products. Equivalently, draw $k$
independent uniform times in $[u,T]$ and sort them: the $k!$
orderings of distinct times give the same ordered integral.

We implement this average using cell labels, since the represented
matrix is constant within each mesh cell. Prepare $k$ independent
labels with probabilities equal to their known cell lengths divided
by $T-u$, sort a working copy reversibly while retaining the
original labels, and select the matrix factors in chronological
order. When $u$ lies inside a cell, use the remaining length of
that first cell. The degree-zero term is the identity, and the
zero-duration branch needs no sampled factors. One degree-$k$
selection uses $O(k)$ matrix queries, regardless of the number
of possible cell labels.

Write $U_j:=U_L(t_j)^\dagger U_R(t_j)$ for the unitary whose
selected block is $A_{\rm mesh}(t_j)/\alpha$. For ordered labels
$j_1\le\cdots\le j_k$, apply each $U_{j_l}$ with a fresh zero
work register $a_l$. Selecting all these registers gives
\begin{equation}
 (\bra{0_{a_1}\cdots0_{a_k}}\otimes I)
 U_{j_k}\cdots U_{j_1}
 (\ket{0_{a_1}\cdots0_{a_k}}\otimes I)
 =\frac{A_{\rm mesh}(t_{j_k})\cdots
             A_{\rm mesh}(t_{j_1})}{\alpha^k}.
\end{equation}
Distinct work registers prevent a failed earlier factor from
returning to the selected block through a later factor.
Uncomputing the sorting and the original label preparation
averages the products with the required known weights.

Repeated cell labels do not create a discretisation error in this
representation. If cell $j$ occurs $r_j$ times, there are
$k!/\prod_jr_j!$ unsorted label sequences with those counts.
Their probabilities include the product of the cell-length weights.
After multiplying by the overall coefficient
$[\alpha(T-u)]^k/k!$, the factorial $k!$ cancels and each
cell contributes its length to power $r_j$, divided by $r_j!$.
These are exactly the coefficients obtained by integrating the
repeated constant matrix within that cell. Thus the finite mesh
is represented exactly before degree truncation, for unequal as
well as equal cells. Refining it increases register sizes and known
preparation costs, but does not multiply the number of matrix
queries in a selected product.

Truncation at degree $K$ has error at most
$\sum_{k>K}z^k/k!$. For $K+2\ge2z$, the factorial tail obeys
\begin{equation}
 \sum_{k>K}\frac{z^k}{k!}
 \le 2\frac{z^{K+1}}{(K+1)!}
 \le2\left(\frac{ez}{K+1}\right)^{K+1}.
\end{equation}
We need the final tail to be a small constant times
$\eps/\bar g$, so that normalising the solution still leaves
error at most $\eps$. Taking logarithms of the tail bound shows
the requirement
\begin{equation}
 (K+1)\log\frac{K+1}{ez}\ge L_g+O(1),
 \qquad K+2\ge2z.
\end{equation}
When $L_g$ is at most of order $z$, a sufficiently large
constant multiple of $1+z$ meets both conditions. When $L_g/z$
is large, the logarithm on the left permits a degree of order
$L_g/\log(e+L_g/z)$, in addition to that baseline. Combining
the regimes gives
\begin{equation}\label{odel:eq:degree}
 K=\mathcal{O}\!\left(1+z+\frac{L_g}{\log(e+L_g/z)}\right)
\end{equation}
The sum of the retained positive LCU coefficients is at most $e^z$.

For the original forced solution
$U(T,0)x_0+\int_0^T U(T,u)b(u)\,du$, use the weight and data
preparations of Eqs.~\eqref{ode:eq:source-interface-mass}--\eqref{ode:eq:source-interface-preparation},
including the supplied phases for signed or complex quadrature coefficients.
Write $B_q\le c_qB$ for the total input/source coefficient norm,
where $c_q$ is the absolute constant in the source-access promise.
Choose the source quadrature to have vector error at most $\delta B$
against the true propagator. Replacing $U$ by $V$ on every branch
costs at most $\delta B_q$ in the resulting sum, since the preceding
propagator bound is uniform in the starting time. Choose $K$ so that
the factorial tail is at most $\delta$ as well. Truncation then costs
at most another $\delta B_q$. Use the common coefficient sequence
$z^k/k!$ for all branches.
For a branch starting at $u$, the coefficient relative to this
common sequence is $((T-u)/T)^k$. It lies in $[0,1]$ and is
known from the branch and degree labels, so a controlled rotation
can place it in the selected amplitude. The zero-degree coefficient
is one even on the zero-duration branch. We may pad the retained
coefficient mass to $e^z$ with cancelling identity terms, giving
the same LCU normalisation on every source branch.

Prepare the weights with $P_w$, load the stated data vectors, apply
the coefficient phase to each source label, and select the appropriate
Dyson block on that label. Apply $P_w^\dagger$ to the label alone.
The same two square-root weights as in
Eq.~\eqref{ode:eq:source-interface-output} now select the truncated
quadrature sum divided by $e^z B_q$. All source nodes, including one
inside a matrix cell, use this one circuit. Its largest degree is
$K$; padding shorter branches with disabled calls gives $O(K)$
matrix queries per preparation.

Write $y/(e^zB_q)$ for the selected vector of the actual preparation.
Let $\xi$ bound the error of this vector relative to the ideal
truncated circuit from finite preparations, phases and rotations.
Their errors add by telescoping: in particular, the error of the
selected bra $\bra0 P_w^\dagger$ equals that of $P_w\ket0$.
The source interface supplies the requested accuracy, and known
rotations can be refined without extra matrix queries. On the
unnormalised solution scale we obtain
\begin{equation}
 \norm{y-x(T)}\le \delta B+2\delta B_q+e^z B_q\xi .
\end{equation}
Take $\delta=\eps/[16(1+c_q)\bar g]$ and
$\xi\le e^{-z}\delta$. Since $c_q$ is absolute, this choice of
$\delta$ gives the degree in Eq.~\eqref{odel:eq:degree} and
$\norm{y-x(T)}\le\eps B/(4\bar g)$. Thus $y\ne0$ and the
error after normalisation is at most
$2\norm{y-x(T)}/\norm{x(T)}\le\eps$. The actual success
amplitude $\norm y/(e^z B_q)$ is $\Omega(1/(e^z\bar g))$.

Amplify the actual circuit using its own inverse. A direct trial
covers amplitudes above $1/2$ with constant probability; for smaller
amplitudes, randomised amplification with a bounded schedule of
$O(e^z\bar g)$ calls does so. Mixing the two trials and repeating
a constant number of times gives heralded success probability at
least $2/3$ without changing the conditional direction. Multiplying
by the cost of one preparation proves Eq.~\eqref{odel:eq:dyson}.
\end{proof}

At fixed positive $s,\tau,\bar g$, this gives
\begin{equation}\label{odel:eq:fixed}
 Q_A=\mathcal{O}\!\left(\frac{\log(1/\eps)}{
                               \log\log(1/\eps)}\right)
 \qquad (\eps\longrightarrow0).
\end{equation}
Consequently a uniform worst-case lower bound proportional to
$\bar g\sqrt{s}\tau\log(\bar g/\eps)$ is false under these access
assumptions: fixing any admissible positive $s,\tau,\bar g$ makes
this comparison upper bound sublinear in $\log(1/\eps)$.
A constant-coefficient upper bound alone would not exclude such a
worst-case assertion for the larger time-dependent class;
Lemma~\ref{odel:lem:dyson} is included precisely to cover that class.
Its exponential $e^{s\tau}$ cost gives no improved comparison with
\eqref{ode:eq:simple} in a joint large-parameter regime, and no
gate-complexity or time-register-size bound is claimed here.

For sufficiently large dimension we therefore have
\begin{equation}
 \Omega(\bar g\sqrt{s}\tau+\sqrt{s\tau r}+r)
 \ \le\ Q_A^{\rm worst}\ \le\
 O(\bar g\sqrt{s}\tau\ell)
 \qquad(s\ge32,\ \tau\ge4,\ \bar g\ge2).
\end{equation}
The cancellation construction proves the necessity of the joint
time--sparsity--normalisation factor for the time-dependent class.
The PSD construction proves a different point: precision hardness
persists even with no Hamiltonian part and with bounded $g$.
Together these results place the known lower bound within one
logarithm of the general upper bound. They do not establish either
the full dissipative precision factor or its multiplication by
$\bar g$.

At fixed time, $r$ grows as
$\log(1/\eps)/\log\log(1/\eps)$. Thus the
\emph{sparsity-dependent} term supplied by the PSD construction
contains the square root of this precision scale; the separate
$r$ term remains present as well. This is weaker than proving
a full $\log(1/\eps)$ multiplying the sparsity-dependent
dissipative time cost. This distinction agrees with Low--Somma's
separation of parameter lower bounds from the unresolved
multiplicative time--precision claim~\cite[Section 7]{LS}.

The fixed-parameter comparison concerns matrix queries; it places
no bound on the known gates or time registers of the Dyson comparison solver.

\section{An optimal sparse-query linear-system algorithm}
\label{part:qls}

We now use the sparse-entry Cayley construction to solve a linear
system. Given an invertible sparse matrix $A$ and a preparation of
the normalised right-hand side $b$, we seek the quantum state
$\ket x=A^{-1}b/\norm{A^{-1}b}$. We assume a supplied spectral
scale $\norm A\le1$ and inverse-norm bound
$\norm{A^{-1}}\le\kappa$. Our aim is a matrix-query cost
$O(\sqrt{s}\,\kappa\log(1/\eps))$, including the work needed
to find a useful starting overlap. The precise access and conditional
output guarantees are given below.

The starting point is to encode the answer in a kernel, the subspace
annihilated by a matrix. Append the column $-b$ to $A$. For a vector
with system coordinates $y\in\mathbb C^N$ and appended scalar
$c\in\mathbb C$, the kernel equation is
\begin{equation}
 [A,-b]\binom yc=Ay-cb=0
 \quad\Longleftrightarrow\quad y=cA^{-1}b.
\end{equation}
Thus
\begin{equation}
 \ker[A,-b]=\operatorname{span}\left\{\binom{A^{-1}b}{1}\right\}.
\end{equation}
The appended coordinate is a known basis state. Projecting it onto
this kernel would produce a vector containing the solution, and
selecting the original system coordinates would then return
$\ket x$. These are two probabilistic operations. If the solution
part is much larger than the appended part, the initial projection
is unlikely; if it is much smaller, extracting the solution is
unlikely. We will balance the two parts by rescaling the appended
column using an estimate of $\norm{A^{-1}b}$.

We also need an inexpensive way to perform the kernel projection.
The entrywise Cayley unitary provides one: multiply it by a known
reflection to obtain a unitary walk whose eigenvalue-one space
encodes the desired kernel. All other eigenphases are separated
from zero by $\Omega(1/(\kappa\sqrt{s}))$. A polynomial in the
walk can therefore preserve the fixed vectors and suppress the
remaining phases. We will construct that filter with nonnegative
coefficients summing to one, so its implementation introduces no
additional coefficient-normalisation penalty.

The final ingredient is the norm estimate used to balance the kernel.
We follow a regularisation path from a matrix with known inverse-vector
norm to the original problem. The norm changes by at most a constant
factor at each stage, so a constant number of nearby guesses suffices.
The stages become geometrically more expensive. By requiring greater
confidence at the cheap early stages, we bound the cost of the whole
path without an extra logarithm in $\kappa$. We develop the walk
and its gap first, then the balanced kernel and filter, and finally
the complete estimation and solution-preparation algorithm.
In Section~\ref{qls:sec:separate}, we use the same walk construction
with a supplied norm estimate to obtain a separate bound on the
number of right-hand-side preparation queries.

\subsection{The sparse-query guarantee and kernel walk}
\label{qls:sec:main}

\begin{theorem}[Sparse linear-system preparation]\label{qls:thm:main}
Let $A\in\mathbb C^{N\times N}$ be invertible, with at most $s\ge1$
nonzero entries per row and column, and suppose
\begin{equation}\label{qls:eq:promise}
 \norm A\le1,\qquad \norm{A^{-1}}\le\kappa,\qquad \kappa\ge1.
\end{equation}
Supply the entry-value and in-place row and column locators
of Section~\ref{hs:sec:control}, including their supplied inverses
and the specified phase-one list action with full unitary extensions. Also supply a circuit description
of $U_b$, with $U_b\ket0=\ket b$ for a normalised nonzero $b$,
and its inverse.

For $0<\eps\le1/10$, an algorithm prepares
$\ket x=A^{-1}\ket b/\norm{A^{-1}b}$ with a classical success flag
of probability at least $2/3$. Conditional on success, its possibly
mixed output $\rho$ obeys
\begin{equation}
 \tfrac12\norm{\rho-\ket x\bra x}_1\le\eps.
\end{equation}
The worst-case number of matrix queries is
\begin{equation}\label{qls:eq:main}
 \mathcal O\!\left(\sqrt{s}\,\kappa\log\frac1\eps\right),
\end{equation}
and calls to $U_b,U_b^\dagger$ have at most the same order.
The constants are absolute; there are no hidden logarithms in
$s$ or $\kappa$.
\end{theorem}

The promise~\eqref{qls:eq:promise} is the standard normalised linear-system
promise. For an unscaled matrix and a known spectral upper bound $\Lambda$,
the corresponding parameter is $\Lambda\norm{A^{-1}}$. It agrees with
the condition number to constant factors when $\Lambda$ is comparably tight.

The matrix model is the canonical model of
Section~\ref{hs:sec:control}. Both locator directions are supplied
for general $A$; for a Hermitian matrix, symmetry can supply one
from the other. The three-query dummy-input construction gives
controlled locators even when their off-promise extensions are
arbitrary, and compute--conditional-copy--uncompute controls the
XOR value oracle.

The supplied circuit description of $U_b$ permits gate-by-gate
control. We count a controlled preparation, an uncontrolled one
or either inverse as one preparation call. An opaque $U_b$ would
instead require controlled access as an additional assumption.
We never query entries of the generally dense vector $b$.
The matrix represented by the entry oracle is exact in the query
model. Known rotations may be synthesised without additional
matrix calls; this theorem does not bound their gate cost or
the total work space.

To build the walk, we first recall what the sparse representation
provides. For a Hermitian matrix $H$, an isometry $J$ embeds the
data space in a larger edge space, and a Hermitian operator $K$
acts there by processing individual entries. The compression
$J^\dagger KJ$ recovers a rescaled $H$. The second compression
$J^\dagger K^2J$ controls $\norm{KJu}^2$, the squared norm
on embedded inputs, and can be much smaller than $\norm K^2$.
We need both estimates: the full norm will locate possible resolvent
singularities, while the compressed second moment controls the lift
of a kernel vector and the phase gap near zero.

\begin{lemma}[Sparse Cayley representation]\label{qls:lem:entry}
Let $H$ be Hermitian, $s$-sparse in both rows and columns, with $\norm H\le1$.
There are an isometry $J$ and a Hermitian $K$ such that
\begin{equation}\label{qls:eq:moments}
 J^\dagger KJ=H/\sqrt{s},\qquad
 J^\dagger K^2J\preceq I,\qquad \norm K\le\sqrt{s}.
\end{equation}
After a sparse-oracle change of edge coordinates, $J$ is implemented
without matrix queries, and the Cayley unitary
\begin{equation}\label{qls:eq:cayley}
 {D_0}=(K+\ip I)(K-\ip I)^{-1}
\end{equation}
and its inverse each use $\mathcal{O}(1)$ primitive sparse queries.
The same holds for the Cayley transform of any known scalar multiple of $K$.
\end{lemma}
\begin{proof}
Use the edge-space construction of Section~\ref{ov:sec:representation}
with $\eta=1$. To recall its normalisation, let
$\mathcal E=\mathbb C^N\otimes\mathbb C^N$ and write
\begin{equation}
 L_{\rm row}\ket i=\sum_{j\text{ in row list }i}\ket{i,j},\qquad
 R_{\rm col}\ket j=\sum_{i\text{ in column list }j}\ket{i,j},\qquad
 Y=\sum_{i,j}H_{ij}\ket{i,j}\bra{i,j}.
\end{equation}
The two maps are unnormalised list embeddings from
$\mathbb C^N$ to the $N^2$-dimensional edge space $\mathcal E$.
The first fixes a row and lists its column indices; the second fixes
a column and lists its row indices. Each list has $s$ distinct
entries, including its prescribed zero padding. The diagonal
operator $Y$ multiplies the basis state $\ket{i,j}$ by the
matrix entry $H_{ij}$. We add a two-dimensional orientation
register to distinguish the two list embeddings, so the private
space is $\mathcal E\oplus\mathcal E$, of dimension $2N^2$.
Set
\begin{equation}\label{qls:eq:JKexplicit}
 J:=\frac1{\sqrt{2s}}\binom{L_{\rm row}}{R_{\rm col}},\qquad
 K=\sqrt{s}\begin{pmatrix}0&Y\\Y^\dagger&0\end{pmatrix}.
\end{equation}
The off-diagonal blocks are adjoints, so $K$ is Hermitian.
Within each list the labels are distinct, and different row inputs
have different first indices, while different column inputs have
different second indices. Hence
$L_{\rm row}^\dagger L_{\rm row}
=R_{\rm col}^\dagger R_{\rm col}=sI$, which gives
$J^\dagger J=I$. The overlap
$L_{\rm row}^\dagger YR_{\rm col}$ equals $H$: its $(i,j)$
entry selects precisely the edge $(i,j)$ whenever that entry is
nonzero. Thus the two orientations give
\begin{align}
 J^\dagger KJ&=(H+H^\dagger)/(2\sqrt{s})=H/\sqrt{s},\\
 J^\dagger K^2J
 &=\tfrac12(L_{\rm row}^\dagger YY^\dagger L_{\rm row}+R_{\rm col}^\dagger Y^\dagger YR_{\rm col})
 \preceq I.
\end{align}
The matrices in the second line are diagonal row and column sums
of squared entries; each is at most $I$ because $\norm H\le1$.
Also $\norm K=\sqrt{s}\max_{ij}|H_{ij}|\le\sqrt{s}$.
Zero padding is annihilated by $Y$, so the two lists need not agree
on their padded entries.

For each edge $(i,j)$, the two orientation states carry
\begin{equation}
 K_{ij}=\sqrt{s}
       \begin{pmatrix}0&H_{ij}\\\overline{H_{ij}}&0\end{pmatrix}.
\end{equation}
Since $K_{ij}^2=s|H_{ij}|^2I$, the full Cayley block is
\begin{equation}
 (K_{ij}+\ii I)(K_{ij}-\ii I)^{-1}
 =I-\frac{2(I-\ii K_{ij})}{1+s|H_{ij}|^2}.
\end{equation}
Let $C_{ij}$ denote the displayed two-dimensional unitary and
let $w_{ij}$ be the entry encoding returned by the value oracle.
For any orientation state $\ket\chi$, the implementation is
\begin{equation}
\begin{aligned}
 \ket{i,j}\ket\chi\ket0
 &\longmapsto\ket{i,j}\ket\chi\ket{w_{ij}}\\
 &\longmapsto\ket{i,j}(C_{ij}\ket\chi)\ket{w_{ij}}\\
 &\longmapsto\ket{i,j}(C_{ij}\ket\chi)\ket0.
\end{aligned}
\end{equation}
These are, respectively, the value query, the encoding-controlled
rotation and the inverse query. Only the orientation changes during
the rotation, so the indices and encoding needed for uncomputation
remain intact. By linearity the same sequence applies all edge
rotations coherently, as in \eqref{hs:eq:edge}; it is not a serial
loop over matrix entries. A known scalar multiple of $K$ changes
only the rotation parameter and has the same query cost.

We finally separate known preparation from the locator calls.
Let $J_{\rm free}$ prepare a uniform list position and the two
orientations, retaining the data index. Let $F$ apply the row or
column locator selected by that orientation, with known routing
of the indices. Then $J=FJ_{\rm free}$. In list coordinates,
we may use the known embedding $J_{\rm free}$ and the private
operation $F^\dagger D_0F$. Each call changes to edge coordinates,
performs the entry rotation, and changes back, at a constant number
of locator and value queries. The supplied full unitary extensions
make this valid throughout the private space, not just on occupied
list states. Below, we write $J,K,D_0$ for the embedding and
operators in these conjugated coordinates.
\end{proof}

The representation gives the exact identity
\begin{equation}\label{qls:eq:R0}
 R_0:=(I-D_0)^{-1}=\frac{I+\ii K}{2}.
\end{equation}
Although $\norm K$ may be $\Theta(\sqrt{s})$, the compressed
second moment gives $\norm{KJu}\le\norm u$. We will use the
full norm to control resolvent singularities and the compressed
norm to control the kernel lift.

We now want a unitary whose fixed vectors identify the kernel
of the compressed matrix. The projector $JJ^\dagger$ selects
the embedded data subspace $\{Ju:u\in\mathbb C^N\}$.
The reflection $I-2JJ^\dagger$ changes the sign on this subspace
and acts as the identity on its orthogonal complement. Choose
\begin{equation}
 W=(I-2JJ^\dagger)D_0.
\end{equation}
One walk step applies the Cayley unitary and then this reflection.
The choice is useful because the resolvent identity converts its
fixed-vector equation into the kernel equation, as we now show. In this section $W$ denotes this unitary
walk, rather than the squared-embedding bound used in the
simulation argument. For the general form of the argument, suppose
$J^\dagger KJ=H/\alpha$, where $\alpha>0$ is the known
compression scale. The sparse representation above has
$\alpha=\sqrt s$. Since $D_0R_0=R_0-I$, applying the walk
to $R_0Ju$ gives
\begin{equation}
 (W-I)R_0Ju
 =Ju-2J(J^\dagger R_0J)u
 =-\ii J(H/\alpha)u.
\end{equation}
Thus $Hu=0$ makes $R_0Ju$ a fixed vector of the walk.
The vector $R_0Ju$ is a mathematical lift of $u$ into the private
space; we are not assuming an operation that prepares this
input-dependent vector directly. The filter will prepare its
direction from a known starting state later.

Two facts are still needed. We must rule out fixed vectors that
do not come from the kernel, and separate all other eigenphases
from zero so that the filter can distinguish them. The next lemma
does both. It allows general bounds $m_2$ on the compressed second
moment and $k_0$ on the full norm, because the appended source
column will require a slightly larger representation.

\begin{lemma}[Kernel correspondence and phase gap]\label{qls:lem:walk}
Suppose $J$ is an isometry, $K$ is Hermitian, and
$\alpha>0$, $m_2\ge0$, $k_0\ge0$ satisfy
\begin{equation}
 J^\dagger KJ=H/\alpha,\qquad
 J^\dagger K^2J\preceq m_2 I,\qquad \norm K\le k_0.
\end{equation}
Let ${D_0}$ be~\eqref{qls:eq:cayley} and set
\begin{equation}\label{qls:eq:walk}
 W:=(I-2JJ^\dagger){D_0}.
\end{equation}
The map $u\mapsto v=R_0Ju$ is a bijection between $\ker H$ and
$\ker(I-W)$. On this kernel,
\begin{equation}\label{qls:eq:lift}
\begin{aligned}
 J^\dagger v&=u/2,\qquad
 \norm v^2=\frac{\norm u^2+\norm{KJu}^2}{4},\\
 \frac{\norm u^2}{4}&\le\norm v^2
       \le\frac{1+m_2}{4}\norm u^2.
\end{aligned}
\end{equation}
If every nonzero eigenvalue of $H$ has absolute value at least
$\Delta_H>0$, then every nonzero principal eigenphase $\theta$ of $W$
with $|\theta|<\pi$ satisfies
\begin{equation}\label{qls:eq:tangap}
 \left|\tan\frac\theta2\right|\ge
 t_*:=\frac{\Delta_H/\alpha}{1+m_2+k_0\Delta_H/\alpha}.
\end{equation}
Consequently the nonzero eigenphase gap obeys
\begin{equation}\label{qls:eq:gap}
 \omega\ge2\arctan t_*.
\end{equation}
The conclusion permits a nontrivial kernel of $H$.
\end{lemma}
\begin{proof}
We first prove that the kernel lift captures the entire fixed
space and determine how much of a lifted vector can be recovered
through the data embedding. We then prove the phase gap. For the
first task, the fixed-vector equation itself supplies an inverse
to the lift. Since the reflection squares to the identity,
$Wv=v$ is equivalent to $(I-D_0)v=2JJ^\dagger v$.
The invertibility of $I-D_0$ then gives $v=R_0Ju$ with
$u=2J^\dagger v$. Its consistency condition is
\begin{equation}
 u=2J^\dagger R_0Ju=(I+\ii H/\alpha)u,
\end{equation}
so $Hu=0$. Conversely the design identity above makes every
such lift fixed. On the kernel, $J^\dagger v=u/2$, and
\begin{equation}
 \norm v^2
 =\langle u,J^\dagger R_0^\dagger R_0Ju\rangle
 =\frac{\norm u^2+\norm{KJu}^2}{4}.
\end{equation}
The norm formula follows from
$R_0^\dagger R_0=(I+K^2)/4$; its lower and upper bounds follow
from $K^2\succeq0$ and $J^\dagger K^2J\preceq m_2I$.
In particular, a nonzero $u$ cannot have zero lift. Conversely,
$u=2J^\dagger v$ recovers the unique kernel vector from a
fixed vector $v$. This proves the bijection and
\eqref{qls:eq:lift}. The identity $J^\dagger v=u/2$ also
specifies the unnormalised amplitude recovered when we later
project back through the embedding.

To prove the gap, we must exclude nearby nonzero phases even
when $H$ has a kernel. A norm estimate for a perturbation of $H$
alone would not do this: the original zero eigenvalues might
remain zero. We therefore show both that the perturbation has a
definite sign and that it is too small to bring an eigenvalue
of the opposite sign to zero.

Let $z=e^{\ip\theta}$ and set $\tau_\theta:=\tan(\theta/2)$.
The tangent parameter turns the phase equation into a Hermitian
matrix equation. Indeed, substituting
$z=(1+\ii\tau_\theta)/(1-\ii\tau_\theta)$ gives
\begin{equation}
 z-D_0=\frac{-2\ii}{1-\ii\tau_\theta}
       (I-\tau_\theta K)(K-\ii I)^{-1}.
\end{equation}
Since $K$ is Hermitian, $K-\ii I$ is invertible. If
$|\tau_\theta|k_0<1$, then $I-\tau_\theta K$ is
positive definite and invertible as well. Thus $z-D_0$ is
invertible in this range, and we may eliminate $v$ from
$Wv=zv$. The remaining consistency condition on
the system space is
\begin{equation}\label{qls:eq:secular}
 \ker\mathcal B(\tau_\theta)\ne\{0\},\qquad
 \mathcal B(\tau_\theta)=J^\dagger(K+\tau_\theta I)
                              (I-\tau_\theta K)^{-1}J.
\end{equation}
To obtain this condition, write
$(z-{D_0})v=-2JJ^\dagger D_0v$ and let $a=J^\dagger D_0v$.
For a nonzero eigenvector in the stated resolvent range, $a\ne0$;
otherwise invertibility of $z-{D_0}$ would force $v=0$. Substitution gives
\begin{equation}
 [I+2J^\dagger {D_0}(z-{D_0})^{-1}J]a=0.
\end{equation}
Using $z=(1+\ip\tau_\theta)/(1-\ip\tau_\theta)$ in the Cayley
formula gives
\begin{equation}
 I+2J^\dagger {D_0}(z-{D_0})^{-1}J
 =\ip J^\dagger(K+\tau_\theta I)(I-\tau_\theta K)^{-1}J
 =\ip\mathcal B(\tau_\theta).
\end{equation}
This reduces the phase-gap task to an eigenvalue crossing of a
Hermitian matrix on the smaller system space.
The advantage is that $\mathcal B(\tau_\theta)$ is Hermitian
for real $\tau_\theta$ and starts at $H/\alpha$. Its change is
\begin{equation}\label{qls:eq:resolventincrement}
 \mathcal B(\tau_\theta)-\mathcal B(0)
 =\tau_\theta J^\dagger(I+K^2)(I-\tau_\theta K)^{-1}J.
\end{equation}
For a scalar eigenvalue $k$ of $K$ we have
$|k|\le k_0$ and
$1-k_0|\tau_\theta|\le1-\tau_\theta k
\le1+k_0|\tau_\theta|$. The denominator is positive in
the range under consideration. Consequently
\begin{equation}
 \frac1{1+k_0|\tau_\theta|}
 \le\frac{1+k^2}{1-\tau_\theta k}
 \le\frac{1+k^2}{1-k_0|\tau_\theta|}.
\end{equation}
Functional calculus lifts these inequalities to $K$, since all
factors are functions of the same Hermitian operator. Compressing
by $J$ and using $J^\dagger K^2J\preceq m_2I$ gives
\begin{equation}\label{qls:eq:incrementbounds}
\begin{aligned}
 J^\dagger(I+K^2)(I-\tau_\theta K)^{-1}J
     &\succeq\frac{I}{1+k_0|\tau_\theta|},\\
 \norm{\mathcal B(\tau_\theta)-\mathcal B(0)}
     &\le\frac{|\tau_\theta|(1+m_2)}
                    {1-k_0|\tau_\theta|}.
\end{aligned}
\end{equation}
The two estimates have complementary roles. For
$\tau_\theta>0$, the increment is positive definite and is
at least $\tau_\theta/(1+k_0|\tau_\theta|)$ times the
identity. Ordered eigenvalues therefore increase strictly, including
those starting at zero. But negative eigenvalues might also move
towards zero. The norm estimate limits that motion to less than
their initial distance $\Delta_H/\alpha$ from zero. For
$\tau_\theta<0$, the signs reverse and the same argument
controls the positive eigenvalues. This is why a bound on the
perturbation norm alone would not suffice when $H$ has a kernel.

The admissible interval follows by solving the inequality that
prevents a crossing:
\begin{equation}
 \frac{|\tau_\theta|(1+m_2)}{1-k_0|\tau_\theta|}
 <\frac{\Delta_H}{\alpha}
 \quad\Longleftrightarrow\quad
 |\tau_\theta|<
 \frac{\Delta_H/\alpha}{1+m_2+k_0\Delta_H/\alpha}=t_*.
\end{equation}
The right side also ensures $|\tau_\theta|k_0<1$.
For $0<|\tau_\theta|<t_*$, we can therefore use both
inequalities in \eqref{qls:eq:incrementbounds}. For $\tau_\theta>0$, Weyl's eigenvalue perturbation inequality keeps
every originally negative eigenvalue of $\mathcal B(0)$ negative.
Strict positivity of the increment makes every originally zero or
positive eigenvalue positive. Thus $\mathcal B(\tau_\theta)$ is invertible.
For $\tau_\theta<0$, the originally positive eigenvalues remain positive and
the others become negative. This excludes the punctured interval
$0<|\tau_\theta|<t_*$ and proves~\eqref{qls:eq:tangap}. Taking
$2\arctan|\tau_\theta|$ gives~\eqref{qls:eq:gap}; a phase at $\pi$ already
satisfies this bound. No invertibility assumption on $H$ or ordinary
block-encoding qubitisation is used.
\end{proof}

\subsection{Kernel augmentation and solution overlap}
\label{qls:sec:augmentation}

The reflected walk can recognise the kernel of a Hermitian matrix.
We now construct a Hermitian kernel problem that contains the
solution of our possibly non-Hermitian $A$, has a known starting
coordinate, and remains easy to project onto during norm estimation.
We do this in two steps: first regularise a Hermitian dilation of
$A$, then append the source column at a chosen scale.

The Hermitian dilation $B$ below has a sector qubit and the original
$N$-dimensional data register. It preserves the singular values of
$A$. We introduce a second qubit, on which the Pauli matrices $X,Z$
act, to regularise $B$. At regularisation parameter $\xi=1$ the
matrix is just $X\otimes I$ and is easy to invert; at $\xi=0$
it becomes $Z\otimes B$ and contains the original inverse. Using
anticommuting Pauli matrices makes its squared singular values
explicit throughout this interpolation. Define
\begin{equation}\label{qls:eq:regularization}
 B=\begin{pmatrix}0&A\\A^\dagger&0\end{pmatrix},\qquad
 \beta_0=\ket0\ket b,\qquad \beta=\ket0\beta_0,
 \qquad
 M_\xi=\sqrt{1-\xi^2}\,Z\otimes B+\xi X\otimes I,
 \quad 0\le\xi\le1.
\end{equation}
Thus $\beta_0\in\mathbb C^{2N}$ places $b$ in the upper
dilation sector, and $\beta\in\mathbb C^{4N}$ also sets the
regularisation qubit to zero. The matrix $M_\xi$ acts on this
$4N$-dimensional space. In its square the two cross terms are
proportional to $(ZX+XZ)\otimes B=0$, while $X^2=Z^2=I$.

Both matrices are Hermitian, and
\begin{equation}\label{qls:eq:regularizedsquare}
 M_\xi^2=I\otimes\big((1-\xi^2)B^2+\xi^2 I\big).
\end{equation}
The eigenvalues of $B^2$ lie in $[1/\kappa^2,1]$.
The squared singular values of $M_\xi$ therefore lie between
$\xi^2+(1-\xi^2)/\kappa^2$ and one. In particular,
$\norm{M_\xi}\le1$, and its smallest singular value is at least
\begin{equation}
 g_\xi=\sqrt{\xi^2+(1-\xi^2)/\kappa^2}
 \ge\max\{\xi,1/\kappa\}.
\end{equation}

For a fixed $\xi$, the unknown norm is now
$r_\xi=\norm{M_\xi^{-1}\beta}$. To test a positive guess
$r$ for this norm, append the column $-\beta/r$. The resulting
kernel equation is
$M_\xi y-(\beta/r)c=0$, so its solution part is
$y=cM_\xi^{-1}\beta/r$. The ratio $r_\xi/r$ therefore
determines the weight of that part relative to the appended
coordinate. We also take a Hermitian dilation of this rectangular
matrix so that the walk lemma applies:
\begin{equation}\label{qls:eq:augmented}
 C_{\xi,r}=[M_\xi,-\beta/r],\qquad
 H_{\xi,r}=\begin{pmatrix}0&C_{\xi,r}\\C_{\xi,r}^\dagger&0\end{pmatrix}.
\end{equation}
Write $e$ for the basis vector of the appended coordinate in the
right-hand sector of this last dilation. This left/right sector
distinction is separate from the two qubits already used inside
$M_\xi$.

Before physical register padding, $C_{\xi,r}$ has $4N$ rows
and $4N+1$ columns, and its Hermitian dilation $H_{\xi,r}$ has
dimension $8N+1$. Since
$M_\xi$ is invertible, $H_{\xi,r}$ has a one-dimensional kernel,
spanned by the normalised vector
\begin{equation}\label{qls:eq:u}
 u_{\xi,r}
 =\frac{(0;\ M_\xi^{-1}\beta/r,\ 1)}{\sqrt{1+(r_\xi/r)^2}},
 \qquad r_\xi=\norm{M_\xi^{-1}\beta}.
\end{equation}
The semicolon separates the left and right sectors of the Hermitian
dilation; within the right sector, the last entry is the appended
coordinate. These are direct-sum components, not tensor factors.
There are no left-sector kernel vectors: $M_\xi$ is invertible,
so $C_{\xi,r}$ has full row rank. The right-sector kernel is
exactly the one-dimensional span displayed above. Moreover,
\begin{equation}
 C_{\xi,r}C_{\xi,r}^\dagger
 =M_\xi^2+\beta\beta^\dagger/r^2\succeq g_\xi^2I.
\end{equation}
The nonzero eigenvalues of a Hermitian dilation are the positive
and negative singular values of its off-diagonal matrix. Hence
the nonzero eigenvalue gap of $H_{\xi,r}$ is at least $g_\xi$,
uniformly in the norm guess $r$.

The kernel matrix is now specified mathematically, but its appended
column may be dense. Treating the whole matrix as a sparse input
would discard the desired access guarantee. We instead implement
the source contribution through the given preparation of $b$,
and keep it separate from the sparse contribution and the known
regularising term. Decompose
\begin{equation}
 H_{\xi,r}=\sqrt{1-\xi^2}\,H_0+\xi T_1+r^{-1}T_2.
\end{equation}
Here $H_0$ is the sparse off-diagonal lift of $Z\otimes B$,
$T_1$ is the corresponding known lift of $X\otimes I$, and
\begin{equation}
 T_2=-\ket{\mathrm{left},\beta}\bra e
       -\ket e\bra{\mathrm{left},\beta}.
\end{equation}
The operator $H_0$ is $s$-sparse with norm at most one.
Each of $T_1,T_2$ exchanges two active subspaces and is zero on
their complement: it squares to the identity on its support, and
has norm one. In particular, $T_2$ exchanges the appended coordinate
with the prepared source state. If we enlarge the $8N+1$ valid coordinates to fit physical
registers, the added coordinates would otherwise contribute
unwanted zero eigenvalues. Let $P_{\rm pad}$ be the projector
onto these unused coordinates. Adding it gives them eigenvalue
one and leaves the valid problem unchanged. We extend
$H_0,T_1,T_2$ by zero on the padding, while $P_{\rm pad}$
is zero on the valid space. The only remaining kernel vector
is therefore the one containing the solution.

We need to combine these four contributions without losing the
compressed moment bounds. Apply Lemma~\ref{qls:lem:entry} to
$H_0$, obtaining $(J_0,K_0)$ with first moment $H_0/\sqrt s$.
For each of the other three operators, use the identity embedding
into a separate copy of the augmented system space. These four
representation spaces are orthogonal branches of a direct sum.

Embed an input with amplitude $1/2$ in each branch. Their squared
weights are all $1/4$, so the four compressed first moments add
with those weights. We choose the branch scales to produce the
common denominator $32\sqrt s$. The factors of eight also
leave enough room to bound the second moment for all guesses
$r\ge1/8$, including guesses produced by failed tests. Set
\begin{equation}\label{qls:eq:composition}
 Ju=\tfrac12(J_0u\oplus u\oplus u\oplus u),\qquad
 K=\frac{\sqrt{1-\xi^2}}8K_0\ \oplus\ 
   \frac{\xi}{8\sqrt{s}}T_1\ \oplus\ 
   \frac1{8r\sqrt{s}}T_2\ \oplus\ 
   \frac1{8\sqrt{s}}P_{\rm pad}.
\end{equation}
The direct-sum notation means that $K$ acts independently within
each representation branch. It is not a tensor product of the
four operators. For example, the sparse first moment contributes
$(1/4)(\sqrt{1-\xi^2}/8)(H_0/\sqrt s)$, while the source
branch contributes $(1/4)(1/(8r\sqrt s))T_2$.
Squaring $K$ also produces no cross terms between branches.
Using the sparse second-moment bound and the norm-one bounds on
the known branches therefore gives, for $r\ge1/8$,
\begin{equation}\label{qls:eq:augmoments}
\begin{aligned}
 J^\dagger KJ&=(H_{\xi,r}+P_{\rm pad})/(32\sqrt{s}),\\
 J^\dagger K^2J
 &\preceq\frac{s(1-\xi^2)+\xi^2+r^{-2}+1}{256s}I
 \preceq\frac{66}{256}I\preceq I,\\
 \norm K&\le\sqrt{s}.
\end{aligned}
\end{equation}
Here the numerator estimate uses
$s(1-\xi^2)+\xi^2=s-\xi^2(s-1)\le s$, so that
$(s(1-\xi^2)+\xi^2+r^{-2}+1)/(256s)\le(s+65)/(256s)\le66/256$.
We also need a query implementation of each branch's Cayley
transform. On the source branch, $T_2$ acts only on the span
of $e$ and $\ket{\mathrm{left},\beta}$. In that ordered basis
it is the known matrix
\begin{equation}
 T_2\big|_{\operatorname{span}\{e,\ket{\mathrm{left},\beta}\}}
 =\begin{pmatrix}0&-1\\-1&0\end{pmatrix}.
\end{equation}
To specify the direction of the basis changes, let $Q_\beta$
apply the supplied preparation of $\beta$ in the left sector
and the identity elsewhere. In particular,
\begin{equation}
 Q_\beta\ket{\mathrm{left},0}=\ket{\mathrm{left},\beta},
 \qquad Q_\beta e=e.
\end{equation}
It also leaves the unused padding unchanged. The inverse
$Q_\beta^\dagger$ maps the two active states to the known
coordinates $\ket{\mathrm{left},0}$ and $e$. In those
coordinates, apply the Cayley transform of the displayed
two-by-two matrix scaled by $1/(8r\sqrt s)$, and denote
this known operation by $C_{\rm known}$. On the orthogonal
complement the generator is zero, so $C_{\rm known}$ includes
the phase $-I$ there. The source-branch Cayley unitary is
therefore
\begin{equation}
 D_{\rm source}=Q_\beta C_{\rm known}Q_\beta^\dagger.
\end{equation}
Reading from right to left, we unprepare the source coordinate,
rotate the known active pair, and prepare the source coordinate
again. This uses $O(1)$ calls to $U_b,U_b^\dagger$, and
does not query individual entries of $b$.
The other known branches have explicit free Cayley transforms, and the
sparse branch is covered by Lemma~\ref{qls:lem:entry}. Thus each walk
\begin{equation}
 W_{\xi,r}=(I-2JJ^\dagger)(K+\ip I)(K-\ip I)^{-1}
\end{equation}
uses $\mathcal{O}(1)$ matrix and state-preparation queries. Its nonzero phase gap is
\begin{equation}\label{qls:eq:auggap}
 \Omega(g_\xi/\sqrt{s}).
\end{equation}
More explicitly, Lemma~\ref{qls:lem:walk} applies with
$\alpha=32\sqrt{s}$, $m_2=66/256$, $k_0=\sqrt{s}$, and
$\Delta_H=g_\xi\le1$. The padding eigenvalues are one and hence do
not reduce the gap. Every nonzero phase with $|\theta|<\pi$ obeys
\begin{equation}\label{qls:eq:augtangent}
 \left|\tan\frac\theta2\right|
 \ge\frac{4g_\xi}{(161+4g_\xi)\sqrt{s}}
 \ge\frac{4g_\xi}{165\sqrt{s}}.
\end{equation}
It follows that all nonzero phases, including $\theta=\pi$, satisfy
\begin{equation}\label{qls:eq:augsine}
 \left|\sin\frac\theta2\right|\ge
 b_\xi:=\frac{4g_\xi}{\sqrt{165^2s+16g_\xi^2}},
 \qquad
 \operatorname{artanh}b_\xi
 =\operatorname{arsinh}\frac{4g_\xi}{165\sqrt{s}}.
\end{equation}
The tangent bound comes from substituting the displayed
$\alpha,m_2,k_0,\Delta_H$ in \eqref{qls:eq:tangap}.
The sine bound then uses
$|\sin(\theta/2)|=|\tan(\theta/2)|/
\sqrt{1+\tan^2(\theta/2)}$, with the phase $\pi$
handled separately. We retain this explicit $b_\xi$ because
it will determine the filter degree. Crucially, the bounds
hold for every tested $r\ge1/8$; their validity does not
depend on the norm estimate being correct.

The gap tells us the cost of an approximate projection. We now use
the known appended-coordinate state $\ket e$ to prepare the walk
input $J\ket e$, which we call the seed. We next determine how its
probability of projection onto the walk's eigenvalue-one space depends
on the norm guess, since this probability will guide our estimation
algorithm. 

Let $P_{\xi,r}$ project onto the walk's
eigenvalue-one space. This space is one dimensional: if
$v=R_0Ju_{\xi,r}$ and
$a_u=\norm{KJu_{\xi,r}}^2$, then
$J^\dagger v=u_{\xi,r}/2$ and
\begin{equation}
 P_{\xi,r}J\ket e
 =v\frac{\langle v,Je\rangle}{\norm v^2}
 =\frac{2v}{(1+a_u)\sqrt{1+(r_\xi/r)^2}}.
\end{equation}
For this calculation, the lift identity gives
$\langle v,Je\rangle=\langle J^\dagger v,e\rangle
=1/(2\sqrt{1+(r_\xi/r)^2})$, and
$\norm v^2=(1+a_u)/4$ because $u_{\xi,r}$ is normalised.
Substituting these two quantities produces the displayed projected
vector. Squaring its norm gives
$1/[(1+a_u)(1+(r_\xi/r)^2)]$.
The unknown factor $a_u$ measures the norm contributed by the
private representation, but the compressed second moment places
it in $[0,1]$. That factor-of-two uncertainty is small enough
to test whether $r$ balances the kernel, without first estimating
the lift norm itself.

For this seed, write
\begin{equation}
 p_\xi(r)=\norm{P_{\xi,r}J\ket e}^2,
 \qquad q=r_\xi/r.
\end{equation}
Equation~\eqref{qls:eq:lift}, using $m_2\le1$, implies the uniform bounds
\begin{equation}\label{qls:eq:probability}
 \frac1{2(1+q^2)}\le p_\xi(r)\le\frac1{1+q^2}.
\end{equation}
For a guess $r\ll r_\xi$, the solution coordinates dominate
the kernel vector, and the known appended coordinate has little
overlap with it. When $r$ is comparable to $r_\xi$, both parts
have constant weight. The factor-of-two bounds above are enough to
test that balance without knowing the matrix-dependent norm of the
lifted kernel vector.

\subsection{Positive filtering of the walk's fixed space}

The augmented walk has the right fixed space and a known gap
around phase zero. We now turn that spectral information into
a circuit approximating its fixed-space projector. On an
eigenvector with walk eigenvalue $e^{\ii\theta}$, the circuit
should act with amplitude one when $\theta=0$ and with small
amplitude for every other allowed phase.

We use a Laurent polynomial, a finite sum of both positive and
negative powers of the walk. There are two implementation
obligations. We need coherent control over forward, reverse and
inactive steps without multiplying the query count. We also need
the sum of the absolute polynomial coefficients to remain small,
because that sum is the LCU normalisation. We first establish
the controlled step and then construct a filter whose coefficients
are nonnegative and sum exactly to one.

The augmented Cayley transform has the form
\begin{equation}\label{qls:eq:commonconjugator}
 {D_0}=V^\dagger C V.
\end{equation}
On the sparse branch, $V$ applies the locator map and computes the
entry value in a clean work register, and $C$ applies the known
entry-controlled Cayley rotation. This rotation leaves the edge label
and stored value unchanged, so the value can be uncomputed. On the
$T_2$ branch, take $V=Q_\beta^\dagger$, with $Q_\beta$
as defined in the source-branch construction above. Its inverse
preparation maps the active source state to a known coordinate;
$V^\dagger C V=Q_\beta C_{\rm known}Q_\beta^\dagger$
is exactly the required source Cayley unitary. The other branches
require only known transformations. Taking their direct sum gives
\eqref{qls:eq:commonconjugator} on the clean work subspace.
Branch and orientation controls inside $V$ are included in its constant
primitive-query cost. In these coordinates $J$ and the reflection
$R=I-2JJ^\dagger$ are known.

\begin{lemma}[Coherent choice of walk direction]\label{qls:lem:signedwalk}
For the preceding implementation, a coherent choice between an active
$W=R D_0$ step, an active $W^\dagger$ step, and an inactive step uses the
same matrix-oracle and state-preparation calls as an active forward
step. The additional controls require only known gates.
\end{lemma}
\begin{proof}
A forward step is $W=RD_0$, whereas its inverse is
$W^\dagger=D_0^\dagger R$. Thus reversing the Cayley rotation
alone would put the reflection on the wrong side. We place one
controlled reflection before the shared Cayley circuit and one
after it, enabling the appropriate one for each direction.
Let the direction register have orthonormal basis labels
$\ket+$ and $\ket-$, for forward and reverse motion.
Write $R_+$ for $R$ controlled on the $+$ label and $R_-$
for $R$ controlled on the $-$ label. The circuit
\begin{equation}\label{qls:eq:signedstep}
 R_+ V^\dagger
 \bigl(\ket+\bra+\otimes C+\ket-\bra-\otimes C^\dagger\bigr)
 V R_-
\end{equation}
acts as $RV^\dagger CV=W$ in the forward branch and as
$V^\dagger C^\dagger VR=W^\dagger$ in the reverse branch.
Here $V$ and $V^\dagger$ act trivially on the direction register.
On an inactive branch, set both conditional reflections and the middle
rotation to the identity; the remaining $V^\dagger V$ cancels.
The oracle calls in $V,V^\dagger$ are shared by all branches, so neither
control requires an additional query. The identities hold for
entangled controls and payloads as well.
\end{proof}

This implementation uses the internal form
\eqref{qls:eq:commonconjugator}, rather than assuming free direction
control of an opaque walk. The supplied circuit for $U_b$ provides its
required controls as specified in the access model.

We now choose the scalar filter. For an unwanted eigenphase
with $|\sin(\theta/2)|\ge b_{\rm gap}$, its cosine lies in
$[-1,1-2b_{\rm gap}^2]$. The affine transformation
\begin{equation}
 x\longmapsto\frac{x+b_{\rm gap}^2}{1-b_{\rm gap}^2}
\end{equation}
maps precisely this interval to $[-1,1]$, where a Chebyshev
polynomial has magnitude at most one. Phase zero is sent to
$(1+b_{\rm gap}^2)/(1-b_{\rm gap}^2)>1$, where the polynomial
grows exponentially with its degree. Normalising by that value
therefore preserves the fixed space and suppresses every
other phase.

Dolph--Chebyshev LCU filtering with both walk directions appears
in Ref.~\cite[Section V]{Costa}. That construction selects
phases zero and $\pi$ of a qubitisation walk. Our affine shift
selects only zero, and the preceding signed-step lemma supplies
its implementation in the present sparse model.

\begin{proposition}[A positive-coefficient fixed-space filter]\label{qls:prop:chebyshev}
Let $W$ be unitary, and let $P$ project onto its eigenvalue-one
subspace. Suppose $0<{b_{\rm gap}}<1$ and every nonzero eigenphase satisfies
$|\sin(\theta/2)|\ge {b_{\rm gap}}$. For $0<\delta<1$, put
\begin{equation}\label{qls:eq:chebdegree}
 d:=\left\lceil\frac{\operatorname{arcosh}(1/\delta)}
                         {2\operatorname{artanh}{b_{\rm gap}}}\right\rceil.
\end{equation}
With $T_d$ the Chebyshev polynomial of the first kind, define
\begin{equation}\label{qls:eq:chebfilter}
 F_{{b_{\rm gap}},d}(e^{\ip\theta})=
 \frac{T_d\!\left((\cos\theta+{b_{\rm gap}}^2)/(1-{b_{\rm gap}}^2)\right)}
      {T_d\!\left((1+{b_{\rm gap}}^2)/(1-{b_{\rm gap}}^2)\right)}.
\end{equation}
This is a Laurent polynomial on powers $-d,\ldots,d$, with
nonnegative coefficients summing to one. It satisfies
\begin{equation}\label{qls:eq:cheberror}
 F_{{b_{\rm gap}},d}(1)=1,\qquad \norm{F_{{b_{\rm gap}},d}(W)-P}\le\delta.
\end{equation}
Its normalisation-one LCU circuit uses at most $d$ signed walk steps.
\end{proposition}
\begin{proof}
The affine Chebyshev construction already suggests how to
suppress unwanted phases. We first establish the less immediate
implementation property: its Laurent coefficients are nonnegative.
Then evaluating the polynomial at phase zero will show that their
sum is one. This is necessary for the success block to equal the
filter itself, without dividing by a large coefficient norm. For this
calculation only, write $\lambda:={b_{\rm gap}}^2/(1-{b_{\rm gap}}^2)\ge0$, so the affine
argument is $x+\lambda(x+1)$. Taylor expansion gives the finite identity
\begin{equation}\label{qls:eq:chebpositive}
 T_d(x+\lambda(x+1))
 =\sum_{k=0}^d\frac{\lambda^k}{k!}(x+1)^kT_d^{(k)}(x).
\end{equation}
The Chebyshev basis consists of $T_0(x),T_1(x),\ldots$.
We will show that every operation on the right-hand side preserves
nonnegative coefficients in this basis. For differentiation, use
$T_n'=nU_{n-1}$, where $U_j$ is the Chebyshev polynomial of
the second kind. Its expansion is
\begin{equation}
\begin{aligned}
 U_{2j}&=T_0+2(T_2+\cdots+T_{2j}),\\
 U_{2j+1}&=2(T_1+T_3+\cdots+T_{2j+1}).
\end{aligned}
\end{equation}
Multiplication by $x$ also preserves them:
$xT_0=T_1$ and $xT_n=(T_{n+1}+T_{n-1})/2$ for $n\ge1$.
Adding the original polynomial shows that multiplication by
$x+1$ preserves nonnegative coefficients too. Starting from
$T_d$, repeated differentiation and multiplication by $x+1$
therefore make every summand in \eqref{qls:eq:chebpositive}
nonnegative in this basis. The scalar weights $\lambda^k/k!$
are also nonnegative. Finally,
\begin{equation}
 T_j(\cos\theta)=\cos(j\theta)
 =\tfrac12(e^{\ii j\theta}+e^{-\ii j\theta})
 \qquad(j\ge1),
\end{equation}
while $T_0=1$. Each positive-frequency Chebyshev coefficient
splits equally between two Laurent powers, so no negative
coefficient is introduced. Their sum is $T_d(1+2\lambda)$ before normalisation,
as follows by setting $\theta=0$.

For a nonzero phase, $\cos\theta\in[-1,1-2{b_{\rm gap}}^2]$, so the affine
argument in~\eqref{qls:eq:chebfilter} lies in $[-1,1]$ and the
numerator has absolute value at most one. Moreover,
\begin{equation}
 1+2\lambda=\frac{1+{b_{\rm gap}}^2}{1-{b_{\rm gap}}^2}
             =\cosh(2\operatorname{artanh}{b_{\rm gap}}),\qquad
 T_d(1+2\lambda)=\cosh(2d\operatorname{artanh}{b_{\rm gap}}).
\end{equation}
The residual amplitude is thus at most
$\operatorname{sech}(2d\operatorname{artanh}{b_{\rm gap}})\le\delta$.
The value at phase zero is one, and the spectral theorem proves
\eqref{qls:eq:cheberror}.

It remains to implement this polynomial with the advertised
number of calls. Write
$F_{{b_{\rm gap}},d}(W)=\sum_{j=-d}^d p_jW^j$, where
$p_j\ge0$ and $\sum_jp_j=1$, and let
$P_f\ket0=\sum_j\sqrt{p_j}\ket j$. For each
$l=1,\ldots,d$, apply one signed step enabled when $|j|\ge l$,
with direction determined by the sign of $j$. After all $d$
steps, the joint state is
\begin{equation}
 \sum_{j=-d}^d\sqrt{p_j}\ket j\otimes W^j\ket\psi
\end{equation}
for any data input $\ket\psi$. The $j=0$ branch is inactive
throughout. Applying $P_f^\dagger$ and selecting label zero
gives the data vector $\sum_jp_jW^j\ket\psi$.
Equivalently, the circuit has selected block
\begin{equation}
 (\bra0P_f^\dagger\otimes I)
 \left(\sum_{j=-d}^d\ket j\!\bra j\otimes W^j\right)
 (P_f\ket0\otimes I)=F_{{b_{\rm gap}},d}(W).
\end{equation}
All powers are selected within the same $d$-step circuit, so
their costs do not add over the $2d+1$ possible labels. All coefficients are known from ${b_{\rm gap}},d$, and their preparation
uses no matrix queries. Lemma~\ref{qls:lem:signedwalk} gives the cost;
there is no additional LCU normalisation or amplitude amplification.
\end{proof}

For the augmented walk, take ${b_{\rm gap}}=b_\xi$ from
\eqref{qls:eq:augsine}. The resulting degree is
\begin{equation}\label{qls:eq:filtercost}
 d_\xi(\delta)=
 \left\lceil
 \frac{\operatorname{arcosh}(1/\delta)}
 {2\operatorname{arsinh}(4g_\xi/(165\sqrt{s}))}
 \right\rceil
 =\mathcal{O}\!\left(\frac{\sqrt{s}}{g_\xi}
                          \log\frac2\delta\right).
\end{equation}
This degree depends only on known parameters and is valid for every
candidate $r\ge1/8$, including candidates reached on failed estimation
transcripts. The circuit has a bounded number of steps without
postselection-dependent repetition.

\begin{remark}[An averaging alternative]\label{qls:rem:averaging}
The elementary filter $f_{n_{\rm av}}(W)={n_{\rm av}}^{-1}\sum_{j=0}^{{n_{\rm av}}-1}W^j$ also has
nonnegative coefficients summing to one. On a nonzero eigenphase,
the geometric-sum identity gives
\begin{equation}
 \left|\frac{1-e^{\ii n_{\rm av}\theta}}
 {n_{\rm av}(1-e^{\ii\theta})}\right|
 \le\frac1{n_{\rm av}|\sin(\theta/2)|}
 \le\frac1{n_{\rm av}b_\xi}.
\end{equation}
Since the filter equals one on the fixed space, the spectral theorem
gives $\norm{f_{n_{\rm av}}(W)-P}\le1/(n_{\rm av}b_\xi)$. Choosing
${n_{\rm av}}=\lceil100/b_\xi\rceil$ gives error at most $1/100$, and taking
its $k$th power gives error at most $100^{-k}$. A polynomial with nonnegative coefficients summing to
one and degree ${d_{\rm av}}$ can be centred directly as
$W^{-\lfloor {d_{\rm av}}/2\rfloor}f(W)$ and implemented in
$\lceil {d_{\rm av}}/2\rceil$ signed steps. This preserves its success
probability and error relative to $P$, because $W$ is unitary and
$W^{-\lfloor {d_{\rm av}}/2\rfloor}P=P$. The Chebyshev filter gives smaller
constants for the same asymptotic dependence.
\end{remark}

\subsection{Norm estimation and bounded solution preparation}

The walk and filter now supply a usable probability test: prepare
the known appended-coordinate input $J\ket e$, apply the
approximate fixed-space projector and observe whether its work
registers return to zero. This probability tells us whether the
guess $r$ is comparable to $r_\xi$. We use such tests to follow
the regularisation path, and then perform a final, more accurate
projection at $\xi=0$ to prepare the solution.

The estimator can occasionally return a poor guess. We therefore
bound the algorithm's query count on every transcript, including
failed estimates, and include successful outputs from those bad
transcripts in the final mixed-state error bound. This is part of
the theorem's conditional-output guarantee, rather than an assumption
that norm estimation always succeeds.

\begin{proof}[Proof of Theorem~\ref{qls:thm:main}]
We use the fixed-space filter at two accuracies. During norm
estimation it only has to estimate a probability to constant
accuracy; in the final preparation it must approximate the
projected vector to the requested solution accuracy.

\emph{Calibrating the norm guess.}
For any ideal filter above, $F(1)=1$ and
$\norm{F(W)-P}\le\delta$. Because $F(1)=1$, the error $(F(W)-P)\psi$ has no component
in the fixed space. It is therefore orthogonal to $P\psi$,
even when the seed $\psi$ has components in both spaces. This
orthogonality makes the probability error quadratic in the filter
error. For a normalised seed,
\begin{equation}\label{qls:eq:filterprobability}
 \norm{F(W)\psi}^2-\norm{P\psi}^2
 =\norm{(F(W)-P)\psi}^2\in[0,\delta^2].
\end{equation}
Use $\delta=1/100$ for probability tests on $J\ket e$, and
$\delta=\eps/1000$ for the final projection. The latter is one
direct Chebyshev filter; its degree already includes the full
precision dependence in~\eqref{qls:eq:filtercost}.

For the probability thresholds below, write $p_F$ for the actual
filter success probability. We require
$|p_F-p_\xi(r)|\le10^{-4}$. The ideal test filter already has
this bound by Eq.~\eqref{qls:eq:filterprobability}; the finite-gate
allocation given after the solution construction preserves it. This
absolute probability bound is all the norm-estimation tests need.

Independent repetitions and Hoeffding's inequality give an empirical
frequency $\widehat p$ with $|\widehat p-p_F|\le0.0099$, except with
probability $\zeta$, after
\begin{equation}\label{qls:eq:hoeffding}
 n_{\rm rep}=\left\lceil\frac{\log(2/\zeta)}{2(0.0099)^2}\right\rceil
 =\mathcal{O}(\log(2/\zeta))
\end{equation}
repetitions. Equation~\eqref{qls:eq:bernoulli} then gives
$|\widehat p-p_\xi(r)|\le0.01$. Each repetition costs
$\mathcal{O}(\sqrt{s}/g_\xi)$ matrix queries and at most the same
order of source-preparation calls. The filtering accuracy is a fixed
constant, independently of $\zeta$; confidence is increased by classical
repetitions.

We find the guess by following the regularisation path from
$\xi=1$, where the inverse norm is known, to $\xi=0$.
The successive norms differ by at most a factor two, allowing a
constant number of probability tests at each stage.
Specifically, the regularisation~\eqref{qls:eq:regularization} obeys
\begin{equation}\label{qls:eq:normladder}
 r_1=1,\qquad
 r_\xi\le r_{\xi/2}\le2r_\xi,\qquad
 r_0/r_\xi\le\sqrt{1+\kappa^2\xi^2}.
\end{equation}
To verify this controlled change, use
$r_\xi^2=\langle\beta,M_\xi^{-2}\beta\rangle$ and
diagonalise $B^2$. It is a weighted sum of terms
\begin{equation}
 \frac1{\sigma^2+(1-\sigma^2)\xi^2},
 \qquad \sigma^2\in[1/\kappa^2,1],
\end{equation}
where $\sigma^2$ ranges over the eigenvalues of $B^2$ and the
nonnegative weights sum to one. Halving $\xi$
decreases each denominator, but leaves it at least one quarter
of its previous value. Each inverse term therefore increases
by a factor between one and four; taking square roots proves
$r_\xi\le r_{\xi/2}\le2r_\xi$. For the comparison with
$\xi=0$, divide a denominator by $\sigma^2$:
\begin{equation}
 \frac{\sigma^2+(1-\sigma^2)\xi^2}{\sigma^2}
 \le1+\kappa^2\xi^2.
\end{equation}
The same weighted-sum argument proves the last inequality in
\eqref{qls:eq:normladder}. At $\xi=1$, $M_1=X\otimes I$
is unitary, so $r_1=1$.

Let $m=\lceil\log_2\kappa\rceil$, $\xi_j=2^{-j}$ for $0\le j\le m$,
and initialise $\widehat r_0=1$. At stage $j$, suppose the previous
estimate satisfies
\begin{equation}
 r_{\xi_{j-1}}\le\widehat r_{j-1}\le4r_{\xi_{j-1}}.
\end{equation}
The new true norm is at most twice the previous one, while
the stored estimate is between one and four times that previous
norm. This places the new norm within a constant range around
$\widehat r_{j-1}$. We cover that range with five guesses,
spaced by factors of two:
\begin{equation}\label{qls:eq:candidates}
 r=\widehat r_{j-1}2^k,\qquad k=-3,-2,-1,0,1,
\end{equation}
in increasing order. Estimate $p_{\xi_j}(r)$ to error $0.01$ and
accept the first estimate at least $0.22$. Set $\widehat r_j:=2r$.
If no candidate is accepted, use any fallback value in $[1,4\kappa]$.
Always clip the stored estimate to $[1,4\kappa]$. Consequently every
candidate satisfies $r\ge1/8$, even on failed estimation branches.

We now verify the threshold rule using only the two-sided bounds
in \eqref{qls:eq:probability}. We do not need the exact test
probability to be monotone in $r$, since its lift-norm factor
may itself vary with the guess. On the event that the probability
estimates are accurate, the first candidate has
$r=\widehat r_{j-1}/8\le r_{\xi_j}/2$. Its exact probability
is at most $1/(1+2^2)=0.2$, so its estimate is below $0.22$.
The last candidate has $r=2\widehat r_{j-1}\ge r_{\xi_j}$.
Its probability is at least $1/[2(1+1^2)]=0.25$, so its
estimate exceeds $0.22$. Thus a first accepted candidate exists
and has a tested, rejected predecessor.
For the first accepted candidate, $p\ge0.21$, so~\eqref{qls:eq:probability}
gives
\begin{equation}
 r\ge\sqrt{\frac{21}{79}}\,r_{\xi_j}.
\end{equation}
Its predecessor $r/2$ was rejected, so its exact probability is less than
$0.23$. The lower bound in~\eqref{qls:eq:probability} gives
\begin{equation}
 r<2\sqrt{\frac{23}{27}}\,r_{\xi_j}.
\end{equation}
The factor of two in the stored value $\widehat r_j=2r$
turns these bounds into a convenient one-sided estimate:
$2\sqrt{21/79}>1$ and $4\sqrt{23/27}<4$. Hence
\begin{equation}\label{qls:eq:estimatebound}
 r_{\xi_j}\le\widehat r_j\le4r_{\xi_j}.
\end{equation}
The clipping does not alter this guarantee since $1\le r_\xi\le\kappa$.

\emph{Confidence along the path.}
There are $\mathcal O(\log\kappa)$ stages, but their costs grow
geometrically. Giving every stage the same failure tolerance would
unnecessarily increase the cost of the expensive final stages.
Instead, impose higher confidence on the cheap early stages. Give
the five tests at stage $j$ a combined failure budget
\begin{equation}
 \zeta_j=\frac{\eps}{1000}\,2^{j-m}.
\end{equation}
Allocate, for example, $\zeta_j/5$ to each test. The total
stage budgets sum to less than $\eps/500$ because
$\sum_{j=1}^m2^{j-m}<2$. Although the next guesses depend
on earlier outcomes, Hoeffding's guarantee holds conditional on
any previous transcript: it concerns fresh independent repetitions
of the currently selected circuit. We may therefore apply the
conditional guarantee at each test and then a union bound over
the stages. No independence between the adaptive stages is required. When $\kappa=1$, there are no estimation stages.
We defer the query sum until the complete solution preparation
has been specified.

\emph{Extracting the solution.}
A successful norm estimate ensures constant overlap both from
the seed into the lifted kernel and from that kernel into the
solution coordinates.
We return to the original unregularised system: set $\xi:=0$
and $r=\widehat r_m$ in the walk. On the event that
the norm estimation succeeds, $\xi_m\le1/\kappa$ and
\eqref{qls:eq:normladder}-\eqref{qls:eq:estimatebound} give
\begin{equation}\label{qls:eq:qrange}
 \tfrac14\le q=r_0/\widehat r_m\le\sqrt2.
\end{equation}
At $\xi=0$, $M_0=Z\otimes B$ and
$M_0^{-1}\beta=\ket0\ket1\,A^{-1}b$, so the required solution is
obtained by discarding known labels after selecting these coordinates.

We now compose the operations that return the data state.
Prepare $J\ket e$ and project onto the walk's phase-zero space.
To apply $J^\dagger$, reverse the known preparation extending
$J$ and select its work registers in their initial state. This
is a heralded contraction from the private space back to the
augmented system space, rather than a unitary inverse on all
private vectors. Finally, select the solution coordinates of
$u_{0,r}$ and discard their known register labels. Let
$a_u=\norm{KJu_{0,r}}^2\in[0,1]$. The ideal combined success probability is
\begin{equation}\label{qls:eq:finalprob}
 p_{\rm sol}=\frac{q^2}{(1+q^2)^2(1+a_u)^2}
 \ge\frac4{289}.
\end{equation}
The state calculation explains both the probability and the data
returned. With $u=u_{0,r}$ and $v=R_0Ju$, the phase-zero projection
and inverse embedding give
\begin{equation}
 J\ket e\ \xrightarrow{P_{0,r}}\
 \frac{2v}{(1+a_u)\sqrt{1+q^2}}
 \xrightarrow{J^\dagger}\
 \frac{u}{(1+a_u)\sqrt{1+q^2}}
 \xrightarrow{\text{solution coordinates}}\
 \frac{q}{(1+a_u)(1+q^2)}\ket x.
\end{equation}
The arrows denote heralded contractions, not unitary state equalities.
Squaring the last coefficient gives \eqref{qls:eq:finalprob};
normalising that branch leaves exactly $\ket x$. For the lower
bound, $1+a_u\le2$, and the function $q^2/(1+q^2)^2$
on $q\in[1/4,\sqrt2]$ is smallest at $q=1/4$, where it
equals $16/289$. This gives $p_{\rm sol}\ge4/289$.
Both projections therefore have a combined constant success
probability once the norm estimate is balanced.

Replace the phase-zero projection by the Chebyshev filter with
target $\delta=\eps/1000$. All subsequent extraction maps are
contractions, so an unnormalised filter error at most $\delta$
gives at most the same output-vector error. The finite-gate
construction below, using the slightly larger degree in
Eq.~\eqref{qls:eq:finitedegree}, includes extraction in this budget.
The ideal unnormalised solution vector has norm at least $2/17$,
so the normalised vector error is at most $2\delta/(2/17-\delta)\le\eps/50$ for
$\eps\le1/10$. This also upper-bounds the trace distance of the
corresponding pure states. Its success probability is at least $1/289$.
Use a fixed number, for example $1000$, of attempts and herald failure if
none succeeds. On good estimation transcripts, their success probability
is at least $1-e^{-1000/289}>0.96$. This is a bounded-query construction;
there is no unbounded repeat-until-success cost on bad norm estimates.

Bad estimation transcripts may also end with a success flag,
so they must remain in the conditional-output analysis. Their
total probability is below $\eps/500$. Since good transcripts
succeed with probability greater than $0.96$, overall success
is at least $0.96(1-\eps/500)>2/3$. After conditioning on that
success, bad transcripts have weight at most $3\eps/1000$.
To see how that weight enters the accuracy guarantee, decompose
the conditional output into the mixture
$(1-w)\rho_{\rm good}+w\rho_{\rm bad}$, where
$w\le3\eps/1000$. Every good transcript has trace-distance
error at most $\eps/50$, whereas any two density operators
have trace distance at most one. Convexity therefore adds at
most $w$ to the error. The classical transcript has already
removed coherence between these cases, so no square root of
the failure probability appears.
Combining this with the error on good transcripts gives
$\eps/50+3\eps/1000<\eps$.

\emph{Query accounting.}
Since $g_{\xi_j}\ge2^{-j}$, the total norm
estimation cost is
\begin{equation}\label{qls:eq:geometriccost}
\begin{aligned}
 \mathcal{O}\left(\sqrt{s}\sum_{j=1}^m
    2^j\left[\log\frac1\eps+m-j\right]\right)
 &=\mathcal{O}\left(\sqrt{s}\,2^m\log\frac1\eps\right)\\
 &=\mathcal{O}\left(\sqrt{s}\,\kappa\log\frac1\eps\right).
\end{aligned}
\end{equation}
To see why the stage count disappears, set $l:=m-j$ in the sum.
After factoring out $2^m$, the remaining series are
$\sum_{l\ge0}2^{-l}$ and $\sum_{l\ge0}l2^{-l}$, both bounded
by absolute constants. We can therefore demand more confidence at
the cheap early stages without charging the expensive final stage
for all $m$ tests. This removes a possible extra $\log\kappa$.

The final walk gap is $\Omega(1/(\kappa\sqrt{s}))$, so each attempt costs
$\mathcal{O}(\kappa\sqrt{s}\log(1/\eps))$ queries. Combining with
\eqref{qls:eq:geometriccost} gives the claimed query bound. We finish
by making the finite-gate accuracy assumptions used above explicit.

\emph{Finite-gate accuracy.}
The probability thresholds and extraction estimates above also
hold for finite gate circuits. We allocate synthesis error
separately to the test probabilities and to the final
unnormalised vector. These choices change no matrix-query
count.

An ideal filter chosen exactly at its target tolerance might
leave no room for synthesis error. We create strict slack by
adding one degree. For either target $\delta$, use
\begin{equation}\label{qls:eq:finitedegree}
 d_\xi^{\rm fin}(\delta)=d_\xi(\delta)+1,\qquad
 \delta_0=\operatorname{sech}\!\left(
 2d_\xi^{\rm fin}(\delta)\operatorname{artanh}b_\xi\right)
 <\delta.
\end{equation}
First regard this as an ideal circuit. Let $\epsilon_{\rm syn}$ bound the change in its
unnormalised success vector due to all known-gate approximations in a
test or final attempt. In a probability test, choose
$\epsilon_{\rm syn}\le(\delta^2-\delta_0^2)/2$. Both success vectors have norm at most
one, so their squared norms differ by at most $2\epsilon_{\rm syn}$. The implemented
success probability $p_F$ therefore obeys the two-sided bound
\begin{equation}\label{qls:eq:bernoulli}
 |p_F-p_\xi(r)|\le\delta_0^2+2\epsilon_{\rm syn}\le10^{-4}.
\end{equation}
Unlike the ideal filter, a gate approximation need not preserve the
kernel exactly; the absolute bound is what the threshold argument uses.
For the final attempt choose $\epsilon_{\rm syn}\le\delta-\delta_0$, including seed
preparation and output extraction in this error budget. The total
unnormalised-vector error relative to exact projection is then at most
$\delta$.

These error allocations require no extra matrix queries. For example,
prepare the coefficient state within Euclidean distance $\eta_c$ and
unprepare with its actual circuit adjoint. The coefficient probability
distribution changes by at most $2\eta_c$ in total absolute difference.
If each implemented signed step differs from the ideal one by at most
$\rho_c$ in operator norm, telescoping bounds its contribution by
$d_\xi^{\rm fin}(\delta)\rho_c$. Allocate these terms and any seed
and extraction errors within $\epsilon_{\rm syn}$. Reverse rotations use the actual
adjoints of the forward circuits. Implement each reflection by a
preparation circuit, an exact computational-basis reflection, and the
actual adjoint preparation, so it remains a Hermitian involution even
when the prepared subspace is approximate. Thus the two directions
remain inverses of one implemented unitary. Gate and coefficient
preparation costs are separate from the query count.

\end{proof}

\begin{remark}[Explicit filtering constants]\label{qls:rem:constants}
The improved filter substantially reduces the prescribed walk budget.
The ratio $g_\xi/\sqrt{s}$ lies in $(0,1]$, and the numerator in the test degree is the constant $\operatorname{arcosh}100$.
Using $\operatorname{arsinh}u\ge u/\sqrt{1+u^2}$ gives
\begin{equation}\label{qls:eq:testconstant}
 d_\xi^{\rm fin}(1/100)
 \le\frac{\frac{\operatorname{arcosh}100}{2}\sqrt{(165/4)^2+1}+2}{g_\xi/\sqrt{s}}
 <\frac{112\sqrt{s}}{g_\xi}.
\end{equation}
This includes integer rounding and the extra degree for finite gates.
For comparison, an averaging filter with the conservative angle bound
$\omega=g_\xi/(512\sqrt{s})$ and
${n_{\rm av}}=\lceil100\pi/\omega\rceil$ uses
$q_{\rm av}=\lceil51200\pi\sqrt{s}/g_\xi\rceil-1$ forward steps. Since
$q_{\rm av}\ge(51200\pi-1)\sqrt{s}/g_\xi$, the Chebyshev schedule reduces this
budget by a factor greater than $1400$.
Here $q_{\rm av}=n_{\rm av}-1$ counts the uncentred forward schedule;
centring as in Remark~\ref{qls:rem:averaging} instead gives
$\lceil(n_{\rm av}-1)/2\rceil$ signed steps. The stated factor compares
the prescribed forward averaging budget with the Chebyshev budget.

The comparison also holds for the final filter. For
$k=\lceil\log(1/\delta)/\log100\rceil$, where $\delta\le1/100$,
$\cosh(k\operatorname{arcosh}100)\ge100^k$ implies
$\operatorname{arcosh}(1/\delta)\le k\operatorname{arcosh}100$.
Including ceilings and the extra degree gives
$d_\xi^{\rm fin}(\delta)\le k d_\xi^{\rm fin}(1/100)$,
whereas the repeated averaging schedule costs $kq_{\rm av}$.
With the same walk implementation, Lemma~\ref{qls:lem:signedwalk}
transfers these comparisons term by term to the prescribed filtering
matrix-query budgets, holding the maximum numbers of tests and attempts
fixed. They do not compare optimised gate costs or individual adaptive
execution times. The remaining repetition count of about
$5102\log(10/\zeta_j)$ per stage test and the allowance of $1000$
final attempts are still chosen for proof convenience; these are not
practical resource estimates.
\end{remark}

\subsection{Separate query bounds with a solution-norm estimate}
\label{qls:sec:separate}

We now seek to reduce the number of calls to the preparation of $b$
while retaining the sparse-matrix query bound. In the preceding
algorithm, the augmented matrix contains $b$, so every step of its
walk can require preparation queries. When a constant-factor estimate
of the solution norm is supplied, we can instead build the walk from
$A$ alone and use it to define a new linear system. This lets us apply
the algorithm of Gao, Ji and Liu~\cite{GJL26}, which separates the
preparation of a state from its subsequent precision refinement:
the latter uses only the block encoding of the matrix. We will show
that our change of linear system increases both the inverse-norm
bound and the solution norm by a factor of order $\sqrt{s}$. Their
ratio, which controls the number of preparation queries, therefore
retains its original order.

\begin{corollary}[Separate matrix and preparation query bounds]
\label{qls:cor:separate}
Under the matrix, oracle and preparation assumptions of
Theorem~\ref{qls:thm:main}, suppose in addition that a classical number
$\widehat\nu$ is supplied with
\begin{equation}\label{qls:eq:separate-estimate}
 \tfrac12\norm{A^{-1}b}\le\widehat\nu\le2\norm{A^{-1}b}.
\end{equation}
For $0<\eps\le1/10$, an algorithm prepares the normalised solution
with a classical success flag of probability at least $2/3$ and
conditional trace-distance error at most $\eps$. Its worst-case
query counts are
\begin{equation}\label{qls:eq:separate-cost}
 Q_A=\mathcal O\!\left(\sqrt{s}\,\kappa\log\frac1\eps\right),
 \qquad
 Q_b=\mathcal O\!\left(\frac{\kappa}{\norm{A^{-1}b}}\right),
\end{equation}
where $Q_A$ counts the primitive sparse-matrix queries and $Q_b$
counts calls to $U_b$, its inverse and their controlled versions.
The bounds hold on every execution history and exclude the cost of
obtaining the supplied estimate.
\end{corollary}

\begin{proof}
We first treat Hermitian $A$. Write $z:=A^{-1}b$ and
$\nu:=\norm z$, so that $z$ is the unnormalised solution and
$\nu\ge1$. Apply Lemma~\ref{qls:lem:entry} to $A$ itself.
In its list coordinates, $J$ is a known isometry from the data
space to the larger space on which $K$ acts, and
\begin{equation}\label{qls:eq:separate-moments}
 J^\dagger KJ=A/\sqrt{s},\qquad
 J^\dagger K^2J\preceq I,\qquad \norm K\le\sqrt{s}.
\end{equation}
We use the same Cayley unitary and reflected walk as in
Lemma~\ref{qls:lem:walk}:
\begin{equation}\label{qls:eq:separate-walk}
 D_0=(K+\ii I)(K-\ii I)^{-1},\qquad
 R_0=(I-D_0)^{-1}=\frac{I+\ii K}{2},\qquad
 W=(I-2JJ^\dagger)D_0.
\end{equation}
Here $D_0$ and $W$ depend only on $A$. Each uses a constant
number of sparse-matrix queries, and neither uses $U_b$.

To identify the linear system we should solve, recall the walk
identity from the proof of Lemma~\ref{qls:lem:walk}:
\begin{equation}\label{qls:eq:separate-identity}
 (I-W)R_0J=\frac{\ii}{\sqrt{s}}JA.
\end{equation}
Thus inverting $I-W$ on an embedded right-hand side produces
$R_0J$ applied to the original solution. We accordingly define
the enlarged matrix and its right-hand side by
\begin{equation}\label{qls:eq:separate-system}
 B:=\frac{I-W}{2},\qquad b_B:=Jb.
\end{equation}
The factor $1/2$ makes $B$ a contraction with an exact block
encoding of normalisation one. To implement that encoding, we
prepare one selection qubit in $\ket+$, apply $I$ on label zero
and $-W$ on label one, and undo the preparation. For every vector
$v$ in the enlarged space, the resulting unitary acts as
\begin{equation}\label{qls:eq:separate-block}
 \ket0v\longmapsto
 \ket0\frac{I-W}{2}v+\ket1\frac{I+W}{2}v.
\end{equation}
Selecting the zero outcome therefore applies $B$. A call to this
block encoding, its inverse or a controlled version costs
$O(1)$ sparse queries. Preparing $b_B$ requires one call to
$U_b$ followed by the known embedding $J$; the same statement
holds for the inverse preparation. No matrix queries enter this
preparation because we work in the list coordinates of
Lemma~\ref{qls:lem:entry}.

We next verify that $B$ is invertible and determine the inverse-norm
bound needed by the block-input solver. Since $A$ is invertible,
the kernel correspondence in Lemma~\ref{qls:lem:walk} excludes
eigenvalue one from the whole space of $W$. Substituting
$\Delta_H=1/\kappa$, $\alpha=\sqrt{s}$, $m_2=1$ and
$k_0=\sqrt{s}$ into its phase-gap bound gives
\begin{equation}\label{qls:eq:separate-gap}
 \left|\tan\frac\theta2\right|
 \ge\frac1{\sqrt{s}(2\kappa+1)}
 \qquad(|\theta|<\pi).
\end{equation}
Because $W$ is unitary, $B$ is normal and its singular values
are $|1-e^{\ii\theta}|/2=|\sin(\theta/2)|$. The phase at
$\pi$, when present, gives singular value one. Using
$|\sin(\theta/2)|=|\tan(\theta/2)|/
\sqrt{1+\tan^2(\theta/2)}$ for the other phases, we obtain
\begin{equation}\label{qls:eq:separate-inverse}
 \norm{B^{-1}}
 \le\sqrt{1+s(2\kappa+1)^2}
 \le4\sqrt{s}\,\kappa.
\end{equation}
We therefore set $\kappa_B:=4\sqrt{s}\,\kappa$, a known valid
inverse-norm bound, even though
the actual smallest singular value of $B$ is not supplied.

We also need the norm of the enlarged solution, because the
preparation-query bound depends on it. Applying
Eq.~\eqref{qls:eq:separate-identity} to $z$ gives the exact solution
\begin{equation}\label{qls:eq:separate-solution}
 y:=B^{-1}b_B=-2\ii\sqrt{s}\,R_0Jz
             =\sqrt{s}(K-\ii I)Jz.
\end{equation}
The Hermiticity of $K$ makes
$(K+\ii I)(K-\ii I)=K^2+I$. Since $J$ is an isometry,
the compressed second-moment bound then gives
\begin{equation}\label{qls:eq:separate-norm}
 \norm y^2
 =s\bigl(\norm z^2+\norm{KJz}^2\bigr),
 \qquad
 \sqrt{s}\,\nu\le\norm y\le\sqrt{2s}\,\nu.
\end{equation}
This is the cancellation we need: both $\kappa_B$ and
$\norm y$ contain a factor $\sqrt{s}$, so their ratio is
$O(\kappa/\nu)$.

The supplied estimate of $\nu$ also suffices for this new system;
we do not estimate the unknown quantity $\norm{KJz}$ separately.
Theorem~5.7 of Ref.~\cite{GJL26} allows an estimate whose ratio
to the solution norm lies in $[3/8,5/2]$. To meet this promise,
we seek a known rescaling $\widehat\nu_B=c\sqrt{s}\,\widehat\nu$.
Equations~\eqref{qls:eq:separate-estimate}
and~\eqref{qls:eq:separate-norm} bound its ratio to $\norm y$
between $c/(2\sqrt2)$ and $2c$. Thus it suffices to choose
$3\sqrt2/4\le c\le5/4$. Taking $c=9/8$ gives
\begin{equation}\label{qls:eq:separate-transferred-estimate}
 \widehat\nu_B:=\frac98\sqrt{s}\,\widehat\nu,\qquad
 \frac9{16\sqrt2}
 \le\frac{\widehat\nu_B}{\norm y}\le\frac94,
\end{equation}
which is inside the required interval. The same theorem applies
to general invertible block-encoded matrices, so it does not
require $B$ to be Hermitian. We apply it with normalisation one,
inverse-norm bound $\kappa_B$, estimate $\widehat\nu_B$
and Euclidean state error $\delta:=\eps/8$. It prepares
$y/\norm y$ to this accuracy with a success flag of probability
at least $2/3$, using
\begin{equation}\label{qls:eq:separate-solver-cost}
 \begin{aligned}
 O\!\left(\kappa_B\log\frac1\delta\right)
 &=O\!\left(\sqrt{s}\,\kappa\log\frac1\eps\right)
 &&\quad\hbox{matrix queries},\\
 O\!\left(\frac{\kappa_B}{\norm y}\right)
 &=O(\kappa/\nu)
 &&\quad\hbox{preparation queries}.
 \end{aligned}
\end{equation}
Each query on the left is implemented by the corresponding
constant-cost operation described above, so these are also the
claimed orders of queries to the original inputs.

It remains to recover the original solution from the enlarged
state. We can remove the factor $R_0$ in
Eq.~\eqref{qls:eq:separate-solution} by applying
\begin{equation}\label{qls:eq:separate-decoder}
 C:=\frac{I-D_0}{2},\qquad Cy=-\ii\sqrt{s}\,Jz.
\end{equation}
The same one-qubit selection circuit as in
Eq.~\eqref{qls:eq:separate-block}, with $D_0$ replacing $W$,
is an exact block encoding of $C$. After its successful outcome,
we undo the known preparation of $J$ and select its zero work
registers. Together these selections apply $J^\dagger C$.
On the exact normalised input $y/\norm y$, the selected vector
is $-\ii\sqrt{s}\,z/\norm y$, whose squared norm obeys
\begin{equation}\label{qls:eq:separate-decode-probability}
 p_{\rm dec}=\frac{s\nu^2}{\norm y^2}\ge\frac12.
\end{equation}
The selected state is therefore the desired solution, up to a
global phase, and its probability is bounded below independently
of $s$, $\kappa$ and $\nu$.

For completeness, this constant probability also controls the
error introduced by normalising after selection. The map
$J^\dagger C$ is a contraction, so an input error at most
$\delta$ changes the selected, unnormalised vector by at most
$\delta$. Its ideal norm is at least $1/\sqrt2$; hence the
normalised state error is at most $2\sqrt2\,\delta<\eps$,
and its selection probability is at least
$(1/\sqrt2-\delta)^2>1/3$. A full trial, including the solver's
success flag, therefore succeeds with probability greater than
$2/9$. Five independent trials, returning the first successful
output, give success probability at least $1-(7/9)^5>2/3$.
Euclidean error bounds the trace distance of the corresponding
pure states, and mixtures of successful outputs obey the same
trace-distance bound by convexity. Decoding adds only $O(1)$
sparse queries per trial and no preparation queries. The fixed
number of trials preserves both worst-case query bounds.

For non-Hermitian $A$, we first apply the argument to
\begin{equation}\label{qls:eq:separate-dilation}
 H:=\begin{pmatrix}0&A\\A^\dagger&0\end{pmatrix},\qquad
 b_H:=\binom b0,\qquad H^{-1}b_H=\binom0{A^{-1}b}.
\end{equation}
The dilation has at most $s$ nonzero entries per row and column,
with $\norm H=\norm A$ and $\norm{H^{-1}}=\norm{A^{-1}}$.
Its row and column locators use the supplied locators of $A$,
interchanging their roles in the adjoint block, and its value
oracle uses complex conjugation there. Each costs a constant
number of original queries. Preparing $b_H$ uses one $U_b$
call, and its solution norm remains $\nu$. We include selection
of the lower data block in the final decoding measurement. The
ideal decoded vector already lies in that block, so this extra
selection leaves its ideal success probability unchanged; being
a contraction, it also preserves the preceding error argument.
This establishes the corollary for general $A$.
\end{proof}

This corollary uses the same sparse-entry construction and phase-gap
lemma as Theorem~\ref{qls:thm:main}, but replaces the augmented-kernel
preparation procedure by an external block-input solver. The supplied
norm estimate is an additional hypothesis. The norm-estimation
procedure above is included in the original theorem's matrix-query
bound, but its preparation-query analysis does not establish
$O(\kappa/\norm{A^{-1}b})$ calls. We therefore do not include that
procedure in Corollary~\ref{qls:cor:separate} or claim the separate
bound without its estimate assumption.

\subsection{Relation to block-input algorithms and scope}
Optimal block-input linear-system algorithms achieve
$\mathcal{O}(\kappa\log(1/\eps))$ calls to a suitably normalised block encoding
\cite{Costa,Dalzell}. Low and Su~\cite{LowSu} obtain optimal solution-overlap
dependence in the number of right-hand-side preparation
queries while retaining nearly optimal matrix-query cost.
Gao, Ji and Liu~\cite{GJL26} simultaneously attain optimal matrix and
right-hand-side query counts given a constant-factor estimate of
$\alpha\norm{A^{-1}b}$, where $\alpha$ is the block-encoding normalisation
and $b$ is normalised. Theorem~\ref{qls:thm:main} gives the
sparse-matrix query bound without a supplied norm estimate.
Corollary~\ref{qls:cor:separate} combines our walk construction
with the block-input result of Gao, Ji and Liu to obtain
$O(\kappa/\norm{A^{-1}b})$ preparation calls at the same
matrix-query order, given a constant-factor solution-norm estimate.

The norm-estimation and augmented-kernel viewpoint
is particularly close in spirit to~\cite{Dalzell}. Our sparse-entry
construction supplies an exact walk with a proved phase gap. Both
linear-system algorithms use that walk directly, and neither replaces
each block call by a separately precision-amplified sparse encoding.
In Corollary~\ref{qls:cor:separate}, the exact block encoding instead
represents the enlarged matrix $(I-W)/2$.
The Chebyshev filter and signed-power construction in the proof of
Theorem~\ref{qls:thm:main} adapt earlier LCU filtering~\cite{Costa}
to this walk. The reduction and gap proof do not
use the enlarged-system block-encoding result of Bravo-Prieto, Harrow,
and Kothari~\cite{BPHK26}, whose concurrent algorithm achieves the same
asymptotic matrix-query complexity by that alternative route.

Recent lower-bound work establishes square-root sparsity dependence
\cite{Mori}; Patel~\cite{Patel} proves the joint bound
$\Omega(\sqrt{s}\,\kappa\log(1/\eps))$ for
$\kappa\ge31$, $s\ge5$, and $0<\eps\le1/64$ (its Theorem~1).
Its Eq.~(3) specifies an out-of-place XOR position oracle. Lemma~5 and its
proof, particularly Eq.~(57), give input-independent row and column
supports, predetermined zero padding, and a two-Boolean-query value oracle
including uncomputation. These properties imply that its lists can also
be padded injectively and extended to known permutations: their in-place
locators and inverses cost zero queries to the hidden Boolean input.
Thus its reduction extends to the stronger location model here; this last
extension is an inference from that construction, not a stated theorem
of~\cite{Patel}. The independent lower-bound proof in Section~\ref{part:qlsl} works directly
in our access model and with our heralded conditional trace-distance output
convention.

The state-preparation circuit and inverse-access assumptions are real
inputs; controlled matrix calls follow from the constructions above. No entry oracle for the generally dense vector $b$ is used.
The output is an amplitude-encoded quantum state, not a classical listing
of $A^{-1}b$. The separate preparation-query bound is the conditional result of
Corollary~\ref{qls:cor:separate}; Theorem~\ref{qls:thm:main} does not
assert it without the norm estimate. Neither result claims a matching
optimised gate complexity. Both results guarantee conditional trace-distance accuracy with
constant heralded success; neither asserts a deterministic pure-state isometry.

\section{A matching linear-system query lower bound}
\label{part:qlsl}

We now show that the preceding matrix-query bound is optimal in
its joint dependence on sparsity, conditioning and precision. The
construction encodes a computation in the inverse of a sparse
matrix. Applying that inverse to a known right-hand side produces
a superposition of successive computation states, each labelled
by its time. The first part of the computation takes the parity
of hidden search positions; the remaining times retain the answer.
An accurate linear-system solution therefore lets us read that
parity with a small but provable advantage.

The three parameters enter the same construction. We use dense
unitary transition blocks whose individual entries reveal only one
Boolean membership bit; their size gives the sparsity dependence.
We attenuate each successive time component by a fixed factor
slightly below one, which controls the inverse norm. Finally, the
requested accuracy determines how long the history can be while
its answer region remains detectable. The resulting history has
length $\Theta(\kappa\log(1/\eps))$, and a Boolean query
argument charges a square root of the transition-block size for
that entire computation.

We first prove the small-advantage parity estimate needed by the
reduction. We then construct the matrix and its exact inverse,
verify the sparse-oracle promises, and choose the attenuation and
history length. The final step converts a heralded, possibly mixed
solution state into an unconditional Boolean-query algorithm, so
that the lower bound applies to the output convention of
Theorem~\ref{qls:thm:main}.

\subsection{Joint linear-system lower bound and XOR estimate}
\label{qlsl:sec:main}

\begin{theorem}[Joint linear-system lower bound]\label{qlsl:thm:main}
Let $\kappa\ge32$, let $s\ge5$ be an integer, and let
$0<\eps\le1/64$. Consider algorithms which, for every promised input,
use at most $Q$ matrix-oracle queries and produce a success flag of
probability at least $2/3$. Conditional on success, the output state
$\rho_A$ must obey
\begin{equation}\label{qlsl:eq:target}
 \trdist(\rho_A,\ket{x_A}\bra{x_A})
 :=\tfrac12\norm{\rho_A-\ket{x_A}\bra{x_A}}_1\le\eps,
 \qquad
 \ket{x_A}=\frac{A^{-1}\ket b}{\norm{A^{-1}b}}.
\end{equation}
The promised matrices satisfy
\begin{equation}
 \norm A\le1,\qquad \norm{A^{-1}}\le\kappa,
 \qquad
 \max\{\text{row sparsity},\text{column sparsity}\}\le s.
\end{equation}
There is a family of real symmetric, indefinite matrices with dimension
$N=\Theta(s\kappa\ln(1/\eps))$, and input-independent right-hand sides,
for which
\begin{equation}\label{qlsl:eq:main}
 Q=\Omega\!\left(\sqrt{s}\,\kappa\ln\frac1\eps\right).
\end{equation}
The proof gives, conservatively,
$Q>\kappa\sqrt{s}\ln(1/\eps)/3072$.
The lower bound holds even when in-place row and column location oracles,
their inverses, and their controlled versions are supplied. Indeed,
these location operations are input-independent in the constructed family.
Only entry-value queries convey the hidden information.
\end{theorem}

This is a worst-case bound in the stated parameter range,
with dimension allowed to grow. Padding with a known identity
block preserves the norm and inverse-norm promises and permits
any larger dimension when the right-hand side is extended by
zero. The matrices are indefinite, so the theorem does not
establish the same result for positive-definite inputs alone.
It also makes no claim at the degenerate endpoints $s=1$ or
$\kappa=1$. Its bounded query count and conditional-output
convention agree with Theorem~\ref{qls:thm:main}.

Patel~\cite{Patel} and, independently, Bravo-Prieto, Harrow, and
Kothari~\cite{BPHK26} also establish the joint lower bound. We use Patel's position-parity idea and a sign-adjusted version
of his dense unitary gadget. Our clock has a first and a last
time label, with no transition back from the last label to the
first; this is the open-clock construction used below. We also
prove the required XOR estimate directly. Its lineage is the weight-growth
argument of Ambainis for the strong direct product theorem for search
\cite{Ambainis}, generalised by \v{S}palek to the multiplicative quantum
adversary \cite{Spalek}. The XOR lemma of Lee and Roland
\cite{LeeRoland} is also based on that method. Here the explicit weight
$(2I-P)^{\otimes n}$ and the parity observable make the needed
specialisation short enough to prove in full.

Let $m$ be even. A string $x\in\{0,1\}^m$ is promised to have exactly one
nonzero bit, at position $j\in\mathbb Z_m$. For $n$ such strings,
$X=(x^{(1)},\ldots,x^{(n)})$, write their marked positions as
$\boldsymbol j=(j_1,\ldots,j_n)$ and define
\begin{equation}
 F(\boldsymbol j)=\bigoplus_{a=1}^n(j_a\bmod2).
\end{equation}
Here $F$ is zero or one according to the parity of the number
of odd marked positions. A Boolean query specifies an instance
$a\in\{1,\ldots,n\}$ and a position $\ell\in\mathbb Z_m$,
and writes the membership bit into a response qubit:
\begin{equation}
 \ket{a,\ell,z}\longmapsto\ket{a,\ell,z\oplus x^{(a)}_\ell}.
\end{equation}
The solver reduction will supply an output correlated with $F$,
but that correlation may shrink with $\eps$. The next lemma
therefore keeps its dependence on the success advantage $\beta$
explicit. In the lemma and its proof, $q$ denotes the number
of queries.

\begin{lemma}\label{qlsl:lem:xor}
Suppose a quantum algorithm uses $q$ Boolean queries and computes $F$
with success probability at least $(1+\beta)/2$ on every promised input,
where $0<\beta\le1$. Then
\begin{equation}\label{qlsl:eq:xor}
 \beta\le 2^{-n}e^{4{q}/\sqrt m},
 \qquad
 {q}\ge\frac{\sqrt m}{4}(n\ln2+\ln\beta).
\end{equation}
\end{lemma}
\begin{proof}
We compare how quickly a weighted norm can grow with how large
it must become to support correlation with all $n$ hidden parities.
The weight is small on the initial uniform hidden-position state
and larger on directions containing input information. Query-free
operations preserve it. A single membership query can increase it
only by a factor close to one, while correlation with the full
parity requires growth that is exponential in $n$. For the analysis, introduce $n$ registers
containing uniformly superposed hidden positions and purify
the algorithm's randomness and measurements. These registers
are not supplied to the algorithm.
Its joint state after $q$ queries is
\begin{equation}
 \ket{\Psi_{q}}
 =m^{-n/2}\sum_{\boldsymbol j}
   \ket{\boldsymbol j}\ket{\psi_{\boldsymbol j,{q}}}.
\end{equation}
On one hidden-position register define
\begin{equation}
 \ket u=m^{-1/2}\sum_{j=0}^{m-1}\ket j,\quad
 P:=\ket u\bra u,\quad w:=2I-P,\quad
 Z=\sum_{j=0}^{m-1}(-1)^j\ket j\bra j.
\end{equation}
The operator $w=2I-P$ assigns weight one to the uniform vector
$u$ and weight two to its orthogonal complement. Let
$W=w^{\otimes n}$ on the hidden-input registers and extend it
by the identity on the algorithm's working space. In this proof
$W$ is a positive analysis weight, not the unitary walk of
Section~\ref{part:qls}. Initially the hidden registers are in
$u^{\otimes n}$ and the working state is independent of them,
so $\norm{W\Psi_0}=1$. Query-free operations act only on the
working space and commute with $W$, preserving this norm.

We next bound one query's effect. In the $X$ eigenbasis of
the response qubit, the $+$ component is unchanged. The $-$
component acquires a minus sign precisely when the queried
position is marked. For a fixed address $\ell\in\mathbb Z_m$
in one selected instance, the query therefore acts on its hidden
position register as
\begin{equation}
 D_{\ell}=I-2\ket\ell\bra\ell.
\end{equation}
The address and instance are controlled by the algorithm registers;
the hidden register is introduced only for the analysis.
The reflection $D_{\ell}$ couples $u$ to the direction of
$\ket\ell$ orthogonal to $u$. Their overlap is
$\langle u|\ell\rangle=1/\sqrt m$. On the orthonormal
two-dimensional basis
$\{u,(\ket\ell-m^{-1/2}u)/\sqrt{1-1/m}\}$,
\begin{equation}
 [P,D_{\ell}]=\begin{pmatrix}0&-2\sqrt{m-1}/m\\
                        2\sqrt{m-1}/m&0\end{pmatrix},
 \qquad \norm{[P,D_{\ell}]}=\frac{2\sqrt{m-1}}m\le\frac2{\sqrt m}.
\end{equation}
It vanishes on the orthogonal complement. This overlap is the source
of the $1/\sqrt m$ limit on per-query progress.
Using $\norm{w^{-1}}=1$ and
$wD_{\ell}w^{-1}=D_{\ell}+[w,D_{\ell}]w^{-1}$, we obtain
\begin{equation}
 \norm{wD_{\ell}w^{-1}}\le1+\frac2{\sqrt m}.
\end{equation}
For a query to one instance, the weight factors on all other
hidden registers cancel under conjugation, leaving the bound
above for that instance. The full query is a direct sum over
its working-register controls, so its conjugated norm is the
largest norm of these blocks, at most $1+2/\sqrt m$.
If $O$ denotes that query, then
\begin{equation}
 \norm{WO\Psi}\le\norm{WOW^{-1}}\norm{W\Psi}
 \le(1+2/\sqrt m)\norm{W\Psi}.
\end{equation}
Applying this estimate at each query and using preservation
between queries gives
\begin{equation}\label{qlsl:eq:weightgrowth}
 \norm{{W}\Psi_{q}}^2
 \le(1+2/\sqrt m)^{2{q}}\le e^{4{q}/\sqrt m}.
\end{equation}
This argument also covers controlled queries and adaptive algorithms
after coherent purification and padding to their worst-case query count.

We now relate this weighted norm to parity information. The
operator $Z$ multiplies each hidden position by its parity sign,
so $Z^{\otimes n}$ supplies the sign $(-1)^{F(\boldsymbol j)}$.
Since $m$ is even, the alternating vector $v=Zu$ is orthogonal
to $u$, and $Z$ exchanges $u$ and $v$. The inverse weight
$w^{-1}$ is one on $u$ and one half on $v$. Consequently,
on their invariant span, $w^{-1}Zw^{-1}$ has matrix
$\left(\begin{smallmatrix}0&1/2\\1/2&0\end{smallmatrix}\right)$.
On its invariant orthogonal complement it equals $Z/4$. Thus
\begin{equation}\label{qlsl:eq:paritynorm}
 \norm{w^{-1}Zw^{-1}}=\frac12,\qquad
 \norm{{W}^{-1}Z^{\otimes n}{W}^{-1}}=2^{-n}.
\end{equation}
Let $M$ be the final output observable, with values $+1$ for output zero
and $-1$ for output one. Its norm is at most one. The average correctness
correlation is
\begin{equation}
 B=\bra{\Psi_{q}}(Z^{\otimes n}\otimes M)\ket{\Psi_{q}}.
\end{equation}
For a fixed input, the expectation of the product of the correct
parity sign and the output sign is twice the success probability
minus one. The worst-case guarantee therefore implies $B\ge\beta$
after averaging uniformly over the inputs. Set
$y=(W\otimes I)\Psi_q$ and insert the inverse weights on both
sides of the observable:
\begin{equation}
 B=y^\dagger\bigl(W^{-1}Z^{\otimes n}W^{-1}\otimes M\bigr)y.
\end{equation}
Using \eqref{qlsl:eq:paritynorm}, $\norm M\le1$ and the
per-query growth bound now gives
\begin{equation}
 \beta\le |B|\le2^{-n}\norm y^2
 \le2^{-n}e^{4q/\sqrt m}.
\end{equation}
Taking logarithms and rearranging proves the query bound as well.
The factor $2^{-n}$ is the cost of correlating with all hidden
parities, while $e^{4q/\sqrt m}$ is the maximum progress the
queries can supply.
\end{proof}

The small-bias regime is essential to the reduction. A
constant-error lower bound for parity would not imply the
claim we need. The explicit estimate above becomes vacuous
when $\beta\le2^{-n}$ and, unlike the more general
Lee--Roland result, does not cover every bias
$2^{-O(n)}$. Our reduction uses $\beta=\eps=e^{-L}$
and $n\ge2L$. Consequently
$\beta\ge e^{-n/2}>2^{-n}$ and
$n\ln2+\ln\beta\ge n(\ln2-1/2)>0$, exactly the
range in which this estimate is useful.

\subsection{Computation encoding and solution decoding}

\begin{proof}[Proof of Theorem~\ref{qlsl:thm:main}]
We construct the matrix in three steps: encode a parity update
in entries that each reveal only one membership bit, place
these updates along an open clock, and choose the attenuation
and length so that the final answer can be read from an
$\eps$-accurate solution.

For the first step, we need a unitary whose logical action
depends on the hidden position's parity but whose individual
entries can be evaluated without locating that position.
Let $m\ge4$ be a power of two, and let $S\ket a=\ket{a+1\bmod m}$.
On an $m$-dimensional computational register, define
\begin{equation}
 \ket{0_L}:=\sqrt{\frac2m}\sum_{a\ {\rm even}}\ket a,\qquad
 \ket{1_L}:=\sqrt{\frac2m}\sum_{a\ {\rm odd}}\ket a,\qquad
 G:=2\ket u\bra u-I.
\end{equation}
For a one-hot string $x=e_j$, set
\begin{equation}\label{qlsl:eq:gadget}
 U_x:=S^{j+1}G.
\end{equation}
The two logical states are orthonormal because they have disjoint
even and odd supports. Since $m$ is even, a cyclic shift exchanges
these supports, so $S$ acts as logical $X$. Also
$u=(\ket{0_L}+\ket{1_L})/\sqrt2$, and the reflection
$G=2\ket u\bra u-I$ exchanges the two logical states as well.
On their span, $U_x=S^{j+1}G$ therefore acts as
$X^{j+1}X=X^{j+2}=X^{j\bmod2}$. Consequently
\begin{equation}\label{qlsl:eq:logical}
 U_x\ket{z_L}=\ket{(z\oplus(j\bmod2))_L}.
\end{equation}
There is no accumulated global phase to track.

We must implement the matrix oracle without first learning the
hidden position $j$. For this, use $S^{j+1}u=u$ to write
$U_x=2\ket u\bra u-S^{j+1}$. The shift has an entry one
exactly when $a=b+j+1\pmod m$, so each entry is
\begin{equation}\label{qlsl:eq:entries}
 (U_x)_{a,b}=\frac2m-x_{a-b-1\bmod m}.
\end{equation}
They are all nonzero for $m\ge4$, so the support is the entire
$m\times m$ block independently of the marked position.
Each entry depends on exactly one queried Boolean bit.
Equation~\eqref{qlsl:eq:gadget} specifies the matrix mathematically; the
reduction implements its entry oracle via~\eqref{qlsl:eq:entries} and does
not compute the hidden index $j$ or apply $U_x$ as a circuit.

\emph{An inverse that contains the computation.}
We now place these logical updates along a clock that only moves
forward and stops at its last label. The shift then becomes zero
after sufficiently many applications. This nilpotence makes the
matrix inverse a finite sum, with one term for every computation
time, rather than an approximation to an infinite history.
The first $n$ transitions compute parity, and the remaining
transitions are identities so that the answer occupies an entire
region of the history.

Use an $m$-dimensional computational register and a $2n$-dimensional
clock with states $\ket0,\ldots,\ket{2n-1}$. Define
\begin{equation}
 L_t=
 \begin{cases}
  U_{x^{(t+1)}},&0\le t<n,\\
  I,&n\le t\le2n-2.
 \end{cases}
 \qquad
 T_X:=\sum_{t=0}^{2n-2}L_t\otimes\ket{t+1}\bra t.
\end{equation}
For $t<2n-1$, the operator maps
$\ket\phi\ket t$ to $L_t\ket\phi\ket{t+1}$; it
annihilates the final clock subspace. Each $L_t$ is unitary,
and distinct times have orthogonal outputs. Hence $T_X$
preserves norm on all nonterminal clock subspaces and has norm
one. This is its partial-isometry property. After $2n$ steps
every component has passed the last label, so $T_X^{2n}=0$.
Define the $2mn$-dimensional linear-system matrix
\begin{equation}\label{qlsl:eq:C}
 C_X=\frac{I-rT_X}{2},\qquad
 r=1-\frac1R,\qquad
 b_0=\ket{0_L}\ket0,
\end{equation}
where $R\ge16$ is a power of two. Then
\begin{equation}\label{qlsl:eq:Cpromise}
 \norm{C_X}\le\frac{1+r}{2}<1,\qquad
 \sigma_{\min}(C_X)\ge\frac{1-r}{2}=\frac1{2R}.
\end{equation}
The singular-value lower bound follows from
$\norm{(I-rT_X)v}\ge(1-r)\norm v$.

Nilpotence gives the exact identity
\begin{equation}
 (I-rT_X)\sum_{t=0}^{2n-1}r^tT_X^t
 =I-r^{2n}T_X^{2n}=I.
\end{equation}
The factor $1/2$ in $C_X$ therefore produces
\begin{equation}\label{qlsl:eq:history}
 C_X^{-1}b_0
 =2\sum_{t=0}^{2n-1}r^t T_X^t b_0
 =2\sum_{t=0}^{2n-1}r^t\ket{\psi_t}\ket t.
\end{equation}
Here $\psi_0=\ket{0_L}$, and for $1\le t\le n$ the
logical value of $\psi_t$ is the parity of the first $t$
marked positions. The remaining transitions preserve the completed
answer. Thus all $\psi_t$ are normalised, and for
$t=n,\ldots,2n-1$ they equal $\ket{F(\boldsymbol j)_L}$.
Orthogonality of the clock labels means that the squared norms
of these history components add, irrespective of their logical
values. The common factor four from the inverse cancels on
normalisation, leaving answer-region probability
\begin{equation}\label{qlsl:eq:p}
 p=\frac{\sum_{t=n}^{2n-1}r^{2t}}
         {\sum_{t=0}^{2n-1}r^{2t}}
   =\frac{r^{2n}}{1+r^{2n}}.
\end{equation}
For the last equality, the weight of the second half is
$r^{2n}$ times that of the first half: replace $t$ by $n+t$
in its geometric sum. This gives the denominator $1+r^{2n}$.
The answer probability is independent of the hidden input,
although the logical answer within that region is not.
No cyclic wraparound, reverse computation or eigenphase assumption
is needed for this inverse.

For Hermitian input, use
\begin{equation}\label{qlsl:eq:hermitian}
 A_X=\begin{pmatrix}0&C_X\\C_X^\top&0\end{pmatrix},
 \qquad b=\binom{b_0}{0}.
\end{equation}
Squaring $A_X$ gives the diagonal blocks $C_XC_X^\top$
and $C_X^\top C_X$. Its eigenvalues are therefore the positive
and negative singular values of $C_X$, so the matrix is real
symmetric and indefinite, with the same singular values as $C_X$
each repeated twice. Solving the two block equations gives
\begin{equation}
 A_X^{-1}b=\binom{0}{C_X^{-1}b_0}.
\end{equation}
The whole history lies in the lower dilation sector, and its
decoder is unchanged apart from selecting that known sector.
The extra sector doubles the dimension to $4mn$.

\emph{The oracle reduction.}
We next check that the matrix description gives an algorithm
using only the original membership oracle. In particular,
the support must not reveal the hidden positions.
Every row and column of $T_X$ contains at most $m$ nonzero entries.
The identity in~\eqref{qlsl:eq:C} adds one entry in a different clock block,
so $C_X$ and $A_X$ have maximum row and column sparsity $m+1$.
The bound is attained on the dense-transition blocks.

All locations are input-independent: a transition block is either the
fixed full support of $U_x$ or a fixed identity block. The absent boundary
transition is also known. One can therefore choose an injective padded
list for each row and column without making Boolean queries.
Complete each known injection to a permutation of all matrix indices.
These permutations implement in-place location access; their inverses
and controlled versions are equally input-independent.
The list padding can also be chosen publicly up to the supplied
bound $s\ge m+1$. It changes no matrix entry. In particular,
neither the forward locator nor its inverse reveals a hidden
position: these are fixed permutations known to the Boolean-query
algorithm before it receives its input.

For a value query, the matrix indices determine whether the
entry is known or depends on the single bit in
\eqref{qlsl:eq:entries}. In the latter case, compute that bit
into a work qubit, XOR the encoding of the corresponding
exact entry into the target, and unquery the bit. This uses
two Boolean queries coherently for all index superpositions.
For a known entry, use a fixed valid Boolean address and
ignore its answer when computing the entry encoding; the
second query still erases it. For a controlled value query,
control only the XOR of the encoding and perform the two
Boolean calls unconditionally. Thus all branches have a
common query schedule and clean work registers, without
assuming free controlled Boolean access. The XOR value
oracle is its own inverse.

Because $m$ and $R$ are powers of two, every entry is an exact dyadic
rational with $\mathcal{O}(\log m+\log R)$ bits. The entry precision does not
increase with $\eps$.
The right-hand side and its preparation circuit are input-independent.
Therefore a solver using $Q$ matrix queries yields a Boolean algorithm
using at most $2Q$ queries.

\emph{Balancing the promises with the answer signal.}
The gadget size must fit the sparsity bound, attenuation
must fit the inverse-norm bound, and the answer region must
retain enough probability to survive the allowed error.
We choose $R$ of order $\kappa$ because the inverse-norm
bound is $2R$, and $m$ of order $s$ because each matrix row
contains at most $m+1$ nonzeros. For the history length, we
need $r^{2n}$ to remain a constant multiple of $\eps$ so
the answer probability survives trace-distance error. Since
$-\ln r$ is of order $1/R$, this permits $n$ of order
$R\ln(1/\eps)$. The following integer choices meet all
three requirements while keeping the entries dyadic:
\begin{equation}\label{qlsl:eq:parameters}
 R=2^{\lfloor\log_2(\kappa/2)\rfloor},\quad
 m=2^{\lfloor\log_2(s-1)\rfloor},\quad
 L=\ln(1/\eps),\quad
 n=\left\lfloor\frac{R-1}{2}\ln\frac1{16\eps}\right\rfloor.
\end{equation}
Then $R\ge16$, $m\ge4$, $\kappa/4<R\le\kappa/2$, and $m>s/4$.
Equations~\eqref{qlsl:eq:Cpromise} and~\eqref{qlsl:eq:hermitian} give precisely
$\norm{A_X}\le1$ and $\norm{A_X^{-1}}\le2R\le\kappa$.
Also $m+1\le s$.

Since $L\ge\ln64$, we have
$\ln(1/(16\eps))\ge L/3$. Consequently
\begin{equation}\label{qlsl:eq:nrange}
 n\ge\frac{(R-1)L}{6}-1
   \ge\frac{RL}{8}
   >\frac{\kappa L}{32},
 \qquad
 n\le\frac{RL}{2}.
\end{equation}
For the middle inequality, subtract $RL/8$ and use
$(R-4)L\ge12\ln64>24$.
In particular $n\ge2L$, $n\ge R/2$, and
$N=4mn=\Theta(s\kappa L)$.

The promised upper bound is already sufficient for the reduction,
but we can also verify that these matrices actually have condition
number of order $\kappa$. This distinguishes an intrinsically
ill-conditioned example from one merely supplied with a loose
inverse-norm promise. To see this, put
$\ket{h_t}=\ket{\psi_t}\ket t$ and
$\ket{\chi}=(2n)^{-1/2}\sum_{t=0}^{2n-1}\ket{h_t}$.
The vectors $h_t$ are orthonormal because of their clock labels,
and $T_Xh_t=h_{t+1}$ before the terminal value. Applying the
inverse series to the uniform history $\chi$, the coefficient at
time $t$ receives contributions from every earlier time. Their
sum is
$\sum_{j=0}^t r^j=(1-r^{t+1})/(1-r)=R(1-r^{t+1})$.
Thus
\begin{equation}\label{qlsl:eq:actualconditioning}
 C_X^{-1}\ket{\chi}
 =\frac{2R}{\sqrt{2n}}\sum_{t=0}^{2n-1}
       (1-r^{t+1})\ket{h_t},\qquad
 \norm{C_X^{-1}}
 \ge\sqrt2 R(1-r^{n+1})
 \ge\sqrt2 R(1-e^{-1/2})>\frac R2.
\end{equation}
For the norm lower bound, retain just the last $n$ orthogonal
components, where $1-r^{t+1}\ge1-r^{n+1}$. Their common
coefficient gives $\sqrt2R(1-r^{n+1})$. The next inequality
uses $r\le e^{-1/R}$ and $n\ge R/2$, so
$r^{n+1}\le e^{-1/2}$.
Also $C_Xh_{2n-1}=h_{2n-1}/2$, so $\norm{C_X}\ge1/2$.
Consequently the Hermitian dilation satisfies
\begin{equation}
 \frac{\kappa}{16}<\frac R4
 <\norm{A_X}\,\norm{A_X^{-1}}\le\kappa.
\end{equation}
No equality in the Neumann-series upper bound is asserted for the
finite open clock.

Using $-\ln(1-1/R)\le1/(R-1)$ and the definition of $n$,
\begin{equation}
 r^{2n}\ge e^{-2n/(R-1)}\ge16\eps.
\end{equation}
Equation~\eqref{qlsl:eq:p} therefore implies
\begin{equation}\label{qlsl:eq:answerprob}
 p\ge\frac{16\eps}{1+16\eps}
   \ge\frac{64}{5}\eps.
\end{equation}

\emph{Decoding the heralded state.}
We finally turn the solution state into a Boolean output. Measure
the dilation sector and clock. In the lower sector and on times
$n,\ldots,2n-1$, read the logical bit by measuring whether
the computational basis position is even or odd. The logical
states have disjoint parity supports, so this returns $F$ with
certainty on the exact answer region. Outside that region,
output an independent fair bit. The exact solution is therefore
decoded correctly with probability
\begin{equation}
 p\cdot1+(1-p)\cdot\tfrac12=\tfrac12+\tfrac p2.
\end{equation} Trace-distance contractivity
makes it correct with probability at least
$1/2+p/2-\eps$ on the solver's conditional output.
We also output a fair bit if the solver reports failure.
Let the solver's success probability on input $X$ be $\pi_X$.
The conditional decoder has correlation at least $p-2\eps$,
and a fair bit on failure contributes zero. Thus the
unconditional correlation is at least $\pi_X(p-2\eps)$.
The signal estimate \eqref{qlsl:eq:answerprob} makes the
parenthesis positive, so $\pi_X\ge2/3$ gives
\begin{equation}\label{qlsl:eq:bias}
 \beta_X:=2\Pr(\text{correct}\mid X)-1
 \ge\frac23(p-2\eps)
 \ge\frac{36}{5}\eps
 \ge\eps.
\end{equation}
This bound permits both input-dependent success probabilities
and mixed conditional states. It is an unconditional
Boolean-query guarantee to which the XOR lemma applies.

Apply Lemma~\ref{qlsl:lem:xor} to the Boolean algorithm simulating the solver.
It uses at most $2Q$ queries. Taking $\beta=\eps$ and using $n\ge2L$,
\begin{align}
 2Q
 &\ge\frac{\sqrt m}{4}(n\ln2-L)\\
 &\ge\frac{n\sqrt m}{4}(\ln2-\tfrac12)
  >\frac{n\sqrt m}{24}.
\end{align}
Here $\ln2>2/3$. Finally~\eqref{qlsl:eq:nrange} and $\sqrt m>\sqrt s/2$ give
\begin{equation}
 Q>\frac{n\sqrt m}{48}
  >\frac{\kappa\sqrt s}{3072}\ln\frac1\eps.
\end{equation}
This proves Theorem~\ref{qlsl:thm:main}.
\end{proof}

For comparison, the inequalities in Patel's proof
\cite[Eq.~(68) and Sections~7--9]{Patel} imply
$Q>\kappa\sqrt{s}\ln(1/\eps)/(1024\sqrt2\,e)$:
they give $2Q>n\sqrt{m}/(8\sqrt2\,e)$,
$n>\kappa\ln(1/\eps)/32$, and $\sqrt m>\sqrt s/2$.
This is a constant of about $1/3936$, versus $1/3072$ above.
Both ranges require $s\ge5$ and $\eps\le1/64$; Patel permits
$\kappa\ge31$, whereas we use $\kappa\ge32$.
The conventions differ: Patel uses a condition-number promise and
pure-state Euclidean error; our theorem uses the normalised inverse-norm
promise and explicitly allows mixed outputs and heralded trace-distance
error. These constants are conservative proof constants, not optimised
performance estimates.

This lower bound couples the three parameters within one family,
rather than multiplying lower bounds proved on unrelated instances.
The condition promise determines how weakly each clock step is
attenuated. Precision then determines how many parity-computation
steps can fit before the answer signal becomes too small. The
membership-oracle reduction and XOR estimate charge
$\Omega(n\sqrt m)$ queries for that entire history. With
$n=\Theta(\kappa\log(1/\eps))$ and $m=\Theta(s)$,
these are exactly the three factors in the upper bound.
Together with Theorem~\ref{qls:thm:main}, it gives
\begin{equation}
 \Theta\!\left(\sqrt{s}\,\kappa\log\frac1\eps\right)
\end{equation}
with the same in-place sparse access, normalised spectral promise,
bounded query count, and heralded trace-distance guarantee, in the
parameter regime of Theorem~\ref{qlsl:thm:main}.

\section{Conclusion}\label{sec:conclusion}

We have developed algorithms for Hamiltonian simulation, differential
equations, and linear systems with square-root dependence on sparsity.
The key ingredient is a carefully chosen unitary embedding of the input matrices through a Cayley transform. For Hamiltonian simulation, combining this representation with transducer
methods gives the baseline query complexity scaling like $\mathcal{O}(\sqrt{s}\tau+\sqrt{s\tau \log(1/\eps)}+ \log(1/\eps))$. This bound holds for both time-independent
Hamiltonians and time-dependent Hamiltonians. The general formulation recovers
additive time and logarithmic precision dependence for supplied Hermitian block
encodings, while extending the earlier time-dependent transducer
construction~\cite{CGWZ} to achieve square-root sparsity dependence
from native sparse queries. For time-independent Hamiltonians,
Theorem~\ref{hs:thm:optimal} improves the precision dependence and
matches the joint lower bound in the stated parameter regime. The mixed term shows
that no bound $\mathcal O(\sqrt{s}L^p+L)$ at $\tau=1$ with $p<1/2$ holds uniformly in sparsity, even with a polylogarithmic factor in $s$ (Corollary~\ref{hsl:cor:power}).
Theorem~\ref{hs:thm:two-norm} also uses a supplied maximum-entry bound,
with matching lower bounds in the parameter setting specified
in Proposition~\ref{hsl:prop:two-norm}. The same simulation results give
optimal square-root dimension dependence at constant error for implementing an arbitrary unitary matrix when given black-box access to its entries.

The main remaining gap regarding sparse Hamiltonian simulation concerns time-dependent simulation and can be stated explicitly as follows. With
$L=\log(1/\eps)$ and $h=\log(e+L/\tau)$, the upper bound contains
$\sqrt{s\tau L}$, whereas the lower bound contains
$\sqrt{s\tau L/h}$, differing by a factor $\sqrt h$.
Closing this factor could require extending the time-independent
precision refinements to noncommuting time-dependent inputs or proving a
stronger lower bound for that class.
Corollary~\ref{hs:cor:hybrid} also leaves a difference in the
precision-only term: the upper bound contains
$L/\log(e+L/(s\tau))$, whereas the lower bound contains
$L/\log(e+L/\tau)$. These have the same asymptotic order at fixed
sparsity and positive time, but differ in their joint dependence on
the parameters.

Combining the simulation algorithm with LCHS gives our differential-equation bounds. Separating the Hamiltonian and dissipative components
is important here: when the source and dissipation vanish, the Hamiltonian result is
recovered, while for general dissipation the integration range itself
grows with the precision logarithm. The lower bound establishes the
joint term $\Omega(\bar g\sqrt{s}\tau)$ for time-dependent
equations and gives separate precision-dependent hardness
for constant positive semidefinite generators. It does not, however,
show that $\bar g$, $\sqrt{s}\tau$, and the full dissipative precision
logarithm must occur multiplicatively in one instance. Consequently an unresolved query-complexity gap remains in
the dissipative ODE problem. Under its mesh-access assumptions,
Lemma~\ref{odel:lem:dyson} gives sublogarithmic precision dependence at
fixed other parameters, but has exponential dependence on $s\eta T$.
For general constant generators, Theorem~\ref{ode:thm:general} retains
square-root sparsity dependence up to logarithmic factors, with additional
costs for propagator growth and the relative intermediate-solution norms.

For linear systems the situation is sharper. The complexity
$\mathcal{O}(\sqrt{s}\kappa\log(1/\eps))$ matches the joint lower bound
$\Omega(\sqrt{s}\kappa\log(1/\eps))$ of Theorem~\ref{qlsl:thm:main} (see also~\cite{Patel,BPHK26}) in the stated
parameter range, including the simultaneous dependence on sparsity,
condition number, and precision. Our open-clock proof gives this lower
bound in the same sparse-access and heralded-output model as the
algorithm. The explicit moment-based phase gap and
normalisation-one Chebyshev filter reduce the filtering constants relative to an averaging filter (Remark~\ref{qls:rem:constants}) while
preserving this scaling and the bounded-query success guarantee.
The independently obtained construction of Bravo-Prieto,
Harrow, and Kothari reaches the same scaling by a different route,
through an enlarged linear system and an exact block encoding~\cite{BPHK26}.
With a supplied constant-factor estimate of $\norm{A^{-1}b}$,
Corollary~\ref{qls:cor:separate} combines our Cayley walk with
Gao, Ji and Liu's block-input solver~\cite{GJL26} to attain the
same sparse-matrix query bound together with
$O(\kappa/\norm{A^{-1}b})$ right-hand-side preparation calls.
The cost of obtaining the estimate is excluded from this additional result.

The gate-efficient implementation of Theorem~\ref{gate:thm:main}
extends to both improved-precision time-independent Hamiltonian bounds.
Under the stated common access and computational assumptions, the additional gate count is the query bound times a polynomial in the data bit lengths, the logarithm of the dimension,
and the logarithms of the query bound and inverse precision. This implementation uses
the ordered-rotation factorisation of Ref.~\cite{CGJLWZ}, with explicit
costs for the filter and coherent reuse. For dissipative ODEs, Theorem~\ref{ode:thm:gates} also accounts for
coherent source durations, frequency selection, coefficient preparation
and bounded success amplification. It covers Lipschitz time-dependent
generators with certified uniform meshes and a positive source quadrature, while retaining explicit
source-preparation costs. These counts exclude gates inside the supplied
matrix oracles; efficient classical compilation of the signed coefficient
tables is not established. Appendix~\ref{app:formalisation} describes
the conditional Lean coverage of both query and gate results, including
the external hypotheses and differences between the formalised interfaces
and the manuscript constructions.

The remaining directions include closing the residual logarithmic
factors for time-dependent Hamiltonian simulation and determining the necessary
precision dependence of dissipative evolution. It would also be useful
to develop gate bounds for the general non-dissipative ODE and
linear-system algorithms, and to reduce the costs of source preparation
in applications.

\begin{acknowledgments}
DWB worked on this project under a sponsored research agreement with Google Quantum AI.
This project is supported by Australian Research Council Discovery Project DP260102543.
SS's research was supported by the U.S. Department of Energy, Office of Science, National Quantum Information Science Research Centers, Co-design Center for Quantum Advantage (C2QA) under Contract No. DE-SC0012704 (Basic Energy Sciences, PNNL FWP 76274). ASR is supported by Sydney Quantum
Academy (SQA). Derivations in this work were developed with GPT-6 Astra and further checked with Fable 5.1.
Lean code was developed with Opus 5.5.
\end{acknowledgments}

\appendix
\renewcommand{\thesection}{\Alph{section}}
\renewcommand{\thesubsection}{\thesection.\arabic{subsection}}
\renewcommand{\thesubsubsection}{\thesubsection.\alph{subsubsection}}
\renewcommand{\thetheorem}{\thesection.\arabic{theorem}}
\renewcommand{\theequation}{\thesection.\arabic{equation}}

\section{Optimal precision for time-independent Hamiltonian simulation}
\label{opt:sec:proof}

Our task is to obtain the optimal precision dependence in
Theorem~\ref{hs:thm:optimal} with a finite unitary circuit. More
precisely, we seek the query count
\begin{equation}
 \mathcal O\!\left(\sqrt{s}\,\eta|t|
 +\sqrt{\frac{s\eta|t|\log(1/\eps)}
 {\log(e+\log(1/\eps)/(\eta|t|))}}
 +\frac{\log(1/\eps)}{\log(e+\log(1/\eps)/(\eta|t|))}\right)
\end{equation}
for $\eta|t|\ge1$ and $0<\eps\le1/64$. As in the main theorem,
$\eta$ bounds the Euclidean norm of every row and column; we do
not assume $\norm H\le\eta$. We also retain the error guarantee on
the complete output, including every work register that starts
clean. We write $\mathsf E$ for the isometry that appends all of
these registers in zero when we need to state that guarantee as
an equation.

The denominator logarithm identifies what remains to be proved.
When it is bounded, the baseline simulation already has the
required order. When it is large, geometric suppression of the
private evolution does not give the full improvement. We will
instead arrange for a finite combination of circuits to cancel
many unwanted contributions exactly, leaving a remainder that
decreases as a high power of its scale divided by the cancellation
order. The main difficulty is to obtain this cancellation without
making the coefficients of the circuit combination so large that
their implementation removes the gain.

We first solve that difficulty for a supplied signed factor
$H=\lambda B^\dagger ZB$. Here $B$ is the successful part of a
unitary preparation, and the known reflection $Z$ assigns a plus
or minus sign to each successful component. We construct the
finite circuit blocks that this access supplies, then choose a
bounded-weight combination that cancels their unwanted low-order
terms. Positivity controls the remaining terms even though the
operators in them need not commute. The resulting simulation
theorem will apply directly to a factor prepared from sparse
entries in one parameter regime.

At larger sparsity, we need a different use of the same signed
preparation. Factoring the whole matrix gives too large a scale
for the remainder, so we first reduce the entry magnitudes and
use the signed factor only for the exact correction. We place
both pieces in one transducer. The reduced entries contribute a
range of private phases that a geometric filter can suppress;
the correction contributes two known phases with small compressed
weights, at which we impose high-order filter zeros. This second
construction completes the remaining regime.

Two finite sizes will appear throughout the proof. Dividing the
physical evolution into more boundary interactions improves its
approximation while enlarging a known circuit. Reusing the resulting
whole interaction on more public inputs costs additional oracle
queries. We will keep these operations separate and show explicitly
that every required accuracy is attained at a finite value of both
sizes. This appendix establishes the query and full-output error
bounds. Appendix~\ref{gate:sec:proof} counts the additional gates
needed to compile these same circuits, as stated in
Theorem~\ref{gate:thm:main}.

\subsection{A finite circuit for signed factors}\label{opt:sec:factor}

We begin by specifying what a signed factor allows us to do.
Let $\mathcal H_{\rm sys}$ be the system space, and let $\mathcal K$
be the space on which the supplied preparation unitary $U_B$ acts.
The known isometry $E:\mathcal H_{\rm sys}\to\mathcal K$ appends
the preparation's zero work registers. A known projection $\Pi$
selects the successful subspace
$\mathcal K_{\rm succ}=\operatorname{ran}\Pi
=\{\Pi v:v\in\mathcal K\}$. Thus the rectangular map
$B:=\Pi U_BE$ describes the successful component of the
preparation; the operation available to the circuit is the full
unitary $U_B$, together with its adjoint, each charged as one factor
query.

To encode a Hermitian matrix that may have both signs, we supply a
known Hermitian unitary $Z$ on $\mathcal K_{\rm succ}$. Its
eigenvalues are $+1$ and $-1$; let $\Pi_+$ and $\Pi_-$ be the
corresponding projections. The two successful components then
give two different system quantities:
\begin{equation}
 \langle\psi,B^\dagger ZB\psi\rangle
 =\norm{\Pi_+B\psi}^2-\norm{\Pi_-B\psi}^2,\qquad
 \langle\psi,B^\dagger B\psi\rangle=\norm{B\psi}^2.
\end{equation}
The difference in the first identity supplies the signed action we
want. The second identity measures the sum
$\norm{\Pi_+B\psi}^2+\norm{\Pi_-B\psi}^2$, namely how much
successful amplitude was prepared before the signs cancel. We must keep both
quantities because the first determines the evolution and the
second controls the remainder in our approximation. We therefore
write $C:=B^\dagger ZB$ and $D:=B^\dagger B$, both operators on
$\mathcal H_{\rm sys}$, and assume known parameters $\lambda$
and $\delta$ satisfying
\begin{equation}\label{opt:eq:factor}
 H=\lambda C,\qquad C=B^\dagger ZB,\qquad
 D=B^\dagger B\preceq\delta I,\qquad
 \lambda>0,\quad 0<\delta\le1.
\end{equation}
The two successful components have positive Gram matrices
$B^\dagger\Pi_\pm B=(D\pm C)/2$. Consequently
$-D\preceq C\preceq D$, and $\norm H\le\lambda\delta$.
This order relation does not make $C$ and $D$ commute. Our
construction and error analysis must therefore preserve the order
of their products. We will never assume an oracle for $B$, $C$ or
$D$ as separate operations: these maps describe components and
compressions of the supplied preparation.

The distinction between signed action and total weight also
separates two time scales. The access normalisation gives
$T=\lambda|t|$, whereas the promised successful weight gives the
smaller upper bound $b=T\delta$ on $|t|\norm H$. The following
theorem uses $T$ to count the required interactions and $b$ to
control the approximation remainder. Its proof extends through
Appendix~\ref{opt:sec:tail}; we first construct the interactions
and the finite weights that the proof needs.

\begin{theorem}[Simulation from a signed factor]\label{opt:thm:signed}
Under the factor-access assumptions \eqref{opt:eq:factor}, put
\begin{equation}
 {T}=\lambda|t|,\qquad b=T\delta,\qquad
 r_b=\frac{\log(1/\eps)}{\log(e+\log(1/\eps)/b)}\quad(b>0),
 \quad 0<\eps\le1/64.
\end{equation}
Set $r_0:=0$ by continuity. There is a simulation with error on all output registers
at most $\eps$ using
\begin{equation}\label{opt:eq:signed-bound}
 \cO\!\left(\sqrt{Tb}+\sqrt{Tr_b}+r_b\right)
\end{equation}
factor queries, including adjoints, controls, and cleanup. If $t=0$ or
$b\le\eps$, the identity suffices.
\end{theorem}

We now turn the preparation into an interaction whose repeated
uses can be combined. The success signs provide two private
phases, while failure provides a third. A boundary rotation
exchanges amplitude between the system, which we call the public
space, and one copy of the preparation space, which we call the
private space. These are the two summands of a direct sum, as in
the reuse framework of Refs.~\cite{BJY,CGWZ}; a private amplitude
is not an independently prepared tensor-factor state.

We choose a real parameter $q\ge1$ and let $K$ act as $qZ$ on
success and as zero on failure. The value of $q$ will eventually
balance the cancellation order against the number of reuses. For
now, its role is concrete: it fixes the two success phases of the
Cayley operation. In the clean input coordinates of the
preparation, that unitary is
\begin{equation}\label{opt:eq:private-unitary}
 P=U_B^\dagger(K+\ii I)(K-\ii I)^{-1}U_B.
\end{equation}
We can implement this operation by applying $U_B$, applying a
known phase on each of its success-sign and failure subspaces,
and applying $U_B^\dagger$. The inverse in the Cayley expression
therefore does not call for a matrix-inversion subroutine. Before
conjugating back, its three phases are simply
\begin{equation}
 \zeta_\pm=\frac{\pm q+\ii}{\pm q-\ii},\qquad -1,
\end{equation}
Their compressed weights at the clean input boundary are,
respectively, $(D+C)/2$, $(D-C)/2$, and $I-D$. Indeed, the first
two are $B^\dagger\Pi_\pm B$, while the last is the preparation
weight outside success. This description records both where the
private phases lie and how strongly a system input couples to
each one.

To distribute the desired evolution over a finite number of
interactions, let $a>0$ be a total coupling strength and let the
positive integer $J_s$ be the number of physical slots. Each slot
uses strength $v:=a/J_s$. We couple the public space to the copy
of $\operatorname{ran}E$ in its private slot with the rotation
\begin{equation}\label{opt:eq:boundary-rotation}
 G_v:=\begin{pmatrix}
 cI&-\beta_v E^\dagger\\
 \beta_v E&I-(1-c) EE^\dagger
 \end{pmatrix},\qquad
 c:=\frac{1-v/2}{1+v/2},\qquad
 \beta_v^2:=\frac{2v}{(1+v/2)^2}.
\end{equation}
We take $\beta_v\ge0$. The identity $c^2+\beta_v^2=1$
shows that $G_v$ is a rotation on the paired vectors
$\psi\oplus0$ and $0\oplus E\psi$; it is the identity on
private vectors orthogonal to $\operatorname{ran}E$. Our local
interaction is $S_v:=G_v(I\oplus P)$, so it first applies the
private phase operation and then exchanges amplitude through
the boundary.

To find the effect of reuse, we need more than the public block
of one interaction. Some amplitude enters the private space and
returns only after later uses. We encode all such returns in a
generating function with variable $w$. Its exponent counts the
delay in uses of the interaction, not physical evolution time.

For a unitary with public/private blocks
$\left(\begin{smallmatrix}A_g&B_g\\C_g&D_g\end{smallmatrix}\right)$,
we call
\begin{equation}
 A_g+wB_g(I-wD_g)^{-1}C_g
 =A_g+\sum_{j\ge1}w^jB_gD_g^{j-1}C_g,\qquad |w|<1,
\end{equation}
its causal transfer. The power of $w$ counts the delay before an
injected private amplitude returns to the public sector. We use this
generating function to calculate the finite circuit's public block.
The inverse exists throughout $|w|<1$: the private block $D_g$ is a
compression of a unitary and hence has norm at most one.

For our specified interaction, the public/private blocks are
$cI$, $-\beta_v E^\dagger P$, $\beta_v E$, and
$(I-(1-c)EE^\dagger)P$. Substituting them in the transfer formula
leaves a private resolvent with a boundary correction. We first
collect the returns of the bare private operation in
$X:=wE^\dagger P(I-wP)^{-1}E$ and set
\begin{equation}\label{opt:eq:response}
 \Omega(w):=I+2X=E^\dagger(I+wP)(I-wP)^{-1}E.
\end{equation}
The factor $I+2X$ counts a direct boundary contribution together
with the private excursions. We use it because eliminating the
private amplitude from $S_v$ puts the entire transfer into a
simple rational function of $\Omega$. In detail, summing the
boundary correction in the private inverse gives
\begin{align}
 cI-w\beta_v^2E^\dagger P
       [I-w(I-(1-c) EE^\dagger)P]^{-1}E
 &=cI-\beta_v^2X(I+(1-c) X)^{-1}\notag\\
 &=\left(I-\frac v2\Omega(w)\right)
   \left(I+\frac v2\Omega(w)\right)^{-1}
 =:u_{J_s}(w).
 \label{opt:eq:finite-transfer}
\end{align}
We can check both steps without treating $\Omega$ as an
additional operation. In the first equality, the inverse of
$I-wP+w(1-c)EE^\dagger P$ sums repeated encounters with
the boundary projector; after compression this gives
$X(I+(1-c)X)^{-1}$. To check the second equality, combine
$cI-\beta_v^2X(I+(1-c)X)^{-1}$ over its common denominator.
Since $\beta_v^2=1-c^2$, the numerator reduces to
$cI-(1-c)X$. Substituting $1-c=v/(1+v/2)$ and
$\Omega=I+2X$ makes the numerator
$(I-v\Omega/2)/(1+v/2)$ and the denominator
$(I+v\Omega/2)/(1+v/2)$. Cancelling the scalar factor gives
Eq.~\eqref{opt:eq:finite-transfer}. Thus the rational transfer
is a consequence of the actual boundary rotation and private
operation already specified.

We next distribute the total coupling across the finite
physical mesh, then reuse the resulting whole interaction
with initially zero private amplitude. These are different
operations and their clock sizes have different costs.

For the physical evolution, we give each of the $J_s$ local
interactions its own private slot. We apply the boundary rotations
in chronological order, carrying the public output of one into
the next. Other factors leave a given private slot untouched.
Eliminating the private amplitude in each slot therefore applies
the same transfer $u_{J_s}(w)$ successively to the public
amplitude. The exact transfer of this finite cascade is
\begin{equation}\label{opt:eq:finite-F}
 F_{J_s}(w)=u_{J_s}(w)^{J_s}.
\end{equation}
Although the cascade has $J_s$ slots, one call to the whole
interaction need not make $J_s$ factor queries. As in
Eq.~\eqref{ov:eq:chronological-circuit}, a private operation on
one slot commutes with every boundary rotation involving a
different slot. We can move all private operations to one side
while keeping the boundary rotations in their original order.
Their combined action is a single direct-sum operation: with a
slot label and a private payload, it is $I_{J_s}\otimes P$.
The same query--phase--inverse-query sequence thus acts coherently
for every slot label. The remaining $J_s$ boundary rotations use
only the known isometry $E$ and incur no factor queries. We will
also show below how to control the global interaction by controlling
only the known middle phase between a factor query and its inverse.

We now reuse that complete interaction, beginning with zero private
amplitude. For a function $f$ analytic at zero, write $[w^r]f$ for
its $r$th Taylor coefficient and define
\begin{equation}\label{opt:eq:cesaro}
 \Ces n(f):=\sum_{r=0}^{n-1}\left(1-\frac rn\right)[w^r]f(w),
 \qquad n\ge1.
\end{equation}
This finite Ces\`aro functional is exactly what a uniform public
reuse clock selects. The following block calculation makes the
connection explicit.

For completeness, if a global unitary has blocks
$\left(\begin{smallmatrix}A_g&B_g\\C_g&D_g\end{smallmatrix}\right)$,
its causal transfer is $A_g+wB_g(I-wD_g)^{-1}C_g$.
Let $\mathsf T_0=A_g$ and $\mathsf T_r=B_gD_g^{r-1}C_g$ for $r\ge1$ be its
Taylor coefficients. Feeding successive public slots through the
initially zero private register gives, for four slots, the public
input-to-output block
\begin{equation}
 \begin{pmatrix}
 \mathsf T_0&0&0&0\\
 \mathsf T_1&\mathsf T_0&0&0\\
 \mathsf T_2&\mathsf T_1&\mathsf T_0&0\\
 \mathsf T_3&\mathsf T_2&\mathsf T_1&\mathsf T_0
 \end{pmatrix}.
\end{equation}
A block $\mathsf T_r$ returns amplitude to a public slot $r$ positions
after its input: injection uses $C_g$, the $r-1$ intervening steps
use $D_g$, and return uses $B_g$. No amplitude returns to an
earlier slot, and each delay has the same block wherever it starts.
These are precisely the lower-triangular and constant-diagonal
properties meant by \emph{triangular Toeplitz}. With $j$ slots,
there are $j-r$ entries on the diagonal of delay $r$.
Preparing and unpreparing the uniform slot state therefore gives
$\sum_{r=0}^{j-1}(1-r/j)\mathsf T_r=\Ces j(F_{J_s})$.
Other outputs are retained as leakage, and the circuit uses $j$
global transducer calls.

For example, $\Ces1(f)=[w^0]f$ and
$\Ces3(f)=[w^0]f+\tfrac23[w^1]f+\tfrac13[w^2]f$.
The private mesh $J_s$ and public reuse horizon $n$ have different
resource meanings: the former enlarges the known boundary network,
whereas the latter counts calls to the global interaction.

We now know which finite blocks the circuit can supply:
$\Ces j(F_{J_s})$ for every finite reuse length $j$. To choose a
useful signed combination of these blocks, we need to identify
which contributions produce the wanted Hamiltonian and which
must cancel. We do this first for an exponential that approximates
the finite cascade. The exponential is used only for analysis;
the final implementation will continue to use $F_{J_s}$.

The three-phase decomposition makes the required separation
possible with scalar algebra. For $|w|<1$, define
\begin{equation}\label{opt:eq:xi-R}
 \xi(w):=\frac{1-w}{1+w},\qquad
 R(w):=\frac{1-\xi(w)^2}{1+q^2\xi(w)^2}.
\end{equation}
The failure phase $-1$ contributes $\xi(w)$ to the response.
For the two success phases, we separate the coefficient of their
weight sum $D$ from the coefficient of their weight difference
$C$. Substitution of the three compressed phase weights gives
\begin{align}
 \Omega(w)&=E^\dagger(I+wP)(I-wP)^{-1}E\notag\\
 &=\xi(w)(I-D)+\frac{1+\zeta_+w}{1-\zeta_+w}\frac{D+C}{2}
                 +\frac{1+\zeta_-w}{1-\zeta_-w}\frac{D-C}{2}\notag\\
 &=\xi(w)I+R(w)\bigl(\ii qC+q^2\xi(w)D\bigr).
 \label{opt:eq:omega}
\end{align}
The last equality follows by separating the symmetric and antisymmetric
parts of the two success phases. With
$r^{\rm Cay}_\pm(w)=(1+\zeta_\pm w)/(1-\zeta_\pm w)$,
\begin{equation}
 \frac{r^{\rm Cay}_+(w)-r^{\rm Cay}_-(w)}2=\ii qR(w),\qquad
 \frac{r^{\rm Cay}_+(w)+r^{\rm Cay}_-(w)}2-\xi(w)=q^2\xi(w)R(w).
\end{equation}
Substituting these scalar identities into the preceding line of
Eq.~\eqref{opt:eq:omega} gives the displayed $C$ and $D$ terms.
They identify the wanted generator and the extra weight terms that
finite moments must cancel.

We analyse the finite circuit by comparing its transfer with
\begin{equation}\label{opt:eq:continuum}
 F(w)=e^{-a\Omega(w)}
 =e^{-a\xi(w)}
   \exp\!\left[-aR(w)\bigl(\ii qC+q^2\xi(w)D\bigr)\right].
\end{equation}
The exponential is an analysis target for the finite cascade, not an
operation assumed available to the circuit. Equation~\eqref{opt:eq:finite-mesh}
will bound the discrepancy on all coefficients used by the algorithm.

The boundary value now reveals how to choose the total coupling.
At $w=1$ we have $\xi=0$ and $R=1$, so the unwanted $D$
contribution vanishes and $F(1)=e^{-\ii aqC}$. We therefore impose
$aq=T=\lambda|t|$ to obtain the desired positive-time evolution.
The circuit does not evaluate its generating function at the
boundary. Instead, we will choose real coefficients $\lambda_j$
such that $\sum_j\lambda_j\Ces j(F)$ reproduces sufficiently
many low-order terms of that boundary value. We then apply those
same weights to the available finite blocks
$\Ces j(F_{J_s})$.

This separates the remaining approximation tasks. First, we must
bound the error of replacing boundary evaluation by a finite
coefficient functional. Second, we must bound the effect of
replacing the analytic transfer $F$ by the finite cascade
$F_{J_s}$. Neither task permits us to commute $C$ and $D$.
Before constructing the weights, we therefore establish the
property that will later control the whole exponential without
estimating each ordered product separately.

The phase decomposition in Eq.~\eqref{opt:eq:omega} gives this
control: the operator weights
$I-D,(D+C)/2,(D-C)/2$ are positive and sum to $I$, and each scalar
fraction has real part
$(1-|w|^2)/|1-\zeta w|^2\ge0$ for $|\zeta|=1$ in the unit disk.
Hence $\Reop\Omega(w)\succeq0$, where
$\Reop X=(X+X^\dagger)/2$. For any initial vector $\phi$,
\begin{equation}
 \frac d{du}\norm{e^{-u\Omega(w)}\phi}^2
 =-2\langle e^{-u\Omega(w)}\phi,
          \Reop\Omega(w)e^{-u\Omega(w)}\phi\rangle\le0.
\end{equation}
The squared norm is nonincreasing for every initial vector, so
$\norm{e^{-u\Omega(w)}}\le1$ for $u\ge0$. We have obtained this
bound from the positive phase weights, rather than from
commutation of the compressed operators. In
Appendix~\ref{opt:sec:tail}, we will use the same argument for a
family of scaled operators to bound the tail left after exact
moment cancellation. We next construct the finite functional
that performs that cancellation.

\subsection{Finite weights for moment cancellation}\label{opt:sec:weights}

Our next task is to make a finite combination of the available
reuse blocks act like evaluation at $w=1$ on the terms we need
to keep. Write $V=e^{-a\xi}$ and
$\Lambda=\sum_j\lambda_j\Ces j$, where the real coefficients
$\lambda_j$ will depend only on known construction parameters.
Here a moment condition is a scalar equation prescribing the
value of $\Lambda$ on a function of $w$; it is distinct from a
compressed operator moment such as $J^\dagger K^jJ$ in the
sparse representation. Recall that the analytic transfer is
\begin{equation}
 F=V\exp[-aR(\ii qC+q^2\xi D)].
\end{equation}
We can see the right conditions by expanding this exponential
while preserving the order of $C$ and $D$. In a term of total
degree $k$, every ordered product containing $m$ copies of $D$
has the same scalar multiplier $VR^k\xi^m$, apart from powers
of $a$, $q$ and $\ii$. The positions of those copies within the
product do not affect this multiplier. At $w=1$, it is one for
$m=0$ and zero for $m>0$. Thus one scalar condition can cancel
all placements of $m$ copies of $D$, without identifying $CD$
with $DC$ or changing any operator order.

For example, the terms of degree two before applying $\Lambda$ are
\begin{equation}
 \frac{a^2}{2}VR^2
 \left[-q^2C^2+\ii q^3\xi(CD+DC)+q^4\xi^2D^2\right].
\end{equation}
If $\Lambda(VR^2)=1$ and
$\Lambda(VR^2\xi)=\Lambda(VR^2\xi^2)=0$, this becomes
$-{T}^2C^2/2$, where $T=aq$. The same conditions at degrees zero and
one give $I-\ii TC$. The theorem below carries out this cancellation
through every degree below $M$.

Two resources must be bounded at the same time. The largest clock
length $n$ determines how many times the oracle circuit is reused.
The total absolute weight $\sum_j|\lambda_j|$ determines the
normalisation and amplification cost of combining those circuits.
An exact cancellation with very large absolute weights would not give
the claimed query bound.

\begin{theorem}[Bounded finite moment-cancelling weights]\label{opt:thm:weights}
For every real $q\ge1$ and integer $M\ge1$, set $a:=qM$ and
$V(w)=e^{-a\xi(w)}$. There are real numbers
$\lambda_1,\ldots,\lambda_n$, depending only on $q,M$, such that
\begin{equation}\label{opt:eq:clock-size}
 n\le27a,\qquad \sum_{j=1}^n|\lambda_j|\le3000,
\end{equation}
and
\begin{equation}\label{opt:eq:clock-moments}
 \sum_{j=1}^n\lambda_j\Ces j(VR^k\xi^m)=\delta_{m0}
 \qquad(0\le m\le k<M).
\end{equation}
\end{theorem}

To construct all of these moments together, we will represent
evaluation at zero on a slightly larger space of rational
functions. This choice makes the poles, rather than the many
individual products $R^k\xi^m$, determine the construction.
Let $p(x)$ be any real rational function whose only finite poles
are $\pm\ii/q$, each of order at most $M-1$, and whose
polynomial part has degree at most $M-1$. The polynomial
part is the polynomial left after separating all terms that have
finite poles; equivalently, it is the quotient obtained by dividing
the numerator by the denominator. We seek a single set
of weights such that
\begin{equation}\label{opt:eq:rational-dual}
 \sum_{j=1}^n\lambda_j\Ces j\bigl(V(w)p(\xi(w))\bigr)=p(0).
\end{equation}
Indeed, as a function of $x=\xi(w)$,
$R^kx^m=[(1-x^2)/(1+q^2x^2)]^kx^m$ has poles only at
$\pm\ii/q$, of order at most $k\le M-1$, and polynomial part of
degree at most $m\le M-1$. Evaluating it at $x=0$ gives one when
$m=0$ and zero otherwise. Thus Eq.~\eqref{opt:eq:rational-dual}
implies all the moment equations at once. The evaluation point
$x=0$ corresponds to $w=1$; we construct finite coefficient weights
that realise its action on this rational space, rather than assume
access to the boundary value.

\begin{proof}[Proof of Theorem~\ref{opt:thm:weights}]
We build the functional by first expressing $p(0)$ as a linear
pairing and then converting that pairing into the finite
coefficient averages supplied by reuse. The poles at $\pm\ii/q$
become oscillations under an inverse Laplace transform. A simple
pole becomes an exponential in the real transform variable, and a
pole of higher order adds a polynomial factor; imaginary poles
therefore produce oscillations. We will
choose two probability mixtures that annihilate those
oscillations, have equal mass, and have different means. Their
difference then extracts a linear term whose slope is $p(0)$.
Finally, a terminating Laguerre expansion will turn the pairing
into the weights $\lambda_j$. We must carry out this last step
quantitatively: termination bounds the largest reuse length, and
a weighted coefficient estimate bounds the LCU normalisation.
Throughout this proof, $t$ denotes the real variable of the
Laplace transform; physical evolution time has already been
absorbed into $T$.

\emph{Cancelling the oscillatory moments.}
For an integer $j\ge0$, the function
$e^{-u}u^j/j!$ on $u\ge0$ is a gamma probability density
with shape $j+1$ and mean $j+1$. Its Fourier transform is
$(1-\ii\omega)^{-(j+1)}$. Therefore, if the coefficients of a
polynomial sum to one and are nonnegative, they can mix these
densities without losing positivity or unit mass; the Fourier
transform of the mixture is that polynomial evaluated at
$(1-\ii\omega)^{-1}$, with one extra factor
$(1-\ii\omega)^{-1}$. We can consequently impose vanishing
Fourier moments at $\omega=\pm1/q$ by placing roots at the two
corresponding complex values.

The constraint is that the polynomial must still have
nonnegative coefficients. To satisfy it, we take a real quadratic
in a suitable power $z^{m_0}$. We choose $m_0$ so that the
desired root raised to the power $m_0$ lies in the left
half-plane. The middle coefficient of that
quadratic will then be nonnegative. We also introduce a shift
$d_0$ in the gamma indices and an averaging length $\ell$:
\begin{equation}
 \zeta_\star=(1-\ii/q)^{-1}=\rho_\star e^{\ii\vartheta_\star},\qquad
 m_0=\ceil{\frac{\pi}{2\vartheta_\star}},\qquad
 d_0=\ceil{12a},\qquad \ell=\ceil a.
\end{equation}
Here $0<\vartheta_\star\le\pi/4$ and
$\pi/2\le m_0\vartheta_\star<\pi/2+\vartheta_\star$, so
$\cos(m_0\vartheta_\star)\le0$. The index $d_0$ will force a high-order
zero of the density at the origin; it also controls the extension
below the eventual delay $a$. The averaging length $\ell$ will
bound individual mixture weights and provide a known difference
of means. First construct the root-bearing polynomial
\begin{equation}
 A(z)=\frac{z^{2m_0}-2\rho_\star^{m_0}\cos(m_0\vartheta_\star)z^{m_0}
                         +\rho_\star^{2m_0}}
                 {1-2\rho_\star^{m_0}\cos(m_0\vartheta_\star)+\rho_\star^{2m_0}}.
\end{equation}
The numerator vanishes at $\zeta_\star$ and its conjugate, and
the choice of $m_0$ makes all three coefficients nonnegative.
The denominator normalises their sum, so $A(1)=1$. Raising
$A$ to the $M$th power gives at least the required multiplicity
at each root while retaining a probability distribution on its
coefficients.

We need two further properties of the mixture. First, the
distributional terms arising later at the delay must pair to
zero. These are terms concentrated at one point that contribute
values of derivatives of the function with which they are paired;
we specify their action when forming the pairing below. Multiplication by $z^{d_0}$ makes every gamma index at
least $d_0$, which forces a zero of that order at the origin.
Second, we need a bound on individual mixture coefficients to
control their eventual absolute weight. Convolution with a
uniform distribution on $\ell$ consecutive indices bounds each
coefficient by $1/\ell$. Finally, a second mixture shifted by
$\ell$ indices has the same mass and Fourier zeros, but its mean
increases by exactly $\ell$. We implement all three choices with
\begin{equation}
 B_\ell(z)=\frac1\ell\sum_{j=0}^{\ell-1}z^j,\qquad
 C_1(z)=z^{d_0}A(z)^M B_\ell(z),\qquad
 C_2(z)=z^\ell C_1(z).
\end{equation}
Writing $C_i(z)=\sum_j c_{i,j}z^j$, all coefficients are nonnegative,
sum to one, and are at most $1/\ell$. They vanish below $d_0$.
Set
\begin{equation}\label{opt:eq:gamma-degree}
 D_* :=d_0+2m_0M+2\ell-1,\qquad
 g_i(u):=e^{-u}\sum_jc_{i,j}\frac{u^j}{j!}.
\end{equation}
For $u\ge0$, $g_i$ is a probability density and
\begin{equation}
 \int_0^\infty e^{\ii\omega u}g_i(u)\,du
 =(1-\ii\omega)^{-1}C_i((1-\ii\omega)^{-1}).
\end{equation}
Its zeros at $\omega=\pm1/q$ have multiplicity at least $M$.
Differentiating the absolutely convergent integral yields
\begin{equation}\label{opt:eq:oscillators}
 \int_0^\infty u^j e^{\pm\ii u/q}g_i(u)\,du=0
 \qquad(0\le j<M).
\end{equation}
Also $g_i$ vanishes to order at least $d_0$ at zero, and
\begin{equation}\label{opt:eq:gamma-means}
 \int_0^\infty g_i(u)\,du=1,\qquad
 \int_0^\infty u(g_2(u)-g_1(u))\,du=\ell.
\end{equation}
We will also evaluate these functions at negative $u$ after a shift.
There we use the entire exponential-polynomial extension of
Eq.~\eqref{opt:eq:gamma-degree}; only the restriction $u\ge0$
is a probability density.

\emph{Extracting the required value.}
We now use these mixtures to isolate $p(0)$. Dividing by $x^2$
creates two distinguished partial-fraction terms at the origin:
$p(0)/x^2+p'(0)/x$. Their inverse Laplace transform is the
affine function $p(0)u+p'(0)$. This division is useful because
the difference of our mixtures has zero mass but nonzero first
moment. Inserting the affine function and using
Eq.~\eqref{opt:eq:gamma-means} gives
\begin{equation}
 \frac1\ell\int_0^\infty[g_2(u)-g_1(u)]
                [p(0)u+p'(0)]\,du=p(0).
\end{equation}
Thus the equal masses remove $p'(0)$, and the mean difference,
after division by $\ell$, returns exactly $p(0)$. The other
poles of $p(x)/x^2$ will produce polynomial multiples of
$e^{\pm\ii u/q}$, which are removed by the oscillatory
moments already established. We must also account for any
polynomial part in $p(x)/x^2$; it gives terms concentrated at
the origin rather than ordinary functions. Introducing a delay
by $a$ will place all of these terms at $t=a$ and supply the
factor $V=e^{-a\xi}$ required by the circuit transfer.

Here $p$ belongs to the rational-function space specified above,
$x$ is its complex argument, and the transform variable $t$ is real
and nonnegative. The polynomial part of $p(x)/x^2$ produces derivatives of delta
functions, so the inverse Laplace transform is understood in the
distribution sense:
\begin{equation}
 \mathfrak h=\mathcal L^{-1}\!\left[\frac{e^{-ax}p(x)}{x^2}\right].
\end{equation}
For an ordinary function $f_{\rm L}$, the pairing $\langle f_{\rm L},\phi\rangle$
is $\int_0^\infty f_{\rm L}(t)\phi(t)\,dt$. For a delta derivative,
the convention is
\begin{equation}
 \langle\delta^{(j)}(t-a),\phi\rangle=(-1)^j\phi^{(j)}(a).
\end{equation}
Thus the pairing below is a finite sum of ordinary integrals and
values of derivatives at the delay. The partial-fraction expansion
of $p(x)/x^2$ is finite: its polynomial part becomes these delta
derivatives, while its pole terms become exponential-polynomial functions.
The factor $e^{-ax}$ shifts every contribution by $a$, so
$\mathfrak h$ is supported on $[a,\infty)$. For $u>0$ the
partial fractions give
\begin{equation}\label{opt:eq:laplace-parts}
 \mathfrak h(a+u)=p(0)u+p'(0)
       +e^{\ii u/q}r_+(u)+e^{-\ii u/q}r_-(u),
 \qquad \deg r_\pm\le M-2.
\end{equation}
When $M=1$, $p$ is constant and the oscillatory terms are absent.
For $M\le2$ there are no terms supported at the delay. In general,
terms supported at the
delay are derivatives of $\delta(t-a)$ of order at most $M-3$, when
this number is nonnegative. Define, using the entire extension,
\begin{equation}\label{opt:eq:phi}
 \phi(t):=\frac{g_2(t-a)-g_1(t-a)}\ell\qquad(t\ge0).
\end{equation}
It is $e^{-t}$ times a polynomial of degree at most $D_*$.
Equations~\eqref{opt:eq:oscillators}--\eqref{opt:eq:gamma-means} cancel the
oscillatory and constant terms of \eqref{opt:eq:laplace-parts} and extract
its affine slope. The zero of order $d_0\ge M$ annihilates all the
delay-supported terms above. For example, an oscillatory term $u^j e^{\ii u/q}$ pairs to zero
by Eq.~\eqref{opt:eq:oscillators}, while a delay-supported term
$\delta^{(j)}(t-a)$ pairs to $(-1)^j\phi^{(j)}(a)=0$ because
$j<d_0$. The affine term pairs to $p(0)$ by the equal-mass and
mean-difference equations. Combining these separately justified
cancellations gives
\begin{equation}\label{opt:eq:dual-pairing}
 \langle\mathfrak h,\phi\rangle=p(0).
\end{equation}
The extension below $a$ does not change this pairing.

\emph{Obtaining finite reuse weights.}
The pairing now returns the correct value, but it is not yet
an operation that our circuit provides. We must express it as
a finite sum of $\Ces j(F_p)$, where
$F_p(w)=V(w)p(\xi(w))$. The fractional-linear map
$\xi(w)=(1-w)/(1+w)$ tells us which coefficient basis to use:
it also appears in the generating function of Laguerre
functions after changing the sign of the generating variable.
Moreover, $\phi$ is $e^{-t}$ times a polynomial. We can
therefore use a finite Laguerre expansion, so this conversion
introduces no truncation error.

For $t\ge0$ and integers $k\ge0$, introduce the Laguerre functions
$\ell_k$ and the differential operator $\mathfrak L$:
\begin{equation}
 \ell_k(t)=e^{-t}L_k(2t),\qquad
 L_k(y)=\sum_{j=0}^k\binom kj\frac{(-y)^j}{j!},\qquad
 \mathfrak L=-t\partial_{t}^2-\partial_{t}+t.
\end{equation}
The generating function and two elementary identities are
\begin{align}
 \sum_{k\ge0}z^k\ell_k(t)
 &=\frac1{1-z}\exp\!\left[-t\frac{1+z}{1-z}\right],
 \label{opt:eq:laguerre-gf}\\
 \int_0^\infty\ell_j(t)\ell_k(t)\,dt&=\frac12\delta_{jk},
 &\mathfrak L\ell_k&=(2k+1)\ell_k.
 \label{opt:eq:laguerre-identities}
\end{align}
For example, integrating the product of the two generating functions
gives $1/[2(1-zw)]$, which proves orthogonality. The differential
identity follows by differentiating \eqref{opt:eq:laguerre-gf}.

To match the generating function to the clock, recall that
dividing a Taylor series by $(1-w)^2$ weights its $r$th
coefficient by the number of later positions to which it can
contribute. This is precisely the weighting in the reuse
functional. In the Laplace representation, however, the same
change of basis introduces $\mathfrak h+\mathfrak h'$.
We therefore choose a function $\psi$ such that
$\psi-\psi'=4\phi$: integration by parts will move the
derivative back onto $\psi$ and recover our established
pairing with $\phi$.

This choice preserves finiteness. Writing
$\phi=e^{-t}P_*(t)$ with $\deg P_*\le D_*$ and seeking
$\psi=e^{-t}Q_*(t)$ reduces the equation to
$2Q_*-Q_*'=4P_*$. Differentiation lowers polynomial degree,
so the inverse of $2-\partial_t$ is a terminating geometric
series on this space. The resulting polynomial solution, and
an integral expression useful for bounding its norm, are
\begin{equation}\label{opt:eq:psi}
 \psi=e^{-t}Q_*(t),\qquad
 Q_*=2\sum_{j=0}^{D_*}2^{-j}P_*^{(j)},\qquad
 \psi(t)=4\int_0^\infty e^{-v_{\rm L}}\phi(t+v_{\rm L})\,dv_{\rm L}.
\end{equation}
Expand $\psi=\sum_{k=0}^{D_*}\beta_k\ell_k$ and put
\begin{equation}\label{opt:eq:weights}
 \lambda_{k+1}:=\frac{k+1}{2}(-1)^k\beta_k,
 \qquad n=D_*+1.
\end{equation}
We verify that these weights multiply exactly the finite blocks
we can supply. Set $F_p(w)=V(w)p(\xi(w))$ and
$S(w)=F_p(w)/(1-w)^2$. The coefficient of order $k$ is
\begin{equation}
 [w^k]S=\sum_{r=0}^k(k+1-r)[w^r]F_p
       =(k+1)\Ces{k+1}(F_p).
\end{equation}
Thus the factor $k+1$ in Eq.~\eqref{opt:eq:weights} has a
specific purpose: it converts a coefficient of $S$ into the
reuse average of length $k+1$.
For $\mathfrak b=(\mathfrak h+\mathfrak h')/4$, the delay removes
an initial boundary term and
$\mathcal L[\mathfrak b](x)=(1+x)\mathcal L[\mathfrak h](x)/4$.
Equation~\eqref{opt:eq:laguerre-gf} therefore gives
\begin{equation}
 [w^k]S=2(-1)^k\langle\mathfrak b,\ell_k\rangle.
\end{equation}
To check the change of basis, the generating function of the right
side is
\begin{equation}
 \frac{2}{1+w}\mathcal L[\mathfrak b](\xi(w))
 =\frac{1+\xi(w)}{2(1+w)\xi(w)^2}F_p(w)
 =\frac{F_p(w)}{(1-w)^2}=S(w).
\end{equation}
The factors in the proposed weights can now be read off from
this equality. The Laguerre generating function supplies the
factor $2(-1)^k$, while the clock identity supplies $k+1$.
Multiplying $\Ces{k+1}(F_p)$ by
$(k+1)(-1)^k\beta_k/2$ therefore leaves
$\beta_k\langle\mathfrak b,\ell_k\rangle$. Summing over
the finite Laguerre expansion of $\psi$ gives the desired
pairing with $\mathfrak b$. This is the step that turns an
analytic extraction rule into weights on actual reuse lengths.

All pairings converge. On a fixed sufficiently small circle about
$w=0$, $\Reop\xi(w)\ge c>0$. Cauchy estimates therefore bound each
coefficient derivative of the generating function by $C_k e^{-ct}$.
This dominates the polynomial growth of the delayed affine and
oscillatory terms, and also permits pairing with the finitely many
delay-supported distributional derivatives. It justifies
differentiation under the pairing. The derivative
$\mathfrak h'$ includes the jump at the delay. Distributional integration
by parts now gives
\begin{equation}
 \sum_{j=1}^n\lambda_j\Ces j(F_p)
 =\langle\mathfrak b,\psi\rangle
 =\frac14\langle\mathfrak h,\psi-\psi'\rangle=p(0).
\end{equation}
This proves \eqref{opt:eq:rational-dual} without a truncation error.

\emph{Bounding the resources of the functional.}
We have proved the exact cancellation identity. It remains to
show that the functional is affordable to implement. Its largest
reuse length is the degree of the extraction polynomial plus
one, so we first bound that degree. To control the total
absolute weight, we then estimate $\psi$ in a norm that weights
its $k$th Laguerre coefficient in proportion to $k+1$, matching
Eq.~\eqref{opt:eq:weights}. The construction of the mixtures
was designed for these estimates: uniform averaging bounds each
coefficient, and the large initial gamma index controls the
extension below the delay.

Here $a=qM\ge1$. Since $\arctan(1/q)\ge1/(2q)$, we have
$m_0\le\pi q+1<5q$.
Also $d_0\le13a$, $\ell-1\le a$, and $\ell\le2a$. Hence
\begin{equation}\label{opt:eq:degree27}
 D_*\le26a,\qquad n\le27a.
\end{equation}
For the norm estimates, observe that
$e^{-u}\sum_{j\ge0}u^j/j!=1$ on $u\ge0$. Since every
$c_{i,j}\le1/\ell$, this gives $g_i(u)\le1/\ell$.
We can turn that pointwise bound into an $L^2$ estimate by
using $g_i^2\le g_i/\ell$ and integrating the probability
density. The same argument applies with a factor $(a+u)^2$,
using the second moments of its gamma components.

For the derivative, differentiating a component gives the
difference between the gamma densities of indices $j-1$ and
$j$. All indices here are positive because they start at
$d_0$. Thus $g_i'$ is the difference of two positive mixtures,
each with total mass one and coefficients at most $1/\ell$.
Applying the preceding estimates to each mixture and using
the triangle inequality yields
\begin{align}
 \norm{g_i}_{L^2(0,\infty)}&\le\ell^{-1/2},&
 \norm{g_i'}_2&\le2\ell^{-1/2},\notag\\
 \norm{(a+u)g_i}_2&\le(a+D_*+2)\ell^{-1/2},&
 \norm{(a+u)g_i'}_2&\le2(a+D_*+2)\ell^{-1/2}.
 \label{opt:eq:gamma-norms}
\end{align}
For instance, the first weighted estimate follows from
$g_i^2\le g_i/\ell$ and the second moment of a gamma density of
shape at most $D_*+1$.

We must also bound the portion $0\le t<a$, where the shifted
argument $u=t-a$ is negative. Positivity is unavailable there,
but the high initial index provides a different bound. For
$-a\le u\le0$, taking absolute values gives
$|e^{-u}u^j/j!|\le e^a a^j/j!$. The lowest index appearing
in either the density or its derivative is
$k_0=d_0-1\ge11a$. Successive terms in the resulting tail
have ratio at most $1/11$, and
$k_0!\ge(k_0/e)^{k_0}$. Hence
\begin{equation}
 e^a\sum_{k\ge k_0}\frac{a^k}{k!}
 \le\frac{11}{10}\exp\!\bigl((12-11\log11)a\bigr)<1.
\end{equation}
Thus $|g_i|\le1/\ell$ and $|g_i'|\le2/\ell$ on that interval.
Combining this with \eqref{opt:eq:gamma-norms}, for the shifted functions
on $t\ge0$, yields
\begin{align}
 \norm{g_i(t-a)}_2&\le2\ell^{-1/2},\notag\\
 \norm{t g_i(t-a)}_2&\le(2a+D_*+2)\ell^{-1/2},\notag\\
 \norm{t g_i'(t-a)}_2&\le2(2a+D_*+2)\ell^{-1/2}.
\end{align}
Using \eqref{opt:eq:phi}, $\ell\ge a$, and \eqref{opt:eq:degree27}, we obtain
\begin{equation}\label{opt:eq:phi-norms}
 \norm\phi_2\le4a^{-3/2},\qquad
 \norm{t\phi}_2\le60a^{-1/2},\qquad
 \norm{t\phi'}_2\le120a^{-1/2}.
\end{equation}
The integral representation in Eq.~\eqref{opt:eq:psi}
now bounds $\norm\psi_2$ by $4\norm\phi_2$: translating
a function to the left on the half-line cannot increase its
$L^2$ norm, and the integrating weight $e^{-v_{\rm L}}$
has unit mass. To obtain the weighted coefficient bound, we
apply the Laguerre differential operator to $\psi$.
From $\psi'=\psi-4\phi$ we get
$\psi''=\psi-4\phi-4\phi'$, and substitution into
$\mathfrak L=-t\partial_t^2-\partial_t+t$ gives
\begin{equation}
 \mathfrak L\psi=4t\phi+4t\phi'-\psi+4\phi,
 \qquad \norm{\mathfrak L\psi}_2\le752a^{-1/2}.
\end{equation}
We can finally return from function norms to the quantity
charged by LCU. Since $\mathfrak L\ell_k=(2k+1)\ell_k$
and the Laguerre functions are orthogonal with squared norm
$1/2$, the finite expansion of $\psi$ gives
$\norm{\mathfrak L\psi}_2^2
=\tfrac12\sum_k(2k+1)^2|\beta_k|^2$.
This controls the factor $k+1$ in each clock weight because
$k+1\le2k+1$. Applying Cauchy--Schwarz to the finite sum of
absolute weights therefore gives
\begin{equation}
 \sum_{k=0}^{n-1}|\lambda_{k+1}|
 \le\frac{\sqrt n}{2}
      \left(\sum_{k=0}^{n-1}(k+1)^2|\beta_k|^2\right)^{1/2}
 \le\sqrt{n/2}\,\norm{\mathfrak L\psi}_2.
\end{equation}
This is the coefficient norm charged by LCU. Substituting the proved
analytic estimates gives
\begin{equation}
 \sum_{j=1}^n|\lambda_j|
 \le\sqrt{\frac{n}{2}}\,\norm{\mathfrak L\psi}_2
 \le752\sqrt{27/2}<3000.
\end{equation}
The two estimates compensate exactly at the level of their
parameter dependence: there are $n=\mathcal O(a)$ clock
lengths, giving a factor $\sqrt n=\mathcal O(\sqrt a)$,
while the weighted function norm is
$\mathcal O(a^{-1/2})$. Consequently the total absolute
weight remains bounded independently of the horizon. Together
with the exact pairing calculation and $n\le27a$, this proves
all parts of the theorem. We now have finite reuse weights
that cancel every required low-order term at constant LCU
normalisation; the next task is to bound the operator tail
and replace the analytic transfer by its finite circuit.
\end{proof}

\subsection{From moment cancellation to a finite simulation circuit}
\label{opt:sec:tail}

The finite functional now has the cancellation properties we
need, and its total absolute weight is bounded independently
of the reuse horizon. We still have to turn those facts into
a simulation guarantee. There are three remaining obligations:
bound the terms beyond the cancellation order, replace the
analytic transfer by the finite cascade, and turn the selected
block into a complete output state with the prescribed error.
We address them in that order so that each error estimate refers
to an operation already available to the construction.

For the first obligation, estimating every ordered product in
$C$ and $D$ separately would discard the positivity that made
the whole transfer contractive. We instead scale both operators
by a single variable. The finite coefficient functional then
becomes a polynomial in that variable, with a uniform bound
on the real interval $[0,1]$. A complex-variable estimate will
convert this bound on the whole polynomial into a bound on
its high-order coefficients at the smaller physical weight
$\delta$. This is how we retain cancellation without ever
commuting $C$ and $D$.

To keep the two costs visible, write $A_*:=3000$ for the absolute
weight bound and $K_*=27$ for the clock-length constant. Let the weights of
Theorem~\ref{opt:thm:weights} have horizon $n$, and define
\begin{equation}
 \mathcal Q(C,D):=\sum_{j=1}^n\lambda_j\Ces j(F),
 \qquad {T}=aq=q^2M,\qquad b=T\delta,
\end{equation}
with $F$ as in \eqref{opt:eq:continuum}. This $\mathcal Q$ is the analytic
approximation; later in this subsection we bound the error of
implementing its weights with $F_{J_s}$.

\begin{lemma}\label{opt:lem:tail}
If $0\preceq D\preceq\delta I$, $-D\preceq C\preceq D$, and
$M\ge2K_*^2b$, then
\begin{equation}\label{opt:eq:factorial-tail}
 \norm{\mathcal Q(C,D)-e^{-\ii TC}}
 \le2(A_*+1)\left(\frac{e^2K_*^2b}{M}\right)^M.
\end{equation}
\end{lemma}

The estimate has the fixed prefactor $2(A_*+1)$ and tail factor
$(e^2K_*^2b/M)^M$. To make the tail of order $\eps$, or of a fixed
power of $\eps$ needed for the final norm guarantee, it suffices to
make $M\log(M/(e^2K_*^2b))$ a sufficiently large multiple of
$\log(1/\eps)$. This gives
\begin{equation}
 M=\mathcal O\!\left(b+
 \frac{\log(1/\eps)}{\log(e+\log(1/\eps)/b)}\right).
\end{equation}
The scale inside the logarithm is $b$, whereas for
$T=q^2M$ with $q\ge1$ the reuse horizon, and hence
the factor-query cost, is $\mathcal O(\sqrt{TM})$.
We handle $T<M$ by attenuation when choosing the final parameters.
Here attenuation reduces the successful preparation amplitude by a
known rotation, while retaining the remaining amplitude in a failure
subspace.

\begin{proof}
We first separate the size of the successful weight from the
operators' shape. Write $C=\delta C_0$ and $D=\delta D_0$,
so $0\preceq D_0\preceq I$ and $D_0\pm C_0\succeq0$.
The symbols $C_0,D_0$ denote system operators only within
this proof. Define
$\mathcal P(x)=\mathcal Q(xC_0,xD_0)$; our desired
approximation is $\mathcal P(\delta)$, whereas the larger
interval $0\le x\le1$ will provide the bound needed to
estimate its coefficients. We will show that the coefficients
below order $M$ agree with those of $e^{-\ii TxC_0}$,
then bound the remaining tails of both expressions.

\emph{Matching the low orders.}
The finite clock is essential here. Since $\xi(0)=1$,
we have $R(0)=0$ and hence $R(w)=\mathcal O(w)$.
Every additional factor from the operator exponential
therefore contributes at least one power of $w$. The
coefficient $[w^k]F$ has total degree at most $k$ in
$C,D$, and the functional uses only coefficients up to
$w^{n-1}$. Consequently $\mathcal P$ is a polynomial
of degree at most $n-1$, despite the exponential used
to describe $F$. At homogeneous degree $j$ in the
operators, its source term is
\begin{equation}
 \frac{(-a)^j}{j!}VR^j(\ii q C+q^2\xi D)^j.
\end{equation}
We now apply the moments to this expression before taking
norms. Each ordered product containing $m$ copies of $D$
has the scalar multiplier $VR^j\xi^m$, together with
the indicated powers of $a,q,\ii$. For $j<M$, the
functional sends that multiplier to zero if $m>0$ and
to one if $m=0$. Only the pure product $C^j$ remains.
After setting $C=xC_0$, its coefficient is
$(-\ii aqC_0)^j/j!$, which is the target coefficient
because $aq=T$. Thus
\begin{equation}\label{opt:eq:matching-jets}
 [x^j]{\mathcal P}(x)=\frac{(-\ii TC_0)^j}{j!}\qquad(0\le j<M).
\end{equation}

\emph{Bounding the whole polynomial.}
Exact matching leaves a tail, but it does not yet bound it:
individual coefficients of $\mathcal P$ may be large.
Our next step is to obtain a bound on the whole polynomial
for real $x\in[0,1]$. Scaling both operators by such
an $x$ preserves the required phase decomposition. In
particular, the three weights
\begin{equation}
 I-xD_0,\qquad \frac{x(D_0+C_0)}2,\qquad
 \frac{x(D_0-C_0)}2
\end{equation}
are positive and sum to $I$. By \eqref{opt:eq:omega}, the corresponding
generator has nonnegative real part throughout $|w|<1$.
The same squared-norm derivative used after
Eq.~\eqref{opt:eq:continuum} therefore makes its exponential contractive. For $0<\varrho<1$, the positive normalised Fej\'er kernel
\begin{equation}
 K_j(\theta)=\frac1j\abs{\sum_{k=0}^{j-1}e^{\ii k\theta}}^2
\end{equation}
satisfies the exact finite-radius identity
\begin{equation}\label{opt:eq:fejer-radius}
 \frac1{2\pi}\int_{-\pi}^{\pi}
 K_j(\theta)F(\varrho e^{\ii\theta};xC_0,xD_0)\,d\theta
 =\sum_{k=0}^{j-1}\left(1-\frac{k}{j}\right)
       \varrho^k[w^k]F(w;xC_0,xD_0).
\end{equation}
To verify this identity, expand the analytic transfer on the
interior circle into its absolutely convergent Taylor series.
The Fourier coefficient of the Fej\'er kernel at an integer
$k$ with $|k|<j$ is $1-|k|/j$, and it vanishes for
$|k|\ge j$. Integration therefore retains exactly the
finite weighted sum on the right. On the left, we average
contractions with a nonnegative scalar weight of integral
$2\pi$, so the operator norm is at most one.

We can now let $\varrho\uparrow1$ in the finite sum.
There is no need to assume that the analytic transfer has
regular boundary values: the right side has finitely many
fixed Taylor coefficients and tends directly to
$\Ces j(F)$. Finally, summing with the signed weights
and using their absolute sum gives
\begin{equation}\label{opt:eq:interval-bound}
 \norm{\Ces j(F(\,\cdot\,;xC_0,xD_0))}\le1,
 \qquad \norm{{\mathcal P}(x)}\le A_*\quad(0\le x\le1),
\end{equation}
This scaling family is used to bound a polynomial in the analysis.
It does not require oracle access to additional, rescaled inputs.

\emph{Converting the interval bound into a tail.}
A bound on a real interval is not itself a Cauchy coefficient
bound, because Cauchy's formula requires a complex contour.
We therefore use a change of variable that turns the interval
into a circle. The map $z=(2+u+u^{-1})/4$ sends
$|u|=1$ onto $[0,1]$; indeed, for $u=e^{\ii\theta}$
it gives $z=(1+\cos\theta)/2$. This is the Joukowski
map in the normalisation appropriate to our interval.
If the degree of $\mathcal P$ is $d$, multiplying
$\mathcal P((2+u+u^{-1})/4)$ by $u^{-d}$ removes
its possible growth at infinity. The exterior maximum
principle then propagates the interval bound to complex
points, yielding
\begin{equation}\label{opt:eq:joukowski}
 \sup_{|z|=\rho}\norm{{\mathcal P}(z)}
 \le A_*\exp\!\bigl(2d\,\arsinh\sqrt\rho\bigr)
 \le A_*e^{2d\sqrt\rho}.
\end{equation}
More explicitly, for any pair of unit vectors we apply the
maximum principle to the corresponding matrix element of
$u^{-d}\mathcal P((2+u+u^{-1})/4)$ on $|u|\ge1$.
It is analytic there, including at infinity, and has modulus
at most $A_*$ on the boundary. Thus
$\norm{\mathcal P(z)}\le A_*|u|^d$ for the root of
$u+u^{-1}=4z-2$ lying outside the unit circle. When
$|z|=\rho$, the quadratic formula and the triangle
inequality give
\begin{equation}
 |u|\le1+2\rho+2\sqrt{\rho(1+\rho)}
       =(\sqrt{1+\rho}+\sqrt\rho)^2.
\end{equation}
The argument applies to every matrix element between unit vectors,
so it gives the stated operator-norm bound.

We now choose the complex-circle radius to balance this
growth against the decay in Cauchy's coefficient estimate.
Using $d<n$, the bound on $|z|=\rho$ is
$A_*e^{2n\sqrt\rho}$. The tail starts at order $M$,
so setting $\rho=(M/n)^2$ makes the exponential factor
$e^{2M}$, while the coefficient estimate contributes
$(\delta/\rho)^M$ at the physical point $x=\delta$.
The hypothesis on $M$ ensures that this point lies well
inside our chosen circle:
\begin{equation}
 \frac{\delta}{\rho}=\frac{n^2\delta}{M^2}
 \le K_*^2q^2\delta=\frac{K_*^2b}{M}\le\frac12.
\end{equation}
Writing ${\mathcal P}(x)=\sum_j{\mathcal P}_jx^j$, Cauchy's estimate gives
$\norm{{\mathcal P}_j}\le A_*e^{2M}\rho^{-j}$. Summing the terms $j\ge M$
at $x=\delta$ therefore gives
\begin{equation}
 \sum_{j\ge M}\norm{{\mathcal P}_j}\delta^j
 \le\frac{A_*e^{2M}(\delta/\rho)^M}{1-\delta/\rho}.
\end{equation}
Using $\delta/\rho\le1/2$ yields
\begin{equation}
 \norm{{\mathcal P}(\delta)-\sum_{j<M}[x^j]{\mathcal P}(x)\delta^j}
 \le2A_*\left(\frac{e^2n^2\delta}{M^2}\right)^M.
\end{equation}
Also $\norm{C_0}\le1$ and $b\le M/2$ imply
\begin{equation}
 \norm{e^{-\ii T\delta C_0}
       -\sum_{j<M}\frac{(-\ii T\delta C_0)^j}{j!}}
 \le\sum_{j\ge M}\frac{b^j}{j!}
 \le2\left(\frac{eb}{M}\right)^M.
\end{equation}
The matched coefficients in Eq.~\eqref{opt:eq:matching-jets}
cancel when we subtract the two approximations. The first
tail is at most
$2A_*(e^2K_*^2b/M)^M$, because
$n\le K_*qM$ and $q^2M=T$. The exponential tail is
at most $2(eb/M)^M$, which is bounded by the same
base raised to $M$. The assumption $M\ge2K_*^2b$
justifies both geometric tail estimates. Adding them gives
the prefactor $2(A_*+1)$ and proves
Eq.~\eqref{opt:eq:factorial-tail}.
\end{proof}

\emph{Returning to an actual finite mesh.}

The tail bound is now available for $\mathcal Q$, the signed
combination formed from the analytic transfer $F$. We must
return to the actual cascade $F_{J_s}$. Recall why we kept
the two finite sizes distinct: making $J_s$ larger refines
the known boundary network, whereas making $n$ larger adds
factor queries through more reuses of the whole interaction.
We can therefore refine the physical mesh until its error
fits the budget without changing the query cost of one
global call.

The finite clock thus supplies $\Ces j(F_{J_s})$ exactly. The
cancellation and tail estimates were established for $F$, so we
must bound the difference between these finite blocks and the analytic
targets $\Ces j(F)$. Only the first $n$ Taylor coefficients enter
the signed combination. We therefore work on an interior circle, where all
resolvent bounds are finite. On $|w|=\rho<1$, the positive weights
in Eq.~\eqref{opt:eq:omega} imply
\begin{equation}
 \norm{\Omega(w)}\le M_\rho:=\frac{1+\rho}{1-\rho},
 \qquad \Reop\Omega(w)\succeq0.
\end{equation}
If $vM_\rho\le1/2$, expansion through quadratic order gives
\begin{equation}
 \norm{u_{J_s}(w)-e^{-v\Omega(w)}}\le v^3M_\rho^3.
\end{equation}
For example, the rational and exponential remainders are bounded by
$(vM_\rho)^3/[4(1-vM_\rho/2)]$ and
$e^{vM_\rho}(vM_\rho)^3/6$, whose sum is smaller than
$(vM_\rho)^3$. Both factors are contractions. For the rational one,
\begin{equation}
 u_{J_s}^\dagger u_{J_s}-I
 =-v\left(I+\frac v2\Omega^\dagger\right)^{-1}
      (\Omega+\Omega^\dagger)
      \left(I+\frac v2\Omega\right)^{-1}\preceq0.
\end{equation}
We can now compare the two full transfers without an
exponential growth factor. Since both one-slot operators
are contractions, telescoping their $J_s$th powers gives
\begin{equation}
 \norm{F_{J_s}(w)-F(w)}
 \le J_s v^3M_\rho^3
 =\frac{a^3M_\rho^3}{J_s^2}.
\end{equation}
For every coefficient with $0\le r<n$, Cauchy's estimate
adds at most the factor $\rho^{-r}\le\rho^{-(n-1)}$.
A reuse functional sums at most $n$ such coefficients
with weights no greater than one, and the outer signed
combination has absolute weight at most $A_*$. We obtain
\begin{equation}
 \norm{\sum_{j=1}^n\lambda_j\Ces j(F_{J_s})-\mathcal Q(C,D)}
 \le A_*n\rho^{-(n-1)}\frac{a^3M_\rho^3}{J_s^2}.
\end{equation}
Taking $\rho=1-1/(2n)$, so $M_\rho\le4n$ and
$\rho^{-(n-1)}<2$, proves
\begin{equation}\label{opt:eq:finite-mesh}
 J_s\ge8an\quad\Longrightarrow\quad
 \norm{\sum_{j=1}^n\lambda_j\Ces j(F_{J_s})-\mathcal Q(C,D)}
 \le\frac{128A_*a^3n^4}{J_s^2}.
\end{equation}
Any positive mesh error budget is therefore met by an explicit finite
$J_s$. Its size affects registers and known gates, but not the number
of factor queries per transducer call.

\emph{Combining and amplifying the finite blocks.}
We now have two independently adjustable error bounds: the
cancellation order controls the analytic tail, and the finite
mesh controls the difference between that analytic expression
and the implemented blocks. To combine the blocks, we use
linear combination of unitaries (LCU). We prepare a label
$j$ with amplitude proportional to $\sqrt{|\lambda_j|}$,
apply the $j$-reuse circuit together with the known sign of
$\lambda_j$, and unprepare the label. Projecting onto its
initial state then multiplies each selected block by
$\lambda_j/\sum_l|\lambda_l|$. The signs are therefore
implemented coherently in the circuit, rather than imposed
by subtracting separately measured outcomes.

Select the finite reuse circuits coherently, pad their query schedules
to length $n$, and encode the signs of $\lambda_j$ by known phases.
On a branch selecting $j$ reuses, each padded step after the $j$th
is the identity, using the controlled construction below.
Linear combination of unitaries gives the successful block
\begin{equation}
 \frac{\mathcal Q_{J_s}}{A_{\rm act}},\qquad
 \mathcal Q_{J_s}=\sum_{j=1}^n\lambda_j\Ces j(F_{J_s}),\qquad
 A_{\rm act}=\sum_j|\lambda_j|\le A_*.
\end{equation}
The zeroth moment and $|\Ces j(V)|\le1$ imply $A_{\rm act}\ge1$.
Pad the normalisation by a known attenuation to $A_*$.
Finite known clock amplitudes can be approximated to any required
accuracy in this query model.

The selected circuits also require controls and inverses, so
we check that these do not change the access assumptions.
To control the private operation, we apply $U_B$, enable
the known Cayley phase only on the active branch, and apply
$U_B^\dagger$. On an inactive branch the two factor
operations cancel exactly. To obtain the adjoint, we replace
the middle phase by its inverse. All boundary operations are
known and can be controlled directly. Thus a selected reuse
call, its inverse, and the inactive steps used to pad shorter
branches each fit the constant factor-query cost model,
without a controlled version of the opaque preparation
unitary being supplied separately.

The successful block after normalisation is close to a
unitary divided by a fixed constant. We next amplify that
block, keeping the effect on the other output registers in
the error analysis.

Choose an odd integer $k_o=\mathcal{O}(A_*)$ such that
\begin{equation}
 s_o=\sin\!\left(\frac{\pi}{2k_o}\right)\le1/A_*.
\end{equation}
A further known attenuation gives successful block $s_o\mathcal Q_{J_s}$.
The usual $(k_o-1)/2$ oblivious amplification rounds use $k_o$ calls
to this circuit or its inverse. For an exact unitary target their
successful block is the target, after a known phase correction.

The amplification identity follows from a two-dimensional rotation
for each singular value of the successful block. If that singular
value is $y=\sin\theta$, one circuit call has successful amplitude
$\sin\theta$ and orthogonal amplitude $\cos\theta$. In this basis
the product of the two reflections in an amplification round is
\begin{equation}
 \begin{pmatrix}
 \cos(2\theta)&\sin(2\theta)\\
 -\sin(2\theta)&\cos(2\theta)
 \end{pmatrix},
\end{equation}
after a known phase choice. It maps
$(\sin\theta,\cos\theta)$ to $(\sin3\theta,\cos3\theta)$.
After $j$ rounds the successful amplitude is
$\sin((2j+1)\theta)$, with an overall known phase removed.
Applying this calculation in every singular-vector subspace gives
$f_{k_o}(y)=\sin(k_o\arcsin y)$ for odd $k_o=2j+1$.

\emph{Including leakage in the error guarantee.}
The preceding rotation calculation gives exact success for
an exactly unitary target at amplitude $s_o$. Our finite
construction is only approximately unitary, so we now bound
what amplification does to both its desired component and
its leakage. Let $\widetilde{\mathcal Q}$ be the implemented
unscaled block, and suppose
$\norm{\widetilde{\mathcal Q}-U}\le\epsilon_0\le1/2$,
where $U=e^{-\ii TC}$. Write the polar decomposition as
$\widetilde{\mathcal Q}=W_p|\widetilde{\mathcal Q}|$. Here
$|\widetilde{\mathcal Q}|=(\widetilde{\mathcal Q}^\dagger
\widetilde{\mathcal Q})^{1/2}$ is the positive factor, whose
eigenvalues are the singular values of $\widetilde{\mathcal Q}$.
The singular values are in $[1-\epsilon_0,1+\epsilon_0]$,
so the block is invertible and $W_p$ is unitary. Moreover,
$\norm{|\widetilde{\mathcal Q}|-I}\le\epsilon_0$, and
the triangle inequality gives
\begin{equation}
 \norm{W_p-U}
 \le\norm{W_p-\widetilde{\mathcal Q}}
       +\norm{\widetilde{\mathcal Q}-U}
 \le2\epsilon_0.
\end{equation}
Amplification keeps this polar unitary and applies the scalar
rotation formula to each singular value, giving
\begin{equation}
 E_{\rm amp}=W_p f_{k_o}(s_o|\widetilde{\mathcal Q}|),\qquad
 f_{k_o}(y)=\sin(k_o\arcsin y).
\end{equation}
On $[0,3/4]$, $|f_{k_o}'|\le2k_o$; also $k_os_o\le\pi/2$.
It follows that
\begin{equation}
 \norm{E_{\rm amp}-U}\le6\epsilon_0.
\end{equation}
Unitarity of the complete circuit $\mathcal A$ then gives, for every
unit vector $\psi$,
\begin{equation}\label{opt:eq:amplified-full}
 \norm{\mathcal A\mathsf E\psi-\mathsf E U\psi}^2
 =2-2\Reop\langle U\psi,E_{\rm amp}\psi\rangle
 \le12\epsilon_0.
\end{equation}
Allocate at most $\eps^2/144$ each to the polynomial, mesh, and
preliminary implementation errors. Their sum is at most $\eps^2/48$,
so \eqref{opt:eq:amplified-full} is at most $\eps/2$ in norm. Approximate
the remaining known gates to total unitary error at most $\eps/2$.
This proves an error on all output registers of $\eps$ with $\cO(n)$ factor queries,
without postselection.
Here the preliminary implementation budget is measured on the
unscaled block $\widetilde{\mathcal Q}$. It suffices to approximate its
attenuated circuit to operator-norm error at most
$s_o\eps^2/144$, since dividing its successful block by $s_o$
then gives the required $\eps^2/144$ budget. The factor $1/s_o$
is an absolute constant. The squared precision budgets are deliberate:
\eqref{opt:eq:amplified-full} converts a successful-block error
$\epsilon_0$ into an error on all output registers of order $\sqrt{\epsilon_0}$.

We have now established a finite circuit once we choose the
cancellation order $M$ and the phase parameter $q$. The
tail estimate dictates how large $M$ must be. The relation
$T=aq=q^2M$ then suggests $q=\sqrt{T/M}$, with reuse
horizon of order $\sqrt{TM}$. This choice is allowed when
$T\ge M$, because the weights theorem assumes $q\ge1$.
When $T<M$, we will attenuate the supplied successful
factor and increase its access normalisation by reciprocal
amounts. That keeps the Hamiltonian and the tail scale $b$
unchanged while allowing us to use $q=1$.

\begin{proof}[Proof of Theorem~\ref{opt:thm:signed}]
The identity case follows from $\norm{e^{-\ii tH}-I}\le b$.
Otherwise choose an integer
\begin{equation}\label{opt:eq:order-choice}
 M=\cO(b+r_b)
\end{equation}
large enough that $M\ge2K_*^2b$ and \eqref{opt:eq:factorial-tail} is at
most $\eps^2/144$. To see this order explicitly, put
$B_{\rm tail}=e^2K_*^2b$ and $\mathcal L=\log(288(A_*+1)/\eps^2)$.
A sufficiently large absolute multiple of
\begin{equation}
 B_{\rm tail}+\frac{\mathcal L}{\log(e+\mathcal L/B_{\rm tail})}
\end{equation}
makes $M\log(M/B_{\rm tail})\ge\mathcal L$.
This is $\cO(b+r_b)$. Integer ceilings are harmless: since
$b>\eps$ and $\log(1/\eps)\ge\log64$, we have
\begin{equation}\label{opt:eq:rb-lower}
 \log\!\left(e+\frac{\log(1/\eps)}b\right)
 <\log\!\left(e+\frac{\log(1/\eps)}\eps\right)
 \le2\log(1/\eps),\qquad r_b\ge\frac12.
\end{equation}
The last logarithmic inequality follows from
$e\eps+\log(1/\eps)\le1/\eps$ for $\eps\le1/64$.

If $T\ge M$, choose $q=\sqrt{{T}/M}$ and $a=qM$. The preceding
construction uses $\cO(qM)=\cO(\sqrt{Tb}+\sqrt{Tr_b})$ queries.

If $T<M$, use a known ancillary rotation to attenuate the supplied
factor. Replace $B$ by
$B'=\sqrt{T/M}\,B$, and replace $\lambda$ by
$\lambda'=\lambda M/T$. Concretely, tensor $U_B$ with a known
rotation whose success amplitude is $\sqrt{T/M}$, and require that
flag as part of the new success projection. On that success space use
$Z'=Z\otimes I$, restricted to the successful flag. Then
\begin{equation}
 \lambda'B'^\dagger Z'B'=H,\qquad
 {T}'=M,\qquad \delta'=(T/M)\delta,\qquad {T}'\delta'=b.
\end{equation}
The new factor and its adjoint cost one original factor query each.
Apply the same construction with $q=1$, $a=M$. Its cost is
$\cO(M)=\cO(b+r_b)$, which is within \eqref{opt:eq:signed-bound} because
$b\le {T}$ implies $b\le\sqrt{Tb}$. Negative time uses the inverse
of the circuit for positive time. This proves the theorem in every case.
\end{proof}

\begin{remark}
The constants in the explicit weight and amplification construction
are large. This theorem determines the asymptotic factor-query count.
Proposition~\ref{gate:prop:signed} supplies its additional-gate and
workspace bounds, including the finite mesh and rounded coefficient
preparation.
\end{remark}

\subsection{Preparing signed factors from sparse entries}\label{opt:sec:native}
The signed-factor theorem is now proved, but its access model
still supplies a preparation unitary. We next build that
preparation from the original sparse locators and entry
queries. Our aim is to encode each entry and its conjugate
in two orientations of an edge. The signs of a local
two-dimensional factor will recover the matrix entry, while
its total successful weight will recover the absolute value
of that entry. Averaging over the row and column lists then
gives the two system compressions required by the theorem.

The absolute row sums that appear in the resulting bound
will be analysis quantities. We do not compute them or
assume an oracle that prepares states normalised by them.
Likewise, the square roots below act only on a known
two-dimensional matrix determined by one queried entry.

On the edge space with basis
$\ket{i,j}$, let
\begin{equation}
 L_{\rm row}\ket i=\sum_{j\text{ in row list }i}\ket{i,j},\qquad
 R_{\rm col}\ket j=\sum_{i\text{ in column list }j}\ket{i,j},\qquad
 J_0=\frac1{\sqrt{2s}}\binom{L_{\rm row}}{R_{\rm col}}.
\end{equation}
The orientation bit distinguishes the row and column summands.
For a matrix $X$ on the supplied support, write
\begin{equation}
 Y_X:=\sum_{i,j}\frac{X_{ij}}\eta\ket{i,j}\bra{i,j},\qquad
 E_X:=\begin{pmatrix}0&Y_X\\Y_X^\dagger&0\end{pmatrix}.
\end{equation}
If $X$ is Hermitian, direct compression gives
\begin{equation}\label{opt:eq:edge-identities}
 J_0^\dagger E_XJ_0=\frac{X}{s\eta},\qquad
 J_0^\dagger|E_X|J_0=\frac{D_X}{s\eta},\qquad
 D_X:=\diag\!\left(\sum_j|X_{ij}|\right).
\end{equation}
We can see the first identity by expanding the two orientations:
\begin{equation}
 J_0^\dagger E_XJ_0
 =\frac{L_{\rm row}^\dagger Y_XR_{\rm col}
       +R_{\rm col}^\dagger Y_X^\dagger L_{\rm row}}{2s}.
\end{equation}
The first term in the numerator is $X/\eta$, and the
second is $X^\dagger/\eta=X/\eta$. For the absolute
value, each edge block has absolute value
$|X_{ij}|I/\eta$, so compression gives the average of
the absolute row and column sums. Hermiticity makes those
sums equal at each system index, producing $D_X/(s\eta)$.
This calculation includes diagonal entries; a padded edge
with zero entry contributes zero, so unmatched zero padding
does not change the identities.
Also, if the row and column Euclidean norms of $X$ are at most $\eta$,
\begin{equation}\label{opt:eq:edge-second}
 J_0^\dagger E_X^2J_0\preceq\frac1sI.
\end{equation}

On an edge with $z_{ij}=X_{ij}/\eta$ and $|z_{ij}|\le1$, put
\begin{equation}\label{opt:eq:edge-factors}
 E(z_{ij}):=\begin{pmatrix}0&z_{ij}\\\overline z_{ij}&0\end{pmatrix},\qquad
 B_\pm(z_{ij}):=\sqrt{\frac{|z_{ij}|I\pm E(z_{ij})}2}.
\end{equation}
The map on that orientation space
\begin{equation}
 v\longmapsto
 \ket+ B_+(z_{ij})v+\ket- B_-(z_{ij})v
       +\ket{\mathrm{fail}}\sqrt{1-|z_{ij}|}\,v
\end{equation}
is an isometry because the sum of the three output Gram
matrices is
$B_+^\dagger B_++B_-^\dagger B_-+(1-|z_{ij}|)I=I$.
Here $\ket+$, $\ket-$ and $\ket{\mathrm{fail}}$ are three
mutually orthogonal labels in a separate output register.
The two successful terms add to $|z_{ij}|I$ and differ
by $E(z_{ij})$. We can therefore extend this map to a
known finite-dimensional unitary controlled by the encoding
of the queried entry. Composing it with $J_0$, retaining
the labels $+$ and $-$ as success, and applying
$Z=\diag(I,-I)$ on those labels gives a factor $B_X$
whose signed and total compressions are
\begin{equation}\label{opt:eq:native-factor}
 B_X^\dagger ZB_X=\frac{X}{s\eta},\qquad
 B_X^\dagger B_X=\frac{D_X}{s\eta}.
\end{equation}
To realise this operation coherently, we first use the
supplied row or column locator, selected by the orientation,
to convert a uniform list position into the corresponding
edge label. For a fixed edge $(i,j)$, let $w_{ij}$ be
the bit string encoding the original entry, and let
$U(z_{ij})$ denote the known local unitary extending the
isometry above. The value-processing part acts as
\begin{equation}
 \ket{i,j}\ket\chi\ket0
 \longmapsto\ket{i,j}\ket\chi\ket{w_{ij}}
 \longmapsto\ket{i,j}U(z_{ij})\ket\chi\ket{w_{ij}}
 \longmapsto\ket{i,j}U(z_{ij})\ket\chi\ket0.
\end{equation}
Here $\ket\chi$ includes the orientation and the initially
clean success-label register. The three steps are the value
query, the encoding-controlled local operation, and the
inverse value query. The local operation preserves both
the edge label and the entry encoding, so uncomputation
works on every branch of a superposition. Whenever
$X_{ij}$ is a known function of the original entry value,
the same queried encoding determines $z_{ij}$ and hence
this local operation. The complete preparation and its
adjoint therefore use a constant number of sparse queries.
All matrix square roots in Eq.~\eqref{opt:eq:edge-factors}
are functions of a single two-dimensional entry block;
no row sum in $D_X$ is evaluated by the circuit.

The same preparation supports the controls required by reuse.
The orientation preparation uses the in-place locators and their
inverses. Their phase-one basis action permits the controlled
construction of Section~\ref{hs:sec:control}, with constant primitive
query cost for every allowed unitary extension. Value controls use
compute, conditional copy and uncompute. The entry-conditioned local
factor then has the same constant-query cost with controls or
adjoints. In the private Cayley operation itself, an external control
need act only on the known middle phase: the unconditional factor
preparation and its inverse cancel when that phase is disabled.

A small example makes the distinction between the signed
compression and its positive weight concrete. Take
$\eta=1$, $s=2$, and
\begin{equation}
 H=\frac12\begin{pmatrix}1&1\\1&0\end{pmatrix}.
\end{equation}
The row norms are $1/\sqrt2$ and $1/2$, so the promise holds.
Both lists may be padded to length two. The matrix is indefinite
because its determinant is $-1/4$. The entrywise construction gives
\begin{equation}
 C=B_H^\dagger ZB_H=\frac14
       \begin{pmatrix}1&1\\1&0\end{pmatrix},\qquad
 D=B_H^\dagger B_H=
       \begin{pmatrix}1/2&0\\0&1/4\end{pmatrix},\qquad H=2C.
\end{equation}
Its plus and minus successful Gram matrices are
\begin{equation}
 G_+=\frac{D+C}{2}=\frac18
       \begin{pmatrix}3&1\\1&1\end{pmatrix},\qquad
 G_-=\frac{D-C}{2}=\frac18
       \begin{pmatrix}1&-1\\-1&1\end{pmatrix}.
\end{equation}
They add to $D$ and differ by $C$. Notice that
\begin{equation}
 CD-DC=\frac1{16}\begin{pmatrix}0&-1\\1&0\end{pmatrix}\ne0.
\end{equation}
Thus the signed-factor theorem must work without commuting its two
compressed operators.

The preparation itself uses only individual entries. Every nonzero
entry has $z_{ij}=1/2$. With
$X_2=\left(\begin{smallmatrix}0&1\\1&0\end{smallmatrix}\right)$,
its local successful factors are
$B_\pm(z_{ij})=(I\pm X_2)/(2\sqrt2)$ and its failure amplitude is
$1/\sqrt2$. The zero entry has only failure amplitude. The row sums
used above to check $D$ describe the resulting operator; the circuit
does not compute them.

We can now apply the signed-factor theorem to the original
Hamiltonian. The row/column Euclidean norm promise controls
the successful weight because each list has length at most
$s$. In each row, Cauchy--Schwarz gives
\begin{equation}
 \sum_j|H_{ij}|
 \le\sqrt{s}\left(\sum_j|H_{ij}|^2\right)^{1/2}
 \le\sqrt{s}\eta.
\end{equation}
Substituting this into the native compression identities yields
\begin{equation}\label{opt:eq:whole-factor}
 D_H\preceq\sqrt{s}\eta I,
 \qquad H=s\eta B_H^\dagger ZB_H,
 \qquad B_H^\dagger B_H\preceq s^{-1/2}I.
\end{equation}
For the remaining sparse-parameter calculation, write
$\tau:=\eta|t|$. Theorem~\ref{opt:thm:signed} is therefore available
with $T=s\tau$ and $b=\sqrt{s}\tau$.

\subsection{A clipped representation and its finite filter}\label{opt:sec:core}

The native factor gives a finite simulation with access scale
$s\tau$ and tail scale $\sqrt{s}\tau$. We now ask when
that tail scale is small enough to produce the optimal
precision dependence. For the rest of the proof, it is useful
to retain the logarithmic quantities
\begin{equation}
 L:=\log(1/\eps),\qquad h:=\log(e+L/\tau),\qquad r:=L/h.
\end{equation}
Thus $L$ is the precision logarithm, $h$ is the denominator
logarithm in the desired bound, and $r=L/h$ is the improved
cancellation scale. At sufficiently large sparsity, the whole-matrix tail scale
need not be small enough for the query count we seek. We develop
a second construction for the regime where $h$ is large and
$s\tau\ge Lh$, using a signed factor only for a correction
with a smaller positive weight.

The construction has two connected parts. We first reduce
the largest entries while keeping the row and column norm
promise. The difference from $H$ is retained exactly and
has absolute row sums bounded by $h\eta$. We then put
the reduced matrix and the correction into one representation
of $H$. This produces a private spectrum with a broad part
coming from the reduced entries and only two additional
phases coming from the signed correction. A filter can
suppress the broad part geometrically while placing exact
high-order zeros at the two correction phases. We prove
the resulting query bound in the following proposition.

\begin{proposition}[The large-sparsity construction]
\label{opt:prop:core}
Let $\tau=\eta|t|\ge1$ and $0<\eps\le1/64$, as in
Theorem~\ref{hs:thm:optimal}. For $h$ above a sufficiently large
absolute constant, if
\begin{equation}\label{opt:eq:core-regime}
 s\tau\ge Lh,
\end{equation}
the original sparse oracles give a simulation with
error on all output registers at most $\eps$ and query cost
$O(\sqrt{s\tau r})$.
\end{proposition}
\begin{proof}
For the construction assume $t>0$; negative time uses the inverse circuit.
Constants denoted by $C$ are absolute and may increase between
estimates.

\emph{Constructing the joint representation.}
We first need an exact decomposition whose two parts have
different useful bounds. The part used directly in the
entry Cayley operation should have uniformly small entries.
The part represented by the signed factor should have
small total absolute row weight. A smooth reduction of
each entry supplies both properties: we divide the entry
by $1+h|H_{ij}|/\eta$, and retain the difference rather
than discarding it. Define
\begin{equation}\label{opt:eq:clip}
 (H_c)_{ij}:=\frac{H_{ij}}{1+{h}|H_{ij}|/\eta},\qquad
 C_H:=H-H_c,\qquad D_c:=\diag\!\left(\sum_j|(C_H)_{ij}|\right).
\end{equation}
The clipped matrix is Hermitian, has the supplied support, and has
row and column Euclidean norms at most $\eta$. Moreover
\begin{equation}
 |(H_c)_{ij}|\le\eta/{h},\qquad
 |(C_H)_{ij}|=
 \frac{{h}|H_{ij}|^2}{\eta+{h}|H_{ij}|},\qquad
 D_c\preceq {h}\eta I.
\end{equation}
For the correction, the denominator is at least $\eta$,
so each absolute entry is at most $h|H_{ij}|^2/\eta$.
Summing along a row and using the Euclidean norm promise
gives at most $h\eta$, which proves the bound on $D_c$.
This is the useful asymmetry of the decomposition: the
reduced matrix has small individual entries, whereas the
correction may have larger entries but has small absolute
row sums. We have not approximated $H$ at this stage;
$H=H_c+C_H$ is an exact identity, so entry reduction
has consumed none of the simulation error budget.
By \eqref{opt:eq:native-factor}, the correction has an entrywise factor $B_c$
with
\begin{equation}\label{opt:eq:correction-factor}
 C_H=s\eta B_c^\dagger ZB_c,
 \qquad B_c^\dagger B_c=\frac{D_c}{s\eta}
 \preceq\frac{{h}} sI.
\end{equation}
Let $J_c$ denote the restriction of the full preparation unitary,
including failure, to initialised system inputs. Thus $J_c$ is an
isometry from the system space into the preparation space, and its
successful component is $B_c$.

We next choose the phase scale. The total coupling will be
$a=s\tau/q$, while the filter degree will be of order
$qr$. Balancing these quantities suggests
$q^2=s\tau/r$, so set $q_*:=\sqrt{s\tau/r}$.
We also want the correction phases to be roots of a known
finite average. For that purpose we choose a nearby scale
$q$ whose two Cayley phases are nontrivial $p$th roots
of unity, with $p$ even:
\begin{equation}\label{opt:eq:phase-choice}
 2q=\cot(\pi/p),\qquad q_*\le q\le2q_*.
\end{equation}
For large $q_*$ this is possible by taking the least suitable even
$p$; consecutive choices change $q$ by a bounded amount.
From \eqref{opt:eq:core-regime}, for sufficiently large $h$,
\begin{equation}\label{opt:eq:core-parameters}
 q\ge {h},\qquad q\le\sqrt{s},\qquad
 a=\frac{s\tau}{q}=\Theta(qr),\qquad
 \frac{q^2}{s}=\mathcal{O}(\tau/r),\qquad \frac aq=\Theta(r).
\end{equation}
The large-$h$ regime makes these choices compatible. In
particular, $r/\tau=(L/\tau)/h=(e^h-e)/h$ grows
faster than every fixed power of $h$. Thus the eventual
upper bounds involving $\tau/r$ become small, while
the hypothesis $s\tau\ge Lh$ ensures that $q$ is
large enough for the separation of phases used below.

We now combine the two pieces before constructing any
evolution. Independently simulating $H_c$ and $C_H$
would introduce another approximation problem because
they need not commute. Instead, we place their
representations in separate summands and arrange that
their first compressed moments add to the actual $H$.
On the direct sum of the reduced-entry orientation space
and the correction preparation space, define
\begin{equation}\label{opt:eq:joint-lift}
 \begin{aligned}
 J&=\frac1{\sqrt2}\binom{J_0}{J_c},\qquad
 D_0=(K+\ii I)(K-\ii I)^{-1},\\
 K&=\diag(2qE_{H_c},\,2qZ\text{ on success},\,0\text{ on failure}).
 \end{aligned}
\end{equation}
The equal boundary amplitudes contribute a factor $1/2$
to each compression, cancelling the factor two in the
definition of $K$. For example, the first moment is
\begin{equation}
 J^\dagger KJ
 =qJ_0^\dagger E_{H_c}J_0+qB_c^\dagger ZB_c
 =\frac{q(H_c+C_H)}{s\eta}
 =\frac{qH}{s\eta}.
\end{equation}
For the second moment, the reduced-entry summand contributes
at most $2q^2/s$ by Eq.~\eqref{opt:eq:edge-second},
and the correction contributes at most $2q^2h/s$ by
Eq.~\eqref{opt:eq:correction-factor}. Consequently
\begin{equation}\label{opt:eq:joint-moments}
 J^\dagger KJ=\frac{H}{\alpha},\qquad
 J^\dagger K^2J\preceq\frac{2q^2(1+{h})}sI,
 \qquad \alpha:=\frac{s\eta}{q}.
\end{equation}
The first identity is the one that will produce the desired
system evolution. The second will control how much
stationary amplitude is needed in the private space. A stationary
private component is one that the interaction returns unchanged while
its public component undergoes the desired unitary action.
To check the scale of the system generator, the row and
column norm promise implies $\norm H\le\sqrt{s}\eta$,
and hence $\norm{H/\alpha}\le q/\sqrt{s}\le1$.
We record the resulting bounds for the squared stationary
private amplitude as
\begin{equation}\label{opt:eq:nu-W}
 \nu=\frac{1+2q^2(1+{h})/s}{2},\qquad W=a\nu,
\end{equation}
Here $\nu$ is a chosen upper bound per unit coupling,
and $W=a\nu$ will bound the squared norm of the
stationary map for the whole cascade. We verify that
interpretation when converting the filter into a circuit.
For sufficiently large $h$, the estimates in
Eq.~\eqref{opt:eq:core-parameters} give $\nu\le1$
and therefore $W\le a$.

The representation is also available at constant sparse-query
cost. We conjugate by the two entrywise preparation
unitaries so that the boundary $J$ is a known preparation.
The private Cayley operation is then implemented by
preparing the chosen summand, applying its known local
phase, and unpreparing. Its inverse and coherent controls
use the same constant-query construction as before. This
is a single representation of $H$, not a product of
evolutions of its noncommuting summands.

When discussing its spectrum, we work on the invariant
space whose value ancilla is clean before this change of
coordinates. Querying the entry, applying the known local
operation, and unquerying returns that ancilla to zero,
so the private operation preserves this space. It contains
$\operatorname{ran}J$ and is preserved by $JJ^\dagger$;
the boundary rotations and hence the full chronological
circuit therefore remain in the same space. All spectral
claims below concern this space, rather than unspecified
behaviour of an oracle extension elsewhere.

For the filter analysis, distinguish the phases contributed by the
clipped entries from the two phases contributed by the correction.
Let $P_0$ be the bulk spectral projection and $P_+,P_-$ the two
isolated-phase projections. On the bulk,
$|K|\le2q/{h}$, including the failure value zero. The isolated
Cayley phases are
\begin{equation}\label{opt:eq:isolated-weights}
 \lambda_+=e^{2\pi\ii/p},\qquad \lambda_-=e^{-2\pi\ii/p},
 \qquad \norm{J^\dagger P_\pm J}\le w_{\rm iso}:=\frac{{h}}{2s},
 \qquad J^\dagger P_0K^2P_0J\preceq\frac{2q^2}sI.
\end{equation}
The projection bound follows because each signed correction factor
has squared norm at most $h/s$, and the direct-sum boundary adds
a factor $1/2$. These projections may depend on the oracle after
conjugation; only their analytic weights are used.
In particular, the construction never needs an oracle for these
spectral projections. They describe how strongly the known input
boundary couples to the corresponding private phases.

\emph{Choosing the filter from the private phases.}
We now put the joint representation into the same finite
chronological interaction used in the signed-factor
construction. Choose $N_s$ physical slots and give each
strength $v:=a/N_s$. Its boundary rotation is
Eq.~\eqref{opt:eq:boundary-rotation} with $E=J$, and
its private unitary is $D_0$. We denote the private block
of the resulting whole interaction by $D_g$. The number
$N_s$ will be fixed by finite-mesh and contour conditions
below. As before, the private operations on all slots
combine into one direct-sum call, so each global interaction
has constant sparse-query cost.

The two known isolated phases suggest the filter directly. Put
\begin{equation}
 Q_0:=\left\lceil\max\{1,2q/h\}\right\rceil,\qquad B_0:=32Q_0,
\end{equation}
and choose positive integers $k,m$, whose sizes will be determined
below. Let
\begin{equation}\label{opt:eq:filter}
 g_j(z):=\frac{1+z+\cdots+z^{j-1}}j,
 \qquad f(z):=g_{B_0}(z)^m g_p(z)^k,
 \qquad d:=m(B_0-1)+k(p-1).
\end{equation}
The average $g_p$ vanishes exactly at both isolated phases because
they are nontrivial $p$th roots of unity. Raising it to the $k$th
power gives a zero of order $k$ at each phase. The other average,
$g_{B_0}$, suppresses the bulk, which lies farther from one.

Both averages have nonnegative coefficients, unit sum, and
coefficients unchanged by reversing their order. Their
product $f$ has the same properties. These facts give a
concrete implementation route: the smoothing identity from
Eq.~\eqref{ov:eq:smoothing} expresses a smoothed version
of $f$ as a bounded signed combination of actual reuse
circuits, after which one oblivious amplification round
produces the desired system operation. We will retain that
identity explicitly at the end of the proof.

The remaining choices of $N_s,k,m$ must achieve two
things. We need $f(D_g)$ to be small, and we need the
filter degree to dominate the squared stationary private
amplitude $W$. It is tempting to estimate $f$ only on
the three parts of the local spectrum, but coupling the
slots to the public space changes the private block to
$D_g$, which need not be normal. For a nonnormal
operator, small values on its eigenvalues do not by
themselves bound the norm of a polynomial in it. We
therefore need a resolvent bound that follows the actual
finite cascade.

\emph{Controlling the finite chronological resolvent.}
Our objective in this part is to bound $(z-D_g)^{-1}$
on contours surrounding the private spectrum. Cauchy's
formula will then bound $f(D_g)$ by integrating the
product of this resolvent and the scalar filter. To retain
the small weights of the correction phases, we need both
the distance to the local spectrum and the norm of the
local resolvent on vectors entering through the boundary
$J$. A bound using spectral distance alone would lose
that weight information.

We derive the estimate in three steps. We first bound
public propagation through consecutive slots. We then
solve for a source in one private slot and follow its
contribution to later slots. Finally, we sum these block
contributions to obtain the norm of the whole inverse.
Every step is at finite $N_s$, so the resolvent we bound
belongs to the circuit already specified.

Take
$N_s$ chronological slots, $v=a/N_s$, and the boundary rotation
\eqref{opt:eq:boundary-rotation} with $E$ replaced by $J$ and $P$ by $D_0$.
Let $D_g$ be the private block of the resulting global transducer.
For $z\notin\spec D_0$, set
\begin{align}
 R_z&:=(z-D_0)^{-1},& F_z&:=J^\dagger D_0R_zJ,\notag\\
 \Omega_z&:=I+2F_z=J^\dagger(z+D_0)(z-D_0)^{-1}J,&
 \rho_z&:=\norm{J^\dagger R_z^\dagger R_zJ},\notag\\
 \Delta_z&:=\dist(z,\spec D_0),&
 m_z&:=(1-|z|^2)_+\rho_z.
\end{align}
The quantity $\Delta_z$ controls the full local inverse,
while $\rho_z=\norm{R_zJ}^2$ measures its squared
norm only on boundary inputs. The latter can be much
smaller because a nearby phase may have little weight
in $\operatorname{ran}J$. The response $\Omega_z$
is the same Cayley transfer response as before, now
written in the resolvent variable $z$ rather than the
delay variable $w$. We use
$(1-|z|^2)_+=\max\{1-|z|^2,0\}$ in $m_z$ to
measure the possible loss of contractivity inside the
unit circle. The following scalar identity explains
exactly why this quantity controls that loss.
For an eigenphase $\lambda$ of the unitary $D_0$,
\begin{equation}
 \Reop\frac{z+\lambda}{z-\lambda}
 =\frac{|z|^2-1}{|z-\lambda|^2}.
\end{equation}
Applying the spectral theorem and compressing through $J$ gives
\begin{equation}\label{opt:eq:real-transfer}
 \Reop\Omega_z=(|z|^2-1)J^\dagger R_z^\dagger R_zJ\succeq-m_zI.
\end{equation}
We first translate Eq.~\eqref{opt:eq:real-transfer} into a
bound on the public amplitude carried between private
slots. Eliminating the private amplitude from one local
interaction gives the same rational transfer as in
Eq.~\eqref{opt:eq:finite-transfer}, namely
\begin{equation}
 u(z)=cI-\beta_v^2F_z(I+(1-c) F_z)^{-1}
     =(I-v\Omega_z/2)(I+v\Omega_z/2)^{-1}.
\end{equation}
Write $G_z:=I+v\Omega_z/2$ for its denominator.
Since $\norm{R_z}\le\Delta_z^{-1}$ and $D_0$ is
unitary, we have
$\norm{\Omega_z}\le1+2/\Delta_z$. Thus the
small-step condition $v(1+2/\Delta_z)\le1/2$
makes $\norm{v\Omega_z/2}\le1/4$, and the
geometric-series inverse bound gives
$\norm{G_z^{-1}}\le4/3$. We do not estimate the
numerator and denominator separately over many steps.
Instead, their difference from unitarity has the exact
form
\begin{equation}
 u^\dagger u-I=-2vG_z^{-\dagger}(\Reop\Omega_z)G_z^{-1}.
\end{equation}
Using $\Reop\Omega_z\succeq-m_zI$ therefore gives
\begin{equation}
 u^\dagger u-I\preceq\frac{32}{9}v m_zI,
 \qquad \norm{u(z)^j}\le\exp\!\left(\frac{16}{9}jv m_z\right).
\end{equation}
The exponent is proportional to the total coupling
$jv$ and to the negative part of the real response,
as measured by $m_z$. In particular, we have not
introduced a quadratic one-step error that would
accumulate with the number of slots. Outside the unit
circle $m_z=0$ and Eq.~\eqref{opt:eq:real-transfer}
is positive, so each public transfer is contractive.

We next solve the private equation in a single slot. The
private diagonal block is
$(I-(1-c)JJ^\dagger)D_0$, which differs from $D_0$
only through the boundary. Inverting this boundary
correction gives the local resolvent
\begin{equation}
 R_{\rm loc}=R_z-(1-c) R_zJ(I+(1-c) F_z)^{-1}J^\dagger D_0R_z.
\end{equation}
The two endpoint maps in the correction are $R_zJ$ and
$J^\dagger D_0R_z$. Because $D_0$ is unitary and
commutes with its resolvent, both have norm
$\sqrt{\rho_z}$. The same small-step inverse bound
therefore controls not only the full local inverse but
also its action at these boundary endpoints:
\begin{equation}
 \norm{R_{\rm loc}}\le\Delta_z^{-1}+C v\rho_z,
 \qquad
 \norm{R_{\rm loc}J},\ \norm{J^\dagger D_0R_{\rm loc}}
 \le C\sqrt{\rho_z}.
\end{equation}
To keep the constant explicit, use $G_z$ as above and set
$A_z=(I+(1-c) F_z)^{-1}$.
Since
\begin{equation}
 I+(1-c) F_z=\frac{G_z}{1+v/2},\qquad
 \norm{G_z^{-1}}\le\frac43,
\end{equation}
the factors from the boundary rotation cancel to give
\begin{equation}
 (1-c)\norm{A_z}\le\frac{4v}{3},\qquad
 \beta_v^2\norm{A_z}^2\le\frac{32v}{9}.
\end{equation}
The endpoint identities are
\begin{equation}
 R_{\rm loc}J=R_zJA_z,\qquad
 J^\dagger D_0R_{\rm loc}=A_zJ^\dagger D_0R_z.
\end{equation}
We now connect this local solution to the chronological
structure. Number the private slots from zero, let $y_j$
be the source supplied to slot $j$ in the equation
$(z-D_g)x=y$, and let $p_{j-1}$ be the public
amplitude arriving from earlier slots during its solution.
There is no incoming public source, so $p_{-1}=0$.
Local elimination gives the coupled recurrence
\begin{equation}
 x_j=R_{\rm loc}y_j+\beta_v R_{\rm loc}Jp_{j-1},\qquad
 p_j=u(z)p_{j-1}-\beta_v J^\dagger D_0R_{\rm loc}y_j,
 \qquad p_{-1}=0.
\end{equation}
This recurrence displays the path of a source explicitly.
Its direct contribution stays in its own slot through
$R_{\rm loc}$. The term
$-\beta_vJ^\dagger D_0R_{\rm loc}y_k$ carries a
source in slot $k$ into the public space. To reach a
later slot $j>k$, it passes through $j-k-1$ public
transfers and then enters that private slot through
$\beta_vR_{\rm loc}J$. Sources cannot travel to
earlier slots. Thus the diagonal inverse blocks are
$R_{\rm loc}$ and the later-slot blocks are
\begin{equation}
 (z-D_g)^{-1}_{jk}
 =-\beta_v^2R_{\rm loc}Ju(z)^{j-k-1}J^\dagger D_0R_{\rm loc}.
\end{equation}
We can now sum the contributions without introducing a
factor proportional to the number of physical slots.
The diagonal norm is at most
$\Delta_z^{-1}+4v\rho_z/3$. Each later-slot block
is bounded by the endpoint estimate
$\beta_v^2\norm{A_z}^2\rho_z\le32v\rho_z/9$
times the propagation factor, which is at most
$e^{16am_z/9}$ throughout the cascade. There are
fewer than $N_s$ such blocks in any row or column,
and $N_sv=a$. Their total is therefore at most
$(32a/9)\rho_z e^{16am_z/9}$.

To pass from these scalar sums to the operator norm,
we use the block Schur estimate: the full operator
norm is at most the geometric mean of the largest
block-row sum and the largest block-column sum of
block norms. This estimate bounds
a block matrix; it is distinct from using a Schur
complement to invert one. Here the same bound holds
for both sums. Enlarging the diagonal contribution
by the nonnegative exponential factor gives
\begin{equation}
 \left(\Delta_z^{-1}
       +\left(\frac{4v}{3}+\frac{32a}{9}\right)\rho_z\right)
 e^{16am_z/9}.
\end{equation}
Since $v\le a$ and $4/3+32/9=44/9<5$, this proves
\begin{equation}\label{opt:eq:global-resolvent}
 \norm{(z-D_g)^{-1}}
 \le5(\Delta_z^{-1}+a\rho_z)e^{2am_z}.
\end{equation}
It also establishes invertibility at these points. The estimate holds
at finite $N_s$.

\emph{Enclosing the spectrum with a finite mesh.}
The resolvent estimate is ready to use once we choose
contours that enclose the spectrum and satisfy its
small-step condition. We use three separate regions:
one for the bulk spectrum and a disk around each
correction phase. The point of keeping the disks
separate is that their small compressed weights and
high-order filter zeros can then be estimated locally.
We will choose a finite mesh fine enough that the
coupling does not move the spectrum out of these
regions.

Define
\begin{equation}\label{opt:eq:bulk-contour}
 Q_0=\ceil{\max\{1,2q/{h}\}},\qquad B_0=32Q_0,
 \qquad \rho_0=\frac1{4Q_0},\qquad R_{\rm out}=1+\frac1{B_0}.
\end{equation}
For a bulk eigenvalue $\kappa$ of $K$, its Cayley
phase has distance $2/\sqrt{1+\kappa^2}$ from one.
The bound $|\kappa|\le2q/h$ and the definition of
$Q_0$ therefore put the bulk spectrum at distance
at least $1/Q_0$ from one. The isolated phases, by
contrast, are at distance of order $1/q$. For large
$h$, they lie in $|z-1|<\rho_0/4$, with separation
$\Theta(1/q)$. The circle $|z-1|=\rho_0$ can
therefore separate the two isolated phases from the
bulk.
For positive integers $k,m$ to be fixed below, set
\begin{equation}\label{opt:eq:zeta}
 \zeta=\frac{c_0{h}\tau}{qk},
\end{equation}
where $c_0=1/100$. This radius balances two effects
that we will quantify below. A smaller disk gives a
smaller filter value near its exact zero, but increases
the local resolvent. The anticipated inequality
$k\ge Ch\tau$ makes $\zeta$ a sufficiently small
multiple of $1/q$. Consequently the two disks
$|z-\lambda_\pm|<\zeta$ are disjoint, lie within
$|z-1|<\rho_0/2$, and lie inside $|z|<R_{\rm out}$.
Use their boundaries together with the boundary of
\begin{equation}
 \Omega_0:=\{z:|z|<R_{\rm out},\ |z-1|>\rho_0\}.
\end{equation}

We must check the spectrum of the coupled block $D_g$,
not just that of $D_0$. The chronological private block
is contractive and block triangular, with the same local
diagonal block $(I-(1-c)JJ^\dagger)D_0$ in every
slot. Its spectrum is therefore the union of those
diagonal-block spectra. The local block differs from
the normal unitary $D_0$ by norm at most $1-c$.
The usual normal-resolvent argument places its eigenvalues
within that distance of $\spec D_0$: outside this
neighbourhood, the perturbation times $(z-D_0)^{-1}$
has norm less than one and the perturbed inverse exists.
We can thus enclose the entire global spectrum by taking
$1-c$ smaller than a fixed multiple of
$\min\{\zeta,\rho_0\}$.

For example, take
\begin{equation}\label{opt:eq:core-mesh}
 N_s\ge\max\!\left\{4a,\frac{8a}{\zeta},24aQ_0,
                         \sqrt{\frac{a^3}{3\eps}}\right\}.
\end{equation}
On either isolated circle $\Delta_z=\zeta$ and $v\le\zeta/8$.
On the inner bulk arc $\Delta_z\ge3/(16Q_0)$ and
$v\le1/(24Q_0)$. Thus the first three terms give the enclosure and
the small-step condition for \eqref{opt:eq:global-resolvent} on both
inner types of contour. The outer arc can instead use contractivity
directly. The geometric inclusions above hold once the fixed threshold
on $h$ is, in particular, larger than $48$.

The same finite mesh must also make the public evolution
accurate. At the boundary value $w=1$, the Cayley
response obeys
$\Omega=\ii J^\dagger KJ=\ii(H/\alpha)$.
We therefore obtain the following exact unitary public
action for the finite cascade:
\begin{equation}
 U_C=\left[(I-\ii v (H/\alpha)/2)(I+\ii v (H/\alpha)/2)^{-1}\right]^{N_s}.
\end{equation}
For an eigenvalue $\lambda_{\rm eig}$ of $H/\alpha$, each factor has phase $-2\arctan(v\lambda_{\rm eig}/2)$. The inequality $|2\arctan(x/2)-x|\le|x|^3/12$ therefore bounds the accumulated phase error by
$N_sv^3|\lambda_{\rm eig}|^3/12$. Since $\norm{H/\alpha}\le1$
and $a(H/\alpha)=tH$, the operator error is at most
$a^3/(12N_s^2)$. The last term in
Eq.~\eqref{opt:eq:core-mesh} makes this at most
$\eps/4$. Thus the chosen finite mesh supplies both
the spectral enclosure needed for the contours and
the desired discretisation accuracy. Its size affects
the known boundary circuit, but not the number of
sparse queries in a global call.

We can now use Cauchy's formula on these finite contours.
Let $\mathcal C$ be the union of the positively oriented
boundary of $\Omega_0$ and the two isolated-phase circles.
The enclosure and invertibility just proved give
\begin{equation}
 f(D_g)=\frac1{2\pi\ii}\int_{\mathcal C}
                 f(z)(z-D_g)^{-1}\,dz,
 \qquad
 \norm{f(D_g)}\le\frac1{2\pi}\int_{\mathcal C}
              |f(z)|\,\norm{(z-D_g)^{-1}}\,|dz|.
\end{equation}
We estimate the two isolated circles and the bulk boundary
separately, then add their contributions in this inequality.

\emph{Suppressing the correction phases.}
We first estimate the Cauchy integrals around the two
correction phases. Near either phase the local resolvent
is large, but the boundary couples to it with squared
weight at most $w_{\rm iso}=h/(2s)$. We keep that
weight in $\rho_z$, then use the order-$k$ zero of
$g_p^k$ to overcome the remaining resolvent growth.
The chosen radius $\zeta$ will make the resulting
bound a high power of $h\tau/k$, rather than merely
geometric in $k$.

For a Cayley transform of $K$, write $\xi_z:=(z-1)/(z+1)$.
Then
\begin{equation}
 R_z=\frac{I+\ii K}{z+1}(I+\ii\xi_zK)^{-1}.
\end{equation}
On a circle $|z-\lambda_\pm|=\zeta$, we have
$|\xi_z|=\mathcal O(1/q)$. For the bulk part of
$K$, this gives $\norm{\xi_zK}=\mathcal O(1/h)<1/2$,
so the final inverse in the displayed formula is
uniformly bounded there. The compressed square of
the numerator is controlled by the bulk zeroth and
second moments, giving a contribution
$\mathcal O(1+q^2/s)$ to $\rho_z$.

There are two further contributions. The other isolated
phase is a distance $\Theta(1/q)$ away, so its
compressed resolvent square is
$\mathcal O(w_{\rm iso}q^2)$. The nearby phase is
exactly a distance $\zeta$ away and contributes at
most $w_{\rm iso}/\zeta^2$. Adding these spectral
pieces gives
\begin{equation}\label{opt:eq:rho-isolated}
 \rho_z\le C(1+q^2/s+ w_{\rm iso} q^2+ w_{\rm iso}/\zeta^2)
          \le C(1+ w_{\rm iso}/\zeta^2).
\end{equation}
Here $ w_{\rm iso} q^2=\mathcal{O}({h}\tau/r)=\mathcal{O}(1)$. Since
$(1-|z|^2)_+\le C\zeta$,
\begin{equation}
 am_z\le C(a\zeta+a w_{\rm iso}/\zeta)\le C k,
 \qquad a w_{\rm iso}/\zeta=\frac{k}{2c_0},\qquad
 a\zeta=\mathcal{O}({h}\tau r/k).
\end{equation}
These bounds identify the benefit of the chosen radius.
The potentially largest contribution to the growth exponent
is $aw_{\rm iso}/\zeta=k/(2c_0)$, which depends on
the filter order rather than on the mesh or sparsity.
The other contribution is
$a\zeta=\mathcal O(h\tau r/k)$, also at most a
constant times $k$ for the orders chosen below.

The contour length is $2\pi\zeta$, and the distance
to the local spectrum on that contour is $\Delta_z=\zeta$.
Multiplying Eq.~\eqref{opt:eq:global-resolvent} by the
length therefore cancels the bare inverse-distance term.
The weighted terms become $a\zeta$ and
$aw_{\rm iso}/\zeta$. The integrated resolvent bound
is consequently at most $C(1+Ck)e^{Ck}$, with no
remaining factor depending on $N_s,q$ or $s$.

The filter has exactly the complementary behaviour on
this contour. Since $\lambda_\pm^p=1$,
\begin{equation}
 1-z^p=\lambda_\pm^p-z^p
 =(\lambda_\pm-z)\sum_{j=0}^{p-1}
                  \lambda_\pm^{p-1-j}z^j.
\end{equation}
Using this in $g_p(z)=(1-z^p)/(p(1-z))$ supplies
one factor $\zeta$ from the zero. The denominator
$|1-z|$ is of order $1/q$, while the remaining sum
is bounded by $p\max\{1,|z|^{p-1}\}$. As $p=\Theta(q)$,
we obtain
\begin{equation}
 |g_p(z)|\le C q\zeta e^{C q\zeta}
 \le\frac{C{h}\tau}{k}\,e^{C{h}\tau/k}.
\end{equation}
A polynomial with nonnegative coefficients summing to one is bounded
by $\max\{1,|z|^{\deg}\}$. Provided $d=\mathcal{O}(qk)$, as verified below,
the other factor is therefore bounded by
$e^{d\zeta}=e^{\mathcal{O}({h}\tau)}$. The Cauchy integral over both
isolated circles is at most
\begin{equation}\label{opt:eq:isolated-error}
 C\left(\frac{C_1{h}\tau}{k}\right)^k.
\end{equation}
Fixed exponential and polynomial factors have been absorbed into
the absolute constant $C_1$.

\emph{Suppressing the bulk.}
For the bulk, geometric suppression is sufficient, but
we must preserve its angular decay on the outer arc.
There are two different resolvent estimates to use.
On the inner arc, which may enter the unit disk, we
use the finite chronological bound just proved. On
the outer arc, contractivity directly bounds the
global resolvent. If we replaced the scalar filter
everywhere by its maximum, the latter bound would
leave an unwanted factor depending on the circle
radius. Keeping two inverse-distance factors in the
filter will cancel it upon integration.

On the inner arc $|z-1|=\rho_0$, write the bulk resolvent as
\begin{equation}
 R_zP_0=
 \left[I+(z-1)(I-D_0)^{-1}\right]^{-1}(I-D_0)^{-1}P_0.
\end{equation}
We use the displayed identity to compare the resolvent
on this inner arc with its value at one. On the bulk,
$\norm{(I-D_0)^{-1}}\le Q_0$ and
$|z-1|Q_0=1/4$, so the first inverse has norm at
most $4/3$. The Cayley identity
$(I-D_0)^{-1}=(I+\ii K)/2$ then shows that its
compressed squared norm depends on the bulk second
moment, not on the largest bulk eigenvalue alone.
That moment gives a contribution
$\mathcal O(1+q^2/s)=\mathcal O(1)$.

The isolated phases lie well inside the removed disk,
at distance $\Theta(\rho_0)$ from the inner arc.
Their combined contribution is
$\mathcal O(w_{\rm iso}/\rho_0^2)
=\mathcal O(q^2/(hs))=\mathcal O(1)$.
Thus both the bulk and correction contributions to
the compressed resolvent are bounded, and
\begin{equation}
 \rho_z\le C,
 \qquad \Delta_z^{-1}\le C Q_0,
 \qquad am_z\le C a/Q_0=\mathcal{O}({h} r).
\end{equation}
The inner arc has length at most $2\pi\rho_0$. Its length times
\eqref{opt:eq:global-resolvent} is thus bounded by
$C(1+C{h} r)e^{C{h} r}$.

On both bulk arcs, $|z|\le R_{\rm out}$ and $|z-1|\ge\rho_0$, so
\begin{equation}\label{opt:eq:bulk-filter-pointwise}
 |g_{B_0}(z)|\le\frac{1+e}{B_0|z-1|}
 \le\frac{1+e}{8}<\frac12,
 \qquad |g_p(z)|\le R_{\rm out}^{p-1}\le e^{C{h}}.
\end{equation}
The first of these estimates suppresses each bulk
filter factor by at least $1/2$; the second allows
growth from the correction-phase filter on this
contour. Combining them with the integrated inner
resolvent bound gives
$C(1+Chr)e^{Chr+Chk}2^{-m}$ for the inner arc.

On the outer arc, $|z|=R_{\rm out}=1+1/B_0$,
contractivity gives
$\norm{(z-D_g)^{-1}}\le1/(|z|-1)=B_0$.
Bounding all $m$ bulk factors by $1/2$ would leave
this factor $B_0$. Instead, for $m\ge2$ we retain
the distance dependence in two of them:
\begin{equation}
 |g_{B_0}(z)|^m
 \le2^{-(m-2)}\frac{(1+e)^2}{B_0^2|z-1|^2}.
\end{equation}
The resulting scalar integrand is nonnegative, so we
can enlarge its integral from the outer arc to the
whole circle. The required identity is
\begin{equation}
 \int_{|z|=R_{\rm out}}\frac{|dz|}{|z-1|^2}
 =\frac{2\pi R_{\rm out}}{R_{\rm out}^2-1}.
\end{equation}
After multiplying by the resolvent bound $B_0$, the
only radius-dependent factor is
\begin{equation}
 \frac{R_{\rm out}}{B_0(R_{\rm out}^2-1)}
 =\frac{B_0+1}{2B_0+1}\le1.
\end{equation}
Thus the outer contribution is at most
$Ce^{Chk}2^{-m}$. We have extended only the
nonnegative inverse-distance integral to the whole
circle; the pointwise bound $|g_{B_0}|<1/2$ was
needed only on the original outer arc. Adding the
inner contribution gives the bulk bound
\begin{equation}\label{opt:eq:bulk-error}
 C(1+C{h} r)e^{C{h} r+C{h} k}2^{-m}.
\end{equation}
This integration step avoids a spurious prefactor proportional to $q$.

\emph{Fixing the two filter orders.}
We can now choose the two orders using the estimates
they must satisfy. The correction-phase error has the
form $(C_1h\tau/k)^k$, so its logarithm involves
$k\log(k/(C_1h\tau))$. For $k$ proportional to
$r=L/h$, the second factor will be of order $h$,
giving the required precision exponent of order $L$.
The bulk error instead asks us to make $m\log2$
larger than the growth terms of order $hr+hk$.
We therefore choose
\begin{equation}\label{opt:eq:filter-orders}
 k=\ceil{C_{\rm order}r},\qquad m=\ceil{C_{\rm bulk}{h} k}.
\end{equation}
The constants are chosen in a definite order. First
take $C_{\rm bulk}$ large enough to make the bulk
bound at most $\eps/16$, using $hr=L$ and $k\ge r$.
This choice also fixes how much the bulk filter can
grow on the isolated circles, and therefore fixes
its contribution to $C_1$ in
Eq.~\eqref{opt:eq:isolated-error}. We then choose
$C_{\rm order}\ge1$ large enough for the isolated
bound. Finally, we choose a fixed lower threshold
on $h$ that makes all required scale separations
valid. This is possible because
\begin{equation}
 \frac{r}{{h}\tau}=\frac{e^{h}-e}{{h}^2}\longrightarrow\infty
\end{equation}
In particular, for that fixed choice of constants,
we can ensure $k\ge C(C_{\rm bulk})h\tau$,
which validates the radii and growth estimates used
above. We also require $e^h-e\ge h^2$ and
$m\ge2$. Increasing the same fixed threshold
further gives
\begin{equation}
 \log\!\left(\frac{k}{C_1{h}\tau}\right)
 \ge\log\!\left(\frac{L}{C_1{h}^2\tau}\right)\ge\frac {h}2.
\end{equation}
Then \eqref{opt:eq:isolated-error} is also at most $\eps/16$.
The following conservative choice verifies explicitly that all
constants can be fixed independently of the parameters:
\begin{equation}\label{opt:eq:core-explicit}
 C_{\rm bulk}=32,\qquad C_{\rm order}=4,\qquad {h}\ge10^4.
\end{equation}
Indeed, write $x:=e^{h}-e=L/\tau$. Then
$q^2/s\le4{h}/x$, $a\le qr$, $a/q\le r$, and
$k\ge4r\ge {h}\tau$. The least eligible even $p$ in
\eqref{opt:eq:phase-choice} satisfies $2\pi q<p\le7q$; the
corresponding choices of $q$ differ by less than one, so
$q_*\le q\le2q_*$. Also $Q_0\le3q/{h}$ and $B_0\le96q/{h}$.
These inequalities give $d<6200qk$.

On either isolated circle the bulk, far isolated, and nearby
isolated spectral contributions give
\begin{equation}
 \rho_z\le8(1+ w_{\rm iso}/\zeta^2),\qquad
 am_z\le16(a\zeta+a w_{\rm iso}/\zeta)\le801k,
 \qquad a w_{\rm iso}/\zeta=50k.
\end{equation}
Including both circles and the Cauchy factor $1/(2\pi)$,
\eqref{opt:eq:global-resolvent} has integrated bound at most
$e^{1611k}$. The root identity for $g_p$ gives
$|g_p(z)|\le {h}\tau/k$, and
$|g_{B_0}(z)|^m\le e^{d\zeta}\le e^{62k}$. Their total
contribution is therefore bounded by
\begin{equation}\label{opt:eq:isolated-explicit}
 E_{\rm isolated}\le
 \left(\frac{e^{1680}{h}\tau}{k}\right)^k
 \le e^{-k{h}/2}\le e^{-2L}=\eps^2\le\eps/16.
\end{equation}
The penultimate estimate follows from
$\log[k/(e^{1680}{h}\tau)]\ge {h}-1681-2\log {h}+\log4\ge {h}/2$.

On the inner bulk arc, the same compressed-moment calculation gives
$\rho_z\le4$, $\Delta_z^{-1}\le6Q_0$, and $am_z\le {h}r$.
Its integrated resolvent bound is at most
$10(1+{h}r)e^{2{h}r}$. On both bulk arcs,
$|g_{B_0}(z)|<1/2$ and $|g_p(z)|\le e^{{h}/8}$.
Keeping two inverse-distance factors on the outer arc, as in the
bulk estimate above, bounds its contribution by
$56e^{{h}k/8}2^{-m}$. Consequently
\begin{equation}\label{opt:eq:bulk-explicit}
 E_{\rm bulk}\le66(1+{h}r)e^{2{h}r+{h}k/8}2^{-m}\le\eps/16.
\end{equation}
For the last inequality, use $1+{h}r\le e^{{h}r}$, $hr=L$,
$m\ge32{h}k$, and
$\log1056\le(32\log2-4.125){h}k$.
The mesh in \eqref{opt:eq:core-mesh} gives both the contour enclosure
and the small-step conditions throughout this range. The large
threshold and constants in \eqref{opt:eq:core-explicit} are an
asymptotic witness, not recommended numerical parameters.
The Cauchy formula on the three spectral domains now gives
\begin{equation}\label{opt:eq:filter-small}
 \norm{f(D_g)}\le\eps/8.
\end{equation}
Since $B_0=\Theta(q/{h})$ and $p=\Theta(q)$, we have
\begin{equation}\label{opt:eq:core-degree}
 d=\mathcal{O}(qk)=\mathcal{O}(qr)=\mathcal{O}(\sqrt{s\tau r}).
\end{equation}
Both parts of the filter have this order. The bulk needs
$m=\mathcal{O}({h} r)=\mathcal{O}(L)$ factors, each of degree $\mathcal{O}(q/{h})$.
The isolated phases need $k=\mathcal{O}(r)$ factors, each of degree $\mathcal{O}(q)$.
Thus the bulk can be suppressed geometrically to precision $\eps$
without restoring an extra factor of $h$ in the total degree.
Also $q\ge q_*$ implies $a\le qr$, while
$\cot(\pi/p)<p/\pi$ and $q\ge1$ give $p-1>5q$.
Thus
\begin{equation}\label{opt:eq:degree-catalyst}
 d\ge k(p-1)>5qr\ge5a\ge5W.
\end{equation}

\emph{Converting the bound into a circuit with the required error on all output registers.}
We have now bounded the operator filter and its degree
for the actual finite chronological block. To finish,
we must connect that operator estimate to the output
of a finite circuit. We use the reuse, smoothing and
amplification operations from the baseline construction,
recalling their identities here so that the leakage
bound and the finite-mesh error remain visible.

First we identify the stationary private amplitude
for one slot: it is the private component that the
interaction returns unchanged while applying its
unitary public action. Solving the local stationary
equation gives the map
\begin{equation}
 \Gamma=\sqrt{2v}(I-D_0)^{-1}J(I+\ii v (H/\alpha)/2)^{-1}.
\end{equation}
The first factor in this expression is controlled by
the joint second moment. In fact,
\begin{equation}
 J^\dagger(I-D_0)^{-\dagger}(I-D_0)^{-1}J
 =\frac14J^\dagger(I+K^2)J
 \preceq\frac\nu2 I.
\end{equation}
The remaining inverse has norm at most one because
$H/\alpha$ is Hermitian. Multiplying by $2v$
therefore gives
$\Gamma^\dagger\Gamma\preceq v\nu I$.
In the chronological composition, each local map
is preceded by a unitary public evolution. These
prefixes do not change its norm, and the private
slots form a direct sum, so their squared norms
add. The whole stationary map $\Gamma_{\rm tot}$
thus has squared norm at most $N_sv\nu=W$,
and the public action is $U_C$.

Let $A_g,B_g,C_g,D_g$ be the blocks of the whole
unitary interaction. The stationary relation is
$S(\psi\oplus\Gamma_{\rm tot}\psi)
=U_C\psi\oplus\Gamma_{\rm tot}\psi$.
To see the identity needed for reuse, apply this
unitary to both vectors in the inner product
between $\psi\oplus\Gamma_{\rm tot}\psi$
and $0\oplus\chi$, where $\chi$ is an arbitrary private vector:
\begin{equation}
 \langle\Gamma_{\rm tot}\psi,\chi\rangle
 =\langle U_C\psi,B_g\chi\rangle
  +\langle\Gamma_{\rm tot}\psi,D_g\chi\rangle.
\end{equation}
Since this holds for arbitrary $\psi,\chi$, it gives
$B_g=U_C\Gamma_{\rm tot}^\dagger(I-D_g)$.
Starting reuse with zero private amplitude and
projecting the public clock back onto its initial
uniform state gives the block
\begin{equation}
 P_j=U_C-B_g g_j(D_g)\Gamma_{\rm tot}.
\end{equation}
Write $f(z):=\sum_{j=0}^d a_jz^j$. Its coefficients are nonnegative,
palindromic, and sum to one. With out-of-range coefficients zero,
\begin{equation}\label{opt:eq:smoothing}
 f(z)g_d(z)=\sum_{j=1}^{2d}\omega_jg_j(z),\qquad
 \omega_j=\frac jd(a_{j-d}-a_j),\qquad
 \sum_j\omega_j=1,\quad\sum_j|\omega_j|=2-2a_0<2.
\end{equation}
The coefficient properties used here follow from
the construction: each uniform average has nonnegative
coefficients, unit sum and reversal symmetry, and
products preserve all three. To verify the smoothing
identity, the coefficient of $z^l$ in $fg_d$ is
$d^{-1}\sum_{r=0}^{d-1}a_{l-r}$. The difference
between two adjacent such coefficients is
$(a_{j-d}-a_j)/d$ at index $j$, giving the stated
$\omega_j/j$. Summing the differences reconstructs
the polynomial in the basis of averages $g_j$.
We can also read off the signs of the weights directly.
For $1\le j<d$, zero extension gives
$\omega_j=-ja_j/d\le0$; at $j=d$, reversal symmetry
gives $\omega_d=a_0-a_d=0$; and for $d<j\le2d$,
we have $\omega_j=ja_{j-d}/d\ge0$.
Reversal symmetry also gives $\sum_{j=0}^d j a_j=d/2$.
The total absolute weight on the negative part is therefore
$\sum_{j<d}|\omega_j|=1/2-a_0$.
Since evaluation at $z=1$ gives $\sum_j\omega_j=1$,
the full absolute sum is
$1+2(1/2-a_0)=2-2a_0<2$.

We use these known weights to select the finite
$j$-reuse circuits. Preparing their coefficient
label, applying the selected signed circuit, and
unpreparing produces an LCU block
$\widetilde P/2$, after padding the normalisation
to two, where $\widetilde P=\sum_j\omega_jP_j$.
The longest branch has $2d$ reuses. Substituting
the formula for $P_j$, using the smoothing identity,
and then using $(I-D_g)g_d(D_g)=(I-D_g^d)/d$
gives
\begin{equation}
 U_C-\widetilde P
 =\frac{U_C}{d}\Gamma_{\rm tot}^\dagger
          f(D_g)(I-D_g^d)\Gamma_{\rm tot},\qquad
 \norm{U_C-\widetilde P}\le\norm{f(D_g)}.
\end{equation}
For the norm estimate, the two stationary factors
contribute at most $W$, the contraction $D_g$
gives $\norm{I-D_g^d}\le2$, and $d\ge5W$.
Thus the displayed inequality follows with room
in its constant. The filter bound is now an
error bound on the successful block of an actual
finite circuit.

One round of oblivious amplification uses at most
$6d$ global calls and has successful block
$\tfrac32\widetilde P
-\tfrac12\widetilde P\widetilde P^\dagger\widetilde P$.
To include all other outputs, set
$E_0=U_C^\dagger\widetilde P-I$ and suppose
$\norm{E_0}\le\epsilon_f\le1$. Substituting
$\widetilde P=U_C(I+E_0)$ into the amplified
block cancels its linear contribution to the
real part of the overlap with $U_C$. The squared
full-output norm error is therefore bounded by
\begin{equation}
 \norm{\Reop(E_0^2)+E_0E_0^\dagger+E_0^\dagger E_0
                   +\Reop(E_0E_0^\dagger E_0)}
 \le3\epsilon_f^2+\epsilon_f^3\le4\epsilon_f^2.
\end{equation}
Thus \eqref{opt:eq:filter-small} gives error on all output registers at most
$\eps/4$ relative to the finite Cayley evolution.
Together with $\eps/4$ discretisation error and at most $\eps/2$ for
known gate approximation, the full error is at most $\eps$.
Every call uses a constant number of sparse queries. This proves the core
bound \eqref{opt:eq:core-degree}.
\end{proof}

We have proved the two improved-precision constructions as
finite circuits, with their leakage included in the error.
The final step is to choose among the baseline, whole-matrix
signed factor and reduced-entry construction from the
supplied values of time, sparsity and precision. We now
show that these regimes cover the stated theorem and
give the same target order at their boundaries.

\begin{proof}[Proof of Theorem~\ref{hs:thm:optimal}]
Recall that $\tau=\eta|t|$, $L=\log(1/\eps)$,
$h=\log(e+L/\tau)$ and $r=L/h$. The query count
to prove is
$\mathcal O(\sqrt{s}\tau+\sqrt{s\tau r}+r)$.
Choose the fixed threshold on $h$ large enough for
all estimates in Appendix~\ref{opt:sec:core}. Below
that threshold, $L=hr$ differs from $r$ by at most
a constant factor. The baseline bound
$\mathcal O(\sqrt{s}\tau+\sqrt{s\tau L}+L)$
therefore already has the required order.

Now let $h$ exceed the threshold. If $s\tau\ge L{h}$, the core
construction applies. Otherwise use the entrywise factor
\eqref{opt:eq:whole-factor} and Theorem~\ref{opt:thm:signed}, with
$T=s\tau$ and $b=\sqrt{s}\tau$.
Put $x:=L/\tau$, so $h=\log(e+x)$ and $s\le x{h}$. Then
\begin{equation}\label{opt:eq:regime-comparison}
 \frac{L}{b}=\frac{x}{\sqrt{s}}\ge\sqrt{\frac{x}{{h}}},
 \qquad
 \frac br=\frac{{h}\sqrt{s}}x\le\frac{{h}^{3/2}}{\sqrt x}.
\end{equation}
Since $x=e^h-e$, the upper bound on $b/r$ tends
to zero as $h$ grows. This makes the factor's
remainder scale no larger than the desired $r$.
We also need its own denominator logarithm to have
order $h$. The lower bound on $L/b$ gives
$\log(e+L/b)\ge\log(e+\sqrt{x/h})$, which
has order $h$ for large $h$; the upper bound
$L/b\le x$, using $s\ge1$, gives
$\log(e+L/b)\le h$. Thus, uniformly in this
remaining regime,
\begin{equation}
 \log(e+L/b)=\Theta({h}),\qquad r_b=\Theta(r),\qquad b=\mathcal{O}(r)
\end{equation}
The signed-factor cost is therefore
$\cO(\sqrt{s\tau r}+r)$. Together with the bounded-$h$ case and the
core case, this proves \eqref{hs:eq:optimal-bound}. Negative times use inverse
circuits throughout.
\end{proof}

For example, with $\tau=1$ and $s=L^2$ as $L\to\infty$, the mixed
term improves from $L^{3/2}$ to
$L^{3/2}/\sqrt{\log L}$.
The argument uses the row and column Euclidean norm promise throughout;
it does not require $\norm H\le\eta$. The matching statement concerns
the time-independent query problem in the stated range.

\section{Proof of the two-norm Hamiltonian bound}
\label{twonorm:sec:proof}

We prove that retaining a maximum-entry bound $\mu$ alongside the
row and column Euclidean-norm bound $\eta$ gives query count
\begin{equation}
 \mathcal O\!\left(\sqrt{s}\,\eta|t|
 +\sqrt{\frac{s\mu|t|\log(1/\eps)}
 {\log(e+\mu\log(1/\eps)/(\eta^2|t|))}}
 +\frac{\log(1/\eps)}{\log(e+\mu\log(1/\eps)/(\eta^2|t|))}\right)
\end{equation}
when $\eta^2|t|/\mu\ge1$ and $0<\eps\le1/64$. The entry bound
limits the full local operator, whereas the Euclidean-norm bound
limits its second moment on embedded system states. We use that
separation in each of the baseline, signed-factor and reduced-entry
constructions. Every preparation still uses the original $s$-slot
row and column lists; the improvement comes from the analytic moment
bounds, without changing the oracle model.

Suppose the positive bounds satisfy
\begin{equation}\label{twonorm:eq:promises}
 \max_j\Bigl(\sum_i|H_{ij}|^2\Bigr)^{1/2}\le\eta,
 \qquad \max_{i,j}|H_{ij}|\le\mu,
 \qquad \mu\le\eta\le\sqrt{s}\,\mu.
\end{equation}
Hermiticity gives the same row bound. Redundant supplied bounds can be
tightened without queries: replace $\eta_0$ by
$\eta=\min\{\eta_0,\sqrt{s}\mu_0\}$ and then set
$\mu=\min\{\mu_0,\eta\}$. The case $t=0$ is immediate. For $t\ne0$ put
\begin{equation}\label{twonorm:eq:parameters}
 \sigma:=\frac{s\mu^2}{\eta^2},\qquad
 \theta:=\frac{\eta^2|t|}{\mu},\qquad
 L:=\log(1/\eps),\qquad {h}:=\log(e+L/\theta),\qquad r_\mu:=L/{h}.
\end{equation}
We assume $\theta\ge1$ and $0<\eps\le1/64$ as in the theorem.
The same analytic sparsity parameter $\sigma$ appeared
in the main theorem. It lies in $[1,s]$ and may be
nonintegral; it describes the relative moment bounds,
not the number of entries accessible through a locator.
For the proof, separate the row-norm and entry-norm
time scales:
\begin{equation}\label{twonorm:eq:times}
 {\tau_{\rm row}}:=\sqrt{s}\,\eta|t|=\sqrt{\sigma}\,\theta,
 \qquad {\tau_{\rm ent}}:=s\mu|t|=\sigma\theta.
\end{equation}
Thus $\theta\le\tau_{\rm row}\le\tau_{\rm ent}$, and
the desired query count is
$O(\tau_{\rm row}+\sqrt{\tau_{\rm ent}r_\mu}+r_\mu)$.
We first verify the new moment bounds on the actual
oracle registers, then apply the baseline, signed-factor
and clipped constructions in their respective regimes.

\begin{proof}[Proof of Theorem~\ref{hs:thm:two-norm}]
\emph{Keeping the native access while changing the scale.}
Normalise each entry by $\mu$ instead of $\eta$.
The entry operator then has norm at most one, and
the row-norm promise supplies compressed second
moment $1/\sigma$. The original list preparation is
unchanged.

On the common edge space, define the row and column list maps
\begin{equation}
 L_{\rm row}\ket i=\sum_{j\text{ in row list }i}\ket{i,j},\qquad
 R_{\rm col}\ket j=\sum_{i\text{ in column list }j}\ket{i,j},\qquad
 J_0=\frac1{\sqrt{2s}}\binom{L_{\rm row}}{R_{\rm col}}.
\end{equation}
Then $J_0$ is an isometry. For Hermitian $X$ supported on these lists, put
\begin{equation}
 Y_X:=\sum_{i,j}\frac{X_{ij}}\mu\ket{i,j}\bra{i,j},\qquad
 E_X:=\begin{pmatrix}0&Y_X\\Y_X^\dagger&0\end{pmatrix},\qquad
 D_X:=\diag\!\left(\sum_j|X_{ij}|\right).
\end{equation}
Direct compression, including zero padding and diagonal entries, gives
\begin{equation}\label{twonorm:eq:compression}
 J_0^\dagger E_XJ_0=\frac{X}{s\mu},\qquad
 J_0^\dagger|E_X|J_0=\frac{D_X}{s\mu}.
\end{equation}
For $X=H$,
\begin{equation}\label{twonorm:eq:moments}
 \norm{E_H}\le1,\qquad
 J_0^\dagger E_H^2J_0\preceq
 \frac{\eta^2}{s\mu^2}I=\frac1{\sigma} I,
 \qquad D_H\preceq\sqrt{s}\,\eta I.
\end{equation}
The last inequality is Cauchy--Schwarz on each row.

The entrywise signed factor also retains the sharper scale. For a represented
entry $c=X_{ij}/\mu$ with $|c|\le1$, let
\begin{equation}
 E(c):=\begin{pmatrix}0&c\\\bar c&0\end{pmatrix},\qquad
 B_\pm(c):=\sqrt{\frac{|c|I\pm E(c)}2}.
\end{equation}
The two success amplitudes $B_+(c),B_-(c)$ and failure amplitude
$\sqrt{1-|c|}I$ form an isometry. Extend this local operation to a
unitary, controlled by the encoded entry value, and combine it with $J_0$.
The resulting success block $B_X$ and known sign reflection $Z$ satisfy
\begin{equation}\label{twonorm:eq:factor}
 B_X^\dagger ZB_X=\frac{X}{s\mu},\qquad
 B_X^\dagger B_X=\frac{D_X}{s\mu}.
\end{equation}
Its unitary preparation and its adjoint each use a constant number of primitive
queries. No row sum is evaluated by the circuit. Decoded values outside the
promised entry range can be clipped by the known decoder so that its
conditional operation is a total unitary; promised entries are unchanged.
For $X=H$, this supplies
\begin{equation}\label{twonorm:eq:wholefactor}
 H=s\mu B_H^\dagger ZB_H,
 \qquad B_H^\dagger B_H\preceq\frac1{\sqrt{\sigma}}I.
\end{equation}
All controls use the entrywise constructions of Appendix~\ref{opt:sec:native}.
In particular, a factor-conjugated known phase is controlled by controlling
only that phase: the unconditional factor and inverse cancel on the inactive
branch. The supplied phase-one locator promise permits the constant-query
controlled-locator construction in Appendix~\ref{opt:sec:native}.

These identities supply the factor-access theorem
with physical tail scale $\tau_{\rm row}$ and access
scale $\tau_{\rm ent}$. To recall its precise form,
Theorem~\ref{opt:thm:signed} gives
\begin{equation}\label{twonorm:eq:factor-theorem}
 O(\sqrt{Tb}+\sqrt{Tr_b}+r_b),\qquad
 T=\lambda|t|,\quad b=T\delta,\quad
 r_b=\frac{L}{\log(e+L/b)}
\end{equation}
for $H=\lambda B_H^\dagger ZB_H$ and
$B_H^\dagger B_H\preceq\delta I$.
Its guarantee concerns a finite unitary
circuit. For the clipped construction we must also
check its spectral weights; we do so in the third
case below.

\emph{Bounded $h$: the baseline.}
The baseline depends on the moment and local-norm
bounds, rather than on integrality of the analytic
parameter.
Set $K_0:=\sqrt{\sigma} E_H$. Then
\begin{equation}
 J_0^\dagger K_0J_0=\frac{H}{\sqrt{s}\,\eta},\qquad
 J_0^\dagger K_0^2J_0\preceq I,\qquad \norm{K_0}\le\sqrt{\sigma}.
\end{equation}
The Cayley-transducer construction underlying
\eqref{hs:eq:degree}, with rescaling $0<\gamma\le1$, has query count
\begin{equation}
 \cO\!\left({\tau_{\rm row}}(\gamma+\gamma^{-1})
             +\lceil\max\{1,\gamma\sqrt{\sigma}\}\rceil L\right).
\end{equation}
Choosing
\begin{equation}
 \gamma=\min\{1,\max\{\sigma^{-1/2},\sqrt{\theta/L}\}\}
\end{equation}
gives $\cO({\tau_{\rm row}}+\sqrt{{\tau_{\rm ent}}L}+L)$. When $h$ is below a fixed absolute
threshold, $L={h}r_\mu$ and this is the required bound.

\emph{Large $h$ with $\tau_{\rm ent}\le Lh$.}
In this regime the whole-matrix factor's smaller tail
scale gives the same precision parameter $r_\mu$ up
to constants.
Now suppose $h$ exceeds that threshold and ${\tau_{\rm ent}}=\sigma\theta\le L{h}$.
The entrywise factor \eqref{twonorm:eq:wholefactor} has ${T}={\tau_{\rm ent}}$ and
$b={\tau_{\rm row}}=\sqrt{\sigma}\theta$. Put $x:=L/\theta$, so
$h=\log(e+x)$ and $\sigma\le x{h}$. Then
\begin{equation}
 \frac{L}{{\tau_{\rm row}}}=\frac{x}{\sqrt{\sigma}}\ge\sqrt{x/{h}},\qquad
 \frac{{\tau_{\rm row}}}{r_\mu}=\frac{{h}\sqrt{\sigma}}{x}
                       \le\frac{{h}^{3/2}}{\sqrt x}.
\end{equation}
Together with $\sigma\ge1$, these give
$\log(e+L/{\tau_{\rm row}})=\Theta({h})$, $L/\log(e+L/\tau_{\rm row})=\Theta(r_\mu)$, and
${\tau_{\rm row}}=\cO(r_\mu)$, uniformly throughout this regime. Therefore
\eqref{twonorm:eq:factor-theorem} has cost
\begin{equation}
 \cO(\sqrt{{\tau_{\rm ent}}{\tau_{\rm row}}}+\sqrt{\tau_{\rm ent}L/\log(e+L/\tau_{\rm row})}+L/\log(e+L/\tau_{\rm row}))
 =\cO(\sqrt{{\tau_{\rm ent}} r_\mu}+r_\mu).
\end{equation}
The finite signed weights, coherent reuse, mesh choice and amplification
are precisely those of the signed-factor theorem.

\emph{Large $h$ with $\tau_{\rm ent}\ge Lh$.}
We now clip at scale $\mu$ and verify the moment
and spectral-weight estimates needed by the finite
resolvent proof. This verification is what permits
the real parameter $\sigma$ to enter that proof
without changing any locator.

Suppose instead $h$ is large and ${\tau_{\rm ent}}\ge L{h}$. Clip at the
entry scale $\mu$:
\begin{equation}\label{twonorm:eq:clip}
 (H_c)_{ij}:=\frac{H_{ij}}{1+{h}|H_{ij}|/\mu},\qquad C_H:=H-H_c.
\end{equation}
Both matrices remain Hermitian on the original support. They obey
\begin{equation}\label{twonorm:eq:clip-bounds}
 |(H_c)_{ij}|\le\mu/{h},\qquad
 \sum_j|(C_H)_{ij}|\le\frac{{h}}\mu\sum_j|H_{ij}|^2
 \le\frac{{h}\eta^2}\mu.
\end{equation}
Consequently the correction factor has
\begin{equation}
 C_H=s\mu B_c^\dagger ZB_c,\qquad
 B_c^\dagger B_c\preceq\frac{{h}}{\sigma} I.
\end{equation}
Let $J_c$ be the isometry obtained by restricting its full preparation
unitary to initialised system inputs, including the failure component.
Choose an even integer $p$ and a real $q$ such that
\begin{equation}
 2q=\cot(\pi/p),\qquad
 q_*\le q\le2q_*,\qquad q_*=\sqrt{{\tau_{\rm ent}}/r_\mu}.
\end{equation}
For the same sufficiently large absolute threshold on $h$ as in the
proof in Appendix~\ref{opt:sec:core}, $q\ge {h}$ and $q\le\sqrt{\sigma}$. Define
\begin{equation}\label{twonorm:eq:joint}
 J=\frac1{\sqrt2}\binom{J_0}{J_c},\quad
 K=\diag(2qE_{H_c},\ 2qZ\text{ on success},\ 0\text{ on failure}),
 \quad D_0=(K+\ii I)(K-\ii I)^{-1}.
\end{equation}
The joint first-moment identity and second-moment bound are
\begin{equation}\label{twonorm:eq:joint-moments}
 J^\dagger KJ=H/\alpha,\qquad
 J^\dagger K^2J\preceq\frac{2q^2(1+{h})}{\sigma} I,
 \qquad \alpha=\frac{s\mu}{q},\quad a=\alpha|t|={\tau_{\rm ent}}/q.
\end{equation}
Also $\norm{H/\alpha}\le q/\sqrt{\sigma}\le1$.
The squared norm of the stationary private map is bounded by
$W=a\nu\le a$, with
$\nu=(1+2q^2(1+{h})/\sigma)/2\le1$.
On the private bulk, $|K|\le2q/{h}$. The two isolated private phases
are $e^{\pm2\pi\ii/p}$, and their compressed spectral weights obey
\begin{equation}\label{twonorm:eq:weights}
 \norm{J^\dagger P_\pm J}\le w_{\rm iso}:=\frac{{h}}{2\sigma},\qquad
 J^\dagger P_0K^2P_0J\preceq\frac{2q^2}{\sigma} I.
\end{equation}
These statements require no spectral-projection oracle.
Conjugating by the entrywise preparation unitaries makes the boundary known,
while each private operation and its adjoint retain constant query cost.

For clarity, the quantities driving both contour estimates are
\begin{equation}
 \frac{q^2}{\sigma}\le\frac{4\theta}{r_\mu},\quad
 a\le qr_\mu,\quad \frac aq\le r_\mu,\quad
 \zeta=\frac{{h}\theta}{100qk},\quad
 \frac{aw_{\rm iso}}{\zeta}=50k.
\end{equation}
The finite-clock resolvent
\eqref{opt:eq:global-resolvent} and spectral enclosures
therefore have exactly the hypotheses used in
Appendix~\ref{opt:sec:core}, with $s$ replaced
analytically by $\sigma$ and $\tau$ by $\theta$.
Those proofs use the displayed moments, spectral
weights and real scalar inequalities. No step
requires a list of length $\sigma$.

For example, let
\begin{equation}
 Q_0=\lceil\max\{1,2q/{h}\}\rceil,\quad B_0=32Q_0,\quad
 k=\lceil4r_\mu\rceil,\quad m=\lceil32{h} k\rceil,
 \quad g_j(z)=\frac1j\sum_{\ell=0}^{j-1}z^\ell.
\end{equation}
For $h\ge10^4$, the explicit estimates
\eqref{opt:eq:isolated-explicit} and~\eqref{opt:eq:bulk-explicit} give
\begin{equation}
 f(z):=g_{B_0}(z)^m g_p(z)^k,\quad
 \norm{f(D_g)}\le\eps/8,\quad
 D_f=m(B_0-1)+k(p-1)=\cO(qr_\mu),\quad D_f>5W.
\end{equation}
Here $D_g$ is the actual finite-clock private block. A sufficient mesh is
\begin{equation}
 {N_s}\ge\max\left\{4a,\frac{8a}\zeta,24aQ_0,
                         \sqrt{\frac{a^3}{3\eps}}\right\}.
\end{equation}
The mesh contributes at most $\eps/4$ error and no mesh multiplier to
the primitive-query count. The signed smoothing
construction~\eqref{opt:eq:smoothing}, followed by one oblivious
amplification round, uses at most $6D_f$ transducer calls,
with error on all output registers at most $\eps/4$ relative to the finite Cayley target.
Known gates are approximated to total error at most $\eps/2$.
The total query cost in this regime is therefore
$\cO(qr_\mu)=\cO(\sqrt{{\tau_{\rm ent}} r_\mu})$.
Together the three regimes prove Theorem~\ref{hs:thm:two-norm}.
For negative $t$, use the inverse circuit.

All three cases give the required norm bound on the complete output
state, including the work registers. Tensoring their
difference with an untouched reference preserves
this norm, so the guarantee is reference-stable.
The finite meshes and coefficient preparations establish the query
bound. Their efficient quantum implementation is proved in
Appendix~\ref{gate:sec:proof}, including the cost of preparing the
coefficient tables and the logarithmic register width of the meshes.
\end{proof}

\section{Resolvent control and finite conversion to a circuit}
\label{sec:conversion-proof}

We now prove the estimate that completes the baseline simulation circuit.
The overview has supplied an exact stationary relation for a rational
time step, a chronological composition with constant cost per call,
and reuse circuits that start with an empty private space. The
outstanding problem is quantitative: we must show that a short
polynomial suppresses the composed private block, even though that
block need not be normal, and then account for the complete circuit
output. We state the result for the compressed resolvent data actually
used in the proof, so that the same argument also applies
to other entry representations.

\begin{lemma}[Conversion from compressed resolvent bounds]
\label{hs:lem:abstract}
Let $D_0$ be a unitary on a finite-dimensional private space and
$J$ a known isometry from the system into that space. Suppose
$I-D_0$ is invertible and, with $R_0:=(I-D_0)^{-1}$,
\begin{equation}\label{hs:eq:interface}
 J^\dagger R_0J=\frac{I+\ii h}{2},\qquad
 J^\dagger R_0^\dagger R_0J\preceq\frac\nu2 I,
 \qquad \norm{R_0}\le Q,
\end{equation}
where $h=h^\dagger$, $\norm h\le1$, $Q\ge1$ is an integer and
$0<\nu\le Q$ is a supplied real bound. Assume access to $D_0$,
$D_0^\dagger$ and their coherent controls. For known $\alpha>0$,
$H:=\alpha h$, $t\in\mathbb R$ and $0<\eps\le1/2$, there is
a unitary simulation with error on all output registers at most $\eps$ using
\begin{equation}\label{hs:eq:abstractcost}
 \mathcal O\!\left(W+Q\log(1/\eps)\right),
 \qquad W:=\alpha|t|\nu,
\end{equation}
calls to the supplied unitaries. The number of physical time slots
does not multiply this count.
\end{lemma}

In this general statement, $\nu$ is an independent bound on the
squared resolvent response of an embedded input. It is neither a
query count nor the full resolvent norm $Q$. For the sparse
construction, Eqs.~\eqref{ov:eq:response},
\eqref{ov:eq:compressed-response} and~\eqref{ov:eq:Q} verify the
hypotheses with $\nu=(1+\gamma^2)/2$ and
$\alpha=\eta\sqrt{s}/\gamma$. The known embedding is obtained
by the locator change of coordinates described above. Thus proving
this lemma and using Eq.~\eqref{ov:eq:gamma} completes
Theorem~\ref{hs:thm:main}.

\begin{proof}
The case $t=0$ is immediate, and negative time uses the adjoint.
Assume $t>0$. Choose a positive integer $M$ with
$w:=\alpha t/M\le1/4$, and form the rotation
\eqref{ov:eq:coupling} with the parameters
\eqref{ov:eq:rotation-parameters}. The balance calculation
\eqref{ov:eq:balance}--\eqref{ov:eq:step} uses only the first
identity in Eq.~\eqref{hs:eq:interface}, so it gives
\begin{equation}\label{proof:eq:local-step}
 S\binom I\Gamma=\binom{U_C}\Gamma,
 \quad U_C=(I-\ii wh/2)(I+\ii wh/2)^{-1},
 \quad \Gamma=\sqrt{2w}\,R_0J(I+\ii wh/2)^{-1}.
\end{equation}
The second relation in Eq.~\eqref{hs:eq:interface} bounds the supporting private component:
\begin{equation}
 \Gamma^\dagger\Gamma
 =2w(I-\ii wh/2)^{-1}J^\dagger R_0^\dagger R_0J
                       (I+\ii wh/2)^{-1}
 \preceq w\nu I.
\end{equation}
Chronological composition on $M$ private summands therefore has
public action $U:=U_C^M$, stationary map
$\Gamma_{\rm tot}\psi=\bigoplus_{r=0}^{M-1}\Gamma U_C^r\psi$,
and $\norm{\Gamma_{\rm tot}}^2\le Mw\nu=W$.
The factorisation~\eqref{ov:eq:chronological-circuit} uses one
coherent application of $D_0$ and known boundary rotations per
complete transducer call. We write $D$ for the private block of
that complete unitary; in particular $\norm D\le1$.

The scalar phase estimate used in Eq.~\eqref{ov:eq:discretisation}
gives
\begin{equation}\label{hs:eq:disc}
 \norm{U-e^{-\ii tH}}
 \le\frac{t^3\norm H^3}{12M^2}
 \le\frac{(\alpha t)^3}{12M^2}.
\end{equation}
We fix any integer
\begin{equation}\label{proof:eq:mesh}
 M\ge\max\left\{1,4\alpha t,
              \frac{(\alpha t)^{3/2}}{\sqrt{3\eps}}\right\},
\end{equation}
so this error is at most $\eps/4$. Increasing $M$ leaves $W$
and $Q$ unchanged. We now hold this finite circuit fixed and analyse
its private block.

\paragraph{From the local response to the chronological resolvent.}
Our first goal is to exclude a disk around one from the spectrum of
$D$ and bound its inverse there. The local unitary already has a
resolvent bound at one; we will extend that bound to nearby $z$,
include the boundary rotation, and then sum how a forcing term in
one private slot propagates to later slots. This separates the
local response from the accumulation through physical time.

For $|z-1|\le r_0:=1/(4Q)$, define
\begin{equation}
 R_0(z):=(zI-D_0)^{-1},\qquad
 F(z):=J^\dagger D_0R_0(z)J.
\end{equation}
The Neumann identity
\begin{equation}
 R_0(z)=[I+(z-1)R_0]^{-1}R_0
\end{equation}
has first-factor norm at most $4/3$, because
$|z-1|\norm{R_0}\le1/4$. Applied to arbitrary private vectors,
it yields the full bound $4Q/3$. Applied after $J$, it preserves
the smaller compressed response $\sqrt{\nu/2}$. We also need a
bound for returning to the public space. Since $D_0$ is normal and
unitary, $R_0(z)$ commutes with $D_0$ and its adjoint, and
\begin{equation}
 (J^\dagger D_0R_0(z))(J^\dagger D_0R_0(z))^\dagger
   =J^\dagger R_0(z)^\dagger R_0(z)J.
\end{equation}
The injection and return norms are therefore equal. Differentiating
$F$ and factoring its two endpoints gives
\begin{equation}\label{hs:eq:localends}
 \begin{aligned}
 \norm{R_0(z)}&\le\frac43Q,\\
 \norm{R_0(z)J},\ \norm{J^\dagger D_0R_0(z)}
       &\le\frac43\sqrt{\nu/2},\\
 \norm{F'(z)}=
 \norm{J^\dagger D_0R_0(z)R_0(z)J}&\le\frac89\nu.
 \end{aligned}
\end{equation}
At one, $D_0R_0=R_0-I$ gives
$F(1)=(-I+\ii h)/2$, whose norm is at most $1/\sqrt2$.
Integrating the derivative along the line from $1$ to $z$ thus gives
\begin{equation}\label{proof:eq:F-bound}
 \norm{F(z)}\le\frac1{\sqrt2}+\frac{2\nu}{9Q}<1.
\end{equation}
This bound is what lets us add the boundary rotation without losing
the compressed endpoint estimates.

The private block of one physical step is
$(I-(1-c)JJ^\dagger)D_0$. Let $R_{\rm loc}(z)$ be its resolvent.
To derive it, solve
$[zI-(I-(1-c)JJ^\dagger)D_0]x=y$ by first applying $R_0(z)$
and then eliminating $J^\dagger D_0x$. This gives
\begin{equation}\label{proof:eq:local-inverse}
 R_{\rm loc}=R_0(z)-(1-c)R_0(z)J
 [I+(1-c)F(z)]^{-1}J^\dagger D_0R_0(z).
\end{equation}
The scalar restriction $w\le1/4$ implies $1-c\le2/9$;
together with Eq.~\eqref{proof:eq:F-bound} it bounds the inverse
in square brackets by $9/7$. The same elimination at the two
endpoints gives
\begin{equation}\label{proof:eq:local-endpoints}
 \begin{aligned}
 R_{\rm loc}J&=R_0(z)J[I+(1-c)F(z)]^{-1},\\
 J^\dagger D_0R_{\rm loc}
   &=[I+(1-c)F(z)]^{-1}J^\dagger D_0R_0(z).
 \end{aligned}
\end{equation}
Hence adding the boundary multiplies the small injection and return
bounds by at most $9/7$. In particular,
\begin{equation}\label{hs:eq:locbounds}
 \norm{R_{\rm loc}}\le2Q,\qquad
 \norm{R_{\rm loc}J},\ \norm{J^\dagger D_0R_{\rm loc}}
 \le\frac{12}{7}\sqrt{\nu/2}.
\end{equation}
For the first inequality, Eq.~\eqref{proof:eq:local-inverse} gives
$4Q/3+16\nu/63\le2Q$. Thus the full inverse remains of order
$Q$, while its two boundary endpoints remain of order $\sqrt\nu$.

We next determine the public amplitude transmitted through one slot
when its private equation has no external forcing. Let $v$ be the
incoming public vector, $x$ the private vector and $v'$ the outgoing
public vector. The boundary rotation gives
\begin{equation}
 [zI-(I-(1-c)JJ^\dagger)D_0]x=\beta Jv,
 \qquad v'=cv-\beta J^\dagger D_0x.
\end{equation}
The first equation gives $x=\beta R_{\rm loc}Jv$. Substituting into
the second and using Eq.~\eqref{proof:eq:local-endpoints} yields
$v'=u(z)v$, where
\begin{equation}\label{proof:eq:transfer}
 u(z):=cI-\beta^2F(z)[I+(1-c)F(z)]^{-1}.
\end{equation}
At $z=1$ this is exactly the unitary $U_C$ from the stationary
calculation. We control its growth away from one by subtracting that
known unitary value, rather than bounding each term separately:
\begin{equation}
 u(z)-u(1)=-\beta^2[I+(1-c)F(z)]^{-1}
 [F(z)-F(1)][I+(1-c)F(1)]^{-1}.
\end{equation}
Using $\beta^2\le2w$, the inverse bounds above, and
Eq.~\eqref{hs:eq:localends}, we obtain
\begin{equation}\label{hs:eq:transfer}
 \norm{u(z)}\le1+3w\nu|z-1|,
 \qquad \norm{u(z)^a}\le e^{3aw\nu|z-1|}
 \quad(a\in\mathbb Z_{\ge0}).
\end{equation}
Thus propagation through many slots depends on their total
coupling $Mw\nu=W$, rather than on the number of slots alone.

We can now solve the complete forced equation $(zI-D)x=y$.
Write $x_r,y_r$ for its private vectors in slot
$r=0,\ldots,M-1$, and let $v_r$ be the public amplitude passed
from slot $r$ to the next. There is no incoming public amplitude,
so $v_{-1}=0$. The private and public equations of the local
rotation give
\begin{equation}\label{proof:eq:forced-system}
 x_r=R_{\rm loc}y_r+\beta R_{\rm loc}Jv_{r-1},
 \qquad
 v_r=u(z)v_{r-1}-\beta J^\dagger D_0R_{\rm loc}y_r.
\end{equation}
These equations determine every $x_r$ successively. They prove
invertibility throughout the disk, and they identify each block of
the inverse: the diagonal block is $R_{\rm loc}$, while a force in
slot $j$ reaches a later slot $r>j$ through
\begin{equation}\label{proof:eq:off-diagonal}
 [(zI-D)^{-1}]_{rj}
 =-\beta^2R_{\rm loc}J\,u(z)^{r-j-1}
                         J^\dagger D_0R_{\rm loc}.
\end{equation}
The right endpoint converts the force into a public amplitude, the
middle factor propagates it through intervening slots, and the left
endpoint injects it into slot $r$. Equations~\eqref{hs:eq:locbounds}
and~\eqref{hs:eq:transfer} bound this block by
$(144/49)w\nu e^{3W/(4Q)}$.
There are at most $M$ off-diagonal blocks in any block row or column.
Applying the block Schur bound
\begin{equation}
 \norm X\le
 \sqrt{\left(\max_r\sum_j\norm{X_{rj}}\right)
       \left(\max_j\sum_r\norm{X_{rj}}\right)}
\end{equation}
and adding the diagonal contribution yields
\begin{equation}\label{hs:eq:globalres}
 \norm{(zI-D)^{-1}}
 \le2Q+\frac{144}{49}We^{3W/(4Q)}
 \le3(Q+W)e^{3W/(4Q)},
 \qquad |z-1|\le\frac1{4Q}.
\end{equation}
We have therefore obtained both a spectrum-free disk and a bound
on the inverse there. The argument used normality of the local
unitary $D_0$ only; it imposed no normality assumption on $D$.

\paragraph{A filter without a size-dependent prefactor.}
We use the spectrum-free disk to justify the averaging filter from
the overview. Let $b:=32Q$, let $k\ge2$ be an integer, and set
\begin{equation}
 g_b(z):=\frac1b\sum_{j=0}^{b-1}z^j,
 \qquad f(z):=g_b(z)^k,\qquad d:=k(b-1).
\end{equation}
The spectrum of the contraction $D$ lies inside the unit disk and
outside $|z-1|\le r_0$. It is therefore enclosed by the positively
oriented boundary of
\begin{equation}\label{proof:eq:contour}
 \Omega:=\{z:|z|<R,\ |z-1|>r_0\},
 \qquad R:=1+1/b,\quad r_0:=1/(4Q).
\end{equation}
This boundary has an outer arc on $|z|=R$ and an inner arc on
$|z-1|=r_0$. Cauchy's operator formula expresses $f(D)$ as the
sum of the integrals of $f(z)(zI-D)^{-1}/(2\pi\ii)$ over those
two arcs. The resolvent estimates allow us to bound the two
contributions separately.

On either arc, $|z|\le R$ and $|z-1|\ge r_0$, so the
geometric-sum identity and $R^b\le e$ give
\begin{equation}\label{hs:eq:gpoint}
 |g_b(z)|\le\frac{1+e}{b|z-1|}
 \le\frac{1+e}{8}<\frac12.
\end{equation}
On the outer arc, contractivity alone gives
$\norm{(zI-D)^{-1}}\le1/(R-1)=b$.
If we combined that bound with $|f(z)|\le2^{-k}$ and a constant
arc length, we would introduce a factor proportional to $Q$.
Removing that factor would then require an extra $\log Q$ in $k$.
Instead, we retain the angular decay of two averaging factors:
\begin{equation}\label{proof:eq:angular-decay}
 |f(z)|\le2^{-(k-2)}
     \frac{(1+e)^2}{b^2|z-1|^2}
 \qquad\text{on the outer arc}.
\end{equation}
The remaining scalar integral is exactly
\begin{equation}\label{hs:eq:circleintegral}
 \int_{|z|=R}\frac{|dz|}{|z-1|^2}
 =\int_{-\pi}^{\pi}\frac{R\,d\theta}{R^2+1-2R\cos\theta}
 =\frac{2\pi R}{R^2-1}.
\end{equation}
We extend this nonnegative integral to the whole outer circle,
without extending the arc-specific suppression assumption
$|g_b|<1/2$. The outer contribution is at most
\begin{equation}\label{proof:eq:outer-bound}
 2^{-(k-2)}(1+e)^2\frac{R}{b(R^2-1)}
 =2^{-(k-2)}(1+e)^2\frac{b+1}{2b+1}
 \le64\,2^{-k}.
\end{equation}
The factor $b$ from the resolvent has been cancelled by the
integrated decay of the filter. This is the decisive improvement
over a bound based only on the maximum filter magnitude.

On the inner arc, we instead use the chronological resolvent bound
\eqref{hs:eq:globalres}. Its length is at most
$2\pi r_0=\pi/(2Q)$, so Eq.~\eqref{hs:eq:gpoint} bounds its
contribution by
\begin{equation}
 \frac34(1+W/Q)e^{3W/(4Q)}2^{-k}.
\end{equation}
Combining the two arcs gives
\begin{equation}\label{hs:eq:filterbound}
 \norm{f(D)}\le65(1+W/Q)e^{3W/(4Q)}2^{-k}.
\end{equation}
For a target $0<\delta<1$, choose
\begin{equation}\label{hs:eq:kchoice}
 k:=\left\lceil4W/Q+\log_2(512/\delta)\right\rceil.
\end{equation}
Indeed $1+x\le e^x$ and $4\log2>7/4$ show, on putting
$x=W/Q$, that the right side of Eq.~\eqref{hs:eq:filterbound}
is at most $65\delta/512<\delta$.
This choice has $k\ge2$, $d\ge2W$, and
\begin{equation}\label{hs:eq:degree}
 d=\mathcal O\!\left(W+Q\log(2/\delta)\right).
\end{equation}
No logarithm of $Q$ or $W$ has entered the degree. We now turn
this polynomial bound into the claimed circuit, keeping track of
the coefficient normalisation and the error on every output register.

\paragraph{Synthesis through finite reuse and the bound on all output registers.}
Write the blocks of the complete transducer as
$S_{\rm tot}=\left(\begin{smallmatrix}A&B\\C&D\end{smallmatrix}\right)$.
The stationary relation gives $C=(I-D)\Gamma_{\rm tot}$,
and the inner-product argument preceding
Eq.~\eqref{ov:eq:reuse-error} gives
$B=U\Gamma_{\rm tot}^\dagger(I-D)$.
For completeness, index the $n$ reuses by $j=0,\ldots,n-1$.
Before reuse $j$, the common private vector is
$v_j=(I-D^j)\Gamma_{\rm tot}\psi/\sqrt n$.
Acting on that private vector and the $j$th public component gives
\begin{equation}\label{proof:eq:reuse-map}
 \left(\frac\psi{\sqrt n};v_j\right)
 \longmapsto
 \left(\frac{(U-BD^j\Gamma_{\rm tot})\psi}{\sqrt n};
       \frac{(I-D^{j+1})\Gamma_{\rm tot}\psi}{\sqrt n}\right).
\end{equation}
All other public labels are untouched. Thus preparing and unpreparing
the uniform label gives the selected block
$P_n=U-Bg_n(D)\Gamma_{\rm tot}$ with $n$ complete transducer calls.
The sector embedding and unpreparation in $\mathcal R_n$ retain
all other outputs as leakage rather than discarding them.

We next verify the signed combination used to improve this average.
Write $f(z)=\sum_{j=0}^d a_jz^j$, and set out-of-range $a_j$ to
zero. Convolution of the uniform coefficients of $g_b$ gives
$a_j\ge0$, $\sum_j a_j=1$ and $a_j=a_{d-j}$.
If $q_j$ is the coefficient of $z^j$ in $f(z)g_d(z)$, then
\begin{equation}
 q_j=\frac1d\sum_{l=0}^{d-1}a_{j-l},\qquad
 q_{n-1}-q_n=\frac{a_{n-d}-a_n}{d}.
\end{equation}
An average $g_n$ has coefficient $1/n$ at every power below $n$.
We can therefore reproduce these adjacent coefficient differences
by choosing
\begin{equation}\label{hs:eq:smooth}
 \lambda_n:=\frac nd(a_{n-d}-a_n),\qquad
 f(z)g_d(z)=\sum_{n=1}^{2d}\lambda_n g_n(z).
\end{equation}
Indeed the coefficient at $z^j$ on the right is
$\sum_{n>j}\lambda_n/n=
\sum_{n>j}(q_{n-1}-q_n)=q_j$.
Evaluation at $z=1$ gives $\sum_n\lambda_n=1$.
For $n<d$ the weights are nonpositive, for $n>d$ they are
nonnegative, and the middle weight vanishes by reversal symmetry.
The negative mass is
$d^{-1}\sum_{n=1}^{d-1}na_n=(1-2a_0)/2$.
It follows that
\begin{equation}\label{proof:eq:weight-norm}
 \sum_{n=1}^{2d}|\lambda_n|=2-2a_0<2.
\end{equation}
Thus the filter can be expressed using available reuse circuits
with uniformly bounded LCU normalisation.

Prepare the absolute coefficient amplitudes, apply the selected
reuse circuit and its sign, and unprepare. Padding the normalisation
to two gives a unitary $\mathcal V$ whose selected block is
$\widetilde P/2$, where $\widetilde P:=\sum_n\lambda_nP_n$.
One may implement the padding by known cancelling branches, or by
a branch with zero selected block. All shorter reuse branches idle
inside one schedule of $2d$ controlled calls. Substituting
Eq.~\eqref{hs:eq:smooth} into the reuse error gives
\begin{equation}\label{proof:eq:block-error}
 U-\widetilde P
 =\frac Ud\Gamma_{\rm tot}^\dagger f(D)(I-D^d)\Gamma_{\rm tot},
 \qquad \norm{U-\widetilde P}\le\frac{2W}{d}\norm{f(D)}\le\delta.
\end{equation}

It remains to promote this selected-block approximation to an
approximation on all output registers. Let $\mathsf E$ be the system embedding with
all work registers zero and let $\Pi=\mathsf E\mathsf E^\dagger$.
One oblivious amplification round, with the known phase convention
chosen appropriately, is
$\mathcal A:=-\mathcal V(I-2\Pi)\mathcal V^\dagger
(I-2\Pi)\mathcal V$.
It uses three LCU calls and has selected block
\begin{equation}\label{proof:eq:amplified-block}
 F_{\rm amp}:=\mathsf E^\dagger\mathcal A\mathsf E
 =\frac32\widetilde P-
   \frac12\widetilde P\widetilde P^\dagger\widetilde P.
\end{equation}
This formula follows by expanding the two reflections and using
$\mathsf E^\dagger\mathcal V\mathsf E=\widetilde P/2$.
When $\widetilde P=U$, it gives exactly $U$.

For the approximate case, set $E:=U^\dagger\widetilde P-I$,
so $\norm E\le\delta$. Unitarity of $\mathcal A$ and $U$
gives an identity for the full isometry difference, not only its
selected part:
\begin{align}\label{hs:eq:amplified-full-output}
 (\mathcal A\mathsf E-\mathsf E U)^\dagger
          (\mathcal A\mathsf E-\mathsf E U)
 &=2I-U^\dagger F_{\rm amp}-F_{\rm amp}^\dagger U\notag\\
 &=\Reop(E^2)+EE^\dagger+E^\dagger E+
                      \Reop(EE^\dagger E).
\end{align}
Here $\Reop X=(X+X^\dagger)/2$; the second line follows by
substituting $\widetilde P=U(I+E)$ into
Eq.~\eqref{proof:eq:amplified-block}. Its norm is at most
$3\delta^2+\delta^3\le4\delta^2$, so
$\norm{\mathcal A\mathsf E-\mathsf E U}\le2\delta$.
This includes the amplitude in all nonzero work-register outputs
and does not assume that the unscaled sum $\widetilde P$ is itself
a contraction.

Finally choose $\delta=\eps/8$ in Eq.~\eqref{hs:eq:kchoice}.
Amplification contributes at most $\eps/4$ and the finite mesh
contributes at most $\eps/4$. We may approximate the remaining known
operations of this finite circuit to total error at most $\eps/2$
without additional oracle queries. The number of complete transducer
calls is at most
\begin{equation}\label{hs:eq:explicitcount}
 6(32Q-1)
 \left\lceil4W/Q+\log_2(4096/\eps)\right\rceil,
\end{equation}
which has the claimed order. Each complete call contains one coherent
private operation and known boundary rotations, independently of the
chosen mesh. Tensoring the final isometry difference with an identity
preserves its norm, so the same guarantee holds with an untouched
reference system. This proves the lemma.
\end{proof}

\section{Gate implementation of the improved-precision constructions}
\label{gate:sec:proof}

The optimal-precision query proofs specify finite circuits, but two of
their known operations can be large: the chronological product may use
polynomially many physical time slots, and the signed combination may
have a coefficient table as long as the reuse horizon. We now show that
neither operation multiplies the query bound by another large factor.
We compress the boundary network into logarithmically many levels, and
prepare the coefficient table once per linear-combination circuit.
After establishing this common implementation, we apply it to the
signed-factor and reduced-entry constructions and account for their
native sparse queries, arithmetic, finite gates and work registers.
The moment and contour estimates are those proved in
Appendices~\ref{opt:sec:proof} and~\ref{twonorm:sec:proof}; here our
remaining task is to implement the operations to which they apply.

\subsection{A compiler for chronological products and signed reuse}
\label{gate:sec:compiler}

We begin with the circuit shared by both precision refinements. Its
boundary isometry specifies where a public system state enters one
private slot; its private unitary processes that slot's entry data.

Let $E\psi=U_E(\ket{0^{a_E}}\otimes\psi)$ be the known boundary isometry,
where $U_E$ has $C_E$ elementary gates. Native preparations may instead
be conjugated into the private operation, leaving $E$ a clean embedding.
We allow $a_E$ to include clean predicate scratch space sufficient for
validity tests on the encoded index and system registers. Extra such
qubits are initialised to zero, and we tensor the private unitary with
the identity on those added work registers. This preserves every
transfer and compressed-moment identity.
For this appendix, $J=2^m$ denotes the number of physical slots, not
the embedding used in Section~\ref{sec:methods}. Let $G_j$ rotate
the public system against the image of $E$ in
private slot $j$ by
\begin{equation}
 R(c)=\begin{pmatrix}c&-\sqrt{1-c^2}\\\sqrt{1-c^2}&c\end{pmatrix},
 \qquad 0\le c\le1,
 \qquad \mathcal G_J=G_{J-1}\cdots G_0.
\end{equation}
It fixes the orthogonal private subspace. If $P$ is the private unitary,
the chronological product equals
\begin{equation}\label{gate:eq:global}
 S=\mathcal G_J\bigl(I_{\rm pub}\oplus(I_J\otimes P)\bigr).
\end{equation}
As in Eq.~\eqref{ov:eq:chronological-circuit}, private operations in one
slot commute with boundary rotations in every other slot. Moving only
these disjoint operations gives Eq.~\eqref{gate:eq:global}, with the
boundary rotations in their original order. The direct sum over slots
uses one coherent call to $P$; we must still implement $\mathcal G_J$
without paying once for each of its $J$ displayed factors.

\begin{lemma}[Boundary compression for the precision constructions]
\label{hs:lem:boundary}
The boundary network, its inverse and its externally controlled version
use $\cO(C_E+a_E+m)$ elementary gates and $\cO(a_E+m)$ reusable
work qubits. No involution condition on $P$ is required.
\end{lemma}
\begin{proof}
We use the ordered-rotation factorisation of Ref.~\cite{CGJLWZ},
including its distinct input and output changes of basis. Let
$\beta:=\sqrt{1-c^2}$ and consider a consecutive group of $n$
private modes. Its public-to-public amplitude is $c^n$.
Its return row and injection column are respectively
\begin{equation}
 -\beta(c^{n-1},\ldots,c,1),\qquad
 \beta(1,c,\ldots,c^{n-1})^{\mathsf T}.
\end{equation}
Both have norm $\sqrt{1-c^{2n}}$. Thus only one normalised input
direction can return to the public mode, and a public input enters
only one normalised output direction. We call the coordinate assigned
to each of these directions the group's root. It is a basis coordinate
within the group, not an extra time slot.

Suppose two consecutive groups of $n$ slots have already been put
in their separate input and output coordinates. Set $z:=c^n$ and
$v:=\sqrt{1-z^2}$. Their two roots have return row $-v(z,1)$ and
injection column $v(1,z)^{\mathsf T}$. We therefore choose the next
input basis to start with $(z,1)^{\mathsf T}/\sqrt{1+z^2}$ and the
output basis with $(1,z)^{\mathsf T}/\sqrt{1+z^2}$. The rotations
\begin{equation}
 B_{\rm in}(z):=\frac1{\sqrt{1+z^2}}
       \begin{pmatrix}z&-1\\1&z\end{pmatrix},\qquad
 B_{\rm out}(z):=\frac1{\sqrt{1+z^2}}
       \begin{pmatrix}1&-z\\z&1\end{pmatrix}
\end{equation}
complete these directions to orthonormal bases.
Only the public mode and the two roots participate in the induction
step. In that three-dimensional space the two rotations multiply to
\begin{equation}
 \begin{pmatrix}z^2&-zv&-v\\v&z&0\\zv&-v^2&z\end{pmatrix}.
\end{equation}
Changing its input and output bases as designed gives
\begin{equation}
 (1\oplus B_{\rm out}(z)^\dagger)
 \begin{pmatrix}z^2&-zv&-v\\v&z&0\\zv&-v^2&z\end{pmatrix}
 (1\oplus B_{\rm in}(z))
 =\begin{pmatrix}
 z^2&-\sqrt{1-z^4}&0\\\sqrt{1-z^4}&z^2&0\\0&0&1
 \end{pmatrix}.
\end{equation}
One private direction is now decoupled and the remaining rotation
has attenuation $z^2=c^{2n}$. This is the decisive reduction; using
the same basis on both sides would not retain the chronological order.

Starting with $U_{1,*}=I$, define the recursive basis changes by
\begin{equation}
 U_{2n,*}:=(U_{n,*}\oplus U_{n,*})B_{*,(0,n)}(c^n),
 \qquad *\in\{{\rm in},{\rm out}\},
\end{equation}
where $B_{*,(0,n)}$ acts on the first coordinate of each half.
The induction leaves one rotation between the public coordinate,
denoted by $p$, and the private root coordinate $0$:
\begin{equation}\label{gate:eq:dyadic}
 \mathcal G_J=(1\oplus U_{J,{\rm out}})
             R_{(p,0)}(c^J)(1\oplus U_{J,{\rm in}})^\dagger.
\end{equation}
We now implement these basis changes on the binary slot register.
Each $U_{J,*}$ uses one rotation $B_*(c^{2^j})$ on bit $j$, in
descending order $j=m-1,\ldots,0$, controlled on all lower bits
being zero. The descending order lets us share the lower-prefix tests:
we first compute them in a chain of clean flags, then erase each flag
after its controlled rotation and before any of its lower bits changes.
Computing and erasing the entire chain costs $\mathcal O(m)$ gates
and work qubits. We add the boundary, private-sector and external-enable
tests and the central public/root rotation at cost $\mathcal O(a_E+m)$.
Finally, we restore the actual boundary coordinates by conjugating the
whole network with private-controlled $U_E$ and $U_E^\dagger$, each
used once. This gives the stated gate count for the network, its inverse
and its external control. The factorisation concerns only the ordered
boundary rotations and places no involution condition on $P$.
\end{proof}

Unused binary words do not require additional oracle promises. Public
states with noncanonical encodings are inactive; boundary gates are
masked by validity and zero-ancilla tests. On invalid index pairs, native
private operations use a known fixed action independently of the returned
value word. The tests cost gates linear in the register widths.

The boundary network now has the required gate cost. We next implement
a signed combination of circuits that repeatedly use this complete
interaction. Write the public/private blocks of $S$ as $A,B,C,D$.
A public input can leave immediately through $A$, or enter through $C$,
spend $k-1$ further interactions in the private space, and return through
$B$. We record these possible delays in the generating function

\begin{equation}\label{gate:eq:transfer}
 F(z):=A+zB(I-zD)^{-1}C
      =A+\sum_{k\ge1}z^kBD^{k-1}C.
\end{equation}
Here $z$ is a formal delay variable (or a complex number with $|z|<1$).
The coefficient of $z^k$ is the amplitude returning after delay $k$.
For a positive integer $j$, define its finite triangular average by
\begin{equation}
 \Ces j(F):=\sum_{k=0}^{j-1}(1-k/j)[z^k]F(z),
\end{equation}
where $[z^k]$ selects the coefficient of $z^k$. Preparing the $j$
public labels uniformly gives the same unknown system input amplitude
$\psi/\sqrt j$ on each label, as in Eq.~\eqref{ov:eq:two-reuses}.
Their successive interactions share one private component, initially
zero. A delay-$k$ return can connect $j-k$ input/output label pairs,
so uniform unpreparation assigns it weight $(j-k)/j$. The selected
public block is therefore exactly $\Ces j(F)$; equivalently, its
unaveraged public matrix is triangular with diagonal $A$ and $k$th
lower diagonal $BD^{k-1}C$.
Denote this complete unitary reuse circuit by $W_j$, retaining all leakage.

\begin{lemma}[Compiled signed reuse]\label{gate:lem:reuse}
For known real weights $w_1,\ldots,w_R$ with
$\sum_j|w_j|\le A_0$, where $A_0>0$ is an absolute constant, an LCU circuit with block
\begin{equation}\label{gate:eq:lcu}
 A_0^{-1}\sum_{j=1}^R w_j\Ces j(F)
\end{equation}
uses at most $R$ controlled private calls and
$\cO(R(C_E+a_E+m+\log(R+1)))$ other elementary gates. Excluding
private-call arithmetic workspace, its auxiliary space is
$\cO(a_E+m+\log(R+1))$. A fixed number of invocations and actual inverse
invocations has the same asymptotic costs.
\end{lemma}
\begin{proof}
We prepare the coefficient label once at the start of an LCU invocation,
so that every later reuse step can share the same length choice. Its
state is
\begin{equation}\label{gate:eq:coef-state}
 \sum_{j=1}^R\sqrt{|w_j|/A_0}\ket j+
 \sqrt{1-A_0^{-1}\sum_j|w_j|}\ket\perp.
\end{equation}
Uniformly controlled rotations prepare this state with $\cO(R)$
rotations and CNOTs~\cite{MottonenStatePrep}. Thus we compile the
coefficient table into gate angles; the gate count includes this
preparation. We can implement its sign diagonal by $R$ equality tests
and phase flips, at cost $\cO(R\log(R+1))$.

The chosen length $j$ must now determine which public labels receive
input amplitude. We use a separate reuse register of width
$r=\ceil{\log_2(R+1)}$ and, conditional on ordinary coefficient label
$j$ and the public sector, prepare
\begin{equation}\label{gate:eq:range}
 \ket{u_j}=j^{-1/2}\sum_{k=0}^{j-1}\ket k.
\end{equation}
To obtain this state exactly, fix $j$ and set
$p=\ceil{\log_2j}$, with $p=0$ at $j=1$. We prepare a uniform
$p$-bit word and a flag whose accepted amplitude is
$\sqrt{2^p}/(2\sqrt j)$. The joint probability of an accepted flag
and a word less than $j$ is then exactly $1/4$, independent of $j$.
One Grover round therefore raises the accepted amplitude to one.
Comparisons, reflections and flag cleanup use
$\cO(\log(j+1))$ gates. By performing this preparation for each
$j$ under its equality control, we obtain all the required range states
coherently using $\cO(R\log(R+1))$ gates. This preparation and its
inverse each occur once per LCU invocation.

We next arrange the shared chronology so that private amplitude advances
from one interaction to the next, while earlier public outputs and later
public inputs retain their labels. At step $k=0,\ldots,R-1$, enabled
on ordinary coefficient labels $j>k$, we apply the private operation
and the compressed boundary network to the component with reuse pointer
$k$, then increment the private pointer modulo $2^r$. More explicitly,
if $S(v\oplus x)=v'\oplus x'$, the active part of this step is
\begin{equation}
 \ket{{\rm pub},k}v+\ket{{\rm priv},k}x
 \longmapsto
 \ket{{\rm pub},k}v'+\ket{{\rm priv},k+1}x'.
\end{equation}
All other public labels remain fixed. Before each active step the
private pointer equals $k$; the increment places it at $k+1$ for
the next step. Since $2^r\ge R+1$, even the final private leakage
cannot wrap to an earlier public label. We have therefore implemented
the chronological reuse whose return amplitudes define
Eq.~\eqref{gate:eq:transfer}, using one private call and
$\cO(C_E+a_E+m+r)$ other gates per step. On the dummy coefficient
branch we instead set a separate rejection flag and skip reuse, so
that branch contributes zero to the successful block.

Finally, we unprepare the uniform public-label state, apply the known
coefficient signs, and unprepare the coefficient label. The success
test selects the public sector and zero on all auxiliary registers,
including the rejection flag. On each ordinary label the first
unpreparation selects $\Ces j(F)$; the coefficient preparation and
unpreparation multiply it by $|w_j|/A_0$, and the sign restores $w_j$.
This yields Eq.~\eqref{gate:eq:lcu}. Adding the two one-time
preparations to the $R$ shared steps gives the stated gate bound.
We reuse the predicate and arithmetic workspace throughout, which
gives the auxiliary-space bound. Every operation is coherent, so the
construction also applies to superpositions of $j$ and to inputs
entangled with an untouched reference.
\end{proof}

\paragraph{Finite gates and actual adjoints.}
Let $G_P(\xi),A_P(\xi)$ be the additional finite-set gates and clean work
qubits for a controlled private operation whose error is at most $\xi$
on clean-work inputs, including all amplitude left in its work registers.
They include decoding, routing and rotation arithmetic, but exclude gates
inside supplied oracles. The primitive-query cost is constant and the
inverse is the adjoint of the actual implemented circuit.

\begin{proposition}[Explicit resource bound]\label{gate:prop:finite}
Assume the finite-gate hypotheses of Section~\ref{gate:sec:result}.
Suppose the ideal full amplified circuit has error at most $\eps/2$
and at most $c_*R$ private calls, with $c_*$ absolute. Define
\begin{equation}\label{gate:eq:budgets}
 \xi=\frac{\eps}{4c_*R},\qquad
 G_0=C\{R(C_E+a_E+m+\log(R+1))+A_P(\xi)+1\},
\end{equation}
where $C$ includes the fixed amplification overhead. Its finite-set
implementation has error at most $\eps$, the same oracle count, and
\begin{align}
 G_{\rm additional}&=\cO\!\left(RG_P(\xi)+G_0
                   \log\frac{G_0+1}{\eps}\right),\label{gate:eq:explicit-gates}\\
 A_{\rm auxiliary}&=\cO\!\left(A_P(\xi)+a_E+m+\log(R+1)\right).
 \label{gate:eq:explicit-space}
\end{align}
\end{proposition}
\begin{proof}
We first include the success reflections, which test the arithmetic
work as well as clock and label registers. Their gate cost accounts
for $A_P(\xi)$ in $G_0$. We allocate $\eps/4$ of the remaining
error to private calls: telescoping the complete amplified circuit
gives $c_*R\xi=\eps/4$. We implement reversible Boolean logic exactly
and decompose the remaining known gates with at most three controls
using the stated gate-set assumption. After absorbing the bounded
decomposition cost in $C$, at most $G_0$ uncontrolled one-qubit
gates require synthesis. Accuracy $\eps/(4G_0)$ per gate therefore
uses at most the final $\eps/4$ of error. The logarithmic synthesis
length gives \eqref{gate:eq:explicit-gates}; the decomposition and
synthesis ancillas can be reused.

The calculation does not assume imperfect work registers are exactly
zero. To see why the same bound covers actual inverse calls, let
$\mathsf E$ append the clean arithmetic-work register to the logical
input; this embedding is distinct from the boundary isometry $E$.
If the ideal $U$ preserves clean work and
$\norm{\widetilde U\mathsf E-\mathsf E U}\le\xi$, then
\begin{equation}
 \widetilde U^\dagger\mathsf E-\mathsf E U^\dagger
 =\widetilde U^\dagger(\mathsf E U-\widetilde U\mathsf E)U^\dagger,
\end{equation}
so its actual adjoint has the same clean-input error. In product
telescoping, each error is evaluated on an ideal intermediate state with
clean work and subsequently multiplied by unitaries. All work can be
reused, proving \eqref{gate:eq:explicit-space}. Tensoring with a reference
does not change any bound.
\end{proof}
The squared successful-block budgets below are distinct from physical
gate budgets: once the complete ideal circuit is fixed, private-call
accuracy $\Theta(\eps/R)$ suffices.

\subsection{Implementing the signed-factor construction}\label{gate:sec:signed}
We first apply the compiler to the signed-factor construction of
Appendix~\ref{opt:sec:factor}. That construction already supplies bounded
signed weights and a factorial remainder estimate. To retain its query
bound with finite gates, we must choose a binary physical mesh, approximate
the coefficient table and include all leakage after amplification. We
recall the analytic inputs as they become needed; their proofs remain in
Appendix~\ref{opt:sec:proof}. Suppose
\begin{equation}\label{gate:eq:factor}
 B=\Pi U_BE,\quad C=B^\dagger ZB,\quad D=B^\dagger B\preceq\delta I,
 \quad H=\lambda C,\quad \lambda>0,\quad 0<\delta\le1,
\end{equation}
where $E$ is a clean embedding, $\Pi$ a known success projection and
$Z$ a known sign reflection on success. Put
$T_f=\lambda|t|$, $b_f=T_f\delta$, and
$r_f=L/\log(e+L/b_f)$ for $b_f>0$, where $L:=\log(1/\eps)$;
set $r_f:=0$ when $b_f=0$. If $b_f\le\eps/2$, use the identity.
Otherwise choose
\begin{equation}\label{gate:eq:factor-order}
 M=\cO(b_f+r_f),\qquad q=\sqrt{T_f/M},\qquad a=qM
\end{equation}
when $T_f\ge M$. If $T_f<M$, attenuate $B$ by $\sqrt{T_f/M}$ and
increase $\lambda$ so that $T_f'=M$; this preserves $H$ and $b_f$.
Use $q=1,a=M$ in that case. The attenuation is a known flag rotation,
with no extra factor queries.

Let $K=qZ$ on success and zero on failure. The private operation is
\begin{equation}\label{gate:eq:factor-private}
 P=U_B^\dagger(K+\ii I)(K-\ii I)^{-1}U_B.
\end{equation}
Only the three known middle phases need an external control. On an
inactive branch the unconditional factor and inverse cancel. The
adjoint reverses the middle phase; no controlled opaque $U_B$ is needed.
Define
\begin{equation}
 x(z)=\frac{1-z}{1+z},\quad R_q(z)=\frac{1-x(z)^2}{1+q^2x(z)^2},\quad
 \Omega(z)=x(z)I+R_q(z)(\ii qC+q^2x(z)D),\quad F(z)=e^{-a\Omega(z)}.
\end{equation}
The moment construction of Appendix~\ref{opt:sec:weights} supplies real weights with
\begin{equation}\label{gate:eq:weights}
 a\le n\le27a,\qquad \sum_{j=1}^n|\lambda_j|\le A_*=3000,
\end{equation}
and the exact identities
\begin{equation}\label{gate:eq:moments}
 \sum_{j=1}^n\lambda_j\Ces j\!\left(e^{-ax(z)}R_q(z)^k x(z)^u\right)
 =\delta_{u0}\quad(0\le u\le k<M).
\end{equation}
The lower bound on $n$ follows from its construction's initial degree
$\ceil{12a}$. Increasing the absolute constant in $M$, Lemma~\ref{opt:lem:tail}
then gives
\begin{equation}\label{gate:eq:tail}
 \norm{\sum_j\lambda_j\Ces j(F)-e^{-\ii tH}}\le\eps^2/144.
\end{equation}
Using $\log(c/\eps^2)$ to choose this order changes $b_f+r_f$ by only
an absolute factor. We describe positive time here; negative time uses
the inverse circuit.

\paragraph{The mesh and the coefficient table.}
On $J$ private slots put $v=a/J$, $c=(1-v/2)/(1+v/2)$ and
\begin{equation}
 F_J(z)=\left[(I-v\Omega(z)/2)(I+v\Omega(z)/2)^{-1}\right]^J.
\end{equation}
This is the exact causal transfer of \eqref{gate:eq:global}. The finite-mesh
estimate in Appendix~\ref{opt:sec:proof} is
\begin{equation}\label{gate:eq:factor-mesh-error}
 J\ge8an\quad\Longrightarrow\quad
 \norm{\sum_j\lambda_j\Ces j(F_J)-\sum_j\lambda_j\Ces j(F)}
 \le128A_*a^3n^4/J^2.
\end{equation}
Choose the least power of two $J$ at least
\begin{equation}\label{gate:eq:factor-mesh}
 \max\left\{2,8an,\sqrt{128\cdot144 A_*a^3n^4/\eps^2}\right\}.
\end{equation}
The mesh error is at most $\eps^2/144$ and
\begin{equation}\label{gate:eq:factor-mesh-log}
 J=\cO(1+n^{7/2}/\eps),\qquad
 m=\log_2J=\cO\!\left(\log\frac{n+1}{\eps}\right).
\end{equation}
Lemma~\ref{hs:lem:boundary} therefore removes the polynomial
inverse-precision mesh overhead in gates.

To avoid testing the sign of an exact real zero, compute dyadic weights
$\widehat\lambda_j$ satisfying
\begin{equation}\label{gate:eq:round}
 \sum_j|\widehat\lambda_j-\lambda_j|\le\eps^2/144.
\end{equation}
Absolute precision $\eps^2/(144n)$ per weight suffices, so the rounded
table has $\cO(\log((n+1)/\eps))$ output bits per entry, plus a fixed
integer part. This is not a bound on the internal precision or classical
time needed to evaluate the Laguerre formula. The known finite formulas
permit offline approximation; their exact signs are not required.
Use the signs of the dyadic weights and omit rounded zero weights.

Each $\Ces j(F_J)$ is a unitary compression of norm at most one, so
rounding contributes at most $\eps^2/144$ to the unscaled block error.
Pad the variation to the fixed constant $\overline A=A_*+1$ with the
dummy branch of Lemma~\ref{gate:lem:reuse}. Its block is
$\widehat Q/\overline A$, where
\begin{equation}
 \widehat Q=\sum_j\widehat\lambda_j\Ces j(F_J),\qquad
 \norm{\widehat Q-e^{-\ii tH}}\le\eps^2/48.
\end{equation}
The $n$-entry coefficient preparation costs $\cO(n)$ rotations,
and its coherent range preparations cost $\cO(n\log(n+1))$ gates.
Neither cost is multiplied by the reuse horizon.

\paragraph{Amplification and the complete output state.}
Choose a fixed odd $k_o$ with $s_o=\sin(\pi/(2k_o))\le1/\overline A$.
Attenuation by $\overline A s_o$ gives block $s_o\widehat Q$.
The $(k_o-1)/2$ oblivious amplification rounds use $k_o$ calls to this
LCU circuit or its inverse; $k_o$ and $1/s_o$ are absolute constants.

For completeness, if $\norm{\widehat Q-U}\le e_0\le1/2$, its polar
factorisation $\widehat Q=W|\widehat Q|$ has
$\norm{W-U}\le2e_0$ and singular values in $[1-e_0,1+e_0]$.
After the known phase correction the amplified successful block is
\begin{equation}
 E_{\rm amp}=W f_{k_o}(s_o|\widehat Q|),\qquad
 f_{k_o}(y)=\sin(k_o\arcsin y).
\end{equation}
As $f_{k_o}(s_o)=1$, $|f'_{k_o}|\le2k_o$ on the relevant interval and
$k_os_o\le\pi/2$, we have $\norm{E_{\rm amp}-U}\le6e_0$.
Let $\mathsf E$ now embed the system with all work registers in zero.
Unitarity of the complete amplified circuit gives
\begin{equation}
 \norm{\mathcal A_{\rm ideal}\mathsf E-\mathsf E U}^2\le12e_0.
\end{equation}
With $e_0\le\eps^2/48$ its error on all output registers is at most $\eps/2$.
Proposition~\ref{gate:prop:finite} applies with $R=n$.

\begin{proposition}[Signed-factor gate bound]\label{gate:prop:signed}
The signed-factor simulation circuit uses
\begin{equation}
 n=\cO(\sqrt{T_fb_f}+\sqrt{T_fr_f}+r_f)
\end{equation}
factor queries and obeys \eqref{gate:eq:explicit-gates}--\eqref{gate:eq:explicit-space}
with $R=n$ and $m=\cO(\log((n+1)/\eps))$. The private-operation cost
includes the factor, its adjoint and the three-phase operation in
\eqref{gate:eq:factor-private}.
\end{proposition}
\begin{proof}
For $T_f\ge M$, $n\le27\sqrt{T_fM}$ with $M=\cO(b_f+r_f)$.
Otherwise attenuation gives $n=\cO(M)=\cO(b_f+r_f)$, and
$b_f\le T_f$ implies $b_f\le\sqrt{T_fb_f}$. The preceding circuit and
error estimates give the other claims.
\end{proof}

\subsection{Implementing the reduced-entry construction}\label{gate:sec:core}
The reduced-entry construction is needed when the whole-matrix
factor's remainder scale is too large. We retain its analytic filter
from Appendix~\ref{twonorm:sec:proof} and check that both its mesh
and its signed smoothing table have affordable implementations. For
this calculation, separate the two accumulated scales:
\begin{equation}\label{gate:eq:effective}
 \sigma=\frac{s\mu^2}{\eta^2}\in[1,s],\quad
 \Theta=\frac{\eta^2|t|}{\mu},\quad
 A=\sqrt\sigma\Theta=\sqrt{s}\eta|t|,\quad
 B=\sigma\Theta=s\mu|t|.
\end{equation}
Here $\Theta=\theta$ from Eq.~\eqref{twonorm:eq:parameters}; the capital symbol
avoids confusion with rotation angles. Physical lists still have $s$
positions. The real $\sigma$ bounds moments, not list lengths.

We reserve half the target error for the analytic construction by
setting $\eps_a=\eps/2$, and write
$L_a=\log(1/\eps_a)$, $h_a=\log(e+L_a/\Theta)$,
$r_a=L_a/h_a$. The clipped regime is large $h_a$ and $B\ge L_a h_a$.
We choose the two isolated correction phases as in
Appendix~\ref{opt:sec:core}, using an even $p$ and a real $q$ with
\begin{equation}
 q_*\le q\le2q_*,\quad q_*=\sqrt{B/r_a},\quad
 2q=\cot(\pi/p),\quad a=B/q.
\end{equation}
At this scale we reduce the entries and retain the correction exactly:
\begin{equation}\label{gate:eq:clip}
 (H_c)_{ij}=\frac{H_{ij}}{1+h_a|H_{ij}|/\mu},\qquad C_H=H-H_c.
\end{equation}
Both pieces enter one joint representation, so its first compressed
moment retains their exact sum.
Its boundary isometry $J_b$ and private Hermitian operator $K_b$ satisfy
\begin{equation}\label{gate:eq:joint-moments}
 J_b^\dagger K_bJ_b=H/\alpha,\qquad
 J_b^\dagger K_b^2J_b\preceq\frac{2q^2(1+h_a)}\sigma I,
 \qquad \alpha=s\mu/q.
\end{equation}
The private Cayley unitary has a bulk with $|K_b|\le2q/h_a$ and two
isolated phases $e^{\pm2\pi\ii/p}$, each of compressed weight at most
$h_a/(2\sigma)$. Section~\ref{gate:sec:native} gives its native gates.

The contour proof in Appendix~\ref{opt:sec:core} supplies the filter with nonnegative coefficients
\begin{align}
 g_j(z)&=j^{-1}\sum_{u=0}^{j-1}z^u,\quad
 Q_0=\ceil{\max\{1,2q/h_a\}},\quad B_0=32Q_0,\notag\\
 k&=\ceil{4r_a},\quad m_f=\ceil{32h_a k},\quad
 f(z)=g_{B_0}(z)^{m_f}g_p(z)^k,\notag\\
 d&=m_f(B_0-1)+k(p-1)=\cO(qr_a).
 \label{gate:eq:filter}
\end{align}
For a sufficiently large fixed threshold on $h_a$ (the explicit
proof uses $10^4$), its analytic conclusions are
\begin{equation}\label{gate:eq:analytic-core}
 \norm{f(D_g)}\le\eps_a/8,\quad d\ge5a\ge5W,\quad
 \norm{U_C-e^{-\ii tH}}\le\eps_a/4.
\end{equation}
Here $D_g$ is the actual finite global private block, $W$ bounds the
squared norm of the stationary private component, and $U_C$ is the
finite Cayley target. These are
the contour and discretisation conclusions we established in the
query proof.

\paragraph{Mesh and filter preparation.}
The contour estimate permits any mesh finer than its stated lower
bound. We choose a binary mesh so that Lemma~\ref{hs:lem:boundary}
applies, setting $\zeta=h_a\Theta/(100qk)$ and taking $J$ as the
least power of two at least
\begin{equation}\label{gate:eq:core-mesh}
 \max\left\{2,4a,8a/\zeta,24aQ_0,\sqrt{a^3/(3\eps_a)}\right\}.
\end{equation}
This is allowed by the analytic proof. Since $h_a\Theta\ge1$,
$p-1>5q$ and $d\ge k(p-1)$,
\begin{equation}
 \zeta^{-1}\le100qk=\cO(d),\quad Q_0=\cO(q)=\cO(d),\quad a\le d/5.
\end{equation}
Thus
\begin{equation}\label{gate:eq:core-mesh-log}
 J=\cO(1+d^2+d^{3/2}/\sqrt\eps),\qquad
 \log_2J=\cO\!\left(\log\frac{d+1}{\eps}\right).
\end{equation}
The compressed boundary circuit has only logarithmic mesh dependence.

Write $f(z)=\sum_{u=0}^d a_u z^u$, setting out-of-range coefficients
to zero. The $a_u$ are nonnegative, palindromic, and sum to one.
The signed smoothing identity is
\begin{equation}\label{gate:eq:smoothing}
 \omega_j=\frac{j}{d}(a_{j-d}-a_j),\quad 1\le j\le2d,
 \qquad f(z)g_d(z)=\sum_{j=1}^{2d}\omega_jg_j(z).
\end{equation}
Telescoping gives $\sum_j\omega_j=1$. The weights are nonpositive
for $j<d$, zero at $j=d$, and nonnegative for $j>d$. Palindromicity
therefore gives $\sum_j|\omega_j|=2-2a_0<2$.
These rational weights can be computed by finite convolution. The table
has $2d$ entries, independently of $J$, so the state
\begin{equation}
 \sum_{j=1}^{2d}\sqrt{|\omega_j|/2}\ket j+\sqrt{a_0}\ket\perp
\end{equation}
costs $\cO(d)$ rotations to prepare. We can therefore use
Lemma~\ref{gate:lem:reuse} with $R=2d,A_0=2$. Its range preparations
add $\cO(d\log(d+1))$ gates only once per LCU invocation, leaving
the linear dependence on the reuse horizon intact.

\paragraph{Error after amplification.}
The stationary relation and its reuse identity give
\begin{equation}
 S\binom I\Gamma=\binom{U_C}\Gamma,\quad \norm{\Gamma}^2\le W,
 \quad P_j=U_C-U_C\Gamma^\dagger(I-D_g)g_j(D_g)\Gamma.
\end{equation}
Thus $\widetilde P=\sum_j\omega_jP_j$ obeys
\begin{equation}
 U_C-\widetilde P=\frac{U_C}{d}\Gamma^\dagger f(D_g)
             (I-D_g^d)\Gamma,\qquad
 \norm{U_C-\widetilde P}\le\norm{f(D_g)},
\end{equation}
using contractivity of $D_g$ and $2W/d\le1$.
One oblivious amplification round on block $\widetilde P/2$ has
successful block
\begin{equation}
 \frac32\widetilde P-\frac12\widetilde P\widetilde P^\dagger\widetilde P
\end{equation}
and uses at most $6d$ private calls. For
$\norm{\widetilde P-U_C}\le e_f\le1$, expanding this expression and
using unitarity of the full amplified circuit bounds its squared
norm error on all output registers by $3e_f^2+e_f^3\le4e_f^2$. Equation~\eqref{gate:eq:analytic-core}
therefore gives conversion error at most $\eps_a/4$ relative to $U_C$. Together with the
mesh error this is $\eps_a/2\le\eps/2$. The remaining error is
available for private arithmetic and known-gate synthesis. Applying
Proposition~\ref{gate:prop:finite} with $R=2d,c_*=3$ shows that
private-call accuracy $\xi=\eps/(24d)$ suffices.

\begin{proposition}[Clipped-transducer gate bound]\label{gate:prop:core}
The clipped construction uses at most $6d$ private calls and
$\cO(d)$ sparse queries, and satisfies
\eqref{gate:eq:explicit-gates}--\eqref{gate:eq:explicit-space} with $R=2d$,
$m=\cO(\log((d+1)/\eps))$ and $d=\cO(\sqrt{B r_a})$.
\end{proposition}
The isolated phases remain exact in the ideal analysis. Gate approximation
perturbs the entire circuit by telescoping; no new contour argument for
slightly displaced zeros is needed.

\subsection{Native signed factors, controls and reversible arithmetic}\label{gate:sec:native}
We verify constant primitive-query cost and efficient additional gates
for the private operations. In particular, we specify the local factor
on the entire private space, not just on its initialised input.

\paragraph{A stable constant-dimensional signed factor.}
The sparse factor must depend continuously on a small entry on the
entire private space, since later reuse does not restrict it to its
initial input. We obtain that property by completing the success
isometry symmetrically. For $c=u e^{\ii\phi}$, $0\le u\le1$, set
\begin{equation}
 E(c)=\begin{pmatrix}0&c\\\overline c&0\end{pmatrix},\quad
 V_\phi=\diag(e^{\ii\phi},1)H_2,\quad
 \Pi_+=V_\phi\ket0\bra0V_\phi^\dagger,\quad
 \Pi_-=V_\phi\ket1\bra1V_\phi^\dagger,
\end{equation}
where $H_2$ is a Hadamard. Then $E(c)=u(\Pi_+-\Pi_-)$ and
\begin{equation}\label{gate:eq:local-factor}
 B_\pm(c)=\sqrt{(uI\pm E(c))/2}=\sqrt u\,\Pi_\pm.
\end{equation}
For orientation qubit $o$ and sign qubit $z$, let
\begin{equation}
 T_\phi=(V_\phi\otimes I_z)\,\mathrm{CNOT}_{o\to z}
                     (V_\phi^\dagger\otimes I_z),
 \qquad T_\phi=T_\phi^\dagger=T_\phi^{-1}.
\end{equation}
With a success flag $f$, define
\begin{align}
 A_\phi&=\diag_f(I,T_\phi),\qquad
 R_u=\begin{pmatrix}\sqrt{1-u}&-\sqrt u\\\sqrt u&\sqrt{1-u}\end{pmatrix},
 \notag\\
 U(c)&=A_\phi(R_u\otimes I)A_\phi^\dagger
   =\begin{pmatrix}\sqrt{1-u}I&-\sqrt u\,T_\phi\\
                     \sqrt u\,T_\phi&\sqrt{1-u}I\end{pmatrix}.
 \label{gate:eq:dilation}
\end{align}
On initial $f=z=0$ this is exactly
\begin{equation}\label{gate:eq:factor-action}
 \psi\longmapsto \ket0_f\ket0_z\sqrt{1-u}\,\psi
       +\ket1_f\ket0_zB_+(c)\psi+\ket1_f\ket1_zB_-(c)\psi.
\end{equation}
Only a constant number of local rotations and bounded-arity gates are
used, conditioned on computed entry data. This is an explicit three-qubit
unitary, not an assumed efficient completion of a large isometry.

This symmetric completion makes the zero-entry operation $U(0)=I$,
independently of its undefined phase. On every input,
\begin{equation}\label{gate:eq:small}
 \norm{U(c)-I}\le1-\sqrt{1-u}+\sqrt u\le2\sqrt u.
\end{equation}
Thus $u\le\xi^2/16$ may be replaced by zero with error at most
$\xi/2$. Above that threshold, magnitude accuracy of order $\xi^2$
and phase accuracy of order $\xi$ suffice after constant budget
adjustments. Resolving the phase then needs only
$\cO(\log(1/\xi))$ additional precision bits. Approximate threshold
comparisons can be made with an overlap region of constant width on
this scale; both choices there satisfy the same bound. No inverse
minimum-nonzero-entry dependence is introduced.

The two-sided completion in \eqref{gate:eq:dilation} is material. A one-sided
$A_\phi R_u$ has the same desired initial-input action, but can retain
order-one phase dependence on other private inputs as $u\to0$. Such a
completion would not justify the full-unitary approximation used here.

\paragraph{Sparse preparation and the joint private operation.}
On the edge space let $L_0,R_0$ be the unnormalised row and column list
maps. Uniform list and orientation preparation followed by the appropriate
locator implements
\begin{equation}
 J_0=\frac1{\sqrt{2s}}\binom{L_0}{R_0}.
\end{equation}
For Hermitian $X$ on the original support with $|X_{ij}|\le\mu$,
we apply Eq.~\eqref{gate:eq:dilation} at $c=X_{ij}/\mu$. We compute
the value word while retaining the entry indices, perform the local
operation on its three separate qubits, and uncompute the value.
The retained indices and unchanged value word make that erasure exact. The successful block $B_X$
and known sign $Z$ obey
\begin{equation}\label{gate:eq:native}
 B_X^\dagger ZB_X=\frac{X}{s\mu},\qquad
 B_X^\dagger B_X=\frac{\diag_i\sum_j|X_{ij}|}{s\mu}.
\end{equation}
These identities follow by compression through the two orientations of
$J_0$. Row sums describe the resulting operator and are not computed.
For $X=H$, Cauchy--Schwarz gives
\begin{equation}\label{gate:eq:whole-factor}
 H=s\mu B_H^\dagger ZB_H,\qquad
 B_H^\dagger B_H\preceq\sigma^{-1/2}I.
\end{equation}
A factor preparation and its inverse have constant sparse-query cost.
Non-power-of-two uniform lists use the fixed-range preparation of
Lemma~\ref{gate:lem:reuse}, with $\cO(\log(s+1))$ gates.

For clipping, write $z=H_{ij}/\mu$ and $x=|z|$. The clipped and correction
entries, normalised by $\mu$, are $z/(1+h_a x)$ and
$z h_a x/(1+h_a x)$. They are computed from the original value
word. The correction magnitude $h_a x^2/(1+h_a x)\le1$ permits
\eqref{gate:eq:dilation}. With $y=2qz/(1+h_a x)$, the bulk private
operation is the two-dimensional unitary
\begin{equation}\label{gate:eq:cayley}
 D(y)=\frac{(|y|^2-1)I+
       2\ii\begin{pmatrix}0&y\\\overline y&0\end{pmatrix}}{1+|y|^2}.
\end{equation}
The correction success phases are $(\pm2q+\ii)/(\pm2q-\ii)$, and
failure has phase $-1$.

Let $U_0$ prepare $J_0$ and $U_c$ prepare the full correction isometry,
including failure. Pad their registers to a common size. With
$\mathcal F=\diag(U_0,U_c)$, the joint boundary is
$J_b=\mathcal F(\ket+\otimes E)$. Before conjugation, its private
operator is the direct sum of \eqref{gate:eq:cayley} and the three correction
phases, denoted $D_{\rm joint}$. Conjugating private coordinates gives
\begin{equation}
 P_{\rm joint}=\mathcal F^\dagger D_{\rm joint}\mathcal F,
 \qquad \text{known boundary }\ket+\otimes E.
\end{equation}
Each private call has constant sparse-query cost, including the
preparations and their inverses. Its boundary preparation is now only a
Hadamard and clean embedding. Alternatively, uniform preparations may
be left in $U_E$ and charged to $C_E$; the asymptotic bound is unchanged.

The local arithmetic uses addition, multiplication, division by positive
$1+h_a x$, square roots, phases and rotation parameters. Its precision
overhead is polynomial in $\log(q+1)$, $\log(h_a+1)$ and
$\log(1/\xi)$. The rational Cayley formula has denominator at least one;
near zero its phase can also be avoided by using $D(0)=-I$ and the
continuity bound on the matrix itself. Standard efficiently decoded
binary or fixed-point data satisfy the scalar-arithmetic assumption.
For an arbitrary costly decoder, use the explicit $G_P,A_P$ formulas
rather than the polynomial corollary.

\paragraph{Controls on native locators.}
For completeness we count gates in the phase-one locator control
construction. Let $U\ket{x_*}=\ket{y_*}$ for known valid $x_*$, with
width-$w$ encoded words. We first apply $U$ to a reference in $x_*$, copy its deterministic
basis output into a clean register, and apply $U^\dagger$ to the
reference. Two queries leave a coherent copy of the deterministic label
$y_*$ and restore the reference. Define the XOR translation
\begin{equation}
 W_y\ket x=\ket{x\mathbin{\oplus}x_*\mathbin{\oplus}y},\qquad
 W_y\ket{x_*}=\ket y.
\end{equation}
This uses $\cO(w)$ gates even when controlled. The unitary
$V=W_{y_*}^\dagger U$ fixes $x_*$ exactly. This fixed input lets us
disable $V$ without controlling the oracle: on the inactive branch
we swap the target with a dummy in $x_*$, apply $V$ unconditionally,
and undo the swap. The dummy is restored on both branches. Applying
controlled $W_{y_*}$ then restores $U$ on the active branch, and we
reverse the two-query reference preparation to clear the copied label. This gives
controlled $U$ with five primitive queries and $\cO(w)$ gates and work
bits. The reversed circuit controls $U^\dagger$. Only the prescribed
action at $x_*$ is used; the rest of the oracle extension is unrestricted.

Value controls use compute, controlled copy and uncompute. For conjugated
private calls control only the middle local operation, so all computations
cancel on inactive branches. No arbitrary large permutation is counted
as an elementary gate. Invalid index words are masked as after
Lemma~\ref{hs:lem:boundary}; zero padding has $c=0$ in the factor and $D(0)=-I$
in the Cayley operation.

\subsection{Completing the gate bound}\label{gate:sec:combine}
\begin{proof}[Proof of Theorem~\ref{gate:thm:main}]
We combine the three available circuits using the scales in
Eq.~\eqref{gate:eq:effective} and analytic accuracy $\eps_a=\eps/2$.
Write $L=\log(1/\eps)$ and
$r_\mu=L/\log(e+L/\Theta)$ for this comparison. Replacing $\eps$
by $\eps/2$ changes both by only absolute factors. For fixed $c\ge1$, the function $L/\log(e+L/\Theta)$ at $cL$
is between its value at $L$ and $c$ times that value.

If $h_a=\log(e+L_a/\Theta)$ is bounded by the fixed threshold, use the
baseline gate-efficient Cayley construction. With $E_H$ the orientation
entry operator normalised by $\mu$, it uses
\begin{equation}
 K_0=\sqrt\sigma E_H,\quad
 J_0^\dagger K_0J_0=H/(\sqrt{s}\eta),\quad
 J_0^\dagger K_0^2J_0\preceq I,\quad \norm{K_0}\le\sqrt\sigma.
\end{equation}
Its rescaling optimisation gives $\cO(A+\sqrt{BL_a}+L_a)$ calls.
Since $L_a=h_a r_a$ with bounded $h_a$, this is
$\cO(A+\sqrt{B r_a}+r_a)$. The baseline gate proof uses exactly the
boundary, native Cayley and reuse circuits of Section~\ref{gate:sec:compiler}.
The actual list length remains $s$.

For large $h_a$ and $B\le L_a h_a$, use \eqref{gate:eq:whole-factor} with
$T_f=B,b_f=A$. Set $x=L_a/\Theta$, so $\sigma\le xh_a$ and
\begin{equation}
 \frac{L_a}{A}=\frac{x}{\sqrt\sigma}\ge\sqrt{x/h_a},\qquad
 \frac{A}{r_a}=\frac{h_a\sqrt\sigma}{x}\le\frac{h_a^{3/2}}{\sqrt x}.
\end{equation}
Uniformly in this regime, $A=\cO(r_a)$ and
$\log(e+L_a/A)=\Theta(h_a)$. The signed-factor cost is consequently
$\cO(\sqrt{B r_a}+r_a)$, and Proposition~\ref{gate:prop:signed} supplies its
gate implementation. At large $h_a$ and $B\ge L_a h_a$,
Proposition~\ref{gate:prop:core} instead gives $\cO(\sqrt{B r_a})$ calls.

The three cases have reuse horizon $R=\cO(\mathcal Q_\mu)$ and
\begin{equation}
 m=\cO\!\left(\log\frac{\mathcal Q_\mu+1}{\eps}\right).
\end{equation}
For the bounded-$h_a$ case the baseline mesh is also polynomial in the
horizon and $1/\eps$, giving the same width bound.
Section~\ref{gate:sec:native} supplies constant sparse queries per private
call and $C_E+a_E$ linear in the needed register widths.

All other scale bit lengths fit the arguments in \eqref{gate:eq:main-gates}.
For example $q,a,M=\cO(R)$ in the signed branch and
$q,k,a=\cO(d)$ in the clipped branch. Also $\Theta\ge1$ implies
$r_a\ge L_a/\log(e+L_a)$, so
$\log(L_a+1)=\cO(\log(r_a+1))$; the clipping parameter does not create
a larger bit-length scale. Thus the assumed native $G_P(\xi),A_P(\xi)$
are polynomial in $b_{\rm data},\log(N+1),\log(s+1)$ and
$\log((\mathcal Q_\mu+1)/\eps)$. Inserting these estimates into
Proposition~\ref{gate:prop:finite} proves the gate and space bounds and the
error on all output registers. Negative times use the actual inverse circuit.
Reference stability follows from Eq.~\eqref{hs:eq:fullerror}.
\end{proof}

\section{Scope of the Lean formalisation}\label{app:formalisation}

The accompanying \texttt{sparse-quantum-lean} package, version
\texttt{0.1.0-rc30}, is available on GitHub~\cite{BerrySparseQuantumLean2026}
and pins Lean and Mathlib~4.32.1. A complete build of the proof sources
and axiom audits has passed with exit code zero, together with the
release checks. Its 14,695 guarded axiom
reports allow only \texttt{propext}, \texttt{Classical.choice} and
\texttt{Quot.sound}, or specified subsets. Compilation verifies the formal
statements under their explicit hypotheses. We give the detailed theorem
correspondence, assumptions and differences in schedules and constants in
\texttt{docs/FORMALISATION.md} in the repository; this appendix summarises
the coverage.

Coverage includes the Hamiltonian query bounds of
Theorems~\ref{hs:thm:main}, \ref{hs:thm:time},
\ref{hs:thm:optimal} and~\ref{hs:thm:two-norm}, their unitary
implementation consequences, and the linear-system solver of
Theorem~\ref{qls:thm:main}, with circuits, worst-case query counts and
output guarantees. The entrywise Cayley and centred averaging results
are included. Lower bounds allow intermediate measurements retaining
quantum memory on the corresponding hard families; the two-norm match
retains the restrictions of Proposition~\ref{hsl:prop:two-norm}.

For Corollary~\ref{qls:cor:separate}, we formalise the sparse reduction,
transfer of the supplied factor-two norm estimate, decoding, success
amplification and separate query bounds on every history, including
clean-target value writers. This proof assumes a stated interface for
Gao, Ji and Liu's block-input solver~\cite{GJL26}: a single heralded circuit
with a fixed output selector and constant success probability. The selector
is a known isometry whose adjoint extracts the unnormalised successful
system vector; it and the circuit are chosen independently of the unknown
oracles. This single-circuit form is intended to follow from their
construction but is not proved in the package. The formal repetition count differs by a constant from the
manuscript, retaining the stated asymptotic bounds. The cost of obtaining
the norm estimate is excluded.

Dissipative ODE query solvers certify success at least $2/3$,
conditional trace-distance accuracy and worst-history resource counts.
Exact-source results also bound actual successful trial directions;
approximate-source results bound sufficiently probable trials and the
conditional weight of the other histories. Time-dependent and Dyson
results allow complex weights, off-mesh nodes and nonaccretive
represented samples. They retain a true accretive generator, a
propagator solving the homogeneous equation classically on each closed
mesh piece, and sampling and source-error budgets. Opposite-branch
certificates bound the sampling errors. These constructions admit exact
clean-column loaders; the time-dependent query theorem also permits approximate
loading on the coherent clean subspace. Dyson factorial-tail and
fixed-parameter $\mathcal O(L/\log L)$ query bounds, with
$L=\log(1/\eps)$, are included.

The improved Hamiltonian gate theorem~\ref{gate:thm:main} includes
finite-set gates, workspace and full-output error for general dimensions,
sparsities and public decoders, under reversible-arithmetic and synthesis
assumptions. Tiny-entry variants avoid their phase computation by
global thresholding, with additional error $|t|s\delta_{\rm cut}$ for threshold
$\delta_{\rm cut}$. Baseline sparse gate accounting constructs the private circuit
and proves its accuracy. The explicit Hadamard construction uses
$s'=2^{\lceil\log_2s\rceil}<2s$; original-$s$ schedules instead take
an exact boundary-preparation circuit. Private-space padding includes
controlled operator identities relating padded calls to flagged original
calls and a known reflection, and the zero-system case; no cheap
elementary implementation of an arbitrary padding is claimed.

The time-dependent ODE finite-gate theorem combines sampling certificates
with nonaccretive represented samples. The actual compiled preparation
and its literal adjoint give success, conditional-output and
successful-trial guarantees, with gates, workspace and matrix calls
bounded. Source executions, adjoints and internal matrix calls are
counted on every history, with source gates charged linearly. This
theorem uses exact binary data on power-of-two registers, nonnegative
source weights at uniform-mesh slot times, positive time and norm
bound, and an adjoint-closed gate set containing the exact zero tests.
General sizes, off-mesh nodes and approximate decoding are not all
combined in this finite-set result.

Clean-target value-writer variants now cover the main query results
and general-size gate results. Source preparation and synthesis retain
their explicit assumptions. Polynomial scalar-arithmetic cost and
efficient classical compilation of signed coefficients are not proved.
No Dyson or linear-system elementary-gate bound is claimed. The general
non-dissipative ODE result assumes an external block-input solver;
the external-algorithm hybrid Corollary~\ref{hs:cor:hybrid} is not
separately formalised.

\end{document}